\documentclass[10pt, oneside]{article} 

\usepackage{geometry}
\usepackage{xparse}

\usepackage[utf8]{inputenc}

\usepackage{slashed}

\usepackage{dot2texi}
\usepackage{pdflscape}

\newcommand\numberthis{\addtocounter{equation}{1}\tag{\theequation}}

\usepackage[most,breakable]{tcolorbox}
\usepackage{changepage}
\usepackage{microtype}
\newtcolorbox{myframe}[1][]{
  enhanced,
  arc=0pt,
  outer arc=0pt,
  colback=white,
  boxrule=0.8pt,
  #1
}

\newtcolorbox{algbox}{breakable,colback=gray,colframe=gray,standard jigsaw,opacityback=0.15,opacityframe=0.5,boxrule=0.75pt,before upper={\parindent15pt}}
\allowdisplaybreaks

\newcommand\precdot{\mathrel{\ooalign{$\prec$\cr
  \hidewidth\hbox{$\cdot\mkern0.5mu$}}}}
\newcommand\succdot{\mathrel{\ooalign{$\succ$\cr
  $\cdot$}}}

\newcommand{\bs}[1]{\boldsymbol{#1}}
\usepackage{amsmath, amsthm, amssymb, wasysym, verbatim, bbm, color, graphics, tikz, braket, tikz-cd}
\usepackage{thmtools,thm-restate}
\usepackage{algorithm2e,setspace}

\usepackage{dsfont}

\usepackage{multirow}
\usepackage{tablefootnote}

\newcommand{\mcA}{\mathcal{A}}
\newcommand{\mcB}{\mathcal{B}}

\newcommand{\mcQ}{\mathcal{Q}}

\newcommand{\mcD}{\mathcal{D}}
\newcommand{\mcF}{\mathcal{F}}

\newcommand{\mcH}{\mathcal{H}}

\newcommand{\mcN}{\mathcal{N}}

\DeclareMathOperator{\id}{\mathds{1}}

\newcommand{\ftwo}{\mathbb{F}_2}

\DeclareMathOperator{\poly}{poly}

\DeclareMathOperator{\polylog}{polylog}

\DeclareMathOperator{\supp}{supp}

\newcounter{relctr} %

\newcommand\labelrel[2]{%
  \begingroup
    \refstepcounter{relctr}%
    \stackrel{\textnormal{(\alph{relctr})}}{\mathstrut{#1}}%
    \originallabel{#2}%
  \endgroup
}
\AtBeginDocument{\let\originallabel\label}

\usepackage[sorting=nyt, backend = bibtex, style=alphabetic, backref, maxbibnames=20]{biblatex}
\usepackage[colorlinks=true,linkcolor=blue,allcolors=blue]{hyperref}
\usepackage[capitalize,nameinlink]{cleveref}

\graphicspath{{figures/}}

\newtheorem{theorem}{Theorem}[section]
\newtheorem{definition}[theorem]{Definition}

\newtheorem{remark}[theorem]{Remark}
\newtheorem{lemma}[theorem]{Lemma}
\newtheorem{cor}[theorem]{Corollary}
\newtheorem{prop}[theorem]{Proposition}

\newtheorem{fact}[theorem]{Fact}
\newtheorem{claim}[theorem]{Claim}

\newtheorem{example}[theorem]{Example}

\newtheoremstyle{named}{}{}{\itshape}{}{\bfseries}{.}{.5em}{\thmnote{#3}}
\theoremstyle{named}

\usepackage{todonotes}

\usetikzlibrary{decorations.pathmorphing}
\usetikzlibrary{cd}
\tikzset{snake it/.style={decorate, decoration=snake}}

\usepackage{float}

\usepackage{physics}

\newcommand{\polyloglog}{\operatorname{polyloglog}}
\newcommand{\loglog}{\operatorname{loglog}}
\newcommand{\logloglog}{\operatorname{logloglog}}
\newcommand{\polylogloglog}{\operatorname{polylogloglog}}
\newcommand{\loglogloglog}{\operatorname{loglogloglog}}
\newcommand{\quasipolylog}{\operatorname{quasipolylog}}
\newcommand{\quasipoly}{\operatorname{quasipoly}}
\newcommand{\CNOT}{\mathsf{CNOT}}
\newcommand{\Had}{\mathsf{H}}
\newcommand{\CCZ}{\mathsf{CCZ}}
\newcommand{\CZ}{\mathsf{CZ}}

\newcommand{\F}{\mathbb{F}_2}
\newcommand{\Fq}{\mathbb{F}_q}
\newcommand{\FF}{\mathbb{F}}
\newcommand{\SWAP}{\mathsf{SWAP}}
\newcommand{\ROT}{\mathsf{ROT}}
\newcommand{\PERM}{\mathsf{PERM}}
\newcommand{\PAULI}{\mathsf{PAULI}}
\newcommand{\CLIFF}{\mathsf{CLIFF}}
\newcommand{\CLIFFM}{\mathsf{CLIFF_M}}
\newcommand{\MEAS}{\mathsf{M}}
\newcommand{\MEASD}{\tilde{\mathsf{M}}}

\newcommand{\stab}{\operatorname{Stab}}
\newcommand{\mixedstate}[1]{\ensuremath{\mathcal{D}\left(\mathcal{H}_2^{\otimes #1}\right)}}

\newcommand{\ndep}[1]{\ensuremath{\widehat{#1}(n_L)}}

\NewDocumentCommand\gadget{}{\ensuremath{\mathsf{Gad}}}
\NewDocumentCommand\qcode{}{\ensuremath{\mathcal{Q}}}
\NewDocumentCommand\qOp{}{\ensuremath{\mathsf{Op}}}
\NewDocumentCommand\weightenum{mm}{\ensuremath{\mathcal{W}{(#1;~#2)}}}
\NewDocumentCommand\fault{o}{%
  \IfNoValueTF{#1}
   {\ensuremath{\mathbf{f}}}
   {\ensuremath{\mathbf{#1}}}%
}

\DeclareMathOperator*{\wtsum}{\boxplus}
\DeclareMathOperator*{\wtprod}{\circledast}

\title{Quantum fault tolerance with constant-space and logarithmic-time overheads}
\date{}
\author{Quynh T. Nguyen\thanks{Harvard. \href{mailto:qnguyen@g.harvard.edu}{qnguyen@g.harvard.edu} } \qquad Christopher A. Pattison\thanks{Caltech. \href{mailto:cpattiso@caltech.edu}{cpattiso@caltech.edu}}}
\date{\today}

\begin{document}

\maketitle

\begin{abstract}

In a model of fault-tolerant quantum computation with noiseless constant-time auxiliary classical computation per quantum operation, we construct a quantum fault tolerance protocol with constant-space and $\widetilde{O}(\log N)$-time overheads, where the notation $\widetilde{O}(\cdot)$ hides sub-polylogarithmic factors. This significantly improves over the previous state-of-the-art protocol of Yamasaki and Koashi that achieved constant-space and quasi-polylogarithmic-time overhead. Our construction is obtained by using constant-rate quantum locally testable codes (qLTC) of appropriately chosen block size and developing new fault-tolerant gadgets on qLTCs and qLDPC codes. In particular, we obtain the following new technical results:

  (1) We develop a magic state distillation protocol with $(\log \frac{1}{\varepsilon})^{\gamma}$ spacetime overhead, where $\gamma \rightarrow 0$ as $\varepsilon \rightarrow 0$. This result differs from recent works in that we use a simple and self-contained construction using Reed-Solomon codes to obtain low \emph{spacetime} overhead (rather than just space overhead as in recent works). We show a similar result for stabilizer state distillation.
  
  (2) We prove that quantum codes based on the cubical complex construction admit sequential and parallel single-shot decoders against adversarial errors of weight scaling with the code distance. In particular, our proof applies to a recent family of almost-good qLTCs  of Dinur-Lin-Vidick and the good qLDPC codes of Dinur-Hsieh-Lin-Vidick.

  (3) Using the introduced distillation schemes, we develop fault-tolerant logical state preparation procedures with $\widetilde{O}(1)$-spacetime overhead on qLTCs. Here, the qLTC property is used to quickly test if a state is too far from the codespace before proceeding.

  (4) We introduce the use of multiple resource states (from a small-sized set) to obtain addressable qLDPC logic that can be easily prepared using our state preparation schemes and transversal qLDPC gates.

    We obtain the main result by combining the above results with carefully chosen parameters. In doing so, we introduce a new \emph{weight enumerator formalism} to prove fault tolerance in a composable way, which is of independent interest.
    To our knowledge, this gives the lowest spacetime overhead to date in the considered model of quantum fault tolerance, which, for the first time, matches that of classical fault tolerance up to sub-polylogarithmic factors.
    We conjecture this is optimal up to sub-polylogarithmic factors.
\end{abstract}

\newpage
\tableofcontents

\newpage
\section{Introduction}
\subsection{Background and main result}
The optimal overhead for fault-tolerant computation is a fundamental question in computer science and engineering. In the setting of classical computation, this question was essentially settled in the 1990s. \cite{pippenger1991lower, gacs1994lower} showed
the existence of a Boolean function whose noiseless circuit complexity is $|C|$ but requires $\Omega(|C| \log |C|)$ noisy gates to compute. This lower bound is met by the original fault tolerance scheme in von Neumann's lectures that initiated the study of fault-tolerant computation~\cite{von1956probabilistic}.

In the setting of quantum computation, this question is largely open and appears to be much more challenging. On one hand, existing rigorous proofs of fault tolerance for quantum circuits are often very complicated. This complexity can be attributed to the inherent subtleties of quantum information that causes new technical and conceptual challenges~\cite{gottesman2009introduction}, such as in conditioning on a subsystem, handling entangled ancillas, tracking propagation of quantum errors, dealing with the lack of universal transversal gate sets~\cite{eastin2009restrictions}, etc.
In fact, it is only recent that linear-distance quantum LDPC codes have been shown to exist~\cite{panteleev2022asymptotically}. 
As such, until recently, there has only been a handful of rigorously proven end-to-end quantum fault tolerance protocols~\cite{aharonov1999fault, aliferis2005quantum}. On the other hand, lower bounds for the fault tolerant overhead in the quantum case are also naturally expected to be more challenging. The currently best lower bound on quantum fault tolerance overhead is from the work of~\cite{fawzi2022lower} who derived a space lower bound of $Q^{-1} n$, where $Q$ is the quantum capacity of the noise channel\footnote{This bound also applies for the computation model considered in this work where small-sized noiseless classical computation is assumed.}. To our knowledge, there is currently no spacetime lower bound that is better than linear.

Motivated by practical implementations of quantum computers, a commonly studied model of fault-tolerant quantum computation is one in which a noisy quantum circuit is executed with the help of a (small-sized) noiseless classical computation. A metric of high interest is the space overhead, i.e., the number of physical qubits used at any point in time of the quantum computation. In his seminal work~\cite{gottesman2013fault}, Gottesman initiated the study of constant-space-overhead quantum fault-tolerance under the assumption that a large auxiliary classical computation is free and noiseless. The works of~\cite{fawzi2020constant} further relaxed the auxilliary classical computation time to polylogarithmic.
Here, we assume that an even smaller, constant-time noiseless classical computation is used per physical quantum time step.

\begin{definition}[Quantum fault tolerance model, simplified]
Let $C$ be a Clifford+$\CCZ$ quantum circuit of depth \(D\) and space \(W\).
For every timestep of the quantum circuit, we run a noiseless classical circuit of constant depth given access to the measurement history and a classical register.
The quantum operations in the following timestep are permitted to be classically controlled on individual output bits of the classical circuit.
    
In this work, we place no restrictions on the connectivity of the operations, neither for the simulated circuit nor the simulating one.
Our noise model is the commonly use locally-stochastic noise: the set of faulty gates \(F\) satisfies that, for any subset of gates \(S\), \(\Pr(S \subseteq F) \le p^{|S|}\) for some \(p\) that should be thought of as the noise rate.
Faulty gates are permitted to apply any CPTP noise channel.
\end{definition}
The goal of the fault-tolerant circuit is to produce samples from a distribution that is \(\varepsilon\)-close in total-variation distance (TVD) to the output distribution of the noiseless quantum circuit.

The question we study in this paper is: \emph{What is the lowest time overhead possible while maintaining a constant space overhead in the above model?} This question is of both practical and theoretical interest. The model of quantum fault tolerance considered here has been the conventional model studied in the field, stemming from the fact that classical computation is practically perfect with current technology.
However, we believe that studying the overhead of fault tolerance in this model could lead to ideas to answer the same question in the fully-quantum model~\cite{aharonov1999fault}.
Previous to the initiation of this work, the state-of-the-art result in this line of work is due to Yamasaki and Koashi~\cite{yamasaki2024time}, who constructed a protocol based on concatenated quantum Hamming codes. They achieved a quasipolylog-time overhead -- which is a function growing faster than $\log^\alpha(|C|)$ for any constant $\alpha$.

The above discussion leads us to our main result.\footnote{A prior version of this paper assumed a \(\polyloglog\)-depth classical circuit. We have made a small modification to the EC gadget to remove this requirement. } (see \cref{thm:main-result-concat} for the full formal statement)

\begin{theorem}[Main result, informal]\label{thm:main} Given any (arbitrary-connectivity) quantum circuit $C$ on $W$ qubits, of depth $D$, and composed of $|C| \leq f(W)$ gates from the Clifford+$\CCZ$ gate set, where $f$ is a function growing faster than any quasipolynomial function, we can efficiently construct a circuit $C_{FT}$ with the following guarantees. $C_{FT}$ uses at most $O(W)$ physical qubits at any point in time. The quantum depth of $C_{FT}$ is $D \times O(\log^{1+o(1)} \frac{|C|}{\varepsilon})$, and the (noiseless) auxilliary classical computation time used per physical quantum time step is $O(1)$. There exists a constant noise threshold $p_*$, such that when executed under the local stochastic noise model with noise rate $p < p_*$, $C_{FT}$ outputs a distribution which has total-variation distance $\varepsilon$ from the output distribution of $C$.
\end{theorem}

Above, we assume the ideal circuit is composed of gates from a fixed local gate set. This assumption is simply to remove a potential polylogarithmic depth overhead in decomposing gates according to the Solovay-Kitaev theorem. Notably, the stated time overhead holds even though we are making no other assumptions about the gate connectivity and specific patterns of gate types in the ideal circuit. In fact, our fault tolerance scheme simulates any logical gate in a time overhead of at most $O(\log^{o(1)} \frac{|C|}{\varepsilon})$, and the extra logarithmic factor arises from potentially adversarially selected connectivity in the simulated circuit. The $o(1)$ exponent roughly scales as $1/\logloglog(|C|/\varepsilon)$. We can also smoothly tune the overheads: the space overhead can become $\widetilde{O}(\log \frac{|C|}{\varepsilon})$ and the time overhead can become as low as a constant, while keeping the spacetime overhead $\widetilde{O}(\log \frac{|C|}{\varepsilon})$ (see~\Cref{subsec:combining}).

\subsection{Overview of ideas and technical results}
The result is based on multiple main ideas and new technical results which we overview below. To keep track of less variables, in the rest of this introduction we take the target error $\varepsilon$ to be a constant as well as assume $|C|=\poly(W)$.

\subsubsection{Revisiting Gottesman's protocol}\label{subsubsec:revisiting}
Our starting point is an observation that a modification of Gottesman's protocol~\cite{gottesman2013fault}, in particular the instantiation by Fawzi, Grospellier, and Leverrier~\cite{fawzi2020constant}, suffices to obtain a time overhead of $\polylog(|C|)$. Let us recall this protocol (see Section 2.5 of the arxiv version of~\cite{fawzi2020constant}) that achieved a linear-time overhead.

The authors of~\cite{fawzi2020constant} partition the $W$ qubits into $\ell=\polylog(W)$ blocks with each block containing roughly $K_\mathrm{FGL}= W/\ell$ qubits. Then they use a square-root-distance quantum LDPC code (namely quantum expander code~\cite{leverrier2015quantum}) of parameters $[[N'=\Theta(K_\mathrm{FGL}), K_\mathrm{FGL}, D'= \Theta(\sqrt{K_\mathrm{FGL}})]]$ to encode each block using $N'$ physical qubits. The fault-tolerant circuit $C'$ alternates between a logical operation cycle and an error correction (EC) cycle. For the EC cycle, they use an efficient decoder that is shown to be single-shot (robust against measurement errors), and hence the quantum depth per EC cycle is $O(1)$. To perform logical gates, they first serialize all gates in $C$ such that each layer only consists of one gate, thereby incurring a $O(K_\mathrm{FGL})$ factor in the time overhead. A logical gate is then fault-tolerantly simulated using gate teleportation~\cite{gottesman1999demonstrating}. Here, the main challenge is to fault-tolerantly prepare the resource state for the logical gate. The solution is Gottesman's idea of using off-the-shelf the concatenated-code FT scheme of Aharonov and Ben-Or~\cite{aharonov1999fault}. The idea is that the concatenated-code protocol can serve as a general-purpose factory of ancilla states \footnote{Strictly speaking, we were not able to find a proof of this important lemma (in particular, the unencoding step of the concatenated-code to obtain the physical state) in the literature. Thus, we give a proof of it in this paper. A very recent work \cite{christandl2024fault} also gives a proof.}. In particular, to prepare the resource state, which is a $O(N')$-qubit state, to error $\varepsilon'= 1/\poly(|C|)$ (which is required to run the computation $C$), concatenated-code FT consumes a circuit spacetime of roughly $N'\polylog(N'/\varepsilon')=N'\polylog(W)$. Hence, choosing $\ell=\polylog(W)$ ensures that only $\Theta(W)$ qubits are used at any point in the fault-tolerant circuit. After some careful analysis the time overhead can be shown to be $O(W)$.

We observe that the choice of code block size and serialization described above is suboptimal. In particular, it suffices to partition $W$ qubits into $W/\polylog(W)$ blocks of $K'=\polylog(W)$ qubits each. The reason is because the logical error suppression due to the quantum expander code $[[N',K']]$ is roughly $2^{-\Theta(\sqrt{N'})}$, so  a blocksize of $N'=\polylog(W)$ already suffices to maintain a low logical error rate throughout the computation \footnote{If we instead use recent good qLDPC codes~\cite{panteleev2022asymptotically,leverrier2022quantum} rather than quantum expander codes, we can in fact choose a code block size of $O(\log |C|)=O(\log W)$. These codes are also known to have a single-shot efficient decoder~\cite{gu2023single}.}. Next, instead of serializing all the gates in $C$, we only need to rearrange them such that there is at most one gate acting on a code block in each circuit layer. This task corresponds to an edge coloring problem on a degree-$O(K')$ graph, which can be done with $O(K')=\polylog(W)$ colors by a generalization of Vizing's theorem~\cite{berge1991short}. Hence, this serialization only incurs a $\polylog(W)$ time overhead. Logical gates can be performed the same way as described earlier, leading to another $\polylog(W)$ spacetime overhead. After some further gate scheduling and serializations, one can obtain constant-space and $\polylog(W)$-time overheads. The details will be described and superseded by our construction, so we do not spell them out here. The degree of the polylog factor comes from multiple sources, including the qLDPC code distance scaling, gate scheduling and serializations, and concatenated-code FT overhead, and has not been explicitly worked out. We expect this to be a large constant (say, at least $4$) with the dominant contributor being the concatenated-code overhead. Hence, new ideas are required to bring this down to the optimal exponent, which we conjecture is 1 as in the classical case.

Motivated by the fundamental question of the optimal overhead of fault-tolerant computation, we reduce the exponent in the time overhead to $1+o(1)$. In particular, we will reduce the time overhead of all logical operations to $\log^{o(1)}(W)$ while keeping the space overhead a constant. The remaining $O(\log W)$ time overhead factor originates from the serialization using edge coloring similar to described above, which we conjecture is inherent -- see discussions in~\Cref{subsec:outlook}. This leads to the claimed main result in~\Cref{thm:main}.

We next overview the extra new ideas and technical contributions in achieving this goal.

\subsubsection{Better overheads with distillation}
The first improvement is to remove the $\polylog(W)$ spacetime overhead arising from resource state preparation using concatenated-code FT. In particular, this factor is $\polylog(N'/\varepsilon')$, where $N'$ is the qLDPC block size and $\varepsilon'$ is logical error rate of the simulated gate. As seen earlier,~\cite{fawzi2020constant} had the parameters $N'= W/\polylog(W)$ and $\varepsilon'=1/\poly(W)$, leading to the $\polylog(W)$ scaling. With our new instantiation, $N'=\polylog(W)$ (in fact, our final choice will be $N'=O(\log W)$). However, the choice $\varepsilon'=1/\poly(W)$ seems unavoidable -- after all, we need this logical gate error rate in order to sustain a $\poly(W)$-long computation. Our idea to circumvent this obstruction is to use state distillation protocols~\cite{bravyi2005universal, knill2005scalable} in combination with concatenated-code FT. In particular, we first use concatenated-code FT to prepare the resource state to only a sufficiently small constant error rate $\varepsilon_0$. This step only incurs a $\polyloglog(W)$ spacetime overhead. The second step is to use a state distillation protocol to suppress this error to the target $\varepsilon'= 1/\poly(W)$. These protocols often employ a quantum code (`distillation code') with a transversal gate related to the distilled resource state. Whether or not this leads to an overhead improvement depends on two factors: The space overhead is measured by the protocol's yield rate, the number of noisy states required to produce a good state. The time overhead depends on the distillation code's encoding/unencoding circuit depth and classical decoder efficiency.
Towards this, we prove the following two new results on distillation overheads that are of independent interest:

\begin{theorem}[Stabilizer state distillation, informal] Let $\ket{\psi}$ be a $O(1)$-qubit stabilizer state. There exists a distillation protocol $\mathsf{SSD}(\varepsilon)$ that uses noiseless CNOT, Pauli operations, Pauli measurements and gives the following guarantees. There exists a constant noise threshold $\varepsilon_\mathrm{SSD}$ such that, provided with $N$ (independent) noisy states whose infidelity with $\ket{\psi}$ is $\varepsilon_\mathrm{in} < \varepsilon_\mathrm{SSD}$, where $N >N(\varepsilon_\mathrm{SSD})$ is a sufficiently large number, $\mathsf{SSD}(\varepsilon)$ produces $\Theta(N)$ states each of which has infidelity at most $\varepsilon$ with $\ket{\psi}$. Furthermore, the quantum depth and classical depth of $\mathsf{SSD}(\varepsilon)$ are both $\polyloglog(1/\varepsilon)$.
\end{theorem}

\begin{theorem}[Magic state distillation, informal] Let $\ket{\CCZ}=\CCZ \ket{+}^{\otimes 3}$. There exists a distillation protocol (using noiseless Clifford operations and measurements) $\mathsf{MSD}(\varepsilon)$ and a constant noise threshold $\varepsilon_\mathrm{MSD}$ such that the following holds. Upon receiving $N$ (independent) noisy states  whose infidelity with $\ket{\CCZ}$ is $\varepsilon_\mathrm{in} < \varepsilon_\mathrm{MSD}$, where $N >N(\varepsilon_\mathrm{MSD})$ is a sufficiently large number, $\mathsf{MSD}(\varepsilon)$ produces $\Theta(N/\log^{o(1)}(1/\varepsilon))$ states each of which has infidelity at most $\varepsilon$ with $\ket{\CCZ}$. Furthermore, the quantum depth and classical depth of $\mathsf{MSD}(\varepsilon)$ are both $\polyloglog(1/\varepsilon)$. Here, the $o(1)$ exponent scales as $O(1/\logloglog(1/\varepsilon))$.
\end{theorem}

These results are proved in~\Cref{sec:state-distill-procedure-single} and~\Cref{sec:magic-state-distillation-code}, respectively. The `batch sizes' $N_\mathrm{SSD}$ and $N_\mathrm{MSD}$ are roughly $\operatorname{quasipolylog}(1/\varepsilon)$. For the first result, we are not aware of any prior results on stabilizer state distillations of this type. We suspect that this is because previous works often consider quantum codes with transversal Cliffords and hence do not run into this question. In our case, and generally on the recently developed qLDPC code families, we often have a limited set of transversal gates and thus require other fault-tolerant Clifford gate techniques such as lattice surgery~\cite{cohen2022low}. Here, we opt to gate teleportation and thus need to derive the stated result. In fact, we will also use this distillation procedure to distill computational basis states. The magic state distillation (MSD) question, on ther hand, has been more intensively studied, starting from~\cite{bravyi2005universal, knill2005scalable}. Our result MSD overhead result gives a yield rate of $\log^{o(1)}(1/\varepsilon)$, improving over a prior work~\cite{hastings2018distillation} that achieved an exponent of $\approx 0.68$. The time overhead (including both quantum and classical) is only $\polyloglog(1/\varepsilon)$, which is crucial to obtain our main result but often disregarded in the MSD literature. The slightly non-constant MSD space overhead is not an issue, as appropriate gate scheduling can turn this into a time overhead. We note that concurrent and independent works~\cite{wills2024constant, nguyen2024good, golowich2024asymptotically} have established that constant-space-overhead MSD is possible (see the explicit scheme in~\cite{wills2024constant}), however we believe the time overhead of these protocols is at least $\polylog(1/\varepsilon)$, arising from the unencoding circuit of the large distillation code. Hence, our protocol gives a better \emph{spacetime} overhead with an arguably simpler construction, which is necessary for the main result.

Applying the above distillation schemes on the logical level of a high-rate qLDPC code is not necessarily straightforward, as we need to address specific logical qubits in a large qLDPC code block with multiple logical qubits. We show how to use the above results to distill important resource states for logical operations on high-rate qLDPC codes in~\Cref{subsec:boost-fidel-state-prep}. Here, the key observation is that the noiseless operations required in the SSD scheme can be implemented transversally on the computation qLDPC CSS code. Once we have distilled qLDPC logical stabilizer states, we can then use them to run the MSD scheme as well. To our knowledge this is the first time this is explicitly worked out, with a rigorous fault tolerance proof.

\paragraph{Techniques.} The SSD protocol is inspired by the concatenated-code scheme of~\cite{yamasaki2024time}. We distill the stabilizer states in multiple levels using quantum Hamming codes $[[2^l-1, 2^l-l-1,3]]$, which have transversal Clifford gates. As its rate quickly approaches 1, the multi-level distillation protocol gives rise to a constant-space overhead. The details of this protocol are in~\Cref{sec:state-distill-procedure-single}. On the other hand, the MSD protocol is based on a new family of good qudit codes supporting transversal CCZ gate, that we describe in~\Cref{sec:magic-state-distillation-code}. There, we give a simple and self-contained construction using (punctured) Reed-Solomon codes over binary extension fields using ideas from~\cite{krishna2019towards}. The claimed time overhead comes from efficient decoders for these codes and that the code block sizes are at most $O(\loglog(1/\varepsilon))$. Let us give an intuition for why distillation improves the time overhead over simply using the concatenated-code FT scheme~\cite{aharonov1999fault}: Multi-level distillation is in a sense also using concatenation, but we unencode the code before going to the next level, whereas concatenated-code FT never unencodes. Note that we are running these protocols on top of a qLDPC code, and `unencode' here means unencoding from the distillation code, not the computational qLDPC code, so we are still protected by the qLDPC code.
Therefore, the time overhead increases additively in distillation in terms of number of levels, rather than multiplicatively as in concatenated-code FT schemes.

\subsubsection{Addressable logic on high-rate codes}
The next question is what resource states to use. This question arises because the gate connectivity in the simulated circuit $C$ can be arbitrary. For example, a logical CZ gate between qubit $1$ of a code block and qubit $7$ of another code block needs a different resource state than CZ between qubits $5$ and $32$. A Hadamard gate on qubit $1$ needs a different resource state from Hadamard on logical qubit 2. We say these gates are of different types. The distillation protocols, as we show, still work to distill any of them. However, the issue is that there are too many such elementary resource state types. For example, there are $O(K'^2)=\polylog(W)$ possible CZs between two distinct code blocks. On the other hand, distillation protocols produce resource states of the same type \emph{in batches} of size roughly $N_\mathrm{SSD}, N_\mathrm{MSD}= \operatorname{quasipolylog}(1/\varepsilon)$, as seen above. Each layer in a circuit $C$ may adversarially consists of gates of all types. Hence, naively, we either have to ruin the constant-space overhead to run distillation for all of these resource states in parallel (many of them will be left ununsed) or serialize the gates in $C$ so that each layer contains gates of a single type. Either way we would incur a factor of at least $\polylog(W)$ in the overhead.

A solution to this is to design a small set of resource states that are capable of performing addressable logic. Here we do so with a set of $O(\log K') =O(\loglog(W))$ resource states. The idea is as follows: We will only consider resource states for logical single-qubit gates (Hadamard and S) that act on the `first' logical qubit in a code block. For the two-qubit gates (CNOT and CZ), we only care about resource states for the gates that act on qubits 1 and 2 of the same block. Similarly, we only have resource states for the CCZ gate that acts on qubits 1, 2, 3 of the same code block. So far this set has $O(1)$ resource states. To turn them into addressable gates, we use $O(\log K')=O(\loglog W)$ special SWAP/permutation multi-qubit gates that allows us to perform low-depth arbitrary qubit permutation. Importantly, we show that the resource states for these multi-qubit states can also be distilled by running the stabilizer distillation protocol on top of the computation qLDPC code. Hence, in total we only use a set of $O(\loglog W)$ `primitive computation states'. Using this primitive gateset, we can then perform a logical gate on any desired locations while only incurring a $O(\log(K'))=O(\loglog W)$ time overhead. The details of this are given in~\Cref{subsec:compilation}.

\subsubsection{Quantum locally testable codes} We now move on to a subtlety that was intentionally omitted in the above discussion. We have claimed that distillation protocols can be applied on the output of low-fidelity resource states to boost the fidelity. However, these protocols are being performed on top of qLDPC codestates, and their performance analysis only applies when the systems are guaranteed to be in the codespace. This is not necessarily the case for noisy codestates prepared by concatenated-code FT. Let $\ket{\psi}$ be the target resource state. Informally, what concatenated-code FT actually says is: With probability $\geq 1-\varepsilon_0$, it outputs a state $\psi'$ that has the form $\psi'=\mcN(\ket{\psi})$, where $\mcN$ is a local stochastic noise channel with parameter $O(p)$. However, in the remaining case, the output state of concatenated-code FT could be arbitrary. It could either be (1) a codestate with some logical error or (2) an arbitrary non-code state. We would like to filter out noisy states of the latter case before running distillation protocols as quick as possible. One idea is to measure the qLDPC code checks and declare `FAIL' if too many checks are violated. However, in general, there are non-code states that are very far from the codespace but only violates a few checks \footnote{This is a well-known issue with toric code. A common solution to this is to perform multiple rounds (roughly as many as the code distance) of syndrome measurements. This incurs a time overhead and will not suffice for our main result.}. This challenge is exactly addressed by local testable codes and is a motivation of classical local testability, as quoted from Spielman's PhD thesis~\cite{spielman1995computationally}: ``The checker can instantly request a retransmission of that block, before the decoder has wasted its time trying to decode the message.'' One can also define quantum locally testable codes (qLTC)~\cite{aharonov2015quantum, eldar2017local}, which are, by definition, qLDPC codes. Roughly speaking, a qLTC with soundness $\rho<1$ ensures that, if $m$ checks are violated, then the state deviates from the codespace on at most $m/\rho$ qubits. Hence, using a constant-soundness qLTC resolves this issue without incurring any extra time overhead, see~\Cref{subsec:low-fidel-state-prep} for more details. In summary, the discussion up to this point \emph{conditionally} proves our main result. Specifically, 

\begin{quote} Suppose that there exists a qLTC family\footnote{There is also a technical requirement that the family is sufficiently `dense', meaning that the code block size does not grow too fast with respect to the family index.} with constant rate, constant relative distance, and constant soundness (called c3-qLTC). Furthermore, suppose that the qLTC family admits an efficient parallel single-shot decoder. Then quantum fault tolerance can be achieved with the overheads stated in~\Cref{thm:main}.
\end{quote}

To our knowledge this work is the first time local testability is used in a fault tolerance protocol, either classical or quantum. The only previous work that used a related notion of locality in codes is~\cite{romashchenko2006reliable}. There, the author constructed a classical fault tolerance protocol using local decodable codes, a notion much stronger than local testability. Such notion, however, is not possible for quantum codes~\cite{aharonov2013guest}.

\subsubsection{Single-shot decoders for cubical complex quantum codes} Unfortunately the c3-qLTC conjecture is still an open question. However, a recent breakthrough~\cite{dinur2024expansion} has nearly approached this goal, providing a family of almost-c3 qLTC. In particular, Dinur, Lin, and Vidick~\cite{dinur2024expansion} generalized the construction of good qLDPC codes from~\cite{dinur2022good} to high-dimensional cubical complexes. In combination with Panteleev and Kalachev's new result~\cite{kalachev2025maximally} on high-dimensional product-expansion of random codes, this gives a qLTC family of constant rate, inverse-polylog relative distance, and inverse-polylog soundness. We show that a suitable instantiation of this construction suffices to obtain our main result. To this end, we establish sequential and parallel single-shot decoders for general quantum codes built on the cubical complex construction:

\begin{theorem}[Single-shot decoder, informal]\label{thm:decoder-informal}
Let $t \geq 2$. Consider an $[[n,k,d]]$ quantum LDPC code built from a $t$-dimensional cubical complex using the local sheaf construction such as in~\cite{dinur2022good,dinur2024expansion}. There exists a linear-time decoder $\mcD_\mathrm{seq}$, and constant $A,B, C>0$ such that the following holds. If the data errors $e$ and syndrome measurement errors $m$ have weights satisfying $A|e|_R+B|m| < d$ (where $|\cdot|_R$ denotes stabilizer-reduced weight), then $\mcD_\mathrm{seq}$ proposes a correction such that the residual error after applying the correction is upper bounded by $C\cdot|m|$. A similar statement holds for a $O(\log n)$-time parallel decoder $\mcD_\mathrm{par}$. In fact, for any positive number $\tau$, running $\mcD_\mathrm{par}$ in $O(\tau)$ time guarantees that the residual error is suppressed to $\frac{1}{2^{-\Theta(\tau)}} |e|_R + C|m|$.
\end{theorem}

We note that the above has been heavily simplified, see~\Cref{sec:qltc} for precise statements. This result is a bounded-distance decoder, i.e., it applies to all (adversarial) errors up to certain weight. It applies to the DHLV code~\cite{dinur2022good} and also the almost-c3 qLTC in~\cite{dinur2024expansion}. In other words, the good qLDPC code family of~\cite{dinur2022good} admits a single-shot decoder up to linear-weight adversarial errors, answering an open question raised by Gu et al.~\cite{gu2023single}. For the almost-c3 qLTC, this result holds for adversarial errors of weight up to $n/\polylog(n)$.

\paragraph{Techniques.} The decoder is a generalization of the small-set flip decoder in~\cite{dinur2022good} to higher-dimensional cubical complexes (\cite{dinur2022good} was on a square complex, i.e., 2-dimensional cubical complex). The main technical challenge is to come up with appropriate high-dimensional generalizations of the procedures and proof techniques there. Furthermore, we show that the decoder is single-shot and parallelizable. The Z-syndrome decoder is shown in~\Cref{sec:Z-decoder}. The sequential version goes roughly as follows: while the syndrome is still nonzero, find a small-set flip in the local view of some vertex that reduces the syndrome weight. In the case when the syndrome is noiseless, we show that if no such flip is found, then the remaining error must be either zero or very large. On the other hand, we show that the chain complex is \emph{small-set co-boundary expanding}, which roughly says that the syndrome weight provides an upper bound on the error weight as long as the error weight is sufficiently small\footnote{The small weight condition makes this property weaker than local testability, but it is all we need for the decoders.}. And since the syndrome weight only reduces throughout the algorithm by design, the error weight cannot become large. Thus when the algorithm terminates, there is no remaining error. In the noisy syndrome case, a similar proof strategy also works by noting that coboundary expansion-based proof is robust to measurement errors. Naturally, the parallel version of the decoder attempts to perform local flips in parallel. We show that, in each parallel decoding round, such flips have large overlap with the underlying error, hence reducing the syndome weight by a multiplicative constant factor per round.

The X-syndrome decoder is not the same due to the asymmetry in the code construction. Similar to~\cite{dinur2022good}, it is obtained via a reduction to a Z-syndrome decoding on a \emph{related} code. This reduction starts by having each vertex locally make a `guess' about what it thinks the error should look like within its local view. The goal is now to fix the inconsistencies between these local guesses to obtain a global guess. The second step of the reduction involves an argument following the ideas in the code distance proof of~\cite{dinur2024expansion}. In particular, for each edge, the opinions of the vertices in that edge are compared, and hence this inconsistency-fixing problem can be moved one level up in the cubical complex. This procedure is repeated until the top level of the cubical complex is reached. Then, it turns out that we can map this inconsistency-fixing problem into a Z-syndrome decoding problem on a related code. So we can invoke the Z-syndrome decoder for that code to obtain a fix, and then propagate the fix back down. Since the above inconsistency-propagation procedure can be done locally, it turns out that the X-syndrome decoder inherits the single-shotness and parallelizability from the Z-syndrome decoder. We refer to~\Cref{sec:X-decoder} for more details.

\subsubsection{Combining everything together}\label{subsec:combining}

Intuitively, combining what we have described so far, including the almost-c3 qLTC with single-shot decoder, we can obtain an FT procotol with overheads as stated in~\Cref{thm:main} but with a sub-constant threshold of $1/\polylog(N')=1/\polyloglog(W)$. To achieve a constant threshold, the final step is to concatenate the described protocol with the Yamasaki-Koashi protocol~\cite{yamasaki2024time}. Their protocol incurs constant-space overhead and a time overhead of $\exp((\loglog 1/\eta)^2)$ to fault-tolerantly simulate a constant-sized operation to logical error rate $\eta$. Hence, setting $\eta=1/\polyloglog(W)$ we only incur an extra factor of $\exp((\operatorname{loglogloglog}W)^2)$ in the time overhead.

This concludes our description of the high-level ideas and new technical results leading to~\Cref{thm:main}. Making the above intuitive descriptions rigorous, however, is significantly involved. To do so, we introduce a new framework to prove (classical or quantum) fault tolerance that we call the \emph{weight enumerator formalism} in~\Cref{sec:weight-enum-formalism}. It can be viewed as a generalization of the technique of Aharonov and Ben-Or~\cite{aharonov1999fault} (dating back to~\cite{GACS1986}) of counting `bad' faulty locations. The motivation of this new framework is the need to combine many different types of fault-tolerant gadgets in non-trivial ways to construct new fault-tolerant
gadgets. In contrast to traditional concatenated-code FT schemes where the same gadgets are recursively used, the large variety of currently existing fault-tolerant gadgets and ways to combine them cause this counting to become unwieldy. Here, we use a framework inspired by the weight enumerator polynomial methods from coding theory. The key observation is that we can associate polynomials to the family of bad fault sets and work with the polynomials similar to probabilities. Multiplication and addition of these polynomials corresponds nicely to operational meaning. We find this framework applicable generally and expect it will be useful in proving fault tolerance in other contexts.

\paragraph{Interpolating space and time overheads.} We briefly describe without proofs how the tradeoff between space and time overheads can be smoothly tuned while keeping the overall spacetime overhead $\widetilde{O}(\log W)$ (the auxiliary classical time per quantum time step is still the same as in~\Cref{thm:main}). The idea is to not use all $K'=\Theta(N')$ logical qubits in a code block, but instead use only $k'$ of them and simply ignore the other logical qubits. This reduces the main slowdown factor $O(K')=O(\log W)$ due to gate scheduling (see~\Cref{subsubsec:revisiting}) to $O(k')$, and hence improves the time overhead to $O(k' \log^{o(1)} W)$, while increasing the space overhead to $O(K'/k')$. Choosing $k'=1$ gives $O(\log W)$-space overhead and $O(\log^{o(1)} W)$-time overhead. To obtain $O(\log^{1+o(1)} W)$-space and $O(1)$-time overheads, we simply shift the ancilla state factories overhead completely to space by preparing ancilla states off-line, and then using the single-shot decoder to preserve them until use. In addition, instead of concatenating with the FT protocol of~\cite{yamasaki2024time} in the step of improving the threshold from $\eta=1/\polyloglog(W)$ to constant, we use the constant-time overhead schemes described in~\cite[Section 8]{gottesman2013fault} or~\cite{bombin2015single}, which only incurs an extra space overhead factor of $\polylog(1/\eta)=\polylogloglog(W)$.

\subsection{Related works}
\paragraph{Classical fault tolerance.} The study of fault-tolerant (FT) computation was started in von Neumann's lectures~\cite{von1956probabilistic}, where he gave a FT scheme for classical computation with log-space and constant-time overhead. His scheme was later made rigorous and explicit by~\cite{dobrushin1977upper, pippenger1985networks}. It was shown by~\cite{pippenger1991lower, gacs1994lower} a matching lower bound on the FT overhead. These works settle the FT spacetime overhead required for classical computation.

\paragraph{Quantum fault tolerance.} The first fault-tolerant quantum computation (FTQC) threshold theorems were shown independently by~\cite{aharonov1999fault, kitaev1997quantum, knill1998resilient}, building on~\cite{shor1996fault}. These works established that FTQC is possible with polylog space and time overheads. Later works such as~\cite{aliferis2005quantum, gottesman2009introduction} devised new proof methods for the concatenated-code FT scheme and improved the rigorously proved threshold value. Motivated by practical considerations, a standard model of FTQC often considers a noisy quantum computer with auxilliary polynomial-time and noiseless classical computation (in this work we only assume polyloglog-time). In this model, Gottesman~\cite{gottesman2013fault} initiated the study of constant-space-overhead FTQC.~\cite{fawzi2020constant} showed that constant-space and polynomial-time overhead is possible. 
Very recent work~\cite{christandl2024fault} studies FTQC with quantum output which is required to prepare resource states of \cite{gottesman2013fault} and \cite{fawzi2020constant}.
In 2022, \cite{yamasaki2024time} achieved a quasipolylog-time overhead which is the state-of-the-art before initiating our work.
Concurrent to our work, we were made aware of a recent independent work initially appearing at the AQIS'24 conference~\cite{tamiya2024aqis}, which gives a proof of FTQC with constant-space and \(\polylog(W)\)-time overheads using substantially different techniques from ours. We note that~\cite{gottesman2013fault, fawzi2020constant, yamasaki2024time} assume the locally-stochastic Pauli noise model, which is weaker than our noise model (we allow faults to be arbitrary CPTP channels).

\paragraph{Good qLDPC code decoders.} Some families of good quantum LDPC codes~\cite{panteleev2022asymptotically, leverrier2022quantum} have been recently shown to admit a single-shot decoder~\cite{gu2023single}. Other provably efficient decoders for qLDPC codes include~\cite{fawzi2020constant,leverrier2023decoding, dinur2022good}. The recent breakthrough work~\cite{dinur2024expansion}, using classical codes constructed in \cite{kalachev2025maximally}, construct almost-good qLTC by generalizing the code from~\cite{dinur2022good}.

\paragraph{Magic states and non-Clifford gates.} The study of magic state distillation protocols are initiated in~\cite{bravyi2005universal, knill2005scalable}. Prior to this work, the state-of-the-art MSD (space) overheads are due to~\cite{hastings2018distillation} (qubit case) and~\cite{krishna2019towards} (qudit case). During preparation of this work, recently posted and independent works~\cite{wills2024constant, nguyen2024good, golowich2024asymptotically} have established that constant-space-overhead MSD by using algebraic-geometry codes. However, the time overhead was not studied in~\cite{nguyen2024good, golowich2024asymptotically}.~\cite{wills2024constant} describe an explicit protocol, although, naively, $\polylog(1/\varepsilon)$ time overhead is required due to the (un)encoding depth and syndrome measurements of the distillation code with large block size. This is insufficient to obtain our main result.
See also very recent works~\cite{zhu2023non, scruby2024quantum, golowich2024quantum, breuckmann2024cups} on quantum LDPC codes with a non-Clifford gate.

\subsection{Outlook}\label{subsec:outlook}
There are multiple immediate follow-up questions to our work, here we list a few:
\begin{itemize}
    \item Can we drop the assumption of noiseless classical computation? We believe that this requirement is not needed and can be removed with careful use of classical fault tolerance along with our constant-time parallel decoder. We leave details of this procedure and its overhead analysis to future work.
    
    \item Can we improve FTQC overhead lower bounds?~\cite{fawzi2022lower}  obtain a space lower bound that is linear in the ideal circuit space. We conjecture that our result in~\Cref{thm:main} is optimal up to sub-polylog factors. We suspect that the classical functions witnessing the $\Omega(\log W)$ classical FT overhead lower bound in~\cite{gacs1994lower} can be extended to the setting of FTQC without auxiliary classical computation.
    \item Toward the question of an overhead lower bound, we describe here an argument for why a time overhead of $O(\log W)$ is likely required for constant-space-overhead FTQC (even with auxiliary classical computation), suggesting that~\Cref{thm:main} is likely optimal up to sub-polylog factors. This comes from the potentially adversarial gate connectivity in the ideal circuit. Consider a FTQC scheme that uses constant-rate code blocks $[[n,k=\Theta(n),d]]$ to simulate an ideal circuit $C$ of $W$ qubits and, for simplicity, size $\poly(W)$. We need the code distance to be $d=\Theta(\log W)$ to maintain a $\poly(W)$-sized computation. Consider $k+1$ codeblocks $\mcB_0,\hdots,\mcB_k$. Suppose we can perform two-qubit gates on any pairs of the $k(k+1)$ logical qubits in $O(1)$ depth. Now, an adversarial layer in the simulated circuit $C$ may consist of two-qubit gates from each logical qubit $j \in [k]$ in $\mcB_0$ to a logical qubit in $\mcB_j$. The implementation of each gate has to touch $\geq d$ physical qubits in $\mcB_0$, thus implementing all $k=\Theta(n)$ of them needs a $\Omega(d)=\Omega(\log W)$ depth by pigeonhole principle. Choosing an initial specific grouping of logical qubits cannot avoid this issue because the connectivity in $C$ may adversarially vary across layers. Regrouping logical qubits mid-circuit with a sorting network would also incur a $\Omega(\log W)$ time factor.
    \item Our construction is rather complicated. Is there a simpler construction that achieves $\widetilde{O}(\log W)$ spacetime overhead? And can we remove the sub-polylog factors? Currently, the $(\log W)^{o(1)}$ factor in the time overhead mainly comes from our magic state distillation scheme. A better MSD scheme, probably ones that use quantum LDPC codes with transversal non-Clifford gates, might improve this factor to $\polyloglog(W)$. There are also other $\polyloglog(W)$ factors
\end{itemize}

\subsection{Organization of paper}
In~\Cref{sec:weight-enum-formalism} we first provide some preliminaries. We then describe our computation model and present the weight enumerator formalism and related notation used in \Cref{sec:main-proof,sec:state-prep}. In~\Cref{sec:main-proof} we present the proof of our main result, with technical lemmas about distillation protocols and the single-shot decoder deferred to later sections. In~\Cref{sec:state-prep} we construct gadgets to distill resource stabilizer and magic states using the code constructed in~\Cref{sec:magic-state-distillation-code}. Finally, we construct the single-shot decoder in~\Cref{sec:qltc}.
For convenience, we provide a diagram (below) of the high-level components of the proof with hyperlinks.

\subsection{Acknowledgements}
This work was done in part while the authors were visiting the Simons Institute for the Theory of Computing, supported by DOE QSA grant \#FP00010905 and NSF QLCI Grant No. 2016245. We thank Anurag Anshu, Nikolas Breuckmann, Louis Golowich, Ting-Chun David Lin, Pedro Paredes, John Preskill, Shiro Tamiya, Hayata Yamasaki, and Guanyu Zhu for useful discussions. We especially thank  Zhiyang (Sunny) He for insightful discussions about quantum fault tolerance proofs. QTN acknowledges support from the NSF Award No. 2238836 and support from the Harvard Quantum Initiative. CAP acknowledges funding from the Air Force Office of Scientific Research (AFOSR) FA9550-19-1-0360 and U.S. Department of Energy Office of Science, DE-SC0020290. The Institute for Quantum Information and Matter is an NSF Physics Frontiers Center.

\begin{figure}[H]
  \centering
  \begin{tikzpicture}[>=latex,line join=bevel,align=center,scale=0.375]
  \pgfsetlinewidth{1bp}
\small%
\begin{scope}
  \pgfsetstrokecolor{black}
  \definecolor{strokecol}{rgb}{0.0,0.0,0.0};
  \pgfsetstrokecolor{strokecol}
  \definecolor{fillcol}{rgb}{1.0,1.0,1.0};
  \pgfsetfillcolor{fillcol}
  \filldraw (05.80,43.77) -- (05.80,58.695) -- (20.60,58.695) -- (20.60,43.77) -- cycle;
  \draw (13.20,57.433) node {Computational QECC};
\end{scope}
\begin{scope}
  \pgfsetstrokecolor{black}
  \definecolor{strokecol}{rgb}{0.0,0.0,0.0};
  \pgfsetstrokecolor{strokecol}
  \definecolor{fillcol}{rgb}{1.0,1.0,1.0};
  \pgfsetfillcolor{fillcol}
  \filldraw (18.10,25.525) -- (18.10,41.89) -- (29.60,41.89) -- (29.60,25.525) -- cycle;
  \draw (24.388,40.627) node {Transversal Gates};
\end{scope}
\begin{scope}
  \pgfsetstrokecolor{black}
  \definecolor{strokecol}{rgb}{0.0,0.0,0.0};
  \pgfsetstrokecolor{strokecol}
  \definecolor{fillcol}{rgb}{1.0,1.0,1.0};
  \pgfsetfillcolor{fillcol}
  \filldraw (30.40,17.965) -- (30.40,42.97) -- (42.80,42.97) -- (42.80,17.965) -- cycle;
  \draw (37.312,41.707) node {Magic States};
\end{scope}
\begin{scope}
  \pgfsetstrokecolor{black}
  \definecolor{strokecol}{rgb}{0.0,0.0,0.0};
  \pgfsetstrokecolor{strokecol}
  \definecolor{fillcol}{rgb}{1.0,1.0,1.0};
  \pgfsetfillcolor{fillcol}
  \filldraw (21.4,43.77) -- (21.4,58.695) -- (41.10,58.695) -- (41.10,43.77) -- cycle;
  \draw (31.50,57.433) node {State Injection};
\end{scope}
\begin{scope}
  \pgfsetstrokecolor{black}
  \definecolor{strokecol}{rgb}{0.0,0.0,0.0};
  \pgfsetstrokecolor{strokecol}
  \definecolor{fillcol}{rgb}{1.0,1.0,1.0};
  \pgfsetfillcolor{fillcol}
  \filldraw (0.80,17.965) -- (0.80,41.89) -- (14.30,41.89) -- (14.30,17.965) -- cycle;
  \draw (7.55,40.627) node {Stabilizer States};
\end{scope}
\begin{scope}
  \pgfsetstrokecolor{black}
  \definecolor{strokecol}{rgb}{0.0,0.0,0.0};
  \pgfsetstrokecolor{strokecol}
  \definecolor{fillcol}{rgb}{1.0,1.0,1.0};
  \pgfsetfillcolor{fillcol}
  \filldraw (6.30,0.80) -- (6.30,17.165) -- (29.90,17.165) -- (29.90,0.80) -- cycle;
  \draw (11.412,15.902) node {Main Result};
\end{scope}
  \pgfsetcolor{black}
  \draw [->] (12.393,51.74) to[line to] (13.614,48.18);
  \draw [->] (24.10,34.916) to[line to] (23.80,29.972);
  \draw [->] (36.60,25.936) to[line to] (36.60,22.387);
  \draw [->] (36.60,33.855) to[line to] (36.60,30.329);
  \draw [->] (26.374,51.74) to[line to] (26.527,48.18);
  \draw [->] (33.3,51.951) to[line to] (28.539,47.914);
  \draw [->] (6.20,34.916) to[line to] (6.20,30.3);
  \draw [->] (6.803,26.05) to[line to] (8.1912,22.388);
  \draw [->] (17.231,11.68) to[line to] (19.029,11.68);
  \draw [->] (24.182,9.4705) to[line to] (24.318,5.967);
  \draw [->] (15.117,10.092) to[line to] (21.828,5.5865);
  \draw [->] (17.939,3.76) to[line to] (19.68,3.76);
  \draw [->] (15.2,44.6) to[out=285,in=140] (20.1,37.7);
  \draw [->] (10.192,51.862) to[bend right] (20.058,29.072);
  \draw [->] (16.869,53.57) to[line to] (22.653,53.57);
  \draw [->] (26.57,26.571) to[line to] (34.05,22.167);
  \draw [->] (21.258,26.658) to[line to] (11.653,22.037);
  \draw [->] (24.10,26.287) to[line to] (24.10,13.889);
  \draw [->] (34.336,18.892) to[line to] (26.692,13.581);
  \draw [->] (27.691,44.622) to[out=285,in=115] (34.128,22.198);
  \draw [->] (23.5,45.471) to[out=225,in=70,looseness=0.9] (9.8721,22.361);
  \draw [->] (13.524,20.565) to[line to] (31.192,20.565);
  \draw [->] (12.615,19.481) to[line to] (22.307,13.734);
\begin{scope}
  \definecolor{strokecol}{rgb}{0.0,0.0,0.0};
  \pgfsetstrokecolor{strokecol}
  \draw (14.20,46.37) ellipse (5.40 and 1.80);
  \draw (14.20,46.2) node {\hyperref[sec:qltc]{Single-Shot Decoder}};
\end{scope}
\begin{scope}
  \definecolor{strokecol}{rgb}{0.0,0.0,0.0};
  \pgfsetstrokecolor{strokecol}
  \draw (11.80,53.57) ellipse (5.04 and 1.80);
  \draw (11.80,53.57) node {almost-good qLTC \\ \cite{dinur2024expansion}};
\end{scope}
\begin{scope}
  \definecolor{strokecol}{rgb}{0.0,0.0,0.0};
  \pgfsetstrokecolor{strokecol}
  \draw (24.10,36.765) ellipse (4.68 and 1.80);
  \draw (24.10,36.765) node {\hyperref[lemma:ec-gadget]{qLDPC Error} \\ \hyperref[lemma:ec-gadget]{Correction Gadget}};
\end{scope}
\begin{scope}
  \definecolor{strokecol}{rgb}{0.0,0.0,0.0};
  \pgfsetstrokecolor{strokecol}
  \draw (24.10,28.125) ellipse (4.68 and 1.80);
  \draw (24.10,28.125) node {\hyperref[lemma:computational-code-transversal-gates]{Transversal Gate} \\ \hyperref[lemma:computational-code-transversal-gates]{Gadgets}};
\end{scope}
\begin{scope}
  \definecolor{strokecol}{rgb}{0.0,0.0,0.0};
  \pgfsetstrokecolor{strokecol}
  \draw (36.60,28.125) ellipse (4.68 and 2.16);
  \draw (36.60,28.125) node {\hyperref[sec:distill-ccz-cczq]{Qubit $|\mathsf{CCZ}\rangle $} \\  \hyperref[sec:distill-ccz-cczq]{Distillation Scheme} \\ \hyperref[sec:distill-ccz-cczq]{($\gamma \rightarrow 0$)}};
\end{scope}
\begin{scope}
  \definecolor{strokecol}{rgb}{0.0,0.0,0.0};
  \pgfsetstrokecolor{strokecol}
  \draw (36.60,20.565) ellipse (5.40 and 1.80);
  \draw (36.60,20.565) node {\hyperref[sec:ft-state-distillation]{$|\mathsf{CCZ}\rangle $} \\ \hyperref[sec:ft-state-distillation]{Preparation Gadget}};
\end{scope}
\begin{scope}
  \definecolor{strokecol}{rgb}{0.0,0.0,0.0};
  \pgfsetstrokecolor{strokecol}
  \draw (36.60,36.765) ellipse (5.40 and 2.88);
  \draw (36.60,36.765) node {\hyperref[sec:pqrs-code-thm]{Punctured Quantum} \\ \hyperref[sec:pqrs-code-thm]{Reed-Solomon Codes} \\ \hyperref[sec:pqrs-code-thm]{over $\mathbb{F}_2^\ell$ qudits}};
\end{scope}
\begin{scope}
  \definecolor{strokecol}{rgb}{0.0,0.0,0.0};
  \pgfsetstrokecolor{strokecol}
  \draw (26.30,53.57) ellipse (3.60 and 1.80);
  \draw (26.30,53.57) node {\hyperref[sec:tester]{qLTC Tester}};
\end{scope}
\begin{scope}
  \definecolor{strokecol}{rgb}{0.0,0.0,0.0};
  \pgfsetstrokecolor{strokecol}
  \draw (36.00,53.57) ellipse (4.32 and 2.0);
  \draw (36.00,53.5) node {\hyperref[fact:AB]{Concatenated} \\ \hyperref[fact:AB]{Code FT} \\ \cite{aharonov1999fault}};
\end{scope}
\begin{scope}
  \definecolor{strokecol}{rgb}{0.0,0.0,0.0};
  \pgfsetstrokecolor{strokecol}
  \draw (26.60,46.37) ellipse (3.60 and 1.80);
  \draw (26.60,46.2) node {\hyperref[sec:injection-gadget]{State Injection} \\ \hyperref[sec:injection-gadget]{Gadget}};
\end{scope}
\begin{scope}
  \definecolor{strokecol}{rgb}{0.0,0.0,0.0};
  \pgfsetstrokecolor{strokecol}
  \draw (6.20,36.765) ellipse (4.32 and 1.80);
  \draw (6.20,36.765) node {Quantum \\ Hamming Code};
\end{scope}
\begin{scope}
  \definecolor{strokecol}{rgb}{0.0,0.0,0.0};
  \pgfsetstrokecolor{strokecol}
  \draw (6.20,28.125) ellipse (4.32 and 2.1);
  \draw (6.20,28.125) node {\hyperref[sec:state-distill-procedure-proof]{Constant Rate} \\ \hyperref[sec:state-distill-procedure-proof]{Stabilizer State} \\ \hyperref[sec:state-distill-procedure-proof]{Distillation}};
\end{scope}
\begin{scope}
  \definecolor{strokecol}{rgb}{0.0,0.0,0.0};
  \pgfsetstrokecolor{strokecol}
  \draw (8.80,20.565) ellipse (4.68 and 1.80);
  \draw (8.80,20.565) node {\hyperref[sec:ft-state-distillation]{Stabilizer State} \\ \hyperref[sec:ft-state-distillation]{Preparation Gadget}};
\end{scope}
\begin{scope}
  \definecolor{strokecol}{rgb}{0.0,0.0,0.0};
  \pgfsetstrokecolor{strokecol}
  \draw (24.10,11.68) ellipse (5.04 and 2.16);
  \draw (24.10,11.8) node {\hyperref[subsec:main-thm-vanishing-threshold]{FT Scheme} \\ \hyperref[subsec:main-thm-vanishing-threshold]{\(\frac{1}{\polyloglog \frac{WD}{\epsilon_L}}\) threshold}};
\end{scope}
\begin{scope}
  \definecolor{strokecol}{rgb}{0.0,0.0,0.0};
  \pgfsetstrokecolor{strokecol}
  \draw (12.90,11.68) ellipse (4.32 and 1.80);
  \draw (12.90,11.68) node {\hyperref[lemma:compiling]{Compilation and} \\ \hyperref[lemma:compiling]{Serialization}};
\end{scope}
\begin{scope}
  \definecolor{strokecol}{rgb}{0.0,0.0,0.0};
  \pgfsetstrokecolor{strokecol}
  \draw (12.50,3.76) ellipse (5.40 and 2.16);
  \draw (12.50,3.76) node {\(O(1)\) Space \\ Concatatenated Code FT \\ \cite{yamasaki2024time}};
\end{scope}
\begin{scope}
  \definecolor{strokecol}{rgb}{0.0,0.0,0.0};
  \pgfsetstrokecolor{strokecol}
  \draw (24.40,3.76) ellipse (4.68 and 2.16);
  \draw (24.40,4.0) node {\hyperref[subsec:main-thm-constant-threshold]{Main Result} \\ \hyperref[subsec:main-thm-constant-threshold]{(const. threshold)}};
\end{scope}
\end{tikzpicture}
\end{figure}

\section{Model of computation and weight enumerator formalism}\label{sec:weight-enum-formalism}
Aharonov and Ben-Or \cite{aharonov1999fault} prove the fault tolerance of their scheme by first working with the set of faulty operations in an adversarial sense with no randomness.
However, the analysis is much more fine-grained than simply the weight of the set: The set of faulty operations is required to exclude certain ``bad'' sets.
This allows one to prove thresholds against many types of noise including stochastic noise.
In this section, we introduce a generalization of this technique inspired by the need to combine many different types of fault-tolerant gadgets in non-trivial ways to construct new fault-tolerant gadgets.
The key observation is that we can associate polynomials (over \(\mathbb{R}_+\)) to the family of bad sets and work with the polynomials similar to probabilities while remaining in the adversarial noise paradigm.
These polynomials are a special case of the weight enumerator polynomials from coding theory, so here we also refer to them as weight enumerators.

We begin by defining our model of computation formally.
This model runs a constant depth classical circuit after every layer in the quantum circuit.
Note that width refers to the number of qubits in the circuit.

\begin{definition}[Adaptive quantum circuit]
    An adaptive Clifford+$\CCZ$ quantum circuit of depth \(D\) and space \(W\) is described by a length-\(w_t\) list of quantum operations \((\mathsf{O_{t,1}}, \dots, \mathsf{O_{t,w_t}})\) for each timestep \(t \in [D]\) such that each quantum operation has disjoint support as well as a sequence of depth-$O(1)$ classical circuits \(C_0, \dots, C_D\).
    The valid quantum operations are one and two-qubit Clifford gates, $\CCZ$ gates, destructive Pauli basis measurements, and state initialization.

    Let \(t \in [D]\) be an arbitrary timestep.
    Let \(M_t = O(W)\) be the number of measurements in the quantum circuit, and let \(\vec{m}_t = (m_{1}, \dots, m_{M_t}) \subseteq \F^{M_t}\) be their result.
    The computation maintains a classical register of polynomial size \(W_C = O(\poly WD)\) with state \(\vec{x}_t \in \F^{W_C}\).
    Each classical circuit \(C_t\) has depth \(O(1)\) and access to both the classical register and measurement results \((\vec{x}_{t}, \vec{m}_{t})\).
  
    Let \(\mathsf{INPUT}\) be the input string to the computation.\footnote{For example, an encoding of a number to be factored.}
    At the start of the computation, the \(W\) qubits are initialized to \(\ket{0}\), and the classical register is initialized to the execution of \(C_0\) on the input string \(x_1 \gets C_0(\mathsf{INPUT})\).
    To execute timestep \(t \in [D]\), we first apply the quantum operations.
    The unitary quantum operations \(u_t \subseteq [w_t]\) are applied conditioned on the classical register \(\vec{x}_{t} = (x_{t,1}, x_{t,2}, \dots)\).
    That is, for each unitary operation \(i \in u_t\), we apply \(\mathsf{O_{t,i}}\) if and only if \(x_{t,i} = 1\).
    The non-unitary quantum operations are applied unconditionally with the measurements producing the measurement outcomes \(\vec{m}_t\).
    Finally, we update the classical register with the output of the classical circuit \(\vec{x}_{t+1} \gets C_t(\vec{x}_t, \vec{m}_{t})\) associated with timestep \(t\).
    The output of the computation is the first \(n\) bits of the classical register after executing all \(D\) timesteps in order.
\end{definition}

Next, we provide some preliminary definitions that will be utilized later.
\subsection{Preliminaries}\label{sec:prelim}
We will use \(\mathcal{H}_2\) to refer to the standard two-dimensional Hilbert space of a qubit.
For a Hilbert space \(\mathcal{H}\), we use \(\mathcal{D}(\mathcal{H})\) to refer to the set of density operators supported on \(\mathcal{H}\).

\begin{definition}[Quantum LDPC code]
  A quantum CSS code \(\qcode =CSS(H_X,H_Z)\) is said to be \(\Delta\)-qLDPC if \(H_X\) and \(H_Z\) have row and column weight at most \(\Delta\).
\end{definition}
Our quantum codes will implicitly carry an encoding map from a dim \(2^k\) logical Hilbert space to a dim \(2^k\) subspace.
This choice is almost arbitrary\footnote{As long as it takes computational basis states to computational basis states.} with the exception of the constant quantum depth encoding circuit for the computational code (\cref{lem:nonFT-state-prep}) and the magic state distillation code (\cref{lemma:prod-codewords}).
For this reason we will refer to the encoding map \(\phi_C\) of \(\qcode\) as \emph{the encoding map}.

\begin{definition}
  A classical code \(\mathcal{C}\subseteq \F^n\) defined by the check matrix \(H \in \F^{r \times n}\) is locally testable with soundness \(\rho\) and locality \(\Delta\) if \(H\) has row and column weight at most \(\Delta\) and for all \(x \in \F^{n}\),
  \begin{align}
    \frac{|H x|}{r} \ge \rho \frac{d(x, \mathcal{C})}{n}.
  \end{align}
\end{definition}

\begin{definition}[Quantum LTC~\cite{aharonov2015quantum, eldar2017local}]
  A quantum CSS code \(\qcode =CSS(H_X,H_Z)\) is said to be \((\rho,\Delta)\)-qLTC if \(H_X\) and \(H_Z\) are both locally testable with soundness \(\rho\) and locality \(\Delta\).
\end{definition}
Note that a quantum code that is \((\rho,\Delta)\)-qLTC is also \(\Delta\)-qLDPC.

\begin{definition}[Stabilizer-reduced weight]
  For a stabilizer code with stabilizer set \(S\) and a Pauli error \(E\), we say that \(E\) is stabilizer-reduced if it is minimal with respect to the application of stabilizers.
  That is, for every \(s \in S\), the weight of \(|sE| \ge |E|\). For a general Pauli error $P$, the stabilizer-reduced weight $|P|_R$ is the weight of a stabilizer-reduced error equivalent to $P$ under applications of stabilizers.
\end{definition}

Let \(\ket{\Phi}=\frac{1}{\sqrt{2}}\left(\ket{00} + \ket{11}\right)\) denote the Bell state.
We recall the standard gate teleportion scheme.
\begin{definition}[Gate teleportation \cite{gottesman1999demonstrating}]
  For an arbitrary \(n\)-qubit unitary \(U\), the state \(\ket{U} = I\otimes U \ket{\Phi}^{\otimes n}\) (\(U\) acting on one qubit from each Bell pair) allows us to perform a teleported application of \(U\) on an input state $\ket{\psi}$: By applying the standard quantum teleportation circuit (using \(X\) and \(Z\) basis measurement and \(\CNOT\)) between a register \(\ket{\psi}\) and \(\ket{U}\), and applying a \(U\)-dependent fix-up operation (\(UX_i U^\dagger\) and \(UZ_i U^\dagger\)) conditioned on the measurement outcomes of the teleportation circuit, we are left with \(U\ket{\psi}\).
\end{definition}

\begin{figure}
    \centering
    \includegraphics[width=0.55\linewidth]{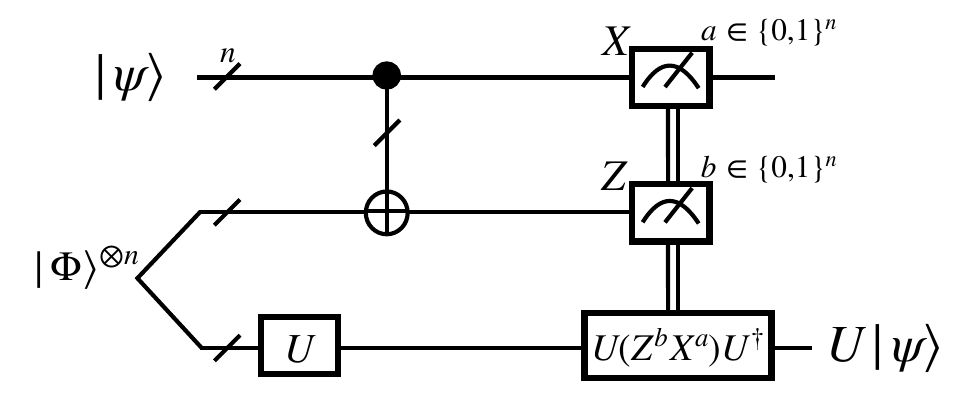}
    \caption{The $n$-qubit gate teleportation circuit. If $U$ is in the $k$-th level of the Clifford hierarchy, then the correction is in level $k-1$. We will use this circuit and the stabilizer resource state $I \otimes U\ket{\Phi}^{\otimes n}$ to implement Clifford gates in our construction.}
    \label{fig:Bell}
\end{figure}

\begin{figure}
    \centering
    \includegraphics[width=0.45\linewidth]{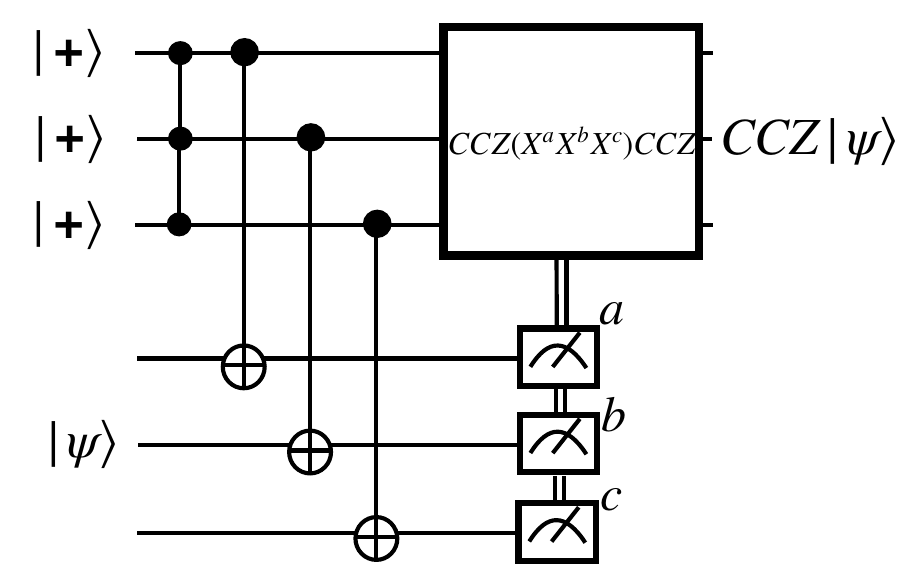}
    \caption{Gate teleportation for $\CCZ$ gate using the magic state $\ket{\CCZ}=\CCZ\ket{+++}$. The fix-up has the form of Pauli operations and $\CZ$ gates. The latter are in turn implemented using~\Cref{fig:Bell}.}
    \label{fig:ccz}
\end{figure}

In particular, if \(U\) is in the \(k\)-th level of the Clifford hierarchy, then the fix-up operation is in the \((k-1)\)-th level.
We will only use \(\CCZ\), a third-level gate, and Clifford gates for which the fix-up is Pauli.
For Clifford gates, the resource state is a stabilizer state and the fix-up includes Pauli operations. In this work we use this standard scheme to perform Clifford gates on the computation qLDPC code. To perform $\CCZ$, we instead use the magic state $\CCZ\ket{+++}$ and the circuit in~\Cref{fig:ccz}.

\begin{remark}[Clifford+T]
  The choice to use Clifford+$\CCZ$ instead of Clifford+T may seem unconventional, but we note that any circuit using the Clifford+T gateset may be converted to the Clifford+$\CCZ$ gateset with \(O(\log (WD/\varepsilon))\) \emph{additive} time cost and constant overhead, by first approximately preparing a \emph{single} \(T\) state to \(1/\poly(WD/\varepsilon)\) error and then using the (exact) catalyzed conversion \(\ket{\CCZ}+\ket{T} \to 2\ket{T}\) to create many \(\ket{T}\) states
  \cite{beverland2020lower}.
\end{remark}

\subsection{Noise model}

We are now ready to define what it means for a circuit to be executed in a faulty way.
Note that we have not yet defined a distribution (or similar) over the ways that a circuit can be faulty.
A major feature is that we can separate the proof of correctness from the distribution\footnote{Most generally, it can be the case that the faults are not drawn from a probability distribution at all (e.g. coherent noise), but we can still upper bound the contribution from ``bad'' fault paths in the output distribution.} the faults are drawn from.
\begin{definition}[Fault]
  For an adaptive quantum circuit \(C = C_D C_{D-1} \dots C_1\) of depth \(D\) and width \(W\), a fault \(\fault\) is a sequence of \(W\)-qubit superoperators \((\fault_1, \fault_2, \dots, \fault_D)\).
  We use \(C[\fault]\) to denote the execution of the circuit \(C\) subject to the fault \(\fault\). That is,
  \begin{align}
      C[\fault] = \fault_D C_D \fault_{D-1} C_{D-1} \dots \fault_1 C_1.
  \end{align}
  We say that a fault is \emph{physical} if each superoperator is completely-positive and trace-preserving (CPTP).

  For a sub-circuit \(C'\) of \(C\), \(C'[\fault]\) should be interpreted as the restriction of \(\fault\) to the locations of \(C'\).
  Whenever such a restriction appears, \(\fault\) will always factorize such that the restriction is well defined.
\end{definition}

We remark that our notion of a fault (singular) approximately corresponds to many faults (plural) in some modern works.
For convenience, we extend multiplication to faults, coordinate-wise.
That is, for two faults \(\fault[f]\) and \(\fault[g]\), \(\fault[f]\fault[g] = (\fault[f]_1 \fault[g]_1, \fault[f]_2 \fault[g]_2, \dots)\).

A particularly special case of a fault will be that of a \emph{Pauli fault}.
Pauli faults are easier to work with for reasons that will be made clear shortly.
\begin{definition}[Pauli superoperator]
  A superoperator is said to be a Pauli superoperator if it is of the form \(\rho \mapsto \alpha A \rho B\) for some Pauli operators \(A\) and \(B\) and complex coefficient \(\alpha\).
  We will further say it is a diagonal Pauli superoperator if it is of the form \(\rho \mapsto \alpha A \rho A\).
\end{definition}

\begin{definition}[Pauli faults]
  For a fault \(\fault = (\fault_1, \fault_2, \dots, \fault_D)\), \(\fault\) is said to be a \emph{Pauli fault} (diagonal Pauli fault) if every superoperator \(\fault_t\) in the sequence is a Pauli superoperator (diagonal Pauli superoperator).
\end{definition}

\begin{remark}
  Since Pauli matrices form a complete basis for the space of matrices, it is clear that every fault can be written as a linear combination of Pauli faults.
  For a circuit subject to a physical fault, we can decompose it into a sum over executions of the circuit subject to Pauli faults.
  We will use this fact to reduce the case of general faults to that of Pauli faults.
  Since Pauli faults are tensor products of single-qubit operators (i.e. factorizes), this makes it convenient to prove properties of circuits in isolation.
\end{remark}

We now define the locations of a circuit.
\begin{definition}[Location]
  A location of an adaptive quantum circuit is a tuple of a timestep \(t \in [D]\) and a set of qubits \(q \subseteq [W]\) such that a quantum operation (including identity) supported on the qubits \(q\) is performed at time \(t\). For each \(t \in [D]\), every qubit must be contained in exactly one location.
\end{definition}
We are now ready to define the analog of the support of a fault.
\begin{definition}[Fault path]
  For an adaptive quantum circuit \(C\) and a fault \(\fault\), the fault path \(\supp\fault\) of \(\fault = (\fault_1, \dots,\fault_D)\) is the subset of locations of \(C\) that \(\fault\) is supported on.
  I.e. a location \((t \in [D], X\subseteq [W])\) of \(C\) is in the fault path if \(\supp \fault_t \cap X \ne \emptyset\).
  Such a location is said to be \emph{faulty}.
\end{definition}

\begin{definition}[Local stochastic noise]\label{def:local-stochastic-noise}
  For an adaptive quantum circuit \(C\) with the set of locations \(L\), the random physical fault \(\mathbf{f}\) is said to be distributed according to an \(\epsilon\)-locally stochastic faults model if, for all \(S \subseteq L\),
  \begin{align}
    \Pr(S \subseteq \supp \mathbf{f}) \le \epsilon^{|S|}.
  \end{align}
\end{definition}

For a quantum circuit \(C\) with depth \(D\) and space \(W\), we define \(|C| = WD\).
Furthermore, we say a circuit is polynomially bounded if \(D = O(\poly W)\).

\subsection{Weight enumerators}

In order to combine fault tolerant gadgets using different techniques such as code concatenation and qLDPC, it becomes convenient to define a notion of a \emph{sparse} error or fault-path that differs from \cite{aharonov1999fault}.
Roughly, these are the errors that are benign and do not affect the final outcome.
We will define these sets by requiring that they exclude certain ``bad'' subsets.
Depending on how these bad subsets are defined, this imposes structure on the error or fault path that can be used to prove correctness of circuit gadgets in the presence of such errors / fault paths.

We now define a notion analogous to a ``sparse'' set of \cite{aharonov1999fault}.
\begin{definition}[Avoiding sets]
  For a set \(\Omega\) and a family of subsets \(\mathcal{F}\subseteq P(\Omega)\), referred to as \emph{the bad sets}, a subset \(X \subseteq \Omega\) is said to be \(\mathcal{F}\)-avoiding if it does not contain an element of \(\mathcal{F}\).
  That is, \(X \subseteq \Omega \) is \(\mathcal{F}\)-avoiding if and only if
  \begin{align*}
    \forall F \in \mathcal{F},~ F \not\subseteq X~.
  \end{align*}
\end{definition}

To this family of sets we can associate a weight enumerator polynomial.
This tracks how many elements of the family of each size there are.
\begin{definition}[Weight enumerators]\label{def:weight-enumerator}
  For a finite set \(\Omega\) and a family of subsets \(\mathcal{F} \subseteq P(\Omega)\), it will be convenient to associate a polynomial \(\weightenum{\mathcal{F}}{x}\) to \(\mathcal{F}\) defined as the sum
  \begin{align}
    \weightenum{\mathcal{F}}{x} = \sum_{w=0}^{\infty} A_w x^w,
  \end{align}
  where \(A_w = |\left\{F \in \mathcal{F} \mid |F| = w\right\}|\) is the number of elements of \(\mathcal{F}\) of weight \(w\).
\end{definition}

Note that our definition is a special case of the weight enumerator polynomial from classical coding theory.
The major motivation of this definition is that \(\weightenum{\mathcal{F}}{x}\) is defined such that, for a random variable \(E\) taking values in \(P(\Omega)\) that is \(\epsilon\)-local stochastic, the probability that \(E\) is not \(\mathcal{F}\)-avoiding is bounded by the evaluation \(\weightenum{\mathcal{F}}{x}\) at \(\epsilon\). In other words,
\begin{align}
  \forall S \subseteq \Omega, ~\Pr(S \subseteq E) \le \epsilon^{|S|} \Rightarrow \Pr(\text{\(S\) is not \(\mathcal{F}\)-avoiding}) \le \sum_{f \in \mathcal{F}} \epsilon^{|f|} = \weightenum{\mathcal{F}}{\varepsilon}.
\end{align}

Having motivated the introduction of weight enumerator polynomials, a natural question is what multiplication and addition of polynomials correspond to.
We begin by defining a sort of sum and product operations on the families of bad sets which produce new families of bad sets with weight enumerator given by the sum or product (\cref{lemma:enumerator-ring}).
The addition operation is essentially the union-bound from probability.
The product operation is slightly more subtle: It is analogous to the probability of simulataneous failure of independent gadgets, but recall that the probability of failure is not entirely independent in our error model.
However, we can recover a notion of independence whenever the families of bad sets do not coincide.

\begin{definition}[Ring of bad sets]
  For a set \(\Omega\) with the decomposition \(\Omega_1\cup \Omega_2 = \Omega\) and families of subsets \(\mathcal{F}_1\subseteq P(\Omega_1)\) and \(\mathcal{F}_2\subseteq P(\Omega_2)\), we define two operations between \(\mathcal{F}_1\) and \(\mathcal{F}_2\) to arrive at a subset \(\mathcal{F}\subseteq \Omega\).
  Let \(\iota_1\) (\(\iota_2\)) be the canonical inclusion map from \(\Omega_1\) (\(\Omega_2\)) to \(\Omega\).

  The first operation is a sort of addition operation.
  \begin{align}
    \mathcal{F}_1\boxplus \mathcal{F}_2 := \iota_1(\mathcal{F}_1) \cup \iota_2(\mathcal{F}_2),
  \end{align}
  where by \(\iota_1(\mathcal{F}_1)\), we mean the application of \(\iota_1\) to each element of the family \(\mathcal{F}_1\).

  When \(\Omega_1\) and \(\Omega_2\) are disjoint, that is \(\Omega = \Omega_1\sqcup \Omega_2\), we define a multiplication operation.%
  \begin{align}
    \mathcal{F}_1\circledast \mathcal{F}_2 := \{\iota_1(f_1) \cup \iota_2(f_2) \mid f_1 \in \mathcal{F}_1,~f_2 \in \mathcal{F}_2\}.
  \end{align}
\end{definition}
First note that the name addition and multiplication are deserved: Multiplication distributes over addition.
Informally, a set \(X \subseteq \Omega\) is \(\mathcal{F}_1\boxplus \mathcal{F}_2\)-avoiding if \(X|_{\Omega_1}\) is \(\mathcal{F}_1\)-avoiding \emph{and} \(X|_{\Omega_2}\) is \(\mathcal{F}_2\)-avoiding whereas \(X\) is \(\mathcal{F}_1\circledast \mathcal{F}_2\)-avoiding if \(X|_{\Omega_1}\) is \(\mathcal{F}_1\)-avoiding \emph{or} \(X|_{\Omega_2}\) is \(\mathcal{F}_2\)-avoiding.
Later, the addition will be used for union-bounds while the second operation will be used for upper bounding the probability of a large simultaneous failure.

Conveniently, the weight enumerators are well behaved under these operations
\begin{prop}\label{lemma:enumerator-ring} It holds that
  \begin{align}
    \weightenum{\mathcal{F}_1\boxplus \mathcal{F}_2}{x} &= \weightenum{\mathcal{F}_1}{x} + \weightenum{\mathcal{F}_2}{x}, \\
    \weightenum{\mathcal{F}_1\circledast \mathcal{F}_2}{x} &= \weightenum{\mathcal{F}_1}{x} \weightenum{\mathcal{F}_2}{x}. 
  \end{align}
\end{prop}
\begin{proof}
  A weight \(w\) element of \(\mathcal{F}_1\boxplus \mathcal{F}_2\) is a weight \(w\) element of \(\mathcal{F}_1\) or of \(\mathcal{F}_2\), so the coefficients of the same degree add.
  To prove the second property, note that any weight \(w\) element of \(\mathcal{F}_1\circledast \mathcal{F}_2\) must be the disjoint union of a weight \(w_1\) element of \(\mathcal{F}_1\) and a weight \(w_2\) element of \(\mathcal{F}_2\) where \(w = w_1+w_2\).
\end{proof}

The product operation will be indispensible when constructing new fault-tolerant gadgets from less resilient older ones.
In several places such as recursive simulations (concatenation) and resource state distillation, we will have fault-tolerant gadgets fail in ways that are recoverable.
Because of this, we will need to analyze the fault sets in a recursive way. Unfortunately, a large degree of generality is required for all potential use-cases. This motivates the following definition of a composition operation.

\begin{definition}[Composition]\label{def:composition}
  For a proposition \(Q\) and a set \(X\), let \(\mathbb{I}_{Q}[X]\) denote \(X\) when \(Q\) holds and \(\emptyset\) when \(\lnot Q\) holds.

  Fix a set \(\Omega\) and \(\mathcal{F}\subseteq P(\Omega)\). For ease of notation, identify \(\Omega \simeq [n]\). Also, consider sets \(\{\omega_i\}_{i \in \Omega}\) and families of sets \(\{\mathcal{S}_i \subseteq P(\omega_i)\}_{i \in \Omega}\) that are indexed by elements of $\Omega$.
  We define an operation \(\mathcal{F} \bullet \{\mathcal{S}_i\}_i \subseteq P\left(\bigsqcup_{i\in\Omega} \omega_i\right)\) where
  \begin{align*}
    \mathcal{F} \bullet \{\mathcal{S}_i\}_i := \boxplus_{f \in \mathcal{F}} \left(\mathbb{I}_{n \in f} [\mathcal{S}_n] \circledast \mathbb{I}_{(n-1) \in f} [\mathcal{S}_{n-1}]  \dots \mathbb{I}_{2 \in f} [\mathcal{S}_2] \circledast \mathbb{I}_{1 \in f} [\mathcal{S}_1]\right).
  \end{align*}
  Note that if \(\mathcal{S}_i = \mathcal{S} \subseteq P(\omega)\) for some set \(\omega\), then \(\mathcal{F} \bullet \{\mathcal{S}_i\}_i \subseteq P(\Omega \times \omega)\) corresponds to all sets which for some element \(f \in \mathcal{F}\) are elements of \(\mathcal{S}\) on each row of \(\Omega \times \omega\) where \(f\) is non-trivial.
We will use the notation \(\mathcal{F} \bullet \mathcal{S}\) for this special case.
\end{definition}

Later, when defining a recursive simulation, \(\Omega\) will label circuit locations to be simulated while each \(\omega_{i \in \Omega}\) will label locations in the simulation of the operation associated with location \(i\).
Different gadgets may have different sets of locations which demands the additional complexity in the definition.
Readers can simply consider code concatenation as a guide:
The outer and inner code blocks have corresponding bad sets \(\mathcal{F}_{\mathrm{outer}}\) and \(\mathcal{F}_{\mathrm{inner}}\).
Informally, the family produced by the \(\bullet\) operation takes a bad set \(f \in \mathcal{F}_{\mathrm{outer}}\) of the outer code block and expands each element \(e \in f\) into one of the bad sets \(f' \in \mathcal{F}_{\mathrm{in}}\) of the inner code block, so that if a subset of qubits of the concatenated code is \(\mathcal{F}_{\mathrm{outer}} \bullet \mathcal{F}_{\mathrm{inner}}\)-avoiding, then the set is \(\mathcal{F}_{\mathrm{inner}}\)-avoiding on all inner code blocks except for an \(\mathcal{F}_{\mathrm{outer}}\) set (of inner code blocks).
More generally, the families of bad sets of the inner set may not be uniform e.g. a gadget of a concatenated code may use different inner code gadgets (performing different operations) with different locations and bad sets.

\begin{remark}
  Frequently, for a weight enumerator \(\weightenum{\mathcal{F}}{x}\), we will only have an upper bound on \(\weightenum{\mathcal{F}}{x}\) on some interval \([0,c]\) for \(c \in (0,1]\).
  \(c\) should be thought of as a multiple (e.g. \(1/2\)) of a threshold value.
\end{remark}
  
\begin{prop}[Composition of weight enumerators]\label{lemma:composition-upper-bound}
  Using the variables as defined in \cref{def:composition}, when there exists a polynomial \(p(x)\) such that on some interval \(x \in I \subseteq \mathbb{R}_+\) for all \(i \in \Omega\), \(\weightenum{\mathcal{S}_i}{x} \le p(x)\),
  \begin{align}
    \weightenum{\mathcal{F} \bullet \{\mathcal{S}_i\}_i}{x} \le \weightenum{\mathcal{F}}{p(x)}
  \end{align}
  for \(x \in I\).
\end{prop}
\begin{proof}
  \(\mathcal{F} \bullet \{\mathcal{S}_i\}_i\) is constructed as a sum over a product for each \(f \in \mathcal{F}\).
  Using the weight enumerator sum rule, the weight enumerator is the sum of weight enumerators for each product.
  Each product has \(|f|\) terms and the corresponding weight enumerator can be evaluated using the weight enumerator product rule.
  Applying the restriction \(x \in I\) and the upper bound \(p(x)\), each term in the sum has upper bound \(p(x)^{|f|}\).
  For each \(d \in \mathbb{N}\), the number of terms in the sum proportional to \(p(x)^d\) correspond to the number of elements of \(\mathcal{F}\) of weight \(d\).
\end{proof}

\subsection{Gadgets}\label{subsec:gadgets}
It will be convenient to group together (physical) qubits into ``blocks.''
Roughly, a block is simply a bookkeeping method to refer to collections of qubits that correspond to the physical qubits of a quantum error correcting code.
To reduce the notation, we intentionally leave implicit a choice of labeling of the qubits.
Recall that for a stabilizer code \(\qcode\) encoding \(k\) qubits into \(n\) qubits, we also leave implicit a choice of encoding map \(\phi_\qcode\) from a \(k\)-qubit Hilbert space to an \(n\)-qubit Hilbert space.

\begin{definition}[Code block]
  For a stabilizer code \(\qcode\) on \(n\) qubits and a set of qubits \([M]\), a block of \([M]\) is an \(n\)-qubit subset of \([M]\) that will be treated as forming an encoded state of \(\qcode\). A block will be further equipped with bad sets \(\mathcal{F}\subseteq P([n])\).
  Such a block will be referred to as a \((\qcode , \mathcal{F})\)-block.
  For a set of qubits \(A \subseteq [M]\), \(A\) is said to be \emph{sparse} on the \((\qcode , \mathcal{F})\)-block \(B \subseteq [M]\) if \(A \cap B\) is \(\mathcal{F}\)-avoiding.
\end{definition}

We would like a notion to compare if the intermediate computational state is equivalent to the correct computional state in a sense that is robust to the presence of benign errors.
Here, we introduce a notion of deviation inspired by \cite{aharonov1999fault, aharonov2008fault} and \cite{kitaev1997quantum}, but suitably generalized for our purposes.
\begin{definition}[Deviation]\label{def:deviation}
  For a set of qubits \(A\), a family of bad sets \(\mathcal{F}\), and two states \(\rho\), \(\sigma\) on \(A\).
  \(\rho\) is said to be \(\mathcal{F}\)-deviated from \(\sigma\) if there exists an \(\mathcal{F}\)-avoiding set \(B \subseteq A\) and a superoperator \(\mathcal{E}_B\) supported only on \(B\) such that \(\rho = \left(\mathcal{I}_{A \setminus B}\otimes \mathcal{E}_B\right)(\sigma)\).
  For a quantum code \(\qcode\), a state \(\rho\) is said to be \(\mathcal{F}\)-deviated from \(\qcode\) if there exists a (possibly mixed) code state \(\sigma\) in \(\qcode\) such that \(\rho\) is \(\mathcal{F}\)-deviated from \(\sigma\).

  We use the less general term ``Pauli \(\mathcal{F}\)-deviated'' (``diagonal Pauli \(\mathcal{F}\)-deviated'') when the superoperator \(\mathcal{E}_B\) is a Pauli superoperator (diagonal Pauli superoperator)
\end{definition}

\begin{remark}
  Our notion of deviation most closely corresponds to the notion of a ``many-to-one'' code of \cite{kitaev1997quantum}.
  A ``many-to-one'' code is the set of states that can be reached from a codestate of a quantum code by application of a superoperator on at most \(l \in \mathbb{N}\) qubits.
  I.e. the space of states that are damaged but not irrecoverably so. Here
  ``many-to-one'' refers to the recovery map sending damaged states to a unique codestate of a standard quantum code.
  
  We avoid comparing the partial traces of states as in \cite{aharonov2008fault} since (i) the reference state is required to be pure (\cite[Section 4.3.3]{aharonov2008fault}) in order for there to exist a superoperator taking one state to another (\cite[Claim 1]{aharonov2008fault}) and (ii) the relation should not be symmetric.
  This difference in definition does not substantially modify the argument of \cite{aharonov1999fault,aharonov2008fault}.
\end{remark}

It will be convenient to refer to a smaller unit of an adaptive quantum circuit which we will call a \emph{gadget}.
A gadget may take a quantum input and outputs the final state of the quantum and classical registers \emph{including the measurement history}.
The requirement that the gadget also preserves the measurement history is a technical one that is needed to deal with non-Pauli errors on inputs to the gadget since the presence of only Pauli faults does not necessarily imply the output error is Pauli (see~\cref{lemma:deterministic-errors}).
At the final step of our analysis, we will discard all measurements except for the result of the logical measurements.
In our case, the gadgets will often operate on a subset of the classical and quantum registers and will be invoked in parallel.
This will be implicit.

The gadgets in our construction will take as input an encoded state and output a state (on the quantum register), possibly encoded in a different code.
This notion of a gadget is somewhat similar to \cite{kitaev1997quantum}.
We now define a fault-tolerant gadget and give several useful propositions related to various forms of composition.
We recommend that readers refer to the diagrams in parallel with the text.
\begin{definition}[Fault-tolerant gadget]\label{def:ft-gadget}
  Fix stabilizer codes \(\qcode\) (\(\qcode'\)) on \(n\) qubits (\(n'\) qubits) encoding \(k\) qubits (\(k'\) qubits) and their respective encoding maps \(\phi_\qcode\) (\(\phi_{\qcode'}\)).
  A fault-tolerant gadget \(\gadget\) operates on a \((\qcode,\mathcal{F})\)-block and returns a \((\qcode',\mathcal{F}')\)-block.
  To reduce notation, we will suppress the dependence on \(\qcode\), \(\mathcal{F}\), and the corresponding primed analogs.

  Consider a gadget \(\gadget\) with a set of locations \([L]\) and a family of bad sets \(\mathcal{G}\subseteq P([L])\).
  Let \(\mathbf{f}\) be a fault.
  For a quantum operation \(\qOp\) mapping unencoded states from $\mixedstate{k}$ to $\mixedstate{k'}$, \(\gadget\) is said to be a \((\qOp, \mathcal{G})\)-FT gadget if the following holds when the fault path \(\supp\fault \subseteq [L]\) of \(\fault\) is \(\mathcal{G}\)-avoiding:
  Without loss of generality, assume the gadget operates on the first \(n\) qubits.
  Then, for any \(m \in \mathbb{N}\), and any state\footnote{See \cref{rmk:reference-system}} \(\rho \in \mixedstate{(k+m)}\), if the input state \(\sigma\in \mixedstate{(n+m)}\) is \(\mathcal{F}\)-deviated from \(\phi_\qcode\otimes \mathcal{I}_m(\rho)\), then the output state \(\sigma' = \gadget[\mathbf{f}]\otimes \mathcal{I}_m(\sigma)\) is \(\mathcal{F}'\)-deviated from \(\left(\phi_{\qcode'}\circ \qOp\right)\otimes \mathcal{I}_m (\rho)\).
  
  A gadget is further said to be \emph{friendly} if, regardless of the input state, when \(\supp\fault\) is \(\mathcal{G}\)-avoiding, there always exists a state \(\rho' \in \mixedstate{(k+m)}\) such that the output state is \(\mathcal{F}'\)-deviated from \((\phi_{\qcode'} \circ \qOp)\otimes \mathcal{I}_m(\rho')\).
\end{definition}

Diagramatically, the following diagram commutes: It is equivalent to apply the operation on the logical information and encode it, or to apply the gadget, and the noisy computation is always ``close'' to the noiseless one.
\begin{align*}
  \begin{tikzcd}[row sep=large, column sep=huge,ampersand replacement=\&]
    \rho \arrow[r,mapsto,"\qOp \otimes \mathcal{I}_{m}"] \arrow[d,mapsto,"\phi_{\qcode}\otimes \mathcal{I}_{m}"] \& \rho' \arrow[d,mapsto,"\phi_{\qcode'}\otimes \mathcal{I}_{m}"]\\
    \overline{\rho} \arrow[r,mapsto,"\gadget \otimes \mathcal{I}_{m}"] \arrow[d,rightsquigarrow,"\text{\(\mathcal{F}\)-deviated}"] \& \overline{\rho'}\arrow[d,rightsquigarrow,"\text{\(\mathcal{F}'\)-deviated}"] \\
    \sigma \arrow[r,mapsto,"\gadget {[}\mathbf{f}{]} \otimes \mathcal{I}_{m}"] \& \sigma'
  \end{tikzcd}
\end{align*}

\begin{remark}
    The friendly property may seem somewhat unusual, but it is required in order for gadgets to be recursively simulated: In a recursive simulation, lower level gadgets may fail, leaving no guarantees on the output state.
    The friendly property allows the gadget following a failed gadget to return the (arbitrarily damaged) state to a well defined logical state, so that the higher simulation level can correct the resulting logical error.
\end{remark}

\begin{remark}\label{rmk:reference-system}
  The presence of the reference system is needed in order for the gadget to be inserted in a larger circuit which may have many more qubits.
  The proof of \cref{lemma:gadget-side-by-side} utilizes this property to show that two fault-tolerant gadgets executed in parallel is itself a fault-tolerant gadget.
\end{remark}

A set of FT gadgets are said to be \emph{compatible} if the parameters of the blocks of the inputs and outputs are matching (or trivial).
That is, for some \((\qcode,\mathcal{F})\), the inputs and outputs of all gadgets in the set are \((\qcode,\mathcal{F})\)-blocks.

In some cases, it will be more convenient to restrict to gadgets that satisfy the fault-tolerant gadget property for Pauli noise.
While we will consider general noise, the analysis of a circuit subject to general noise will be reduced to the Pauli noise case, so that the gadgets themselves only need to be fault-tolerant to Pauli noise.
\begin{definition}[Pauli fault-tolerant gadgets]\label{def:pauli-ft-gadget}
  A Pauli fault-tolerant gadget \(\gadget\) is a gadget satisfying a weaker version of \cref{def:ft-gadget} with all faults restricted to be Pauli faults and the input and output taken to be Pauli deviated.
\end{definition}

What follows is a formalization of properties that allow us to assembly complicated gadgets out of simple ones.
The first proposition allows us to compose gadgets.
The second proposition allows us to execute two gadgets side-by-side in parallel on separate inputs.
The fault analysis in both of these cases will be a union-bound.

\begin{prop}[Composition of Pauli fault-tolerant gadgets]\label{lemma:gadget-composition}
  For a \((\qOp_1, \mathcal{G}_1)\)-Pauli FT gadget \(\gadget_1\) mapping a \((\qcode_1, \mathcal{F}_1)\)-block to a  \((\qcode_2, \mathcal{F}_2)\)-block and a \((\qOp_2, \mathcal{G}_2)\)-Pauli FT gadget \(\gadget_2\) mapping a \((\qcode_2, \mathcal{F}_2)\)-block to a  \((\qcode_3, \mathcal{F}_3)\)-block, the composition is a \((\qOp_2 \circ \qOp_1, \mathcal{G}_1 \boxplus \mathcal{G}_2)\)-Pauli FT gadget that maps a \((\qcode_1, \mathcal{F}_1)\)-block to a \((\qcode_3, \mathcal{F}_3)\)-block.
\end{prop}
\begin{proof}
    Let \(\fault\) be a fault with a \(\mathcal{G}_1 \boxplus \mathcal{G}_2\)-avoiding fault path.
    Let \(\sigma_1\) be the input state such that there exists \(\rho_1\) such that \(\sigma_1\) is \(\mathcal{F}_1\)-deviated from \(\phi_{\qcode_1}(\rho_1)\).
    Consider the following diagram:
    \begin{align*}
        \begin{tikzcd}[row sep=large, column sep=huge,ampersand replacement=\&]
        \rho_1 \arrow[r,mapsto,"\qOp_1\otimes \mathcal{I}_{m}"] \arrow[d,mapsto,"\phi_{\qcode_1}\otimes \mathcal{I}_{m}"] \& \rho_2 \arrow[d,mapsto,"\phi_{\qcode_2}\otimes \mathcal{I}_{m}"] \arrow[r,mapsto,"\qOp_2\otimes \mathcal{I}_{m}"] \& \rho_3 \arrow[d,mapsto,"\phi_{\qcode_3}\otimes \mathcal{I}_{m}"]\\
    \overline{\rho}_1 \arrow[r,mapsto,"\gadget_1\otimes \mathcal{I}_{m}"] \arrow[d,rightsquigarrow,"\text{\(\mathcal{F}_1\)-deviated}"] \& \overline{\rho}_2\arrow[d,rightsquigarrow,"\text{\(\mathcal{F}_2\)-deviated}"] \arrow[r,mapsto,"\gadget_2\otimes \mathcal{I}_{m}"] \& \overline{\rho}_3\arrow[d,rightsquigarrow,"\text{\(\mathcal{F}_3\)-deviated}"]\\
        \sigma_1 \arrow[r,mapsto,"\gadget_1{[}\fault{]}\otimes \mathcal{I}_{m}"] \& \sigma_2 \arrow[r,mapsto,"\gadget_2{[}\mathbf{f}{]}\otimes \mathcal{I}_{m}"] \& \sigma_3
        \end{tikzcd}
    \end{align*}
    Using the assumption on the input, the only relations that need to be established is that \(\sigma_2\) is \(\mathcal{F}_2\)-deviated from \(\overline{\rho}_2\) and \(\sigma_3\) is \(\mathcal{F}_3\)-deviated from \(\overline{\rho}_3\).
    Using the definition of an FT gadget, the assumption that the fault path is \(\mathcal{G}_1\)-avoiding, and that \(\sigma_1\) is \(\mathcal{F}\)-deviated from \(\overline{\rho}_1\) establishes \(\sigma_2\) is \(\mathcal{F}_2\)-deviated from \(\overline{\rho}_2\).
    Applying the argument a second time establishes that \(\sigma_3\) is \(\mathcal{F}_3\)-deviated from \(\overline{\rho}_3\).
\end{proof}

It is also the case that two Pauli FT-gadgets acting in parallel also form a new Pauli FT-gadget.
This is the reason for the \(m\)-qubit subsystem in the definition of an FT gadget (\cref{def:ft-gadget})
\begin{prop}[Parallel Pauli FT-gadgets]\label{lemma:gadget-side-by-side}
  For a \((\qOp_1, \mathcal{G}_1)\)-Pauli FT gadget \(\gadget\) mapping a \((\qcode_1, \mathcal{F}_1)\)-block to a  \((\qcode_2, \mathcal{F}_2)\)-block and a \((\qOp', \mathcal{G}')\)-Pauli FT gadget \(\gadget'\) mapping a \((\qcode_1', \mathcal{F}_1')\)-block to a  \((\qcode_2', \mathcal{F}_2')\)-block, \(\gadget \otimes \gadget'\) is a \((\qcode \otimes \qcode', \mathcal{G} \boxplus \mathcal{G}')\)-Pauli FT gadget that maps a \((\qcode_1\otimes \qcode_1', \mathcal{F}_1\boxplus \mathcal{F}_1)\)-block to a \((\qcode_2\otimes \qcode_2', \mathcal{F}_2'\boxplus \mathcal{F}_2')\)-block.
\end{prop}
\begin{proof}
  Recall that in the definition of a FT gadget that the properties continue to hold when the gadget is supported on only a subsystem.
  If \(\fault[f]\) is a Pauli-fault with a \(\mathcal{G} \boxplus \mathcal{G}'\)-avoiding fault path then we can write it in terms of the restrictions \(\fault[f]=\fault[h]\otimes\fault[h']\) such that \(\left(\gadget \otimes \gadget'\right)[\fault[f]] = \gadget[\fault[h]]\otimes \gadget'[\fault[h']]\).
  The claim follows after utilizing the presence of the \(m\)-qubit subsystem in the definition to apply the Pauli-FT property to each gadget individually.
  In other words, we use the definitions show that the following diagram commutes for any \(m \in \mathbb{N}\).
  \begin{align*}
    \begin{tikzcd}[row sep=1cm, column sep=3cm,ampersand replacement=\&]
      \rho \arrow[r,mapsto,"\qOp\otimes \qOp'\otimes \mathcal{I}_{m}"] \arrow[d,mapsto,"\phi_{\qcode_1}\otimes \phi_{\qcode_1'} \otimes\mathcal{I}_{m}"] \& \rho' \arrow[d,mapsto,"\phi_{\qcode_2}\otimes \phi_{\qcode_2'}\otimes \mathcal{I}_{m}"]\\
      \overline{\rho} \arrow[r,mapsto,"\gadget\otimes \gadget'\otimes \mathcal{I}_{m}"] \arrow[d,rightsquigarrow,"\text{\(\mathcal{F}_1\boxplus\mathcal{F}_1'\)-deviated}"] \& \overline{\rho'}\arrow[d,rightsquigarrow,"\text{\(\mathcal{F}_2\boxplus\mathcal{F}_2'\)-deviated}"] \\
      \sigma \arrow[r,mapsto,"\gadget {[\fault[h]]} \otimes \gadget'{[\fault[h']]} \otimes \mathcal{I}_{m}"] \& \sigma'
    \end{tikzcd}
  \end{align*}  
\end{proof}

We are now ready to give an example employing many of the definitions of the last several pages.
\begin{example}[Circuit correctness union bound]
    Consider a circuit with classical input and output that is the composition of \(V\) compatible FT gadgets \(C:=\gadget_V \circ \dots \circ \gadget_2 \circ \gadget_1\) implementing the classical input-classical output operation \(\qOp:=\qOp_V \circ \dots \circ \qOp_2 \circ \qOp_1\) with bad sets \(\{\mathcal{G}_i\}_{i\in[V]}\).
    Let \(\fault\) be a random Pauli fault of \(G\) distributed according to \(\epsilon\)-local stochastic noise.
    
    We can apply \cref{lemma:gadget-composition} inductively in time\footnote{If \(C\) utilized gadgets in parallel, we would first need to apply \cref{lemma:gadget-side-by-side} inductively on each timestep.} to arrive at a circuit with classical output that is correct if \(\supp\fault\) is \(\mathcal{G}:=\mathcal{G}_V \boxplus \dots \boxplus \mathcal{G}_2 \boxplus \mathcal{G}_1\)-avoiding.
    Suppose that for all \(i \in [V]\), \(f(x)\) is an upper bound for \(\weightenum{\mathcal{G}_i}{x}\) on some interval of \([0,1]\) that contains \(\epsilon\).
    Then, the probability that \(\supp\fault\) is not \(\mathcal{G}\)-avoiding is at most 
    \begin{align*}
        \Pr(\text{\(\supp \fault\) is not \(\mathcal{G}\)-avoiding}) \le \weightenum{\mathcal{G}}{\epsilon} = \sum_{i \in [V]} \weightenum{\mathcal{G}_i}{\epsilon)} \le V \cdot f(\epsilon)
    \end{align*}
    where we have used \cref{def:weight-enumerator} to upper bound the probability that \(\supp\fault\) is not \(\mathcal{G}\)-avoiding and \cref{lemma:enumerator-ring} to evaluate \(\weightenum{\mathcal{G}}{x}\) as a sum over \(\{\weightenum{\mathcal{G}_i}{x}\}_{i\in [V]}\).
\end{example}
For this example, much of the machinery could have been skipped and a standard union bound applied.
However, later we will need to use a similar union-bound type argument in combination with a sort of independence utilizing the product operation.
That is, it is rare for many gadgets with disjoint locations to fail simultaneously.
The weight enumerator polynomial machinery allows us to easily analyze such a scenario which, for example, will appear in the distillation of resource states.

\subsection{Decoherence of errors}
In several of our gadgets, we will need to accept states that are Pauli-deviated from the codespace, but the analysis of our error correction (including state distillation) gadgets only considers states that are diagonal Pauli-deviated from the codespace (differs from a codestate by a diagonal Pauli superoperator).
We show that measurement of the checks of a quantum code reduces the former case to the latter.

\begin{lemma}[Decoherence of errors]\label{lemma:deterministic-errors}
  For a \([[n,k,d]]\) stabilizer code \(\qcode\), let \(\mathcal{M}\) be the channel that measures the \(r\) stabilizer checks of \(\qcode\) and outputs the measurement outcomes.
  For a superoperator \(\mathcal{E}_B\) supported on a recoverable (in the sense of erasure) subset of qubits\footnote{This can be thought of as \(|B| < d\), but, in fact, many sets larger than \(B\) are also recoverable. This fact is used for robustness to stochastic noise.} \(B \subseteq [n]\) and a codestate \(\rho\) of \(\qcode\), the application of the noise superoperator followed by (ideal) measurement of the checks can be written as
  \begin{align}
    \mathcal{M}\circ \mathcal{E}_B(\rho) = \sum_{x \in \F^r} \alpha_x \ketbra{x}{x}\otimes E_x \rho E_x,
  \end{align}
  where each \(E_x\) is a Pauli operator supported on \(B\) and the \(\alpha_x\) are complex coefficients.
  When \(\mathcal{E}_B\) is a physical noise channel, they satisfy \(\sum_x \alpha_x = 1\).
  In other words, measurement of the checks collapses the error into a single Pauli error \emph{as long as we remember the measurement outcome}.
\end{lemma}
\begin{proof}
  Let \(\{S_i\}_{i \in [r]}\) be a generating set for the stabilizer group of \(\qcode\) such that measurement of \(S_i\) produces the \(i\)-th syndrome bit.
  We can write the projector into each syndrome eigenspace \(x \in \F^r\) as \(\Pi_x = \prod_{i \in [r]}\frac{1}{2}(1 + (-1)^{x[i]}S_i)\), so that
  \begin{align*}
    \mathcal{M}(\rho) = \sum_{x \in \F^r} \ketbra{x}{x}\otimes \Pi_x \rho \Pi_x.
  \end{align*}
  
  The noise superoperator \(\mathcal{E}_B\) can be written as a linear combination of Pauli superoperators supported only on \(B\).
  Thus, for some Pauli operators \(\{K_{\mu}\}_{\mu}\), \(\{K'_{\nu}\}_{\nu}\) and complex coefficients \(\{\beta_{\mu\nu}\}_{\mu,\nu}\), we can write
  \begin{align*}
    \mathcal{E}_B(\rho)  = \sum_{\mu,\nu} \beta_{\mu\nu} K_{\mu} \rho K'_{\nu}.
  \end{align*}
  
  We will decompose this representation by separating operators by syndrome. Let \(\sigma\) be the syndrome map from Pauli operators to their syndromes.
  Recall that \(\rho\) is a code-state so that \(\Pi_0 \rho \Pi_0 = \rho\).
  The post-measurement state is
  \begin{align*}
    \mathcal{M}\circ \mathcal{E}_B(\sigma) &= \sum_{x \in \F^r} \ketbra{x}{x}\otimes \sum_{\mu,\nu} \beta_{\mu\nu} \left(\Pi_xK_{\mu} \rho K'_{\nu} \Pi_x\right)\\
                    &= \sum_{x \in \F^r} \ketbra{x}{x}\otimes \sum_{\mu,\nu} \beta_{\mu\nu} \left(\Pi_xK_{\mu} \Pi_0\rho\Pi_0 K'_{\nu} \Pi_x\right)\\
                    &= \sum_{x \in \F^r} \ketbra{x}{x}\otimes \sum_{\substack{\mu \mid \sigma(K_{\mu}) = x \\ \nu \mid \sigma(K'_{\nu}) = x}} \beta_{\mu\nu} \left(\Pi_xK_{\mu} \Pi_0\rho\Pi_0 K'_{\nu} \Pi_x\right).
  \end{align*}
  The projectors will annihilate any term in the Pauli decomposition that do not have syndrome \(x\), so we can restrict the summation.
  For any two Pauli operators \(a,b\) supported on \(B\) with the same syndrome \(x\), their product \(ab\) has trivial syndrome and is also supported on \(B\).
  By assumption, \(B\) is recoverable from erasure so any Pauli operator with trivial syndrome supported on \(B\) must be in the stabilizer.
  For any codestate of \(\qcode\) the action of an element of the stabilizer is trivial, so that we can write
  \begin{align*}
    \Pi_x a \Pi_0 = \Pi_x a (ab)\Pi_0 = \Pi_x b \Pi_0.
  \end{align*}

  This implies that, for an arbitrary set of Pauli operators \(\{E_x\}_x\) supported on \(B\) satisfying \(\sigma(E_x)=x\), we can write
  \begin{align*}
    \sum_{\substack{\mu \mid \sigma(K_{\mu}) = x \\ \nu \mid \sigma(K'_{\nu}) = x}} \beta_{\mu\nu} \left(\Pi_xK_{\mu} \Pi_0\rho\Pi_0 K'_{\nu} \Pi_x\right) &= \sum_{\substack{\mu \mid \sigma(K_{\mu}) = x \\ \nu \mid \sigma(K'_{\nu}) = x}} \beta_{\mu\nu} \left(\Pi_xE_x \Pi_0\rho\Pi_0 E_x \Pi_x\right) \\
    &= \left(\sum_{\substack{\mu \mid \sigma(K_{\mu}) = x \\ \nu \mid \sigma(K'_{\nu}) = x}} \beta_{\mu\nu}\right) \left(E_x\rho E_x\right). 
  \end{align*}
  The sum in parenthesis is the coefficient in the lemma statement.
\end{proof}

\begin{remark}
  It is often said that coherent errors are ``decohered'' by syndrome measurement into Pauli errors.
  \Cref{lemma:deterministic-errors} is the formal statement of this fact.
  The gadgets retain the measurement record in order to ``purify'' the error and more easily track it.
  This is also the only point where we utilize our different definition of deviation (\Cref{def:deviation}) from \cite{aharonov1999fault, aharonov2008fault}.
  Equivalently, one could ensure that the logical state of the computation is pure and employ \cite[Claim 1]{aharonov2008fault} before using \cref{lemma:deterministic-errors}.
\end{remark}

\section{Proof of main result}\label{sec:main-proof}
In this section, we prove the main result (\cref{thm:main-result-concat}). We will first define some notations. Then, we state several lemmas in~\Cref{subsec:deferred-lemmas}, most of whose proofs will be deferred to later sections. These include the error correction gadget and the state preparation / distillation gadgets. In~\Cref{subsec:compilation}, we describe our choice of a set of primitive logical operations, for which the lemmas from~\Cref{subsec:deferred-lemmas} apply, and show how to compile the circuit using these operations. In~\Cref{subsec:main-thm-vanishing-threshold}, we combine the fault-tolerant gadgets and the compilation lemma to prove the main result, albeit with a slightly vanishing threshold due to the `almost-goodness' of the qLTC family~\cite{dinur2024expansion}. Finally, in~\Cref{subsec:main-thm-constant-threshold} we obtain the main result with a constant threshold by a concatenation step with the scheme of~\cite{yamasaki2024time}.

Let us start by introducing some notations. We will need to refer to sets of qubits at the logical level that are grouped together in some way (think the logical qubits in a quantum code).
We will refer to such groups of qubits as a \emph{register} of qubits.
Later, registers will be encoded into error-correcting code blocks.
With respect to registers, we will refer to an operation as transversal if it acts on every qubit identically.
When necessary, we will denote the target register e.g. \(H(A)\) or \(\CNOT(A,B)\) where \(A\) and \(B\) should be thought of as ``unbound'' variables indicated the target register unless otherwise stated.
For two-qubit gates, and two registers \(A\) and \(B\) with qubits coordinates \(\{1, \dots, k\}\) and \(\{k+1, \dots, 2k\}\), a transversal two-qubit gate such as \(\CNOT(A,B)\) acts as \(\CNOT(1,k+1)\CNOT(2, k+2)\dots \CNOT(k,2k)\).

In the following definitions, for an indexed set \([W]\) separated into registers of equal size \(k_L\), \([W] = A \sqcup B \sqcup ...\), and for an index \(i \in [k_L]\), we use the following abuse of notation \(i_A, i_B \in [N]\) to denote an index of the larger set specified by a choice of one of the registers and an index into that register.
E.g. if \(A = \{1, \dots, k_L\}\) and \(B = \{k_L+1, \dots 2k_L\}\), the notation means \(i_A = i\) and \(i_B = k_L + i\).

We will use qLTC code blocks of the same parameters $[[n_L,k_L,d_L]]$, referred to as the \emph{computational code}. Throughout, we fix an encoding map of the computational code \(\mathcal{E}\) mapping from \(k_L\) qubits to \(n_L\) qubits that is implementable by a Clifford circuit and compatible with the constant quantum-depth encoding circuit from \cref{lemma:constant-depth-encoding}.
Let \([W]\) be the set of qubits of the simulated circuit \(C\).
We will label this Hilbert space (with basis) by \(\mcH_S\) and refer to it as the \emph{simulated Hilbert space}.
We further introduce a second Hilbert space \(\mcH_L\) with basis given by \(m\) registers each of size \(k_L\) qubits, such that $m k_L \geq W$, that we refer to as the \emph{logical Hilbert space}.
We will label the qubits of \(\mcH_L\) by \([m]\times[k_L]\).
Much like how it is standard to decompose an \(N\)-qubit Hilbert space into \(\mcH_2^{\otimes N}\), there are two relevant levels of decomposition of the logical Hilbert space \(\mcH_L\): \(\mcH_L = \left(\mcH_2^{\otimes k_L}\right)^{\otimes m}\) is the tensor product of single-register Hilbert spaces\footnote{This decomposition is particularly important, because the fundamental indivisible unit of space for our compiled circuits will be registers. E.g. quantum operations will act on registers instead of qubits.} which is further a tensor product of single-qubit Hilbert spaces.
Finally, the application of \(\mathcal{E}\) to each register in \(\mcH_L\) induces a third Hilbert space \(\mcH_P= \left(\mcH_2^{\otimes n_L}\right)^{\otimes m}\), the physical Hilbert space.

\subsection{Deferred lemmas}\label{subsec:deferred-lemmas}

\subsubsection{Quantum LDPC and LTC}
We first define the computational code on which the main logical computations in our scheme will be performed. We will eventually choose this code to be the qLTC from~\cref{sec:qltc}.

\begin{definition}[Computational qLDPC code]
  We define the computational code \(\qcode_L=CSS(H_X,H_Z)\) to be an element of a family of \(\Delta\)-qLDPC codes.
  It will have parameters \([[n_L,k_L,d_L]]\), at most \(r_L\) rows
  in \(H_X\) and \(H_{Z}\) (the subscript `L' refers to `LDPC'), and will be \((\rho,\Delta)\)-LTC.
  For two (to be determined) integers \(t_{\mathrm{corr}},~t_L \in [n_L]\) with \(t_L \leq t_{\mathrm{corr}}/3\), we will further define the family of bad sets \(\mathcal{F}_{\mathrm{corr}}\) (\(\mathcal{F}_L\)) associated with the computational code to be all subsets of \([n_L]\) of size \(t_{\mathrm{corr}}\) (\(t_L\)).
  Thus,
  \begin{align}
    \weightenum{\mathcal{F}_{\mathrm{corr}}}{x} &= \binom{n}{t_{\mathrm{corr}}} x^{t_{\mathrm{corr}}}\\
    \weightenum{\mathcal{F}_L}{x} &= \binom{n}{t_L} x^{t_L}.
  \end{align}
  
\end{definition}
Note that, by construction any set \(X \subseteq [n_L]\) of size \(|X| \ge t_L\) will contain, as a subset, at least one element of \(\mathcal{F}_L\), i.e., \(\mathcal{F}_L\) is the smallest set of witnesses for the property \(|X| \ge t_L\).
In other words, this definition identifies all sufficiently large sets as bad sets -- those that could cause a logical error. This definition is compatible with our later use of an adversarial noise decoder for the qLTC from~\cref{sec:qltc}.
One can more carefully define bad sets for the setting of stochastic noise, which is crucial if the code relative distance is smaller than the noise rate.
However, their definition is more involved as it requires the use of the \(\Delta\)-qLDPC property (e.g. see the proof of theorem 3 of \cite{gottesman2013fault}).

With the exception of the statement of \cref{lemma:ec-gadget} and within the proofs of some gadgets, all codeblocks of the computational code will be \((\qcode_L,\mathcal{F}_L)\) blocks.
Note that if a set is \(\mathcal{F}_L\)-avoiding then it is also \(\mathcal{F}_{\mathrm{corr}}\)-avoiding, so nearly all statements will hold with \(\mathcal{F}_L\) replaced by \(\mathcal{F}_{\mathrm{corr}}\).

Later, \(t_{\mathrm{corr}}\) and \(t_L\) will be picked in a way related to the maximum distance for which our bounded-distance decoder (from~\Cref{sec:qltc}) succeeds.
If we have a stochastic-error decoder, we may also pick the sets \(\mathcal{F}_{\mathrm{corr}}\) and \(\mathcal{F}_L\) accordingly: The majority of the arguments in the paper do not explicitly depend on the precise choice.
Our end goal will be to construct a set of compatible gadgets where all inputs and outputs are \(\mathcal{F}_L\)-deviated from the codespace.
Since some of our intermediate gadgets will create or spread errors, we will construct an error correction gadget that (roughly) takes a state \(\mathcal{F}_{\mathrm{corr}}\)-deviated from the codespace to a state that is \(\mathcal{F}_L\)-deviated from the codespace.\footnote{Actually, the majority of our final gadgets will only be Pauli fault-tolerant gadgets, so the states will satisfy the stronger condition of being Pauli deviated from the code space.}

To define our gadget, we introduce the standard definition of a single-shot decoder~\cite{bombin2015single, gu2023single} (for adversarial Pauli errors).
\begin{restatable}[Single-shot decoding]{definition}{definitionSingleShotDecoding}
 Let \(\qcode\) be a quantum CSS code specified by the parity check matrices $H_X \in \mathbb{F}_2^{r_X \times N}$ and $H_Z \in \mathbb{F}_2^{r_Z \times N}$. Suppose the data state carries the Pauli errors $e_X, e_Z \in \mathbb{F}_2^{N}$ and the syndrome measurements incur the errors $m_X \in \mathbb{F}_2^{r_Z}$, $m_Z \in \mathbb{F}_2^{r_X}$, such that the associated noisy syndromes are $\tilde{\sigma}_X = H_Z e_X + m_X$ and $\tilde{\sigma}_Z = H_X e_Z + m_Z$. Let $\mathcal{D}$ be a decoding algorithm that receives $\tilde{\sigma}_X$, $\tilde{\sigma}_Z$ and proposes the Pauli corrections $f_X$, $f_Z$. The decoder is said to be $(\alpha(N), \beta(N), \gamma(N),\eta)$-single-shot, where $\alpha, \beta, \gamma: \mathbb{N}^+ \mapsto \mathbb{R}^{+}$ and $0\leq \eta <1$, if the following holds. For each $P \in \{X,Z\}$, provided that
\begin{align*}
    |e_P|_R + \alpha(N) |m_P| \leq \beta(N) N,
\end{align*}
then the proposed Pauli correction $f_P$ satisfies
\begin{align*}
    |e_P + f_P|_R \leq \eta |e_P|_R + \gamma(N) |m_P|.
\end{align*}
Here $|\cdot|_R$ denotes the stabilizer-reduced weight. When $\eta=0$ we simply say $(\alpha,\beta,\gamma)$-single-shot.
\end{restatable}

In~\Cref{sec:qltc} we prove the following decoder result on the code construction from~\cite{dinur2024expansion}.

\begin{restatable}[Computational almost-good qLTC with single-shot decoding]{theorem}{theoremAlmostGoodQltcWithDecoder}\label{thm:single-shot-qltc}
There exists an explicit $i$-indexed family of CSS quantum LDPC codes with parameters $[[N_i,K_i,D_i]]$, where $K_i = \Theta(N_i)$, $D_i = N_i/\polylog(N_i)$. Any instance in the family is locally testable with soundness parameter $\rho_i = 1/\polylog(N_i)$. Furthermore, the code family admits
\begin{itemize}
    \item A $O(\log N_i)$-time parallel decoder that is $(\alpha, \beta(N_i), \gamma,0)$-single-shot, where $\alpha, \gamma$ are constants and $\beta(N_i) = 1/\polylog(N_i)$.
    \item A $O(N_i)$-time sequential decoder with similar single-shot decoding statement.
    \item A $O(1)$-constant time parallel decoder that is $(\alpha, \beta(N_i), \gamma, \eta)$-single-shot, where  $\alpha, \gamma$ are constants, $\beta(N_i) = 1/\polylog(N_i)$, and $\eta \leq 1/8$ is a constant.
\end{itemize}
In addition, there exists a constant $c$ such that for any $i \in \mathbb{N}^+$, it holds that $\frac{N_{i+1}}{N_i} < c$.
\end{restatable}

Above, we introduce an extra requirement on the growth rate of the code family. This is a useful condition to achieve constant-space overhead for fault-tolerant computation when the required code block size is large.
We refer to code families satisfying this as `computational code family'.
With the exception of the transversal logical measurement gadget in \cref{lemma:computational-code-transversal-gates}, we will use the constant-time decoder parallel decoder from the above theorem.

\subsubsection{Global variables}\label{sec:globals}
Here we introduce ``global'' variables and assumptions that will appear in the remaining sections of the proof.
We will pick \(n_L\) and \(k_L\) in the proof of the main theorem.
The remainder of the paper will assume that that \(k_L = 2^{\nu}\) for some integer \(\nu\).
Since our codes are almost-good, we will introduce \(n_L\)-dependence in a very controlled way:
Let \(\zeta \colon \mathbb{R}^+ \to \mathbb{R}^+\) be a monotonic invertible function \(\zeta(x) = x/\polylog(x)\) such that for all \(N \in \mathbb{N}^+\), \(\zeta(N) \le \beta(N) N\) and \(\zeta(N) \le \rho(N) N\).
For a constant \(\epsilon\), we will use the notation \(\ndep{\epsilon}\) to mean \(\frac{\zeta(n_L)}{n_L} \epsilon\).
In all of our results, when substituted with a good qLTC, one should think of \(\zeta(n_L)\) as being proportional to \(n_L\), e.g., \(n_L/100\).

We will also set
\begin{align}\label{eq:tL-size}
  t_L &= \left\lfloor \frac{\zeta(n_L)}{4} \right\rfloor \le \frac{\beta(n_L)n_L}{4}\\
  t_{\mathrm{corr}} &= 3 t_L~.\label{eq:tcorr-size}
\end{align}
We take \(n_L\) to be at least a large enough constant such that certain parameters within \cref{lemma:ec-gadget} \((s)\) and \cref{lemma:ec-gadget-tester} \((s, t_{\mathrm{test}})\) and \(t_L\) are large enough (at least 1 or 2). %

\subsubsection{Resource state preparation}\label{sec:state-prep-defer}

We also require gadgets to prepare resource states, which will enact the gates in our primitive computational gate set defined in \cref{subsec:compilation}.
These lemmas are proved in \cref{sec:ft-state-distillation}.

The first lemma shows how to prepare various types of stabilizer states on our computational qLTC code.

\begin{restatable}[Stabilizer resource state preparation]{lemma}{lemmaStabilizerResourceStatePrep}\label{lemma:stab-resource-state}
  For a constant \(b \in \mathbb{N}\), \(b < 10\)\footnote{Any absolute constant suffices.}, representing the number of blocks, let \(\{\ket{\psi_i}\}_{i=1}^{k_L}\) be a set of \(b\)-qubit stabilizer states such that \(\ket{\psi} := \ket{\psi_1}\otimes \dots \otimes \ket{\psi_{k_L}}\) is preparable by a constant-depth unitary Clifford circuit \(\qOp_{\mathrm{stab}}\) acting on \(\ket{0}^{\otimes b k_L}\).
  In terms of a variable \(M_{\mathrm{stab}} = e^{\Theta(\log n_L \cdot \loglog n_L)}\), let \(K_{\mathrm{stab}} = \Theta(M_{\mathrm{stab}})\).
  Then, there exists a \((\qOp_{\mathrm{stab}},\mathcal{G}_{\mathrm{stab}})\)-Pauli FT gadget \(\gadget_{\mathrm{stab}}\) preparing \(K_{\mathrm{stab}}\) copies of \(\ket{\overline{\psi}}\), each encoded in \(b\) codeblocks of the computational code such that the \(i\)-th logical qubits of the \(b\) codeblocks are jointly in the state \(\ket{\psi_i}\)\footnote{I.e. preparing \(\psi_1\otimes \dots \otimes \psi_{k_L}\) ``coordinate-wise.''} and \(\weightenum{\mathcal{G}_\mathrm{stab}}{x} \le e^{-\zeta(n_L) + O\left(\polylog n_L\right)}\) for \(x \in [0,\ndep{\epsilon_{*,\mathrm{stab}}})\).
  
  Furthermore, \(\gadget_{\mathrm{stab}}\) has depth \(O(\polylog n_L)\) and width \(O\left(M_{\mathrm{stab}} n_L\right)\)
\end{restatable}

The second lemma creates magic states for performing \(\CCZ\) gates on the first three qubits of a computational code block.

\begin{restatable}[Magic resource state preparation]{lemma}{lemmaMagicResourceStatePrep}\label{lemma:magic-resource-state} Let \(\qOp_\mathrm{magic}\) be the depth-2 circuit that prepares the state \(\CCZ(1,2,3) \ket{+}^{\otimes k_L}\). In terms of a variable \(M_{\mathrm{magic}}\) satisfying \(M_{\mathrm{magic}} = e^{O(\log n_L \cdot \loglog n_L)}\), \(M_{\mathrm{magic}} = \Omega(n_L)\),  let
\(K_{\mathrm{magic}} = \Omega(M_{\mathrm{magic}} n_L^{-\tilde{\gamma}(n_L)}) \)
  where \(\tilde{\gamma}(n_L) = O_{n_L\to\infty}\left(\frac{1}{\loglog n_L}\right)= o_{n_L \to \infty}(1)\).
  There exists a \((\qOp_\mathrm{magic},\mathcal{G}_\mathrm{magic})\)-Pauli FT gadget \(\gadget_\mathrm{magic}\) preparing \(K_\mathrm{magic}\) copies of \(\ket{\overline{\CCZ(1,2,3)}}\) such that \(\weightenum{\mathcal{G}_\mathrm{magic}}{x} \le e^{-\zeta(n_L) + O\left(\polylog n_L\right)}\) for \(x \in [0,\ndep{\epsilon_{*,\mathrm{magic}}})\).

  Furthermore,
  \(\gadget_\mathrm{magic}\) has depth \(O(\polylog n_L)\) and width \(O(M_\mathrm{magic} n_L)\).
\end{restatable}

\subsubsection{Computational code EC gadget}
We now construct an error correction gadget for the computational code.

\begin{lemma}[Computational code EC gadget]\label{lemma:ec-gadget}
  For \(n_L\) larger than some constant,\footnote{Dependent on constants of the code family and decoder.} there exists a constant \(\epsilon_{*,\mathrm{EC}}\in(0,1)\) and gadget \(\gadget_\mathrm{EC}\) for the computational code of depth \(O(1)\), width \(n_L\) qubits, and \(A_{\mathrm{EC}}=O(n_L)\) locations such that:
  \begin{itemize}
  \item \(\gadget_\mathrm{EC}\) takes a \((\qcode_L, \mathcal{F}_{\mathrm{corr}})\)-block to a \((\qcode_L, \mathcal{F}_L)\)-block.
  \item There is family of bad fault paths \(\mathcal{G}_{\mathrm{EC}}\subseteq P([A_{\mathrm{EC}}])\) such that \(\gadget_\mathrm{EC}\) is a \((\mathsf{Id}, \mathcal{G}_{\mathrm{EC}})\)-Pauli fault-tolerant gadget (i.e. it performs identity).
  \item \(\mathcal{G}_{\mathrm{EC}}\) satisfies \(\weightenum{\mathcal{G}_{\mathrm{EC}}}{x} \le e^{-\zeta(n_L)}/4\) on \(x \in [0,\ndep{\epsilon_{*,\mathrm{EC}}})\).
  \end{itemize}
\end{lemma}

\begin{proof}
  Our circuit and argument will be a standard construction (e.g. see~\cite{gottesman2013fault}) of a syndrome extraction circuit for a qLDPC code:
  Introduce at most \(r_L\) ancilla qubits, one for each row of \(H_X\), initialized in \(\ket{+}\).
  We can efficiently find an edge coloring \(c : [r_L]\times [n_L] \to [\Delta]\) of the Tanner graph of \(H_X\) using \(\Delta\) colors.
  For each color \(i \in [\Delta]\), we perform \(\CNOT\) controlled on the ancilla for all pairs of data and ancilla qubits in \(c^{-1}(i)\) with a depth-1 circuit.
  Finally, perform \(\Had\) and measure the ancilla qubits in the \(Z\) basis to obtain a syndrome \(\tilde{\sigma}_Z\) (detecting \(Z\) errors).
  We repeat this procedure for the \(Z\) basis with \(\CNOT\) replaced by \(\CZ\) and \(H_X\) by \(H_Z\).
  In the absence of faults, this circuit measures the syndrome of the input state by construction.

  We invoke the decoder on the measurement outcomes \((\tilde{\sigma}_X, \tilde{\sigma}_Z)\) to receive the corrections \((f_X,f_Z)\).
  Let \(\tau = O(1)\) be the runtime of the decoder measured as the number of circuit steps.
  All qubits experience (noisy) identity gates for \(\tau\) timesteps while the decoder runs.
  The final step is to apply the correction \(F_c = X^{f_X} Z^{f_Z}\) (i.e. the restriction of \(F_c\) to a qubit \(i\) is \(F|_i = X^{f_{X}[i]}Z^{f_{Z}[i]}\)).

  Here, we will use the code parameters defined in the statement of \cref{thm:single-shot-qltc}.
  That is, the decoder is implementable by a constant-depth classical circuit and is \((\alpha,\beta(n_L), \gamma, \eta)\)-single shot with \(\eta \le 1/8\).

  Let \(\ell = 2\cdot 2^{2\Delta}\),\footnote{A more careful bound using the structure of the syndrome extraction circuit with \(\ell = O(\Delta)\) is possible e.g. see \cite[lemma 5.2]{pattison2025hierarchical}. Here, we use a lightcone argument for simplicity.} and take \(\mathcal{G}_{\mathrm{EC}}\) to be all subsets of the spacetime locations of size \(s = \left\lfloor \frac{t_L}{\ell \cdot 10 \max(\alpha,\gamma, 8)}\right\rfloor\)
  (recall the definitions of \(t_L, t_\mathrm{corr}\) from~\cref{sec:globals}).
  For \(n_L\) larger than some constant, \(s\) satisfies \(s \ge 2\).
  We now prove that this is a Pauli fault-tolerant gadget.
  Recall from the definition of a Pauli fault-tolerant gadget (\cref{def:pauli-ft-gadget}), we are considering Pauli input errors and Pauli faults.
  To avoid cluttering notation, let us neglect the details of the \(m\)-qubit reference system in \cref{def:ft-gadget}: Everything that we do here extends straightforwardly to a state entangled with a reference system (see \cref{rmk:reference-system}).
  
  First we restrict to the following case:
  Suppose that the input state \(\tilde{\rho}\) is \(\mathcal{F}_{\mathrm{corr}}\)-diagonal Pauli deviated from an encoded state \(\overline{\rho} = \phi_{\qcode_L}(\rho)\).
  That is, \(\tilde{\rho} = E_{\mathrm{in}} \overline{\rho} E_{\mathrm{in}}\) for some Pauli operator \(E_{\mathrm{in}}\) that supported on an \(\mathcal{F}_{\mathrm{corr}}\)-avoiding set.
  Further suppose that the \(\gadget_\mathrm{EC}\) is subject to a diagonal Pauli fault \(\mathbf{f}\) that is \(\mathcal{G}_{\mathrm{EC}}\)-avoiding. By construction of \(\mathcal{G}_{\mathrm{EC}}\), \(\mathbf{f}\) is supported on at most \(s-1\) locations, each with at most \(2\) qubits.
  
  Since \(\gadget_\mathrm{EC}\) (excluding the measurements and decoding circuit) is Clifford and \(\mathbf{f}\) is Pauli, we can ``push'' the Paulis to the end and write
  \begin{align}\label{eq:ec-gadget:push-paulis}
    \gadget_\mathrm{EC}[\mathbf{f}](\tilde{\rho}) = \mathbf{f}_{\mathrm{out}} ~ \gadget_\mathrm{EC}[\mathbf{m}] ~ E_{\mathrm{in}}(\overline{\rho}) = \mathbf{f}_{\mathrm{out}} ~ E_{\mathrm{in}} ~ \gadget_\mathrm{EC}[\mathbf{m}\mathbf{\sigma}](\overline{\rho}),
  \end{align}
  where \(\mathbf{m}\) is a Pauli fault stemming from \(\mathbf{f}\) that is supported only on measurements and is decomposed into a part supported only on syndrome measurements detecting \(X\) errors and syndrome measurements detecting \(Z\) errors \(\mathbf{m} = \mathbf{m}_X\mathbf{m}_Z\).
  If \((\sigma_X, \sigma_Z)\) is the syndrome of \(E_{\mathrm{in}}\), then, associating \(\mathbf{m}_X\) with a corresponding bitstring \(m_X\), \(\tilde{\sigma}_X = \sigma_X + m_X\) and likewise for \(Z\).
  \(\mathbf{\sigma}\) is also a Pauli fault,\footnote{It is a consequence of pushing the input error past the syndrome measurement circuit. We are still measuring the syndrome, but it is unfortunate consequence of our notation that it is technically called a ``fault.''} but it corresponds to the flipped syndrome measurements \(E_{\mathrm{in}}\) i.e. to the tuple \((\sigma_{X}, \sigma_Z)\).
  Let us also decompose \(E_{\mathrm{in}} = X^{e_X}Z^{e_Z}\).

  Since there are at most \(2\Delta\) entangling gates, a fault in a single location can propagate to at most \(\ell\) measurements or output qubits.
  It follows that the measurement error and output error stemming from circuit faults is bounded as
  \begin{align} \label{eq:measurement-error-bound}
    |m_X| &\le \ell (s-1) \\ \label{eq:output-error-bound}
    |\supp \mathbf{f}_{\mathrm{out}} | &\le \ell (s-1)~.
  \end{align}
  Using \cref{eq:measurement-error-bound} and the assumption that the input error is supported on an \(\mathcal{F}_{\mathrm{corr}}\)-avoiding set, we are now ready to show that the preconditions for the decoder (\cref{thm:single-shot-qltc}) are satisfied (utilizing definitions of \(t_L, t_\mathrm{corr}\) from \cref{sec:globals}).
  \begin{align*}
    \alpha |m_X| + |e_X| &\le \ell (s-1) + t_{\mathrm{corr}} \\
     &\le \left(\frac{\alpha}{10 \cdot \max(\alpha, \gamma, 8)} + 3\right) t_L\\
     &\le  \left(\frac{1}{10} + 3\right) \frac{\zeta(n_L)}{4}\\
                  &\le \beta(n_L) n_L
  \end{align*}
  
  It follows that the decoder (which is \((\alpha,\beta(n_L), \gamma, \eta)\)-single shot) returns a correction \(f_X\) such that the reduced weight of \(f_X+e_X\) is bounded.
  That is, there exists an operator \(E_{\mathrm{out},X}\) such that \(E_{\mathrm{out},X} X^{f_X + e_X}\) is in the stabilizer of the code and
  \begin{align*}
  |\supp E_{\mathrm{out},X} | &\le \gamma |m_X| + \eta |e_X| \\
   &\le \gamma \ell (s-1) + \eta t_{\mathrm{corr}}~.
  \end{align*}
  An identical argument holds with the roles of \(Z\) and \(X\) exchanged.

  Let \(F_c\) be the correction applied by the gadget, then since codestates are left invariant by measurement of the stabilizer generators, the output state can be written as (recall \cref{eq:ec-gadget:push-paulis})
  \begin{align}\label{eq:ec-gadget:residual}
     \gadget_\mathrm{EC}[\mathbf{f}](\tilde{\rho}) = \mathbf{f}_{\mathrm{out}} ~ E_{\mathrm{in}} ~ \gadget_\mathrm{EC}[\mathbf{m}\mathbf{\sigma}](\overline{\rho}) = \mathbf{f}_{\mathrm{out}} ~ E_{\mathrm{in}} ~ F_c (\overline{\rho}) = \mathbf{f}_{\mathrm{out}}~E_{\mathrm{out},X}~E_{\mathrm{out},Z}(\overline{\rho}).
  \end{align}
  Finally, summing the bounds on the supports and using \(\eta \le 1/8\), the support of \(\mathbf{f}_{\mathrm{out}}E_{\mathrm{out},X}E_{\mathrm{out},Z}\) is at most
  \begin{align*}
    \left|\supp \left(\mathbf{f}_{\mathrm{out}}E_{\mathrm{out},X}E_{\mathrm{out},Z}\right)\right| &\le 2\left(\gamma \ell (s-1) + \eta t_{\mathrm{corr}}\right) + \ell (s-1) \\
    &\le \left(\frac{2}{10} + 6\eta + \frac{1}{80}\right)t_L\\
    &\le t_L~.
  \end{align*}
  That is, the output is \(\mathcal{F}_L\)-diagonal Pauli deviated from \(\overline{\rho}\).

  It now remains to upper bound the weight enumerator \(\weightenum{\mathcal{G}_{\mathrm{EC}}}{x}\).
  Since the decoder runtime is constant, the number of spacetime locations \(A_{\mathrm{EC}} = \Theta(n_L)\) is linear in the blocklength.
  The size of elements of \(\mathcal{G}_{\mathrm{EC}}\) is  \(s = \Theta(t_{L}) = \Theta(\zeta(n_L))\), so the weight enumerator of \(\mathcal{G}_{\mathrm{EC}}\) can be bounded as 
  \begin{align*}
    \weightenum{\mathcal{G}_{\mathrm{EC}}}{x} &= \binom{A_{\mathrm{EC}}}{s} x^s\\
                     &\le \left(\frac{A_{\mathrm{EC}} e x}{s}\right)^s\\
                     &\le \left((\mathrm{const.})\frac{n_L}{\zeta(n_L)} x\right)^{(\mathrm{const.})\zeta(n_L)},
  \end{align*}
  where the terms \((\mathrm{const.})\) depend on \(\Delta\), \(\gamma\), \(\alpha\), and \(\tau\) (the runtime of the decoder) only.
  Thus, there exists a constant \(\epsilon_{*,\mathrm{EC}} \in (0,1)\) such that for \(x \in (0, \ndep{\epsilon_{*,\mathrm{EC}}})\),
  \begin{align*}
    \weightenum{\mathcal{G}}{x} &\le \left[(\mathrm{const.})\frac{n_L}{\zeta(n_L)} \ndep{\epsilon_{*,\mathrm{EC}}} \right]^{(\mathrm{const.})\zeta(n_L)} \\
                       &\le \left[(\mathrm{const.})  \epsilon_{*,\mathrm{EC}} \right]^{(\mathrm{const.})\zeta(n_L)} \\
                       &\le e^{-\zeta(n_L)}/4.
  \end{align*}

  We now remove the assumption on the input state and the fault:
  Suppose that \(\tilde{\rho}\) is now only (non-diagonal) Pauli-deviated from the codestate \(\overline{\rho}\) i.e. \(\tilde{\rho} = \mathcal{E}_{\mathrm{in}}(\overline{\rho})\) and that \(\fault\) is a general Pauli fault.
  If \(\fault\) were a diagonal Pauli fault, \cref{lemma:deterministic-errors} would immediately apply, so that either: 1) \(\mathcal{E}_{\mathrm{in}}\) can be picked to be a diagonal Pauli superoperator, or 2) the state is in the kernel of the measurement projector.
  In the case that \(\fault\) is not a diagonal Pauli fault, \(\fault[\sigma]\fault[m]\) in \cref{eq:ec-gadget:push-paulis} still must be a diagonal Pauli superoperator due to the measurement projectors.
  Thus, after moving \(\mathcal{E}_{\mathrm{in}}\) across the syndrome measurement circuit, \(\mathcal{E}_{\mathrm{out}}\) can only be non-diagonal on a portion of qubits that is in the lightcone of some non-diagonal portion of \(\fault\).
  In other words, the argument holds for the general case except that \(\fault_\mathrm{out}\), \(E_{\mathrm{out},X}\),  and \(E_{\mathrm{out},Z}\) in \cref{eq:ec-gadget:residual} are non-diagonal Pauli superoperators that resulted from the spreading of \(\fault\) as it was pushed through the circuit (and so are already accounted for).
\end{proof}

The next lemma gives transversal operations on our computational codes.

\begin{lemma}[Computational code transversal gates]\label{lemma:computational-code-transversal-gates}
  There exists a constant \(\epsilon_{*,\mathrm{OP}} \in (0,1)\) such that for any
  \begin{align*}
    \qOp \in \{\CNOT(A,B), \SWAP(A,B), M_Z(A), M_X(A), \PAULI, \mathsf{Id}\},
  \end{align*}
  there exists a gadget \(\gadget_\qOp\) for the computational code of uniform depth \(O(\log n_L)\) and \(A_\qOp = O(n_L \log n_L)\) locations such that:
  \begin{itemize}
  \item \(\gadget_\qOp\) takes \((\qcode_L, \mathcal{F}_L)\)-blocks to \((\qcode_L, \mathcal{F}_L)\)-blocks.
  \item There is family of bad fault paths \(\mathcal{G}_\qOp\subseteq P([A_\qOp])\) such that \(\gadget_\qOp\) is a \((\qOp, \mathcal{G}_\qOp)\)-Pauli fault-tolerant gadget.
  \item \(\mathcal{G}_\qOp\) satisfies \(\weightenum{\mathcal{G}_\qOp}{x} \le \polylog(n_L) e^{-\zeta(n_L)}\) on \(x \in [0,\ndep{\epsilon_{*,\mathrm{OP}}})\).
  \end{itemize}
\end{lemma}
\begin{proof}
  As in the proof of \cref{lemma:ec-gadget}, we avoid cluttering the notation by neglecting the \(m\)-qubit auxillary system in (\cref{def:ft-gadget}).
  However, all steps will be compatible with its presence.
  All gadgets will be padded to have identical length using the error correction gadget (\cref{lemma:ec-gadget}) at most \(O(\log n_L)\) times.
  Let \(\mathcal{G}_{\mathrm{EC,pad}}\) be the bad fault paths of these error correction gadgets.

  We first prove the case of 2-qubit unitary gates on two encoded registers \(A\) and \(B\):
  Our gadget will be to perform the gate transversally on the two blocks followed by an error correction gadget (\cref{lemma:ec-gadget}) on each of the blocks.
  Clearly, an encoded \(\SWAP(A,B)\) or \(\CNOT(A,B)\) can be executed transversally (recall the computational code is CSS).
  The input state \(\tilde{\rho}\) is Pauli deviated from \(\overline{\rho} = \phi_{\qcode_L}\otimes \phi_{\qcode_L} \left(\rho \right)\) by an error \(E_{\mathrm{in}}\) that is \(\mathcal{F}_L\)-avoiding on each of the two blocks.
  Define the output state \(\rho' = \qOp(\rho)\) and \(\overline{\rho}' = \phi_{\qcode_L}\otimes \phi_{\qcode_L} (\rho')\).

  Let \(\mathcal{G}_{\mathrm{gate}}\) be all subsets of size \(t_L\) of the transversal gates layer and let \(\mathcal{G}_{\mathrm{EC},A}\), \(\mathcal{G}_{\mathrm{EC},B}\) be the families of bad fault paths associated with the error correction gadget defined in \cref{lemma:ec-gadget} associated with the error correction gadget acting on the \(A\) and \(B\) blocks, respectively.
  Our family of bad fault paths is \(\mathcal{G}_\qOp = \mathcal{G}_{\mathrm{gate}} \boxplus \mathcal{G}_{\mathrm{EC},A} \boxplus \mathcal{G}_{\mathrm{EC},B} \boxplus \mathcal{G}_{\mathrm{EC,pad}}\).
  Let \(\mathbf{f}\) be a Pauli fault with a \(\mathcal{G}_\qOp\)-avoiding fault path.
  Let us decompose \(\mathbf{f}\) into a part supported only on the layer of transversal gates, a part supported only on the \(A\) block EC gadget, and a part supported only on the \(B\) block EC gadget \(\mathbf{f} = \mathbf{f}_{\mathrm{gate}} \mathbf{f}_{\mathrm{EC},A} \mathbf{f}_{\mathrm{EC},B}\).

  We proceed to bound the error on the state output by the layer of gates:
  Since \(\supp\mathbf{f}_{\mathrm{gate}}\) is \(\mathcal{G}_{\mathrm{gate}}\)-avoiding, the support of \(\mathbf{f}_{\mathrm{gate}}\) on each block is strictly less than \(t_L\).
  The support of \(E_{\mathrm{in}}\) is \(\mathcal{F}_L\)-avoiding on each of the two blocks (\(\mathcal{F}_L\boxplus \mathcal{F}_L\)-avoiding), so it also has weight strictly less than \(t_L\) on each of them.
  It follows that the state after the layer of transversal gates is then \(E_2 \overline{\rho}' E_2\) where \(E_2\) has weight strictly less than \(3t_L = t_{\mathrm{corr}}\) on each of the blocks.
  By construction of \(\mathcal{G}\), \(\mathbf{f}\) is \(\mathcal{G}_{\mathrm{EC}}\)-avoiding on each of the two EC gadgets, so we may apply \cref{lemma:ec-gadget} to conclude that the output state \(\qOp(\tilde\rho)\) is \(\mathcal{F}_L\boxplus \mathcal{F}_L\)-Pauli deviated from \((\phi_{\qcode_L}\otimes \phi_{\qcode_L})\circ \qOp(\rho)\).

  For \(x \in [0,\epsilon_{*,\mathrm{EC}})\), we can now bound (using \cref{lemma:enumerator-ring})
  \begin{align*}
    \weightenum{\mathcal{G}_\qOp}{x} &= \polylog(n_L) \weightenum{\mathcal{G}_{\mathrm{EC}}}{x} + \weightenum{\mathcal{G}_{\mathrm{gate}}}{x} \\
                     &\le \polylog(n_L)  e^{-\zeta(n_L)}/2 + \binom{n_L}{t_L}x^{t_L} \\
                     &\le \polylog(n_L)  e^{-\zeta(n_L)}/2 + \left[(\mathrm{const.})  x \right]^{(\mathrm{const.})\zeta(n_L)}.
  \end{align*}
  Thus, there exists a constant \(\epsilon_{*,\mathrm{OP}} \in (0, \epsilon_{*,\mathrm{EC}}]\) such that for \(x \in [0,\epsilon_{*,\mathrm{OP}}]\),
  \begin{align*}
     \weightenum{\mathcal{G}_\qOp}{x} \le \polylog(n_L)  e^{-\zeta(n_L)}.
  \end{align*}

  The argument for one-qubit unitaries is identical since any logical Pauli operation can be implemented by a corresponding depth-1 pattern of physical Pauli operations.
  
  For measurements, we simply transversally measure all the data qubits in the appropriate basis.
  As in the proof of the error correction gadget, by \cref{lemma:deterministic-errors} (invoked on the trivial code with one codeword stabilized by all \(Z\)s or all \(X\)s.), any non-diagonal Pauli error is in the kernel of the measurement projector, so it suffices to consider diagonal Pauli errors.
  Constructing \(\mathcal{G}_\qOp\) in the same way, we can obtain a bitstring that has Hamming distance less than \(2t_L/3 \le \beta(n_L) n_L\) from the codespace.
  This allows us to invoke the computational code decoder and perform a (classical) decoding to obtain the measurement outcomes.
  More precisely, we invoke the \((\alpha,\beta(n_L), \gamma, 0)\)-single shot decoder with \(O(\log n_L)\) parallel runtime to obtain a noiseless codeword.
  The eigenvalue of each logical operator can be extracted by computing a size \(O(n_L)\) parity of a subset of bits which has \(O(\log n_L)\) parallel runtime.
\end{proof}

\subsection{Compilation}\label{subsec:compilation}
We will need to turn an arbitrary Clifford+\(\CCZ\) circuit into one that only uses the primitive operations that we will build gadgets for.
Before we give our gate set, we first need to define gates for permutation of qubits.

\subsubsection{Register permutation}
Because we will distill our resource states in bulk, we may only distill a few types.
The number of types will be far fewer than the number of distinct gates that we will need to perform due to the need to address the targets.
Thus, we will need a means of targeting the gates our generic resource states perform.
To do this, we introduce the notion of a \(\SWAP\) state.
\(\SWAP\) states perform teleportations of the qubits within a block.
By conditionally applying a teleportation circuit consuming a \(\SWAP\) state, we will obtain one bit of addressing capability.
We will use a small sized generating set of size \(O(\log k_L)\) to obtain full gate targeting.

\begin{definition}[\(\SWAP\) State]
  Given a set of qubits \([N]\) and a matching \(M \subset [N] \times [N]\), the corresponding \(N\)-qubit \(\SWAP\) state is a state \(\ket{\psi_M}\) where each pair of qubits corresponding to an element of \(M\) is a Bell state \(\ket{\phi_+}\), and each qubit that does not appear in \(M\) is a \(\ket{0}\) state.
  I.e. letting \(U = [N]\setminus\left(\bigcup_{m\in M} m\right)\), \(\ket{\psi_M}\) is the state stabilized by the set of stabilizer generators
  \begin{align}
      \left(\bigcup_{(u,v) \in M} \left\{X_u X_v, Z_u Z_v\right\}\right)\bigcup \{Z_u\}_{u \in U}.
  \end{align}
\end{definition}

\begin{remark}
  Intuitively, before we have the \(\SWAP\) states, we are only capable of coordinate-wise fault-tolerant operations.
  State injection allows us to prepare low-fidelity states that are entangled between coordinates.
  Fortunately, distillation of the \(\SWAP\) states we will use requires only coordinate-wise gates and targeted \(\PAULI\) operations, so we can distill these without first having a gate that generates entanglement between coordinates.
\end{remark}

The \(\SWAP\) state is a product state of Bell states and single qubit states, so we can use this to teleport qubits using the traditional teleportation circuit.

The main operations the \(\SWAP\) states will implement we call \(\ROT\):
\begin{definition}[\(\ROT\)]
  For a register of size \(k_L\), and for any \(i \in \mathbb{Z}_{k_L}\), the operation \(\ROT(i)\) implements a cyclic shift of the register by \(i\) places.
  I.e. the qubit at coordinate \(m\) goes to the coordinate \(m+i\text{ mod }k_L\).

  When multiple register are present, we will subscript by the register that the operation acts on e.g. \(\ROT_A(i)\) implements a shift by \(i\) on register \(A\).
\end{definition}

\begin{definition}[Permutation states and implementing \(\ROT(i)\)]
  Let \(\sigma\) be a permutation on \([k_L]\) and define two registers of \(k_L\) qubits \(A\) and \(B\).
  We define the corresponding permutation state on \(AB\), \(\ket{\Phi_\sigma}_{AB}\) as the \(\SWAP\) state defined by the matching
  \begin{align}
    \left\{(i_A, \sigma[i]_B)\right\}_{i\in[k_L]}.
  \end{align}
\end{definition}
Given a \(k_L\)-qubit state \(\ket{\psi} = \sum_{x \in \F^n} \alpha_x \ket{x}\) and a permutation state \(\ket{\Phi_\sigma}_{AB}\), the application of a teleportation circuit between \(\ket{\psi}\) and the \(A\) register will result in the \(B\) register being left in the state \(\sum_{x\in\F^n} \alpha_x \ket{\sigma x}\).
I.e. the permutation \(\sigma\) is applied to the state.
We call this operation \(\PERM(\sigma)\).
An example of this state is illustrated in \Cref{fig:perm-swap-state}.
For \(i \in \mathbb{Z}_{k_L}\), we use \(\ket{\ROT(i)}\) to refer to the permutation state with \(\sigma\) such that consumption of the permutation state implements the operation \(\ROT(i)\).

\begin{prop}[Classically controlled permutation]\label{prop:controlled-perm}
  Given a permutation state \(\ket{\Phi_\sigma}_{AB}\), using classically controlled Pauli operations, transversal \(\CNOT\), and transversal \(\SWAP\), we can apply classically controlled \(\PERM(\sigma)\) in depth \(O(\log n_L)\).
\end{prop}
\begin{proof}
  Let \(C\) be the input register.
  We will conditionally apply the standard teleportation circuit transversally between registers \(A\) and \(C\):
  If the control bit is \(1\), we perform the teleportation circuit \(M_Z(C) M_X(A)\CNOT(A,C)\) to yield measurement outcomes \(z,x \in \F^{k_L}\).
  For each \(i \in [k_L]\), we apply the controlled Pauli operation \(X^{z[i]}_{\sigma(i)} Z_{\sigma(i)}^{x[i]}\) to the \(\sigma(i)\)-th qubit of \(B\).
  Finally, we peform \(\SWAP(B,C)\) to swap register B with register C.
  If the control bit is \(0\), then we do nothing for the \(O(1)\) timesteps required to execute the other branch of the circuit.
\end{proof}

\begin{prop}[Arbitrary \(\SWAP(i_A,j_B)\)]\label{prop:arb-swap}
  Using \(4\nu\) classically controlled \(\ROT\) operations that are independent of \(i,j\) and one \(\SWAP(1_A,1_B)\) operation we can implement \(\SWAP(i_A,j_B)\).
  We can further implement \(\SWAP(i_A,j_A)\) using an additional two \(\SWAP(1_A,1_B)\) operations and a second register where the state is arbitrary and preserved.
\end{prop}
\begin{proof}
  We first show how to implement \(\ROT(s)\) for arbitrary \(s \in \mathbb{Z}_{k_L}\):
  Recall that \(k_L = 2^{\nu}\).
  Write the binary expansion of \(s = \sum_{m=1}^{\nu} a_m2^{m-1}\) where \(a_m \in \{0,1\}\).
  For each \(m \in [\nu]\), implement \(\ROT(2^{m-1})\) conditioned on \(a_m\).
  \(\ROT(a)\ROT(b) = \ROT(a+b)\), so the implemented operation is \(\prod_{m=1}^{\nu}\ROT(a_m 2^{m-1}) = \ROT\left(\sum_{m=1}^{\nu} a_m2^{m-1}\right) = \ROT(s)\).

  The result for \(\SWAP(i_A,j_B)\) follows from the decomposition:\footnote{Recall that the composition of maps should be read from right to left.}
  \begin{align*}
    \SWAP(i_A,j_B) = \ROT_B(j-1)\ROT_A(i-1)~\SWAP(1_A,1_B)~\ROT_B(1-j)\ROT_A(1-i).
  \end{align*}
  We rotate such that the qubits to be swapped are at the first coordinate, swap, and then apply the inverse rotation.
  For \(\SWAP(i_A,j_A)\), we will reduce to a version of the previous case using a scratch register \(C\).
  We rotate qubit \(j_A\) to the first position and then use the register \(C\) to hold it while we perform \(\SWAP(1+i-j)_A,1_C\).
  \begin{align*}
    \mathsf{SETUP} &= \ROT_A(j-i)\SWAP(1_A,1_C)\ROT_A(1-j)\\
    \mathsf{SETUP}^{-1} &= \ROT_A(j-1)\SWAP(1_A,1_C)\ROT_A(i-j) \\
    \SWAP(i_A,j_A) &= \mathsf{SETUP}^{-1} \SWAP(1_A,1_C) \mathsf{SETUP}.
  \end{align*}
  The qubit initially at position \(j_A\) follows the sequence of positions (left to right) \((j_A, 1_A, 1_C, 1_C, 1_A, (1+i-j)_A, (1+i-j)_A, i_A)\), and the qubit initially at position \(i_A\) follows the sequence of positions \((i_A,(1+i-j)_A,(1+i-j)_A,1_A,1_C,1_C,1_A,j_A)\).
  Since register \(C\) is never rotated, the qubit initially at position \(1_C\) ends at position \(1_C\).
\end{proof}
 
\begin{figure}
    \centering
    \includegraphics{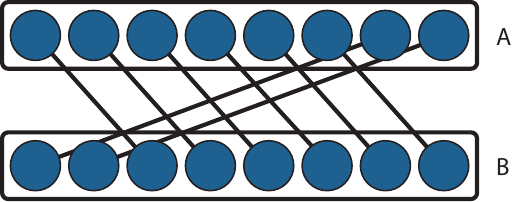}
    \caption{Permutation state corresponding to \(k_L = 8\) with the permutation \(\begin{pmatrix}1 & 2 & 3 & 4 & 5 & 6 & 7 & 8\\ 3 & 4 & 5 & 6 & 7 & 8 & 1 & 2\end{pmatrix}\) i.e. a cyclic shift by \(2\). We also refer to this state as \(\ket{\ROT(2)}\). Lines indicate Bell pairs between registers \(A\) and \(B\).}
    \label{fig:perm-swap-state}
\end{figure}

\subsubsection{Primitive operations}\label{sec:primitive-gateset}
We now give our primitive gate set that we will construct gadgets for.

Let \(\PAULI\) denote the set of 1-qubit Pauli operators.
We take \(H\), \(S\), and \(\CNOT\) as a generating set of the Clifford group and denote this set \(\CLIFF\).
When we augment \(\CLIFF\) by (projective) \(Z\)-basis measurement, we will use \(\CLIFFM\).
We use \(\MEASD_\mathsf{P}\) to denote standard (``destructive") measurements in the basis \(\mathsf{P}\).
Whereas we use \(\MEAS_Z\) to denote computational basis measurements that additionally output a computational basis state initialized to the measurement outcome (``quantum non-demolition").
As a first step, we will decompose any circuit into the gate set \(\PAULI+\CLIFFM+\CCZ\).

However, we will not implement all of these operations directly.
Let \(A\) and \(B\) be arbitrary registers.
Our primitive gateset can be split into three groups.
The first group will be used for implementing all other gadgets via teleportation:
\begin{itemize}
\item \(\MEASD_Z(A)\)
\item \(\MEASD_X(A)\)
\item \(\CNOT(A,B)\)
\item \(\PAULI\)
\end{itemize}
The second group will be used for gate targeting (i.e. addressing logical qubits):
\begin{itemize}
\item \(i \in [\log_2 k_L]\)~\(\ROT(2^i)\)
\item \(\SWAP(1_A, 1_B)\)
\end{itemize}
The final group is used only for computation:
\begin{itemize}
\item \(H(1_A)\)
\item \(S(1_A)\)
\item \(\MEAS_Z(1_A)\)
\item \(\CNOT(1_A, 2_A)\)
\item \(\CCZ(1_A, 2_A, 3_A)\)
\end{itemize}
The second and third groups of gates we call the \emph{primitive computational gates}.

\subsubsection{Register assignment}
\begin{definition}[Register assignment]
  For a circuit \(C\) on the set of qubits \([W]\) and \(m,k\in\mathbb{N}\) such that \(W \le mk\), an \((m,k)\)-register assignment is an injective map \(\phi \colon [W] \to [m]\times[k]\).
\end{definition}

\begin{prop}[Gate compatible register assignment]\label{prop:gate-compat-assignment}
  For \(k\ge 12\), \(m \ge 4W/k\), and one layer \(C\) of non-overlapping Clifford + CCZ gates on \(W\) qubits, there exists a \((m,k)\)-register assignment \(\phi \colon [W] \to [m] \times [k]\) such that each gate in \(C\) is supported only on the qubits \(\phi^{-1}(a,[k])\) for some \(a \in [m]\).
  Furthermore, this register assignment can be found in time \(O(|C|)\).
\end{prop}
\begin{proof}
  First note that the maximum support size of any gate in \(C\) is at most \(3\).
  This is an instance of the \emph{bin packing problem} which requires that we assign a large number of parcels taking different sizes to bins such that no bin exceeds its capacity.
  Our case is that of \emph{high-multiplicity bin packing} for which the number of distinct sizes is \(O(1)\), and efficient approximate algorithms are known \cite{filippi2005asymptotically}.
  
  Here, we describe the first-fit bin packing algorithm \cite{garey1972worst} with a single open bin:
  Maintain pointers \(a \in [m]\) and \(b \in [k]\) initialized to \(1\), called the open bin and load, respectively.
  For each gate \(c\) in the gate list \(C\):
  \begin{enumerate}
    \item If \(\alpha := |\supp c| > k - b + 1\), set \(b = 1\) and increment \(a\).
    \item Assign the qubits \(\supp c\) to positions \(\{b, \dots, b+\alpha\}\) of register \(a\).
    \item Set \(b = b + \alpha\).
  \end{enumerate}
  At least \(\lfloor \frac{k}{3}\rfloor\) gates can be assigned to each register and there are at most \(W\) gates, so at most \(\frac{W}{\frac{k}{3}-1} \le \frac{4W}{k} \le m\) registers will be assigned to.
\end{proof}

\subsubsection{Serialization}
\begin{definition}[Serialized circuits]
  For a circuit \(C\) acting on \(m\) registers of size \(k\), an \(\alpha\)-serialization of \(C\) is a new circuit \(C'\) such that, for every timestep, there are at most \(\frac{m}{\alpha}\) gates and no two gates act on the same register.
\end{definition}

Clearly, given a \(1\)-serialization, it is straightforward to construct an \(\alpha\)-serialization \(\alpha \ge 1\) by a greedy assignment when we do not need to worry about gate conflicts.
\begin{prop}[Serialization subdivision]\label{lemma:serialization-single}
  Given a \(1\)-serialized depth-1 circuit \(C\) acting on \(m\) registers of size \(k\) such that no gate is supported on more than one register and \(\alpha \in [1,m]\)\, we can efficiently compute an \(\alpha\)-serialized circuit \(C'\) of \(C\) such that the depth of \(C'\) is \(2\alpha\).
\end{prop}
\begin{proof}
    Eagerly execute the gates in \(C\) in one of \(2\alpha\) timesteps such that no more than \(m/\alpha\) gates are executed in each timestep.
    In each timestep, we are able to execute \(\lfloor m/\alpha \rfloor\) gates.
    \(C\) contains at most \(m\) gates since it is \(1\)-serialized.
    Thus, after \(2\alpha\) timesteps, there are \(2 \alpha\lfloor m/\alpha \rfloor \ge m\) opportunities to execute a gate.\footnote{Bounding the cases \(\alpha \in [0,m/2]\) and \(\alpha \in [m/2, m]\) separately.}
\end{proof}

We now construct a serialization suitable for general circuits with gates supported on multiple registers.
\begin{lemma}[Serialization]\label{lemma:serialization}
  Given a depth-1 circuit \(C\) acting on \(m\) registers of size \(k\) such that no gate is supported on more than two registers and \(\alpha \in [1,m]\), we can efficiently compute an \(\alpha\)-serialized circuit \(C'\) of \(C\) such that the depth of \(C'\) is \(4k\alpha\).
\end{lemma}
\begin{proof}
  Two gates are said to be \emph{non-conflicting} if they are supported on qubits in different registers.
  Our procedure will divide the gates of \(C\) into \(2k\) sets of pair-wise non-conflicting gates, and each set will be executed over \((\alpha + 1)\) timesteps.

  Let \(U\) be the set of gates in \(C\) supported on two registers.
  Then, consider the multigraph \(G = ([m],E)\) where each element \((u,v) \in E\) corresponds to a gate of \(U\) with targets in registers \((w_1,w_2)\).
  \(G\) has degree and multiplicity at most \(k\), so a proper edge coloring \(c \colon E \to [2k]\) using at most \(2k\) colors can be efficiently computed \cite{berge1991short}.

  We first partition \(C = Q_1\sqcup Q_2 \dots \sqcup Q_{2k}\) where, for \(i \in [2k]\), \(Q_i\) contains all gates of \(U\) corresponding to the set of edges \(E_i = c^{-1}(i)\) as well as an arbitrarily selected subset of single register gates \(C \setminus U\) such that all elements of \(Q_i\) are pairwise non-conflicting.
  There are at most \(k\) gates of \(C\) supported on each register, and each gate in \(C \setminus U\) is supported on one register.
  Thus, a greedy assignment of elements of \(C \setminus U\) to each \(Q_i\) will succeed.
  
  For each \(i \in [2k]\), we further subdivide \(Q_i = \sqcup_{j \in [\alpha + 1]}Q_{(i,j)}\) such that each subset has at most \(m  /\alpha\) elements i.e. \(|Q_{(i,j)}| \le m/\alpha\) using an argument identical to \cref{lemma:serialization-single}.

  \(C'\) has timesteps \((i,j) \in [2k]\times [\alpha + 1]\) where in timestep \((i,j)\), we execute the gates in \(Q_{(i,j)}\).
  The precise ordering of the timesteps is arbitrary since all gates in \(C\) had disjoint qubit support.
\end{proof}

\begin{lemma}[Compiling]\label{lemma:compiling}
  Assume \(k_L \ge 12\) and \(m \ge 4W/k_L\).
  Given a Clifford + \(\CCZ\) circuit \(C\) acting on \(\mcH_S\) of depth \(T\), we can efficiently construct a new circuit \(C'\) acting on \(\mcH_L\) of depth \(T'\) where
  \begin{itemize}
  \item All operations are primitive computational operations (\cref{sec:primitive-gateset}).
  \item Only one operation is supported on each register per timestep. (\(1\)-serialized)
  \item Only one type of primitive computational operation is used per timestep.\footnote{We treat \(\ROT\) with different parameters as distinct types.}
  \item \(\frac{T'}{T} = O(\nu k_L) = O(k_L\log k_L)\).
  \item $C$ and $C'$ implement the same quantum channel.
  \end{itemize}
\end{lemma}
\begin{proof}
  Without loss of generality, assume that \(C\) contains only operations in \(\CLIFF_{\MEAS}+\CCZ\) with only one type of gate per layer.
  Otherwise, we may first replace any two-qubit gate in the Clifford group by a constant number of gates in \(\CLIFF_{\MEAS}\), and then replace each layer by a constant number of layers, each containing only one type of gate per layer.
  
  For each layer \(C_t\) of \(C\), we compute a register assignment \(\phi_t\) using \cref{prop:gate-compat-assignment} such that the gates of \(C_t\) have support only on one register each.
  Let \(S_t=S_t^{(2)}S_t^{(1)}\) be a depth-2 layer of \(\SWAP\) gates acting on \(\mcH_L\) that implements the permutation sending qubit \(\phi_t(w)\) to \(\phi_{t+1}(w)\) for each \(w \in [W]\).
  We will seek to implement the circuit:
  \begin{align}
    \left(C_D\right)_{\phi_D} S_{D-1} \left(C_{D-1}\right)_{\phi_{D-1}} \dots S_2 \left(C_2\right)_{\phi_2} S_1 \left(C_1\right)_{\phi_1}.
  \end{align}
  However, we must decompose it further.
  
  For each \(v \in \{1,2\}\), \(S_t^{(v)}\) is the product of operations of the form \(\SWAP(i_A,j_B)\).
  We can decompose \(S_t^{(v)}\) into primitive computational operations by first constructing the \(1\)-serialization (\cref{lemma:serialization}) and then further decompose each \(\SWAP(i_A,j_B)\) operation into power-of-two \(\ROT\) and \(\SWAP(1_A,1_B)\) (\cref{prop:arb-swap}).
  Let \(S_t'\) refer to this decomposition of \(S_t\).
  The depth of \(S_t'\) is \(O(\nu k_L)\).

  Let \(C_t^{(k_L)} C_t^{(k_L-1)}\dots C_t^{(1)}\) be the composition of single layers of operations that is a 1-serialization of \((C_t)_{\phi_t}\).\footnote{Since \((C_t)_{\phi_t}\) only contains gates acting on single registers, this is trivial.}
  For each \(v \in [k_L]\), we (in parallel) swap the support of each gate in \(C_t^{(v)}\) (\cref{prop:arb-swap}) to the first coordinates of the register, perform the corresponding primitive computational operation, and swap the support back.
  Since each gate has support at most \(3\), this can be done in \(O(\nu)\) steps.
  Let \(C_t'\) refer to this sequence of operations (depth \(O(\nu k_L)\)).

  \(C'\) is now simply:
  \begin{align}
    C'  = C_D' S_{D-1}' C_{D-1}' \dots S_2' C_2' S_1' C_1'.
  \end{align}
  By construction \(C'\) satisfies all conditions of the lemma statement.
\end{proof}

\subsection{Main theorem with vanishing threshold}\label{subsec:main-thm-vanishing-threshold}
Here we now prove the main result with a \emph{sub-constant} threshold. This sub-constant threshold is due to our use of the adversarial-noise decoder for the~\emph{sub-linear-distance} quantum locally testable code from~\cref{thm:single-shot-qltc}.
It is reasonable to expect that a decoder capable of decoding ``most'' errors of linear weight exists (a stochastic-noise single-shot decoder).
However, we find that it is simpler to perform a threshold amplification step (\cref{subsec:main-thm-constant-threshold}) instead to amplify the threshold to constant.
\begin{theorem}[Main result with vanishing threshold]\label{thm:main-result-vanish-threshold}
  There exists a function \(f(x)\) growing faster than any \(\quasipoly(x)\) and a value \(\epsilon_{*}\in (0,1)\) such that for any \(\epsilon_L \in (0,1)\) and (Clifford+\(\CCZ\)) classical input / classical output quantum circuit \(C\) with width \(W\) and depth \(D\) satisfying
  \begin{align}
    \frac{WD}{\epsilon_L} &\le f(W)
  \end{align}
  There is a corresponding efficiently constructable  classical input / classical output quantum circuit \(\overline{C}\) with width \(\overline{W}\) and depth \(\overline{D}\) satisfying
  \begin{align}
    \frac{\overline{W}}{W} &= O_{W \to \infty}(1)\\
    \frac{\overline{D}}{D} &= O_{W \to \infty}\left(\left(\log\frac{W D}{\epsilon_L}\right)^{1 + o(1)}\right)
  \end{align}
  and using auxiliary $O(1)$-time classical computation per quantum time step,
  such that the following guarantees hold. There is a family of bad fault paths \(\mathcal{G}\) such that
  \begin{itemize}
  \item For any \(\mathcal{G}\)-avoiding physical fault \(\mathbf{f}\), the output distribution of \(\overline{C}\) subject to \(\mathbf{f}\), \(\overline{C}[\mathbf{f}]\), is equal to the output distribution of \(C\).
  \item \(\weightenum{\mathcal{G}}{x} \le \epsilon_L\) for \(x \in \left[0, \epsilon_* \cdot c\left(\frac{WD}{\epsilon_L}\right)\right]\) where \(c\left(\frac{WD}{\epsilon_L}\right) = \Omega_{W\to\infty}\left(\frac{1}{\polyloglog \frac{W D}{\epsilon_L}}\right)\).
  \end{itemize}
\end{theorem}

It is an immediate corollary of the above theorem that the output distribution of \(\overline{C}\) subject to a random fault distributed according to a \(p\)-locally stochastic faults model is \(\epsilon_L\)-close in total-variation distance (TVD) for \(p \le O\left(\frac{1}{\polyloglog \frac{W D}{\epsilon_L}}\right)\).
However, maintaining the weight enumerator allows us to work with the circuit further in the next section.

Our proof is split into roughly two parts which we organize into parts for readability that should be read as a single proof.

  \noindent \emph{Proof.} \paragraph{Setup and compilation}  We use the code family from \cref{thm:single-shot-qltc}.
  If the number of logical qubits is not a power-of-two, we ignore some fraction \(<1/2\) of them, such that the remaining number of logical qubits is \(k_L = 2^{\nu}\).
  This harms our rate by at most \(1/2\).
  There exists a minimum (constant) computational code size \(n_{\mathrm{min}}\) such that all properties required in \cref{sec:globals} hold.
  
  Since the error analysis depends on \(n_L\), we will do it for a general \(n_L\) and pick \(n_L\) at the end.
  We will indicate the circuits that appear in the main steps with a superscript \((n_L)\) to denote that \(n_L\) has not yet been selected.
  For now, we simply require \(n_{\mathrm{min}} \le n_L \).
  
  We use a number of registers \(m = \left\lceil W / k_L \right\rceil\) (using the constant rate).
  We apply \Cref{lemma:compiling} to \(C\) to arrive at an equivalent circuit \(C'^{(n_L)}\) with depth overhead \(O(k_L \log k_L)\) which: 1) Only contains primitive computational operations. 2) Only applies one type of primitive computational operation per timestep. 3) At most one primitive computational operation is applied per register per timestep (is \(1\)-serialized).
  
  \paragraph{Fault-tolerant circuit construction} Our fault-tolerant circuit \(\overline{C}^{(n_L)}\) will have a computational code block for each register of \(C'^{(n_L)}\).
  Let \(D' = D \cdot O(k_L \log k_L)\) be the depth of \(C'^{(n_L)}\) and for each timestep \(t \in [D']\) of \(C'^{(n_L)}\), let \(C'_t\) be the corresponding layer of gates in \(C'^{(n_L)}\).
  The timesteps of \(\overline{C}^{(n_L)}\) will be partitioned into \(D'\) \emph{work periods} where, in each work period, we accomplish the work of a single layer of gates of \(C'^{(n_L)}\).
  
  Let \(K\) be the number of copies of the resource state prepared by a single stabilizer state preparation gadget (\cref{lemma:stab-resource-state}), and define \(\beta = \lceil m/K \rceil\) the number of stabilizer state preparation gadgets required to \(m\) output states.

  We begin the circuit by executing \(\beta\) stabilizers state preparation gadgets to prepare \(m\) \(\ket{\overline{0}^{\otimes k_L}}\) states.\footnote{It is notable that our gadget does not require any input \(\ket{\overline{0}}\) states to operate on, so we can use it to initialize the computation.}
  This takes time \(O(\polylog n_L)\).
  
  Now, consider a given work period \(t \in [D']\).
  By construction, only a single type of primitive computational operation \(\mathsf{O}\) (\cref{sec:primitive-gateset}) is applied in \(C_t'\) and each register is only involved in at most one primitive computational operation.
  I.e. there are at most \(m\) operations to be performed.
  In the following, the $\mathsf{Id}$ gadget in~\Cref{lemma:computational-code-transversal-gates} is repeatedly applied to any of the \(m\) computational code blocks that are not otherwise involved in a gadget.
    
  If \(\mathsf{O}\) is a unitary Clifford operation or \(\MEAS_Z(1_A)\) (projective computational basis measurement), let \(\ket{\mathsf{O}}\) denote the corresponding stabilizer resource state (defined in \cref{sec:prelim}).
  This resource state is supported on at most \(4\) registers, so we can use the stabilizer state preparation gadget (\cref{lemma:stab-resource-state}) to prepare it.
  We run \(\beta = \lceil m/K \rceil\) state preparation gadgets in parallel to produce \(m\) copies of the encoded resource state \(\ket{\overline{\mathsf{O}}}\) and then consume it via state teleportation (\cref{lemma:computational-code-transversal-gates}) to perform \(\mathsf{O}\) (\cref{sec:prelim}).
  If \(\mathsf{O}\) is \(\MEAS_Z(1_A)\) we additionally store the logically decoded value.
  This takes time \(O(\polylog n_L)\).

  Otherwise, \(\mathsf{O}\) is \(\CCZ(1,2,3)\).
  We use the magic state preparation gadget (\cref{lemma:magic-resource-state}) to prepare \(\ket{\overline{\CCZ(1,2,3)}}\).
  Let \(M_\textrm{magic}\) be the variable from the statement of \Cref{lemma:magic-resource-state}, so that a single magic state preparation gadget produces \(K_\textrm{magic} = \Omega\left(M_\textrm{magic} n_L^{-\tilde{\gamma}(n_L)}\right)\) copies of \(\ket{\overline{\CCZ(1,2,3)}}\) and uses space \(O(M_\textrm{magic} n_L)\).
  We execute an \(\alpha\)-serialized version\footnote{Recall that \(C_t\) is already \(1\)-serialized.} of \(C_t'\) (\cref{lemma:serialization-single}) over \(2\alpha\) steps with \(\alpha = \lceil M_\textrm{magic}/K_\textrm{magic} \rceil = O\left(n_L^{\tilde{\gamma}(n_L)}\right)\).
  In each step, we run \(\beta_\textrm{magic} = \left\lceil m/M_\textrm{magic} \right\rceil\) magic state preparation gadgets to produce the \(\frac{m}{\alpha}\) magic states required to execute a single step of the \(\alpha\)-serialized circuit.
  In parallel, we also run 3 sets of \(\beta\) stabilizer state preparation gadgets, preparing \(m\) each of the resource states:
  \(\ket{\overline{\CZ(1,2)}}\), \(\ket{\overline{\CZ(2,3)}}\), and \(\ket{\overline{\CZ(1,3)}}\).
  For each of the \(\frac{m}{\alpha}\) gates in each step, we use these states and the transversal gate gadgets from~\Cref{lemma:computational-code-transversal-gates} to execute a teleported \(\CCZ(1,2,3)\) gate.
  Overall, this takes time \(O(n_L^{\tilde{\gamma}(n_L)} \polylog n_L)\). Again, while waiting for state preparation gadgets, the $\mathsf{Id}$ gadget in~\Cref{lemma:computational-code-transversal-gates} is repeatedly applied on each of the $m$ computation data registers.

  At the end of the last layer of gates, we output the decoded results of any \(\MEAS(1_A)\) gadgets and discard the remainder of the measurement record and classical memory.

  Thus, the overall depth of \(\overline{C}^{(n_L)}\) is \(O\left(D' n_L^{\tilde{\gamma}(n_L)} \polylog n_L \right) = O\left(D n_L^{1+\tilde{\gamma}(n_L)} \polylog n_L\right)\).

  \paragraph{Fault analysis and reduction to Pauli noise}
  Having constructed the circuit \(\overline{C}^{(n_L)}\), we proceed to the fault analysis.
  Following standard techniques (\cite{aharonov1999fault},\cite{aliferis2005quantum},\cite{gottesman2013fault}), we will reduce the analysis of an arbitrary fault to that of Pauli faults.
  We will define \(p_{*} = \min(\epsilon_{*\mathrm{EC}},\epsilon_{*,\mathrm{OP}},\epsilon_{*,\mathrm{stab}},\epsilon_{*,\mathrm{magic}})\) from \cref{lemma:ec-gadget}, \cref{lemma:computational-code-transversal-gates}, \cref{lemma:stab-resource-state}, and \cref{lemma:magic-resource-state}, respectively.
  
  We will denote (non-diagonal) Pauli faults with a superscript \((\mathsf{P})\).
  All gadgets used are Pauli fault-tolerant gadgets that accept and output \((C_L,\mathcal{F}_L)\)-blocks, so they are compatible.\footnote{Recall the definition of compatible FT gadgets (\cref{def:ft-gadget}).}
  This allows us to inductively apply the gadget composition lemmas (\cref{lemma:gadget-composition} and \cref{lemma:gadget-side-by-side}), so the overall circuit \(\overline{C}^{(n_L)}\) is a \((C,\mathcal{G}^{(n_L)})\)-Pauli fault-tolerant gadget where \(\mathcal{G}^{(n_L)}\) is the sum \((\boxplus)\) of the bad fault paths for each gadget.
  I.e. for a Pauli fault \(\mathbf{f}^{(\mathsf{P})}\), the output distribution\footnote{Distribution is a slight abuse of terminology since, in general, there is an overall complex amplitude that will be summed over at the end. It is at this point that we discard the measurement record of the gadgets and sum over all possible measurement outcomes (see \cref{lemma:deterministic-errors}).} of \(\overline{C}^{(n_L)}[\mathbf{f}^{(\mathsf{P})}]\) is proportional to the output distribution of \(C\) when \(\mathbf{f}^{(\mathsf{P})}\) is \(\mathcal{G}^{(n_L)}\)-avoiding.

  Let \(\ket{\mathsf{input}}\) be the classical input string (as a classical register).
  For an arbitrary \(\mathcal{G}^{(n_L)}\)-avoiding physical fault \(\mathbf{f}\) (not necessarily Pauli), it can be decomposed into a sum of Pauli faults with the same fault path (operations are component-wise) and complex coefficients \(\alpha\).
  \begin{align*}
    \mathbf{f} = \sum_{\substack{\text{Pauli faults }\mathbf{g}^{(\mathsf{P})} \\ \supp\mathbf{g}^{(\mathsf{P})}\subseteq \supp\mathbf{f}}}\alpha_{\mathbf{g}^{(\mathsf{P})}}\mathbf{g}^{(\mathsf{P})}.
  \end{align*}
  We can now compute the output probability of \(\overline{C}^{(n_L)}\) subject to an arbitrary fault.
  If \(\overline{C}^{(n_L)}\) is subject to a Pauli fault \(\mathbf{g}^{(\mathsf{P})}\) that is \(\mathcal{G}^{(n_L)}\)-avoiding, then the output distribution is proportional to \(C\), i.e. \(\overline{C}^{(n_L)}[\mathbf{g}^{(\mathsf{P})}] \propto C\).
  Using linearity, for some real constants \(c\) and \(\{c_{\mathbf{g}^{(\mathsf{P})}}\}_{\mathbf{g}^{(\mathsf{P})}}\) independent of \(x\), we can write (recall that quantum operations are superoperators)
  \begin{align*}
    \Pr(\overline{C}^{(n_L)}[\mathbf{f}]\text{ outputs }x) &= \bra{x}\left(\overline{C}^{(n_L)}[\mathbf{f}]\left(\ketbra{\mathsf{input}}{\mathsf{input}}\right)\right)\ket{x}\\
                                                             &= \sum_{\substack{\text{Pauli faults }\mathbf{g}^{(\mathsf{P})} \\ \supp\mathbf{g}^{(\mathsf{P})}\subseteq \supp\mathbf{f}}} \alpha_{\mathbf{g}^{(\mathsf{P})}} \bra{x}\left(\overline{C}^{(n_L)}[\mathbf{g}^{(\mathsf{P})}]\left(\ketbra{\mathsf{input}}{\mathsf{input}}\right)\right)\ket{x}\\
                                                             &=\left(\sum_{\substack{\text{Pauli faults }\mathbf{g}^{(\mathsf{P})} \\ \supp\mathbf{g}^{(\mathsf{P})}\subseteq \supp\mathbf{f}}} \alpha_{\mathbf{g}^{(\mathsf{P})}} \cdot c_{\mathbf{g}^{(\mathsf{P})}} \right)\bra{x}\left(C^{(n_L)}\left(\ketbra{\mathsf{input}}{\mathsf{input}}\right)\right)\ket{x}\\
    &= c \bra{x}\left(C^{(n_L)}\left(\ketbra{\mathsf{input}}{\mathsf{input}}\right)\right)\ket{x}.
  \end{align*}
  Since \(\mathbf{f}\) is a physical fault (i.e. the channels are CPTP) and \(C\) is composed of quantum operations (also CPTP channels), \(P(x) = \Pr(\overline{C}^{(n_L)}[\mathbf{f}]\text{ outputs }x)\) is a normalized probability distribution over bitstrings.
  Therefore, it must be the case that \(c=1\).
  We conclude that \(\overline{C}^{(n_L)}\) is a \((C,\mathcal{G}^{(n_L)})\)-fault-tolerant gadget (under the restriction that the fault is a physical fault) i.e. it is fault tolerant to arbitrary physical faults with a \(\mathcal{G}^{(n_L)}\)-avoiding fault path.

  \paragraph{\(n_L\) selection}
  We now upper bound \(\mathcal{G}^{(n_L)}\) and pick \(n_L\) appropriately.
  There are at most \(A=O(W D \poly(n_L))\) gadgets, each with a weight enumerator upper bounded by \(p_L(x) \le e^{-\zeta(n_L) + O(\polylog n_L)}\) on \(x \in [0, \ndep{\epsilon_{*}}]\).
  Thus, on \(x \in [0, \ndep{\epsilon_{*}}]\),
  \begin{align*}
    \weightenum{\mathcal{G}^{(n_L)}}{x} &\le A p_L(x) \\
                              &\le O(W D \poly(n_L)) e^{-\zeta(n_L) + O(\polylog n_L)}\\
                              &\le W D e^{-\zeta(n_L) +  g(n_L)}.
  \end{align*}
  for some \(g(n_L) \le O(\polylog(n_L))\)
  Let \(y\) be the smallest solution to
  \begin{align}
    \zeta(y) \ge \log \frac{W D}{\epsilon_L} + g(y).
  \end{align}
  Since \(\zeta(x) = x/\polylog x\) and \(g(n_L) \le O(\polylog(n_L))\), \(y = \Theta\left(\log \frac{W D}{\epsilon_L} \cdot \polyloglog \frac{W D}{\epsilon_L}\right)\).
  Now let \(n_L'\) to be the smallest element of the computational code \(N_i\) greater than \(\min(n_{\mathrm{min}},y)\).
  Since our computational code satisfies \(N_{i+1}/N_i = C\) for some constant \(C > 1\), we have the bounds
  \begin{align}
    n_L' = \Theta\left(\log \frac{W D}{\epsilon_L} \cdot \polyloglog \frac{W D}{\epsilon_L}\right).
  \end{align}
  We now set \(\overline{C} = \overline{C}^{(n_L')}\).
  So that for \(x \in [0,\ndep{\epsilon}_*]\),
  \begin{align*}
      W_{\mathcal{G}}(x) \le \epsilon_L~.
  \end{align*}
  \(\ndep{\epsilon_*} = \epsilon_* \frac{\zeta(n_L')}{n_L'} = \epsilon_* \cdot \Omega\left(\frac{1}{\polyloglog \frac{W D}{\epsilon_L}}\right)\)
  
  Using the upper bound on \(\frac{WD}{\epsilon_L}\), \(m \le W\), and the space bounds of the gadgets, it follows that
    \begin{align}
      \frac{\overline{W}}{W} &= \frac{1}{W} O\left(\beta M + \beta_\mathrm{magic} M_\mathrm{magic} + m n_L \right) \\
                           &= \frac{1}{W} O\left(\left(\frac{m}{K} + 1\right) e^{\Theta(\log n_L \cdot \loglog n_L)} + \left(\frac{W}{n_L} + 1\right) n_L\right)\\
                           &= O\left(\frac{e^{\Theta(\log n_L \cdot \loglog n_L)}}{W}  + \frac{n_L}{W} + 1\right) \\
                           &= O\left(\frac{\quasipoly n_L}{W} + 1\right)\\
                           &= O\left(\frac{\quasipolylog \frac{WD}{\epsilon_L}}{W} + 1\right)\label{eq:log-vanishing-width-bound}
    \end{align}
    We now select \(f(x)\).
    There exists absolute constants \(c \in (0,1) \) and \(c' \ge 0\) such that we can define \(f(x) = \exp\left(\exp\left(c' (\log x)^{c}\right)\right)\) to be a quickly growing function such that when \(\frac{WD}{\epsilon} \le f(W)\), the right hand side of \cref{eq:log-vanishing-width-bound} is \(O(1)\).
    We note that \(f(x)\) grows faster than any \(\quasipoly(x)\) but slower than any \(\exp(\poly(x))\).
    Using \(\frac{WD}{\epsilon} \le f(W)\), we have that 
    \begin{align*}
        \frac{\overline{W}}{W} = O(1)~.
    \end{align*}
    
    We move on to the time bound which is
    \begin{align*}
    \frac{\overline{D}}{D} &= O\left(\left(\log \frac{W D}{\epsilon_L}\right)^{1+\tilde{\gamma}\left(\tilde{\Theta}\left(\log \frac{W D}{\epsilon_L}\right)\right)} \right)
  \end{align*}
  where \(\tilde{\Theta}(\cdot)\) suppresses doubly-logarithmic factors and \(\tilde{\gamma}(n) = o(1)\) is defined in the proof of \cref{lemma:magic-resource-state}.~\qedsymbol{}

\subsection{Main theorem}\label{subsec:main-thm-constant-threshold}
Having established that we can achieve nearly-logarithmic time overhead and constant space overhead with a threshold that vanishes as \(\frac{1}{\polyloglog \frac{W D}{\epsilon_L}}\) (\cref{thm:main-result-vanish-threshold}), the last task is to amplify our threshold to a constant.
One (nearly trivial) option is to use the AB concatenated code construction again (as defined in \cref{sec:concatAGP}) with \(r \approx \log\log\log\log \frac{W D}{\epsilon_L}\) where gates and qubits in \(\overline{C}\) are replaced by their concatenated code counterparts.
\cref{lemma:composition-upper-bound} will give us an upper bound on the weight enumerator of the bad fault paths of the resulting circuit that allows us to establish a constant threshold.
However, this introduces a space and time overhead of \(O\left(\polylogloglog \frac{W D}{\epsilon}\right)\) and we desire \(O(1)\) space overhead.

The strategy will be to simulate our circuit with the concatenated codes of \cite{yamasaki2024time}.
Intuitively, the simulated circuit will see a noise model that is not very different from independent noise.
A precise implementation of this program is the subject of \cite{he2025composable} that greatly expands the weight enumerator formalism from \cref{sec:weight-enum-formalism} and proves an extension of \cref{lemma:gadget-composition} suitable for the recursive simulation.
A proof of \cref{thm:yk-sim} (using definitions from \cref{sec:weight-enum-formalism}) can be accomplished using the standard gadgets from \cite{yamasaki2024time} or \cite{aharonov1999fault} on the code family with parameters \(r \in \mathbb{N}\), \([[\left(2^r - 1\right)^2, \left(2^r - r - 1\right)^2, 9]]\) given by concatenating each quantum Hamming code with itself once and then forgetting about the concatenated structure.
This is essentially standard.
Here, we omit these details and state the following implication of \cite{yamasaki2024time} within the framework of our formalism.

We remark that \cite{yamasaki2024time} uses the gate set Clifford+T.
\(\CCZ\) can be exactly simulated by Clifford+T, so this distinction is not important in our use.

For circuit on a set of qubits partitioned into registers, a \emph{register location} is the straightforward generalization of location with qubits replaced by registers.
Likewise, we generalize fault paths to \emph{register fault paths}.
Occasionally, we will promote a fault path to a register fault path in the canonical way by replacing each location by the register location that it is supported in.
\begin{claim}[Modified version of \cite{yamasaki2024time}]\label{thm:yk-sim}
  There exists a constant value\footnote{Think of \(\epsilon_{*,YK}\) as something like \(1/2\) the threshold value computed in \cite{yamasaki2024time}.} \(\epsilon_{*,YK} \in (0,1)\), such that: for \(r \in \mathbb{N}\) and a classical input-classical output circuit \(C\) using \(W\) qubits and depth \(D\), and a partitioning of qubits into registers of size \(k_{\mathrm{YK}} = e^{\Theta(r^2)}\).
  \begin{itemize}
  \item There is a new circuit \(\overline{C}\) with width \(\overline{W}\) and depth \(\overline{D}\) such that \(\overline{W}/W = O(1)\) and \(\overline{D}/D = O(\poly k_{\mathrm{YK}})\) constructed out of gadgets of \cite{yamasaki2024time} that act only on single and pairs of registers at a time and take a single work period of size \(\poly k_{\mathrm{YK}}\).
  \item Let \(\Omega\) be the set of register locations of \(\overline{C}\).
  \item This circuit is equipped with families of bad fault paths parameterized by register fault paths \(\{\mathcal{G}_P\}_{P \subseteq \Omega}\) such that if a physical fault \(\mathbf{f}\) is \(\mathcal{G}_P\)-avoiding, then the output distribution of \(\overline{C}[\mathbf{f}]\) is equal to the output distribution of \(C[\mathbf{g}]\) for some register physical fault \(\mathbf{g}\) with a \(\{P\}\)-avoiding register fault path.
  \item There is a function \(p_{\mathrm{YK}}(x)\) that satisfies \(p_{\mathrm{YK}}(x) \le e^{-O(2^r)}\) on \(x \in [0,\epsilon_{*,YK}]\) such that \(\weightenum{\mathcal{G}_P}{x} \le \left(p_{\mathrm{YK}}(x)\right)^{|P|}\).
  \end{itemize}
\end{claim}

\begin{remark}[Threshold amplification using qLDPC codes]
    A natural question is why constant space overhead qLDPC code constructions are unsuitable for the threshold amplification step.
    Conveniently, because the inner codes used for the amplification step are very small, techniques that gave unacceptably high time-overhead in the outer fault-tolerance scheme can be used.
    Instead of using the large distillation gadgets as in \cref{thm:main-result-vanish-threshold}, one could prepare resource states using a concatenated code as in \cite{gottesman2013fault} with many small sized qLDPC codes (see \cref{subsubsec:revisiting}).

    The reason this does not immediately work turns out to be somewhat subtle: In order to simulate an outer fault-tolerance scheme, it must be the case that the behavior of the simulating circuit is well defined with respect to failures of the inner fault-tolerance gadgets.
    In order for this to be the case, the gadgets must be capable of accepting an arbitrarily damaged state.
    This is the behavior the friendliness property in \cref{def:ft-gadget} captures and is the reason why we needed to construct the qLTC tester in \cref{sec:tester}.
    One should note that standard single-shot error-correction gadgets \emph{are not friendly}.
    We are aware of two ways to establish friendly error correction gadgets which we briefly outline here.

    \begin{itemize}
        \item  
            The first method is to measure the syndrome of the qLDPC code a number of times equal to the distance \(d\) as in \cite{gottesman2013fault} and decode the resulting spacetime syndrome.
            Then, regardless of the input state, the output state is sparsely deviated from a code state.
            Unfortunately, in most cases, an efficient decoder for the resulting spacetime code is not known (it is unknown whether a single shot decoder implies a decoder for the spacetime code).
            A notable exception is the minimum-weight perfect matching decoder for the surface code~\cite{dennis2002topological}.
            Alternatively, one can simply use a brute-force approach since the codes are of extremely small size.
        \item
            The second method to obtain friendly gadgets is to utilize a qLTC.
            Given a source of \(\ket{0}\) states, the error correction gadget in \cref{lemma:ec-gadget} can be made friendly by applying the tester gadget \cref{lemma:ec-gadget-tester} before each error correction gadget.
            In the case where the tester rejects, we swap in the known-good \(\ket{0}\) state as in \cite{aharonov1999fault}.
            The preparation of logical \(\ket{0}\) states using a qLTC is straightforward: Prepare a zero product state and measure the checks of the code. Use a bit-flipping decoder for the \emph{syndrome error}, and then finally use Gaussian elimination to solve for a correction that satisfies the measured syndrome.
            The resulting state has small syndrome and hence must be close to the code space by local-testability.
    \end{itemize}
\end{remark}

We are now ready to state and prove the main result.
\begin{theorem}[Main result]\label{thm:main-result-concat}
  There exists a function \(f(x)\) growing faster than any \(\quasipoly(x)\) and a value \(\epsilon_{*}\in (0,1)\) such that for any \(\epsilon_L \in (0,1)\) and (Clifford+\(\CCZ\)) classical input / classical output quantum circuit \(C\) with width \(W\) and depth \(D\) satisfying
  \begin{align}
    \frac{WD}{\epsilon_L} &\le f(W)
  \end{align}
  There is a corresponding efficiently constructable  classical input / classical output quantum circuit \(\overline{C}\) with width \(\overline{W}\) and depth \(\overline{D}\) satisfying
  \begin{align}
    \frac{\overline{W}}{W} &= O_{W \to \infty}(1)\\
    \frac{\overline{D}}{D} &= O_{W \to \infty}\left(\left(\log\frac{W D}{\epsilon_L}\right)^{1 + o(1)}\right)
  \end{align}
  and using auxiliary $O(1)$-time classical computation per quantum time step,
  such that the following guarantees hold. For a random physical fault \(\mathbf{f}\) distributed according to \(\epsilon\)-locally stochastic faults model with \(\epsilon \in [0,\epsilon_{*}]\), the output distribution of \(\overline{C}\) subject to \(\mathbf{f}\) is \(\epsilon_L\)-close in TVD to the output distribution of \(C\).
\end{theorem}

  \paragraph{Construction} \emph{Proof.} First we construct a circuit \(\overline{C}_{\mathrm{qLDPC}}\) using \cref{thm:main-result-vanish-threshold}.
  This has a family of bad fault paths \(\mathcal{G}_{\mathrm{qLDPC}}\) such that if the physical fault \(\mathbf{f}\) is \(\mathcal{G}_{\mathrm{qLDPC}}\)-avoiding, then the output distribution of \(\overline{C}_{\mathrm{qLDPC}}[\mathbf{f}]\) is equal to the output distribution of \(C\), and \(\weightenum{\mathcal{G}_{\mathrm{qLDPC}}}{x} \le \epsilon_L\) for \(x \in [0,\epsilon_{*,\mathrm{LDPC}} \cdot c]\) where \(c = \Omega\left(\frac{1}{\polyloglog \frac{W D}{\epsilon_L}}\right)\).

  We will simulate this circuit using \cref{thm:yk-sim} with \(r = \Theta\left(\log \log \frac{1}{\epsilon_{*,\mathrm{LDPC}} c}\right)\) in a particular way (We will only pick \(r\) at the end):
  Recall that the blocks of \cref{thm:main-result-vanish-threshold} are of size at most \(O\left(\log \frac{WD}{\epsilon_L} \cdot \polyloglog \frac{WD}{\epsilon_L}\right)\) and the largest gadget is of size at most \(J = O\left(\quasipolylog\frac{WD}{\epsilon_L}\right)\).
  We will arrange the qubits of \(\overline{C}_{\mathrm{qLDPC}}\) into registers of size \(k_{\mathrm{YK}} = \exp\left[\Theta(\log\log \left(\frac{1}{\epsilon_{*} c}\right)^2\right]\) such that no two qubits participating in a gadget (of \(\overline{C}_{\mathrm{qLDPC}}\)) are assigned to the same register.
  
  Let \(W_{\mathrm{qLDPC}}\) be the number of qubits in \(\overline{C}_{\mathrm{qLDPC}}\).
  To compute the register assignment, for each timestep \(t\) of \(\overline{C}_{\mathrm{qLDPC}}\), we compute a greedy vertex coloring on the graph \(G_t=([W_{\mathrm{qLDPC}}], E_t)\) where two vertices \(i,j \in [W_{\mathrm{qLDPC}}]\) share a vertex if they participate in the same gadget.
  \(G\) has degree at most \(J\), so there is a vertex coloring \(c_t \colon [W_{\mathrm{qLDPC}}] \to [J+1]\) with at most \(J+1\) colors.
  
  We use \(J+1\) groups of (potentially many) YK registers corresponding to the \(J+1\) colors.
  At time \(t\), we assign each qubit \(q\in [W]\) to occupy a register with color \(c_t(q)\).
  Between timesteps \(t\) of \(\overline{C}_{\mathrm{qLDPC}}\), we swap qubits between registers in parallel to maintain the coloring constraint.
  This adds \(\poly(k_{\mathrm{YK}})\) time overhead.
  Now, we apply \cref{thm:yk-sim} to this circuit\footnote{Actually, to a serialization of this circuit, but this is not important to the argument.} to arrive at the circuit \(\overline{C}_{n}\) which is deeper by a factor \(O(\poly k_{\mathrm{YK}})\).

  \paragraph{Fault analysis} Let \(\mathcal{G}_{\mathrm{qLDPC}}^{\mathrm{(reg)}}\) be a family of bad register fault paths where each fault path in \(\mathcal{G}_{\mathrm{qLDPC}}\) has been replaced by the corresponding register fault path with respect to the register partitioning.
  By ensuring that no two qubits participating in a qLDPC gadget are assigned to the same register, it is the case that \(\weightenum{\mathcal{G}_{\mathrm{qLDPC}}^{\mathrm{(reg)}}}{x} = \weightenum{\mathcal{G}_{\mathrm{qLDPC}}}{x}\).
  Let \(\{\mathcal{G}^{(\mathrm{YK})}_P\}_P\) be the parameterized families of bad fault paths from the statement of \cref{thm:yk-sim}.
  Now consider the fault set
  \begin{align}\label{eq:main-result-register-fault-paths}
    \mathcal{G} = \boxplus_{P \in \mathcal{G}_{\mathrm{qLDPC}}^{\mathrm{(reg)}}} \mathcal{G}_P^{(\mathrm{YK})}.
  \end{align}
  A \(\mathcal{G}\)-avoiding fault \(\mathbf{f}\) is \(\mathcal{G}_P^{(\mathrm{YK})}\)-avoiding for every register fault path \(P \in \mathcal{G}_{\mathrm{qLDPC}}^{\mathrm{(reg)}}\).
  Suppose that \(\mathbf{f}\) is a \(\mathcal{G}\)-avoiding physical fault, then (by \cref{thm:yk-sim}) there exists a physical fault \(\mathbf{g}\) for which its register fault path is \(\mathcal{G}_{\mathrm{qLDPC}}^{\mathrm{(reg)}}\)-avoiding such that the output distribution of \(\overline{C}[\mathbf{f}]\) is the same as \(\overline{C}_{\mathrm{LDPC}}[\mathbf{g}]\).
  Now, by construction of \(\mathcal{G}_{\mathrm{qLDPC}}^{\mathrm{(reg)}}\), if \(\mathbf{g}\) has a register fault path that is \(\mathcal{G}_{\mathrm{qLDPC}}^{\mathrm{(reg)}}\)-avoiding, then its fault path must be \(\mathcal{G}_{\mathrm{qLDPC}}\)-avoiding, so the output distribution of \(\overline{C}_{\mathrm{LDPC}}[\mathbf{g}]\) is equal to the output distribution of \(C\).

  Now let \(\mathbf{f}\) be a random physical fault distributed according to an \(\epsilon\)-locally-stochastic faults model.
  Recall that this means that for any fault path \(S\), the probability that the fault path of \(\mathbf{f}\) contains \(S\) is at most \(\epsilon^{|S|}\).
  We would like to bound the total-variation distance between the output distribution of \(\mathbb{E}_{\mathbf{f}}\overline{C}[\mathbf{f}]\) and that of \(C\).
  We recall that the two distributions are equal conditioned on the event \(\mathbf{f}\) is \(\mathcal{G}\)-avoiding, so
  \begin{align*}
    |\Pr(\overline{C}[\mathbf{f}]\text{ outputs }x) - \Pr(C\text{ outputs }x)| &\le |\Pr_{\mathbf{f}}(\mathbf{f}\text{ is not \(\mathcal{G}\)-avoiding})|\\
                                                                                   &\le \sum_{S \in \mathcal{G}} \Pr(\supp \mathbf{f} \subseteq S) \\
                                                                                   &\le \sum_{S \in \mathcal{G}} \epsilon^{|S|} = \weightenum{\mathcal{G}}{\epsilon}~.
  \end{align*}
  I.e. \(\weightenum{\mathcal{G}}{\epsilon}\) upper bounds the TVD between the two distributions.

  \paragraph{Error rate} It remains to upper bound \(\weightenum{\mathcal{G}}{\epsilon}\).
  Since \(\epsilon \in [0,\epsilon_{*,\mathrm{YK}}]\), we have (using \cref{eq:main-result-register-fault-paths}, \cref{lemma:enumerator-ring}, and \cref{thm:yk-sim})
  \begin{align*}
    \weightenum{\mathcal{G}}{\epsilon} &\le \sum_{P \in \mathcal{G}^{\mathrm{(reg)}}_{\mathrm{qLDPC}}} \weightenum{\mathcal{G}_P^{\mathrm{(YK)}}}{\epsilon} \\
                       &\le \sum_{P \in \mathcal{G}^{\mathrm{(reg)}}_{\mathrm{qLDPC}}} \left(p_{\mathrm{YK}}(\epsilon)\right)^{|P|}\\
                       &\le \weightenum{\mathcal{G}^{\mathrm{(reg)}}_{\mathrm{qLDPC}}}{p_{\mathrm{YK}}(\epsilon)}
  \end{align*}
  We now select the minimum \(r\) such that  \(p_{\mathrm{YK}}(\epsilon) \le \frac{1}{\epsilon_{*,\mathrm{LDPC}} c}\), so that
  \begin{align*}
    \weightenum{\mathcal{G}}{\epsilon} \le \weightenum{\mathcal{G}^{\mathrm{(reg)}}_{\mathrm{qLDPC}}}{p_{\mathrm{YK}}(\epsilon)} \le \weightenum{\mathcal{G}^{\mathrm{(reg)}}_{\mathrm{qLDPC}}}{\frac{1}{\epsilon_{*,\mathrm{LDPC}} c}} \le \epsilon_L~.
  \end{align*}
  \(p_{\mathrm{YK}}(\epsilon) = e^{-O(2^{-r})}\), so \(r = O\left(\loglogloglog \frac{W D }{\epsilon_L}\right)\) suffices.

  \paragraph{Circuit size}
  Let \(f'(x)\) be the function from the statement of \cref{thm:main-result-vanish-threshold}.
  If \(\frac{WD}{\epsilon_L} \le f'(W)\) then \(\frac{W_{\mathrm{qLDPC}}}{W} = O(1)\).
  Note that for our choice of \(r\), \(k_{\mathrm{YK}} = e^{O\left(\left(\loglogloglog \frac{W D }{\epsilon_L}\right)^2\right)}\).
  \((J+1) k_{\mathrm{YK}}\) is the minimum number of registers required to assign each qubit of a gadget of \(\overline{C}_{\mathrm{qLDPC}}\) to distinct YK registers, so the space overhead of \(\overline{C}\) relative to \(\overline{C}_{\mathrm{qLDPC}}\) is at most
  \begin{align}\label{eq:const-width-bound}
      O\left(1+\frac{J k_{\mathrm{YK}}}{W_{\mathrm{qLDPC}}}\right) = O\left(1 + \frac{\quasipolylog \frac{WD}{\epsilon_L}}{W_{\mathrm{qLDPC}}}\right)
  \end{align}
  We now select \(f(x)\).
  As in \cref{thm:main-result-vanish-threshold}, there exists absolute constants \(c \in (0,1) \) and \(c' \ge 0\) such that we can define \(f(x) = \exp\left(\exp\left(c' (\log x)^{c}\right)\right)\) to be a quickly growing function such that 1) \(f(x) \le f'(x)\) and 2) when \(\frac{WD}{\epsilon} \le f(W)\), the right hand side of \cref{eq:const-width-bound} is \(O(1)\).
\(f(x)\) again grows faster than any \(\quasipoly(x)\) 

  Finally, relative to \(\overline{C}_{\mathrm{qLDPC}}\), \(\overline{C}\) has depth overhead
  \begin{align*}
    \poly k_L =O\left(e^{O\left(\left(\loglogloglog\frac{W D }{\epsilon_L}\right)^2\right)}\right)
  \end{align*}
  proving the theorem with \(\epsilon_{*} = \epsilon_{*,\mathrm{YK}}\). \qedsymbol{}

\section{State preparation gadgets}\label{sec:state-prep}
In this section, we develop state preparation gadgets for both logical stabilizer states (\Cref{lemma:stab-resource-state}) and magic states (\Cref{lemma:magic-resource-state}) of the computational qLTC. As described in the introduction section, the main idea is to (1) use existing concatenated-code fault tolerance schemes to prepare the codestates to a constant logical fidelity, and then (2) apply state distillation schemes to boost the fidelity to an arbitrary target error. We will first review the concatenated-code fault tolerance proof of Aharonov and Ben-Or~\cite{aharonov1999fault} and describe how to use it to prepare qLTC codestates to constant fidelity in~\Cref{subsec:low-fidel-state-prep}. Next, in~\Cref{subsec:boost-fidel-state-prep} we describe how to perform state distillation protocols (whose constructions are deferred to later sections) at the logical level of the computional qLTC code to improve the logical fidelity. In~\Cref{sec:state-distill-procedure-single}, we present our stabilizer state distillation scheme with constant-space and polyloglog-time overheads. The magic state distillation scheme with almost-constant space and polyloglog-time overheads is deferred to~\Cref{sec:magic-state-distillation-code}.

\subsection{Concatenated code FT scheme}\label{sec:concatAGP}
We now briefly introduce the scheme of concatenated code fault-tolerance scheme of Aharonov and Ben-Or~\cite{aharonov1999fault} which uses distance-9 CSS codes and Clifford+\(\CCZ\) \footnote{
More precisely they use Clifford+Toffoli gate set, but they use self-dual quantum codes with a transversal Hadamard and hence essentially the same scheme works for the Clifford+$\CCZ $ gate set.}
We use this scheme over the similar \cite{aliferis2005quantum} to avoid introducing the analog of exRecs for the weight enumerators which makes the notation burdensome.
We remark that we are using a modified definition of deviation from \cite{aharonov1999fault}, but their argument is unmodified (e.g. \cite{kitaev1997quantum}).
Nearly any concatenated code scheme suffices.
The code family will be parameterized by a parameter \(r \in \mathbb{N}\) corresponding to a number of concatenation levels.

\begin{claim}[\cite{aharonov1999fault}]\label{fact:AB}
  There is a family of concatenated codes \(\qcode_{\mathrm{AB},r}\) with compatible and friendly gadgets for Clifford+\(\CCZ\) with the property that: There exists constants \(W_{\mathrm{AB}},~A_{\mathrm{AB}},~B_{\mathrm{AB}},~c_{\mathrm{AB}}>0\) such that
  \begin{itemize}
  \item Each code block holds one logical qubit.
  \item The size of a block is at most \(B_{\mathrm{AB}}^r\).
  \item The number of locations in a gadget is at most \(A_{\mathrm{AB}}^r\).
  \item The maximum number of qubits used at any point in time in a gadget is at most \(W_{\mathrm{AB}}^r\).
  \item The gadgets are compatible and friendly as in~\cref{def:ft-gadget}.
  \item For each gadget, the corresponding bad fault sets \(\mathcal{G}\) have weight enumerator \(\weightenum{\mathcal{G}}{x} \le (c_{\mathrm{AB}} x)^{2^r}\) on \([0,1)\).
  \end{itemize}
  we will refer to the \(r\)-th element of this family of code and gadgets as \(r\)-AB.
\end{claim}

\subsection{Noisy state preparation}\label{subsec:low-fidel-state-prep}
We begin by describing a fault-tolerant gadget that prepares code states of qLDPC codes. The main idea is to use the concatenated-code scheme~\cite{aharonov1999fault} to prepare the qLDPC codestate while being encoded in an outer concatenated code, and then unencode the outer code. This idea originates from Gottesman's seminal work~\cite{gottesman2013fault}, but our goal here differs from his work. While Gottesman used this gadget to prepare the codestates to a low error rate that is inverse polynomial in the computation size, we only do so to a sufficiently small constant error rate. This choice is crucial to obtain the main overhead result of this paper, as concatenated-code FT induces a $\polylog(1/\varepsilon)$ overhead that prevents us from choosing $\varepsilon$ to be too small. Additionally, we present a fault tolerance proof of the unencoding procedure, which was omitted in~\cite{gottesman2013fault}. The main lemma proved in this subsection is~\Cref{lemma:noisy-prep-gadget}.

\subsubsection{Constant depth qLDPC encoding}\label{subsec:const-depth-encoding}
We start by showing how to prepare simple qLDPC codestates (all ancilla states used in this paper satisfy the precondition below) non-fault-tolerantly in constant depth.
    
\begin{lemma}[Constant-depth qLDPC encoding~\cite{gottesman2013fault}]\label{lemma:constant-depth-encoding}
\label{lem:nonFT-state-prep} Let $\ket{\psi}$ be a $k$-qubit state that can be prepared by a (Clifford+\(\CCZ\)) quantum circuit of constant depth. Then the encoding of $\ket{\psi}$ into a $[[n,k]]$ CSS $\Delta$-qLDPC code can be done by a non-fault-tolerant procedure using $O(n)$ qubits in quantum depth $O(\Delta)$ and $O(\log^2 n)$ classical depth.
\end{lemma}
\begin{proof}
  Let $m_X$, $m_Z$ be the number of independent X and Z checks, so that $k=n-m_X-m_Z$. First, we observe that for any $[[n,k]]$ CSS code, up to qubit permutations, it is possible to choose the logical Pauli operators to be of the form $\overline{X}_i= X_i \otimes P_i \otimes P'_{i}$ and $\overline{Z}_i= Z_i \otimes Q_i \otimes Q'_i$ for each $i \in [k]$, where $X_i$ ($Z_i$) acts on the $i$-th physical qubit, $P_i$ and $Q_i$ are Pauli X strings on $m_Z$ qubits, and $P'_i$, $Q'_i$ are Pauli Z strings on $m_X$ qubits\footnote{A similar statement can be made for general stabilizer codes.}. Indeed, we can perform Gaussian elimination (and possibly qubit permutations) to put the check matrices into the canonical form~\cite[Section 4]{gottesman1997stabilizer}, $S_X = (A_1|A_2|\id_{m_X})$ where $A_{1} \in \ftwo^{m_X \times k}$, $A_{2} \in \ftwo^{m_X \times m_Z}$, and $S_Z = (B_1|\id_{m_Z}|B_2)$ where $B_{1} \in \ftwo^{m_Z \times k}$, $B_{2} \in \ftwo^{m_Z \times m_X}$. Next, we imagine applying the Hadamards $H^{\otimes m_z}$ on the middle block. This sounds strange at first because it makes the code non-CSS, but we will shortly see why the Hadamards are useful. We write the new stabilizers in the symplectic representation~\cite{gottesman1997stabilizer}
  \begin{align}
  \left(
  \begin{array}{c|c}
       A_1|A_2|\id_{m_X}&\mathbf{0}  \\
      \mathbf{0} & B_1|\id_{m_Z}|B_2
  \end{array}
  \right)
      \overset{H^{\otimes m_z}}{\longrightarrow} 
        \left(
  \begin{array}{c|c|c|c|c|c}
       A_1&\mathbf{0} &\id_{m_X}& \mathbf{0}&A_2&\mathbf{0} \\
      \mathbf{0}&\id_{m_Z}&\mathbf{0} & B_1&\mathbf{0} &B_2
  \end{array}
  \right).
  \end{align}
  In this Hadamard-transformed code, given a logical Pauli X (that contains Pauli Zs on the middle block and Pauli Xs on the first and third block) we can multiply it with stabilizers from the first $m_X$ rows to remove the Pauli Xs from the third block. Similarly, we can multiply with the stabilizers from the last $m_Z$ rows to remove Pauli Xs from the middle block of each logical Pauli Z. Therefore, we can choose a logical basis for this Hadamard-transformed code such that $\overline{X}'_i= X_i \otimes P''_i $ and $\overline{Z}'_i= Z_i \otimes Q''_i$ where $P''_i,Q''_i$ are Pauli Z strings on $m_Z+m_X$ qubits. Applying $H^{\otimes m_Z}$ to return to the original code we obtain the stated logical basis.

  Having chosen a convenient logical basis, we can encode any state $\ket{\psi}$ as follows. We prepare $\ket{\psi}\ket{+}^{m_Z} \ket{0}^{\otimes m_X}$.
  We then measure the qLDPC stabilizers. Importantly, we measure the LDPC presentation of the stabilizers rather than the canonical generators from the previous paragraph, so that the measurement depth is constant $O(\Delta)$. The resulting state will be the desired logical state encoded in some coset of the code, and we can shift back to the correct code as detailed below.
        
  For simplicity we start with the case when $\ket{\psi}$ is a product state in the Z basis, $\ket{\psi}=\ket{a}$ for $a \in \mathbb{F}_2^{k}$. The procedure starts with the state $\phi = \ket{a}\ket{0}^{\otimes (n-k)}$, which is stabilized by $(-1)^{a_i} \overline{Z}_{i}$ for $i\in [k]$. Next we perform the (LDPC) syndrome measurements, obtaining syndromes $\sigma_X \in \mathbb{F}_2^{m^\mathrm{LDPC}_Z}, \sigma_Z \in \mathbb{F}_2^{m^\mathrm{LDPC}_X}$. Using Gaussian elimination we solve for the Pauli corrections $e_X, e_Z \in \mathbb{F}_2^{n}$ such that $H_X e_Z = \sigma_Z$ and $H_Z e_X = \sigma_X$, where $H_X, H_Z$ are matrices representing the LDPC presentation of the stabilizers. We then multiply $e_X$ with $\overline{X}_i$ for each $i \in [k]$ such that $\overline{Z}_i$ anticommutes with $e_X$, yielding a Pauli string $e'_X$. Similarly, we multiply $e_Z$ with $\overline{Z}_i$ for each $i \in [k]$ such that $\overline{X}_i$ anticommutes with $e_Z$, yielding a Pauli string $e'_Z$. Note that $e'_Z$ and $e'_X$ can be of mixed Pauli types (including both X and Z), . Applying $e'_Z$ and $e'_X$
  produces the desired encoded state $\ket{\overline{a}}$. The same procedure can be seen to work for X basis product states as well. And thus it works for any general input state $\psi$ by linearity.

  The quantum depth consists of the physical circuit preparing $\psi$, the LDPC measurements, and the Pauli correction. This is a $O(\Delta)$ when $\psi$ is a constant-depth state and uses $O(n)$ physical qubits if $m^\mathrm{LDPC}_Z, m^\mathrm{LDPC}_X=O(n)$. The classical depth is the time needed to compute the Pauli corrections, which can be done by, e.g., Gaussian elimination with gate complexity $O(n^3)$. This can be parallelized to $O(\log^2n)$ classical depth by the algorithms in~\cite{csanky1975fast}.

\end{proof}

\subsubsection{Concatenated code unencoding}
Concatenated code fault tolerance allows us to simulate the encoding circuit of our computational code in a concatenated code, but we will be left with an encoded state of both the computational code and concatenated code. Hence, we need to unencode this state from the concatenated code.
Here, we establish a gadget for unencoding from concatenated codes that has controlled errors\footnote{A very recent work \cite{christandl2024fault} also gives a proof of this fact.}. A similar claim was made in~\cite{gottesman2013fault,fawzi2020constant} without proof.
Recall that for each \(r \in \mathbb{N}\), the set of gadgets was compatible, i.e. each encoded concatenated code block input and output a \((\qcode_{\mathrm{AB},r},\mathcal{F}_r)\)-block for some family of bad fault sets \(\mathcal{F}_r\) (the exact definition is not needed).
Denote the trivial block corresponding to a single qubit on the set \(\{1\}\) and the family of bad sets \(\{\{1\}\}\) as a \((\qcode_{\mathrm{AB},r}, \mathcal{F}_0)\)-block.

\begin{lemma}[Concatenated code unencoding]\label{lemma:ab-unencode}
  There exists a constant \(c_{\mathrm{unencode,AB}} > 0 \) such that for each \(r \in \mathbb{N}\) the following holds. There exists an unencoding gadget \(\gadget_{\mathrm{unencode,AB,r}}\) that takes a single qubit state \(\psi\) encoded in a \((\qcode_{\mathrm{AB},r},\mathcal{F}_r)\)-block of the \(r\)-th level of AB concatenated codes, and outputs the single qubit state \(\psi\) if the input is \(\mathcal{F}_r\)-deviated from \(\phi_{\qcode_{\mathrm{AB},r}}(\psi)\) and the fault set is \(\mathcal{G}_{r,\mathrm{unencode}}\)-avoiding where \(\weightenum{\mathcal{G}_{r,\mathrm{unencode}}}{x} \le c_{\mathrm{unencode,AB}}  \cdot x\) on \([0,1/(2c_{\mathrm{AB}}))\).
\end{lemma}
\begin{proof}
  Let \(\qOp_{\mathrm{unencode}}\) be the unencoding operation for the code used in the recursive concatenation procedure.
  We first define \(\widehat{\gadget}_{\mathrm{unencode,AB,\mathfrak{r}}}\) as the EC-gadget followed by the simulation of  \(\qOp_{\mathrm{unencode}}\) using \((\mathfrak{r}-1)\)-AB.
  This takes a state encoded in a \((\qcode_{\mathrm{AB},\mathfrak{r}},\mathcal{F}_{\mathfrak{r}})\)-block and outputs the same state encoded in a \((\qcode_{\mathrm{AB},\mathfrak{r}-1},\mathcal{F}_{\mathfrak{r}-1})\)-block.
  We will define \(\gadget_{\mathrm{unencode,AB,\mathfrak{r}}}\) recursively as
  \begin{align*}
    \gadget_{\mathrm{unencode,AB,\mathfrak{r}}} = \gadget_{\mathrm{unencode,AB,\mathfrak{r}-1}}\circ \widehat{\gadget}_{\mathrm{unencode,AB,\mathfrak{r}}}
  \end{align*}
  with \(\gadget_{\mathrm{unencode,AB,0}}\) trivial.

  Let \(A_{\mathrm{unencode}} = O(1)\) be the number of locations in \(\qOp_{\mathrm{unencode}}\).
  We use \cref{lemma:gadget-composition} and \cref{lemma:enumerator-ring} to upper bound the weight enumerator polynomial for the family of bad sets at each step of the recursion.
  For each recursion step \(\mathfrak{r}>1\), there is a set of bad fault paths \(\mathcal{H}_{\mathfrak{r},\mathrm{unencode}}\) with weight enumerator \(\weightenum{\mathcal{H}_{\mathfrak{r},\mathrm{unencode}}}{x} \le A_{\mathrm{unencode}} (c_{\mathrm{AB}} \cdot x)^{2^\mathfrak{r-1}}\) such that if the fault path is \(\mathcal{H}_{\mathfrak{r},\mathrm{unencode}}\)-avoiding then the gadget takes a \((\qcode_{\mathrm{AB},\mathfrak{r}},\mathcal{F}_{\mathfrak{r}})\)-block \(\rho_{\mathfrak{r}}\) that is \(\mathcal{F}_{\mathfrak{r}}\)-deviated from \(\phi_{\qcode_{\mathrm{AB},\mathfrak{r}}}(\psi)\) to a \((C\qcode_{\mathrm{AB},\mathfrak{r}-1},\mathcal{F}_{\mathfrak{r}-1})\)-block \(\rho_{\mathfrak{r}-1}\) that is \(\mathcal{F}_{\mathfrak{r}-1}\)-deviated from \(\phi_{\qcode_{\mathrm{AB},\mathfrak{r}-1}}(\psi)\).
  For \(\mathfrak{r}=1\), there is no simulation and \(\weightenum{\mathcal{H}_{1,\mathrm{unencode}}}{x} \le A_{\mathrm{unencode}} x\) as any error in one of the \(A_{\mathrm{encode}}\) locations will cause an error on the output.
  The family of fault paths \(\mathcal{H}_{\mathfrak{r},\mathrm{unencode}}\) corresponds to failure of any of the \(A_{\mathrm{unencode}}\) level-\((\mathfrak{r}-1)\) simulation gadgets i.e. it is the sum of the bad fault paths for each of the gadgets.
  \Cref{fact:AB} gives an upper bound on the weight enumerator polynomial for each of these gadgets.
  
  The family of bad fault paths for the entire gadget \(\mathcal{G}_{r,\mathrm{unencode}}\) is the sum of the bad fault paths in each step of the recursion.
  That is, for \(\mathfrak{r}>1\),
  \begin{align}
    \weightenum{\mathcal{G}_{\mathfrak{r},\mathrm{unencode}}}{x} &\le \weightenum{\mathcal{G}_{\mathfrak{r}-1,\mathrm{unencode}}}{x} + \weightenum{\mathcal{H}_{\mathfrak{r},\mathrm{unencode}}}{x}
  \end{align}
  with \(\weightenum{\mathcal{G}_{1,\mathrm{unencode}}}{x} = \weightenum{\mathcal{H}_{1,\mathrm{unencode}}}{x}\).
  Suppose that \(x \in (0,1/(2c_{\mathrm{AB}}))\), then evaluating the recursion gives the upper bound:
  \begin{align}
    \weightenum{\mathcal{G}_{r,\mathrm{unencode}}}{x} &\le A_{\mathrm{unencode}}~ x + \sum_{\mathfrak{r}=2}^{r} A_{\mathrm{unencode}} (c_{\mathrm{AB}}\cdot x)^{2^{\mathfrak{r}-1}} \\
                         &\le A_{\mathrm{unencode}} x \left(1 + c_{\mathrm{AB}}  \sum_{\mathfrak{r}=2}^{\infty} (c_{\mathrm{AB}}\cdot x)^{2^{\mathfrak{r}-1}-1}\right) \\
                         &\le A_{\mathrm{unencode}} x \left(1 + c_{\mathrm{AB}} \int_1^{\infty} (1/2)^{2^n-1}\mathrm{d}n\right) \\
    &\le A_{\mathrm{unencode}} \left(1 + c_{\mathrm{AB}}\right) x.
  \end{align}
  The integral can be evaluated to be less than \(0.343\).
  Thus, the result follows with \(c_{\mathrm{unencode,AB}} = A_{\mathrm{unencode}} \left(1 + c_{\mathrm{AB}}\right)\).
\end{proof}

\subsubsection{qLTC tester}\label{sec:tester}
Combining~\Cref{lem:nonFT-state-prep} and~\Cref{lemma:ab-unencode} allows us to prepare computational code states to a constant error rate. However, we do not have any guarantee on the prepared state when this procedure fails. Hence, we now construct a verification gadget to verify that the state prepared is not too damaged. We need this verification to be fast, and this is where the qLTC property is used.

The final ingredient will be a tester gadget that uses the qLTC property to test (in low depth) if we have prepared a state close to the code space. We use the notations and variables defined in~\Cref{sec:globals}.

\begin{lemma}[Testing variant of computational code EC gadget]\label{lemma:ec-gadget-tester}
  There exists a constant \(\epsilon_{*,\mathrm{tester}}\in (0,1)\) and gadget \(\gadget_{\mathrm{tester}}\) for the computational code of depth \(O(\log n_L)\), width \(O(n_L)\), and \(A_{\mathrm{EC}}=O(n_L \log n_L)\) locations such that
  there is family of bad fault paths \(\mathcal{G}_{\mathrm{tester}}\subseteq P([A_{\mathrm{EC}}])\) for which \(\weightenum{\mathcal{G}_{\mathrm{tester}}}{x} \le \polylog(n_L) e^{-\zeta(n_L)}\) for \(x \in [0,\ndep{\epsilon_{*,\mathrm{tester}}}]\).
  For a \(\mathcal{G}_{\mathrm{tester}}\)-avoiding Pauli fault \(\fault\), and some pure state in the codespace \(\ket{\psi}\), \(\gadget_{\mathrm{tester}}[\fault]\) satisfies
  \begin{itemize}
  \item \(\gadget_{\mathrm{tester}}[\fault]\) outputs to the classical memory a FAIL bit \(y_{\mathrm{fail}}\).
  \item If \(y_{\mathrm{fail}} = 0\) is output, the output is \(\mathcal{F}_{\mathrm{L}}\)-Pauli deviated from the code space.
  \item There exists a family \(\mathcal{F}_{\mathrm{test}}\subseteq P([n_L])\) such that \(\weightenum{\mathcal{F}_{\mathrm{test}}}{x} = \binom{n_L}{t_{\mathrm{test}}} x^{t_{\mathrm{test}}}\) with \(t_{\mathrm{test}} = \Theta(t_L)\).
  \item If the input is \(\mathcal{F}_{\mathrm{test}}\)-deviated from a codestate \(\sigma\), then \(y_{\mathrm{fail}} = 0\) is always output and the output is \(\mathcal{F}_{\mathrm{L}}\)-Pauli deviated from \(\sigma\).
  \end{itemize}
\end{lemma}
\begin{proof}
  First execute the same quantum circuit as in \cref{lemma:ec-gadget}, without the decoding and correction steps.
  This will allow us to borrow the majority of the analysis.
  We call this first layer of syndrome extraction the test step.
  While the classical circuit computing \(y_{\mathrm{fail}}\) is running (computing a size \(n_L\) threshold\footnote{One construction is to sort the bits of the input bitstring \(x\). The bit at position \(k\) is \(|x| \ge k\).}), we will perform \(O(\log n_L)\) applications of \(\gadget_{\mathrm{EC}}\) (\cref{lemma:ec-gadget}).
  We parallel the analysis by first restricting to a state that is diagonal Pauli-deviated from some codestate \(\sigma\) and a diagonal Pauli fault \(\fault\).
  
  Let \(\ell = 2\cdot 2^{2\Delta}\), and define \(\mathcal{G}_{\mathrm{test~step}}\) to be all subsets of locations of the test step of a size \(s\) that will be selected later.
  Let \(\mathcal{G}_{\mathrm{tester}}\) be the sum of \(\mathcal{G}_{\mathrm{test~step}}\) and the bad fault paths \(\mathcal{G}_{\mathrm{EC}}\) for each application of \(\gadget_{\mathrm{EC}}\).
  Consider the application of \(\gadget_{\mathrm{tester}}[\fault]\) for a \(\mathcal{G}_{\mathrm{tester}}\)-avoiding diagonal Pauli fault \(\fault\).
  For now, suppose that the input state \(\tilde{\rho}\) differs from the codestate \(\sigma\) by a diagonal Pauli superoperator \(E\).
  By the same argument as in \cref{lemma:ec-gadget}, the state output from the test step then differs from the codestate by a (stabilizer reduced) diagonal Pauli superoperator \(E_{\mathrm{out}}\) which satisfies (note that we are defining \(s\) to be slightly smaller)
  \begin{align}
    |\supp E_{\mathrm{out}} | \le |\supp E| + \ell(s-1).
  \end{align}
  Let \(r_{\mathrm{min}} = \min(r_x,r_z) = \Omega(n_L)\) be a lower bound for the number of \(X\) and \(Z\) checks for the computational code.
  Using the \((\rho,\Delta)\)-qLTC property of the computational code, the Hamming weight of the noisy syndrome measured \(|\tilde{\sigma}_X| + |\tilde{\sigma}_Z|\) satisfies
  \begin{align}
    \frac{\rho r_{\mathrm{min}}}{n_L} |\supp E| - \ell(s-1) &\le |\tilde{\sigma}_X| + |\tilde{\sigma}_Z|  \le 2 \Delta |\supp E| + \ell(s-1).
  \end{align}
  We reject (setting \(y_{\mathrm{fail}} = 1\) using a \(\log n_L\) depth circuit) when the measured syndrome weight \(|\tilde{\sigma}_X| + |\tilde{\sigma}_Z|\) is \(\sigma_{\mathrm{reject}}\) or more.
  \begin{align*}
    \sigma_{\mathrm{reject}} = \frac{\rho r_{\mathrm{min}}}{n_L} (t_{\mathrm{corr}} - s (\ell-1)) - s(\ell-1).
  \end{align*}
  When \(y_{\mathrm{fail}} = 0\), the measured syndrome satisfies \(|\tilde{\sigma}_X| + |\tilde{\sigma}_Z| < \sigma_{\mathrm{reject}}\), so the test step output error is at most
  \begin{align*}
    |\supp E_{\mathrm{out}} | &\le |\supp E| + \ell(s-1) \\
                              &< \left(\sigma_{\mathrm{reject}} + \ell(s-1)\right) \frac{n_L}{\rho r_{\mathrm{min}}} + \ell(s-1) \\
    &< t_{\mathrm{corr}}.
  \end{align*}

  We now analyze under what conditions this test will always accept.
  Let \(t_{\mathrm{test}} = \left\lfloor \frac{\sigma_{\mathrm{reject}} -\ell (s-1)}{2 \Delta} \right\rfloor\) and suppose that \(|\supp E| < t_{\mathrm{test}}\).
  Then, the input state is always accepted:
  \begin{align*}
    |\tilde{\sigma}_X| + |\tilde{\sigma}_Z| &\le 2\Delta |\supp E| + \ell (s-1)\\
                                            &< 2\Delta t_{\mathrm{test}} + \ell(s-1) \le \sigma_{\mathrm{reject}}.
  \end{align*}
  \(\mathcal{F}_{\mathrm{test}}\) will be all subsets of \([n_L]\) of size \(t_\mathrm{test}\) which implies \(\weightenum{\mathcal{F}_{\mathrm{test}}}{x} = \binom{n_L}{t_{\mathrm{test}}} x^{t_{\mathrm{test}}}\).

  Since \(\frac{\rho(n_L) r_{\mathrm{min}}}{n_L} = \Theta\left(\frac{\zeta(n_L)}{n_L}\right)\), for \(n_L\) larger than some constant, we can always pick a \(s = \Theta(\zeta(n_L)) = \Theta(t_{\mathrm{corr}})\) such that \(s \ge 2\) and \(t_{\mathrm{test}}\ge 1\).
  By an identical analysis to \Cref{lemma:ec-gadget}\, there exists an absolute constant \(\epsilon_{*,\mathrm{tester}}\in (0,\epsilon_{*,\mathrm{EC}})\) such that \(\weightenum{\mathcal{G}_{\mathrm{test step}}}{x} \le e^{-\zeta(n_L)}\) for \(x \in [0,\ndep{\epsilon_{*,\mathrm{tester}}}]\).
  Computing the weight enumerator sum, it follows that \(\weightenum{\mathcal{G}_{\mathrm{tester}}}{x} \le \polylog(n_L) e^{-\zeta(n_L)}\) for \(x \in [0,\ndep{\epsilon_{*,\mathrm{tester}}}]\).

  Since \(t_{\mathrm{test}} \le t_{\mathrm{corr}}\), a set that is \(\mathcal{F}_{\mathrm{test}}\)-avoiding is also \(\mathcal{F}_{\mathrm{corr}}\)-avoiding.
  We can now release the assumption on the input state and the fault as in \cref{lemma:ec-gadget}: Apply \cref{lemma:deterministic-errors} to ensure that the input and test step output must be diagonal-Pauli superoperators except where there is a fault.
  If the input state \(\tilde{\rho}\) is \(\mathcal{F}_{\mathrm{test}}\)-Pauli deviated from a codestate \(\sigma\), then it follows from the earlier argument that the test step output state is \(\mathcal{F}_{\mathrm{corr}}\)-Pauli deviated from \(\sigma\).
  Otherwise, if \(y_{\mathrm{fail}}=0\), the test step output state is \(\mathcal{F}_{\mathrm{corr}}\)-Pauli deviated from some codestate \(\phi_{\qcode_L}(\rho')\) where \(\rho'\) is arbitrary.
  It follows that the same holds for the output of the gadget (after one or more rounds of  \(\gadget_{\mathrm{EC}}\)) with \(\mathcal{F}_{\mathrm{corr}}\) replaced by \(\mathcal{F}_{\mathrm{L}}\).
\end{proof}

\subsubsection{Noisy state preparation gadget}\label{sec:injection-gadget}
We are now ready to state and prove the main goal of this section.
\begin{lemma}[Noisy state preparation gadget]\label{lemma:noisy-prep-gadget}
  For \(\epsilon_{\mathrm{noise}}\in (0,1)\) and a pure \(k_L\)-qubit state \(\ket{\psi}\) that can be prepared by a Clifford+\(\CCZ\) circuit \(\qOp_{\mathrm{prep}}\) using depth \(D_\mathrm{prep}\) and space \(W_\mathrm{prep}\) such that \(W_\mathrm{prep} \geq \Theta(k_L) = \Theta(n_L)\), there exists a constant \(\epsilon_{*,\mathrm{inject}} \in (0,1)\) (independent of \(n_L\)) and a Pauli fault-tolerant gadget \(\gadget_{\mathrm{noisyprep}}\)\footnote{We omit the parameters, since this is a slightly extended notion of fault-tolerant gadget.} preparing \(\ket{\overline{\psi}}\) in a \((\qcode_L,\mathcal{F}_L)\)-block with \(L\) locations such that the following properties hold:
  \begin{itemize}
  \item If the fault is Pauli and its fault path is \(\mathcal{G}_{\mathrm{bad}}\)-avoiding, then the output state is \(\mathcal{F}_L\)-Pauli deviated from \(\qcode_L\) or FAIL is output.
  \item If the fault is Pauli and its fault path is also \(\mathcal{G}_{\mathrm{noisy}}\)-avoiding, then the output state is \(\mathcal{F}_L\)-Pauli deviated from \(\phi_{\qcode_L}(\psi)\).
  \end{itemize}
  Where, for \(x \in [0,\ndep{\epsilon_{*,\mathrm{inject}}}]\), \(\mathcal{G}_{\mathrm{bad}},~\mathcal{G}_{\mathrm{noisy}}\subseteq P([L])\) satisify 
  \begin{align}
    \weightenum{\mathcal{G}_{\mathrm{bad}}}{x} &\le \polylog(n_L) e^{-\zeta(n_L)} \\
    \weightenum{\mathcal{G}_{\mathrm{noisy}}}{x} &\le \epsilon_{\mathrm{noise}}
  \end{align}

  \(\gadget_{\mathrm{noisyprep}}\) has depth \(D_\mathrm{prep}'\) and width \(W_\mathrm{prep}'\).
  \begin{align}
    D_\mathrm{prep}' &= O\left(\left(D_\mathrm{prep} + \log^2 n_L\right) \polylog \left(\frac{W_\mathrm{prep}D_\mathrm{prep}}{\epsilon_{\mathrm{noise}}}\right)\right), \\
    W_\mathrm{prep}' &= O\left(W_\mathrm{prep} \polylog \left(\frac{W_\mathrm{prep}D_\mathrm{prep}}{\epsilon_{\mathrm{noise}}}\right)\right).
  \end{align}
\end{lemma}

The family of fault paths \(\mathcal{G}_{\mathrm{bad}}\) should be thought of as those that cause the output to be so bad that \emph{further computation is not possible} e.g. the state is too far from the codespace for the decoder to operate.
We cannot let such fault paths to occur, so we will perform a testing step at the end of our gadget to increase their weight to be \(\Omega(n_L)\). 
In order to avoid incurring a depth overhead, we use the qLTC property to perform the testing in \(O(\log n_L)\) depth.
The family of fault paths \(\mathcal{G}_{\mathrm{noisy}}\) should be thought of as those fault paths that cause the wrong logical state to be prepared.
This is not immediately fatal as we will later be performing distillation of these states.

\begin{proof}[Proof of \cref{lemma:noisy-prep-gadget}]
  Let \(\mathcal{E}_{\qcode_L}\) be the encoding circuit of \(\qcode_L\) from \cref{lem:nonFT-state-prep} which has space at most \(O(n_L)\) and depth \(O(\log^2 n_L)\).
  Let \(\gadget_{\mathrm{prepsim}}\) be the simulation of \(\mathcal{E}_{\qcode_L} \circ \qOp_{\mathrm{prep}}\) by \(r\)-AB with \(r = \lceil \log_2 \log_2 \frac{3\left(|\mathcal{E}_{\qcode_L}|+|\qOp_{\mathrm{prep}}|\right)}{\epsilon_{\mathrm{noise}}}\rceil\).

  Our gadget \(\gadget_{\mathrm{noisyprep}}\) will execute \(\gadget_{\mathrm{prepsim}}\), run the AB concatenated code unencoding gadgets (\cref{lemma:ab-unencode}), the tester gadget (\cref{lemma:ec-gadget-tester}), and then the error correction gadget (\cref{lemma:ec-gadget}). I.e.
  \begin{align*}
    \gadget_{\mathrm{noisyprep}} = \gadget_{\mathrm{EC}}\circ \gadget_{\mathrm{tester}} \circ \left(\gadget_{\mathrm{unencode,AB,r}}\right)^{\otimes n_L} \circ \gadget_{\mathrm{prepsim}}
  \end{align*}
  The bit produced by the tester gadget is the FAIL bit that this gadget outputs.
  The depth and width of \(\gadget_{\mathrm{noisyprep}}\) bounds follow from the depth/width bounds on the components.
  To simplify the bounds, we use the fact that \(W = \Omega(n_L)\).
  Recall that an \(r\)-AB concatenated code gadget has width and depth at most exponential in \(r\):
  \begin{figure}[H]
  \centering
    \begin{tabular}{|c|c|c|c|}\hline
      &  \(\gadget_{\mathrm{tester}}\) & \(\left(\gadget_{\mathrm{unencode,AB,r}}\right)^{\otimes n_L}\) & \(\gadget_{\mathrm{prepsim}}\)\\\hline
      Depth & \(O(\log n_L)\) & \(O\left(\polylog \left(\frac{W_\mathrm{prep}D_\mathrm{prep}}{\epsilon_{\mathrm{noise}}}\right)\right)\) & \(O\left(\left(D_\mathrm{prep} + \log^2 n_L\right)\polylog \left(\frac{W_\mathrm{prep}D_\mathrm{prep}}{\epsilon_{\mathrm{noise}}}\right)\right)\) \\\hline
      Width & \(O(n_{L})\) & \(O\left(n_L \polylog \left(\frac{W_\mathrm{prep}D_\mathrm{prep}}{\epsilon_{\mathrm{noise}}}\right)\right)\) & \(O\left(W_\mathrm{prep} \polylog \left(\frac{W_\mathrm{prep}D_\mathrm{prep}}{\epsilon_{\mathrm{noise}}}\right)\right)\)\\\hline
    \end{tabular}
  \end{figure}

  We now perform the fault analysis beginning with \((\left(\gadget_{\mathrm{unencode,AB,r}}\right)^{\otimes n_L} \circ \gadget_{\mathrm{prepsim}})\).
  Let \(\mathbf{f}\) be the Pauli fault path.

  Using \cref{lemma:gadget-composition}, we can take the sum (\(\boxplus\)) over the bad fault paths for each of the \((|\mathcal{E}_{\qcode_L}|+|\qOp_{\mathrm{prep}}|)\) concatenated code gadgets comprising \(\gadget_{\mathrm{prepsim}}\) and construct a family of bad fault paths \(\mathcal{G}_1\) for \(\gadget_{\mathrm{prepsim}}\) with \(\weightenum{\mathcal{G}_1}{x} \le (|\mathcal{E}_{\qcode_L}|+|\qOp_{\mathrm{prep}}|)(c_{\mathrm{AB}} x)^{2^r} \le \epsilon_{\mathrm{noise}}/3\) for \(x \in [0,\frac{1}{3c_{\mathrm{AB}}}]\) such that: If \(\supp\fault\) is \(\mathcal{G}_1\)-avoiding, then \(\gadget_{\mathrm{prepsim}}[\fault]\) outputs a state \(\tilde{\psi}_1\) that is \(\mathcal{F}_r\)-deviated on each concatenated code block (recall \(\mathcal{F}_r\) are the bad sets for \(r\)-level AB concatenated codes) from \(\phi_{\qcode_{\mathrm{AB},r}}\circ \phi_{\qcode_L}(\ket{\psi})\).
  
  The bad fault paths \(\mathcal{G}_1\) cause an error in preparing the initial state (encoded in the concatenated code), but this is not immediately fatal.
  Next, we must unencode the state and verify that it is close to the codespace.

  We run the unencoding gadget on each \((\qcode_{\mathrm{AB},r},\mathcal{F}_r)\)-block of \(\tilde{\psi}_1\) in parallel.
  Let \(L_{r,\mathrm{unencode}}\) be the set of locations of a single unencoding gadget \(\gadget_{\mathrm{unencode,AB,r}}\).
  \cref{lemma:ab-unencode} guarantees that the output of a single unencoding gadget is perfect if the fault set is \(\mathcal{G}_{r,\mathrm{unencode}}\)-avoiding on \(L_{r,\mathrm{unencode}}\).
  In order for the output to be \(\mathcal{F}_{\mathrm{test}}\)-deviated from \(\phi_{\qcode_L}(\psi)\), it must be the case that the fault path of \(\mathbf{f}\) must be \(\mathcal{G}_{r,\mathrm{unencode}}\)-avoiding on all but a \(\mathcal{F}_{\mathrm{test}}\)-avoiding set of unencoding gadgets.
  That is, \(\supp\mathbf{f}\) must be \(\mathcal{G}_2:=\mathcal{F}_{\mathrm{test}}\bullet \mathcal{G}_{r,\mathrm{unencode}}\)-avoiding on \(\sqcup_{i \in [n_L]} L_{r,\mathrm{unencode}}\).

  \cref{lemma:gadget-composition} computes the weight enumerator of \(\mathcal{G}_2\) in terms of \(\weightenum{\mathcal{G}_{r,\mathrm{unencode}}}{x}\) and \(\weightenum{\mathcal{F}_{\mathrm{test}}}{x}\).
  Thus for \(x \in \left[0,\frac{1}{2c_{\mathrm{unencode,AB}}}\right]\), using \(t_{\mathrm{test}} = \Theta(\zeta(n_L))\), we can then bound
  \begin{align*}
    \weightenum{\mathcal{G}_2}{x} &= \weightenum{\mathcal{F}_{\mathrm{in}}}{\weightenum{\mathcal{G}_{\mathfrak{r},\mathrm{unencode}}}{x}} \\
                       &\le \binom{n_L}{t_{\mathrm{test}}} \left( c_{\mathrm{unencode,AB}} \cdot x \right)^{t_{\mathrm{test}}} \\
                       &\le \left((\mathrm{const.})\frac{n_L}{t_{\mathrm{test}}} x\right)^{(\mathrm{const.}) {t_\mathrm{test}}} \\
                       &\le \left((\mathrm{const.})\frac{n_L}{\zeta(n_L)} x\right)^{(\mathrm{const.}) \zeta(n_L)}.
  \end{align*}
  It follows that we can pick a constant \(\epsilon_{*,2} \in \left(0,\frac{1}{2c_{\mathrm{unencode,AB}}}\right]\)  such that for \(x \in [0, \ndep{\epsilon_{*,2}}]\),
  \begin{align*}
    \weightenum{\mathcal{G}_2}{x} &\le \left((\mathrm{const.})\frac{n_L}{\zeta(n_L)} \ndep{\epsilon_{*,2}})\right)^{(\mathrm{const.}) \zeta(n_L)} \\
                         &\le \left((\mathrm{const.}) \epsilon_{*,2}\right)^{(\mathrm{const.}) \zeta(n_L)} \\
                       &\le \epsilon_{\mathrm{noise}}/3.
  \end{align*}

  Let \(\mathcal{G}_{\mathrm{bad}} := \mathcal{G}_{\mathrm{tester}}\) be the bad fault sets for the tester gadget \(\gadget_{\mathrm{tester}}\).
  Assume that \(\supp \fault\) is \(\mathcal{G}_{\mathrm{bad}}\)-avoiding.

  Define \(\mathcal{G}_{\mathrm{noisy}}:=\mathcal{G}_1\boxplus \mathcal{G}_2\).
  At this point, the fault analysis splits into two cases, depending on whether \((\left(\gadget_{\mathrm{unencode,AB,r}}\right)^{\otimes n_L} \circ \gadget_{\mathrm{prepsim}})\) is faulty.
  
  If \(\supp\fault\) is \(\mathcal{G}_{\mathrm{bad}} \boxplus \mathcal{G}_{\mathrm{noisy}}\)-avoiding we are left with a post-preparation state \(\tilde{\psi}_2\) that is \(\mathcal{F}_{\mathrm{test}}\)-deviated from \(\ket{\overline{\psi}}\).
  The test always passes, and the post-tester state is \(\mathcal{F}_{\mathrm{L}}\)-Pauli deviated from \(\ket{\overline{\psi}}\).

  If \(\supp\fault\) is only \(\mathcal{G}_{\mathrm{bad}}\)-avoiding there is no guarantee on the post-preparation logical state.
  However, we do have the guarantee that if the tester does not output FAIL, then the post-tester state is  \(\mathcal{F}_{\mathrm{corr}}\)-Pauli deviated from some (possibly mixed) codestate, so the output of \(\gadget_{\mathrm{EC}}\) is \(\mathcal{F}_L\)-Pauli deviated from a codestate.

  The desired fault-tolerance properties of the gadget follows for \(x \in [0,\epsilon_{*,\mathrm{inject}}]\) with \(\epsilon_{*,\mathrm{inject}} = \min\left(\frac{1}{3c_{\mathrm{AB}}}, \epsilon_{*,\mathrm{tester}}, \epsilon_{*,\mathrm{EC}}, \epsilon_{*,2} \right)\).
\end{proof}

\subsection{State distillation}\label{subsec:boost-fidel-state-prep}
We now proceed to define a state distillation procedure.
Informally, a state distillation procedure takes \(M\) imperfect copies of a state \(\psi\) and outputs \(K\) perfect copies of \(\psi\) subject to the condition that the subset of imperfect states is sparse in a certain sense.
\begin{definition}[State distillation procedure]
  Let \(M \in \mathbb{N}\) and \(\mathcal{A}\subseteq P([M])\) be a family of sets.
  For an \(n\)-qubit state \(\psi\), an \((M,K,\mathcal{A})\)-state distillation procedure for \(\psi\) takes an \(n\cdot M\) qubit state state \(\rho\) such that there is an \(\mathcal{A}\)-avoiding set \(A \subseteq [M]\) for which \(\rho|_{[n]\times ([M]\setminus A)} = \psi^{\otimes (M-|A|)}\) and outputs \(\psi^{\otimes K}\).
\end{definition}
As in section \cref{sec:main-proof}, we use an overline to denote a register state encoded in a computational code block.

\subsubsection{State distillation procedures}
We will construct a state distillation scheme for stabilizer states and a particular magic state.
To avoid a lengthy excursion, we will defer the proof of these lemmas to later sections.
\begin{restatable}[State distillation procedure for stabilizer states]{lemma}{lemmaStabDistillationProcedure}\label{lemma:stab-distillation-procedure}
  For a stabilizer state \(\psi\) on \(b\) qubits, there exists constants \(\epsilon_{*,distill} \in (0,1)\), \(\beta_{\mathrm{distill}} > 0\) and a family indexed by \(l \in \mathbb{N}\) of \((M_l,K_l,\mathcal{A}_l)\)-state distillation procedure where \(\weightenum{\mathcal{A}_l}{x} \le (\beta_{\mathrm{distill}} x)^{2^{\ell}}\) on \(x \in [0,\epsilon_{*,distill}]\), \(\frac{K_l}{M_l} = \Theta(1)\), and \(M_l = e^{\Theta(l \log l)}\).

  Furthermore, the procedure satisfies:
  \begin{itemize}
      \item The depth is \(O(l^3)\), and the width is \(O(M_l)\).
      \item No additional input states are required beyond the \(M_l\) noisy inputs.
      \item Only the operations \(\CNOT\), \(Z\)/\(X\) basis measurement, and classically-controlled Pauli gates are used.
      \item With the exception of Pauli gates, the quantum gates applied do not depend on the state \(\psi\) to be distilled, only \(b\) and \(l\).
      \item No two-qubit gates are applied between qubits with different \([b]\) coordinates.\footnote{In other words, the gates are applied ``transversally'' with regard to the qubits \([b]\).}
  \end{itemize}
\end{restatable}
The proof of the lemma utilizes the same constant-rate concatenated Quantum Hamming code family as used in \cite{yamasaki2024time} as a state distillation method with modifications to reduce the time overhead.
Since the construction is mostly standard ideas from state distillation, we defer the proof to \cref{sec:state-distill-procedure-proof}.

The second state distillation procedure will be for magic states.
\begin{restatable}[State distillation procedure for magic states]{theorem}{lemmaMagicDistillationProcedure}\label{lemma:magic-distillation-procedure}
  There exists constants \(\epsilon_{*,Mdistill} \in (0,1)\), \(\beta_{\mathrm{Mdistill}} > 0\) and a family indexed by \(l \in \mathbb{N}\) of \((M_l,K_l,\mathcal{A}_l)\)-state distillation procedure for \(\ket{\CCZ}\) where \(\weightenum{\mathcal{A}_l}{x} \le (\beta_{\mathrm{Mdistill}} x)^{(l+1)!}\) on \(x \in [0,\epsilon_{*,Mdistill}]\), \(\frac{K_l}{M_l} = \Omega\left(c^l\right)\), and \(M_l = e^{\Theta ( l \log l)}\) for some absolute constant \(c > 0\).

  Furthermore, the procedure satisfies:
  \begin{itemize}
      \item The depth is \(O(l^4 (\log l)^3)\), and the width is \(O(M_l)\).
      \item \(O(M_l)\) perfect \(\ket{+}\) and \(\ket{0}\) input states are required
      \item \(O(M_l)\) classically-controlled \(\CZ\) gates are required
      \item Otherwise, only \(\CNOT\), \(Z\)/\(X\) basis measurement, and classically-controlled Pauli gates are used.
  \end{itemize}
\end{restatable}
The proof of this state distillation procedure will require the construction of new quantum codes over qudits of prime-power dimension.
The codes are constructed in \cref{sec:magic-state-distillation-code}, and the proof of \cref{lemma:magic-distillation-procedure} is deferred until \cref{lemma:magic-distillation-procedure-proof}.

\subsubsection{FT state distillation gadgets (proofs of \texorpdfstring{\cref{lemma:stab-resource-state}}{\cref*{lemma:stab-resource-state}} and \texorpdfstring{\cref{lemma:magic-resource-state}}{\cref*{lemma:magic-resource-state}})}\label{sec:ft-state-distillation}

We are now ready to prove the main lemmas for preparation of resource states.
At this point, we are ready to switch to asymptotic notation.
The first gadget exploits the fact that the stabilizer state distillation procedure in \cref{lemma:stab-distillation-procedure} uses the same gates regardless of the state to be distilled.
Thus, we can distill resource states that differ on different coordinates of the computation code.

\lemmaStabilizerResourceStatePrep*

\begin{proof}
  The gadget will first consist of parallel repetitions of preparation of noisy resource states.
  We will then collect those resource states which pass the qLTC check for being close to the codespace and then proceed with the stabilizer state distillation procedure.
    
  Set \(l = \log_2\log_2 \frac{1}{e^{-\zeta(n_L)}} = O(\log n_L)\) in \cref{lemma:stab-distillation-procedure} and \(\epsilon_{\mathrm{noise}}= \min\left(\frac{1}{2\beta_{\mathrm{distill}}}, \frac{1}{8}\right)\) in \cref{lemma:noisy-prep-gadget}.
  Let \(\epsilon_{*,\mathrm{stab}} = \min\left(\epsilon_{*,\mathrm{inject}}, \epsilon_{*,\mathrm{noise}}\right)\).
  This choice of parameters implies that the number of states input to the stabilizer state distillation procedure (\cref{lemma:stab-distillation-procedure}) \(M_{\mathrm{stab}}\equiv M_l = e^{\Theta(\log n_L \loglog n_L)}\).
  Let \(\alpha = \min(n_L^2, M_{\mathrm{stab}}) = \Omega(n_L^2)\) and \(\beta = \left\lceil \frac{M_{\mathrm{stab}}}{\alpha}\right\rceil\).
  By assumption, \(\qOp_{\mathrm{stab}}\) is constant depth and has width \(O(n_L)\).
  
  The gadget is constructed as follows:
  \begin{enumerate}
  \item Consider the noisy simulation of \(\qOp_{\mathrm{stab}}\) using the gadget \(\gadget_\mathrm{noisyprep}\) from \cref{lemma:noisy-prep-gadget}.
  \(\gadget_\mathrm{noisyprep}\) uses width \(O(n_L c(n_L))\) for some \(c(n_L) = O(\polylog n_L)\) and depth \(O(\polylog n_L)\).
  Execute \(\gadget_\mathrm{noisyprep}\) \(\left\lceil\frac{2 \alpha \beta}{c(n_L)}\right\rceil\) times in parallel \(c(n_L)\) times for a total of \(2 \alpha\beta\) states.\footnote{The organization of states in groups of \(\alpha\) is required to avoid a \(\log M_l\) time overhead when collecting the outputs. The preparation over \(c(n_L)\) steps is required to avoid incurring space overhead.}
  Overall, this requires depth \(O(\polylog n_L)\) and width \(O(n_L \polylog n_L)\), so this step has depth \(O(\polylog n_L)\) and width \(O(\alpha \beta n_L) = O(M_{\mathrm{stab}} n_L)\).
    Each of the \(2\alpha\beta\) executions of \(\gadget_\mathrm{noisyprep}\) produce a state on \(b\) code blocks.
    Since the code blocks are naturally indexed by \([\beta] \times [2\alpha] \times [b]\), we will refer to each group of \(b\) code blocks as a ``row'' and each group of \([\beta]\) rows as a ``column.''
  \item Let \(J \subseteq [\beta]\times [2\alpha]\) be the subset of state preparation gadgets for which the FAIL bit is not set (i.e. the qLTC checks are passed).
    In parallel, for each column \(i \in [\beta]\), sort the \(2\alpha\) rows of computation code blocks such that the rows originally in \(J\) are the first rows in column \(i\).
    This can be done in depth \(O(\log \alpha)\) and width \(O(\alpha \beta n_L)\) using a classically controlled transversal \(\SWAP\) of the computational code blocks.
  \item Define the set
    \begin{align*}
      I = \{(i,j) \in [\beta] \times [2\alpha] \mid (i < \beta~\text{and}~j\le \alpha)~\text{or}~(i=\beta~\text{and}~j \le M_{\mathrm{stab}} - (\beta - 1) \alpha\}
    \end{align*}
    of size \(M_{\mathrm{stab}}\) given by taking \(\alpha\) rows from all but the last column and the remaining rows from the last column. Execute the distillation in \cref{lemma:stab-distillation-procedure} on the rows in \(I\).
    Since the gates in the distillation procedure do not depend on the state and act only coordinate-wise, this can be done using the transversal gate gadgets of the computation code (\cref{lemma:computational-code-transversal-gates}).
    This has depth \(O(\polylog n_L)\), width \(O(M_{\mathrm{stab}} n_L)\), and produces \(\Omega(M_{\mathrm{stab}})\) resource states.
  \end{enumerate}

  In the following, we restrict \(x\in [0,\ndep{\epsilon_{*,\mathrm{stab}}}]\) and upper bounds do not depend on \(x\) unless indicated.

  For \((i,j) \in I\), let \(\mathcal{H}^{(i,j)}_{\mathrm{bad}}\) and \(\mathcal{H}^{(i,j)}_{\mathrm{noisy}}\) denote the families of bad fault paths \(\mathcal{G}_{\mathrm{bad}}\) and \(\mathcal{G}_{\mathrm{noisy}}\) from \cref{lemma:noisy-prep-gadget} on the \((i,j)\)-th state preparation gadget in step-1.
  Let \([L]\) index the set of computational code gadgets applied in steps 2 and 3.
  By construction, \(L = O(M_{\mathrm{stab}} \polylog(n_L))\).
  For each computational code gadget \(l \in [L]\), let \(\mathcal{H}_l\) be the corresponding set of bad fault paths.
  
  The output is correct as long as none of the following events occur:
  \begin{enumerate}
  \item A computational code gadget or the state verification step in \cref{lemma:noisy-prep-gadget} fails:
    \begin{align}
      \mathcal{G}_1 &= \left(\wtsum_{l \in [L]} \mathcal{H}_l \right) \wtsum \left(\wtsum_{(i,j) \in I} \mathcal{H}^{(i,j)}_{\mathrm{bad}}\right)\\
      \weightenum{\mathcal{G}_1}{x} &\le \sum_{l \in [L]}\weightenum{\mathcal{H}_l}{x}  + \sum_{(i,j) \in I} \weightenum{\mathcal{G}_{\mathrm{bad}}}{x}\le (L+|I|) e^{-\zeta(n_{L})} \\
                         &= O\left(M_{\mathrm{stab}} \polylog(n_L) e^{-\zeta(n_L)}\right).
    \end{align}
  \item There exists a column with less than \(\alpha\) rows for which FAIL was not set:\footnote{We slightly abuse notation and consider \(A_l\subseteq I\). This detail is not important.}
    \begin{align}
      \mathcal{G}_2 &= \wtsum_{i \in [\beta]} \wtsum_{\substack{J \subseteq [2\alpha] \\ |J| = \alpha+1}} \left(\wtprod_{j \in J} \mathcal{H}^{(i,j)}_{\mathrm{noisy}}\right)\\
      \weightenum{\mathcal{G}_2}{x} &\le \beta \binom{2\alpha}{\alpha+1} \left( \weightenum{\mathcal{G}_{\mathrm{noisy}}}{x}\right)^{\alpha+1}\\
                         &\le \beta \cdot 2^{2\alpha} \left(\epsilon_{\mathrm{noise}}\right)^{\alpha+1}\\
                         &\le \beta \cdot \left(\frac{1}{2}\right)^{\alpha}\\
                         &\le \beta e^{-\Omega(\alpha)} \le M_{\mathrm{stab}} e^{-\Omega(n_l^2)}.
    \end{align}
  \item The distillation output contains an error: \(\mathcal{G}_3\),
    \begin{align}
      \mathcal{G}_3 &= \wtsum_{A \in \mathcal{A}_l} \wtprod_{(i,j) \in A} \mathcal{H}_{\mathrm{noisy}}^{(i,j)} \simeq \mathcal{A}_l\bullet \mathcal{G}_{\mathrm{noisy}}\\
      \weightenum{\mathcal{G}_3}{x} &\le \weightenum{\mathcal{A}_l}{ \weightenum{\mathcal{G}_{\mathrm{noisy}}}{x}}\\
                           &\le \left(\frac{\beta_{\mathrm{distill}}}{2 \beta_{\mathrm{distill}}}\right)^{2^l}\\
                           &\le e^{-\zeta(n_L)}.
    \end{align}
  \end{enumerate}
  Thus, \(\mathcal{G}_{\mathrm{stab}} = \wtsum_{i=1}^3 \mathcal{G}_{i}\) with
  \begin{align}
      \weightenum{\mathcal{G}_{\mathrm{stab}}}{x} &\le O(M_{\mathrm{stab}} \polylog(n_L) e^{-\zeta(n_L)}).
  \end{align}
\end{proof}

The construction of an analog of \cref{lemma:stab-resource-state} for magic states is essentially similar with minor differences due to the need for a conditional Clifford operation and the preparation of the resource state only on the first coordinate instead of for all coordinates.

\lemmaMagicResourceStatePrep*

\begin{proof}
  Let \(w(x)\) be the principal branch of the Lambert \(W\) function.
  I.e. a solution to \(w(x) e^{w(x)} = x\) for \(x \ge 0\) satisfying \(w(x) \ge 0\).
  We start by defining the function \(g(x)\) which satisfies the following:
  \begin{align}
    g(x) = \frac{\log(x)}{w(\log(x^{1/e}))}, & & \left(\frac{g(x)}{e}\right)^{g(x)} = x.
  \end{align}
  This allows us to obtain an upper bound for the inverse of the factorial:
  \begin{align}
   \lceil g(x)\rceil! \ge \left(\frac{g(x)}{e}\right)^{g(x)} = x.
  \end{align}
  We will follow the argument of \cref{lemma:stab-resource-state} almost exactly.
  We apply the magic state distillation procedure of \cref{lemma:magic-distillation-procedure} using the transversal gate gadget (\cref{lemma:computational-code-transversal-gates}).

  We use \cref{lemma:magic-distillation-procedure} with \(l\) such that
  \begin{align}
    l+1 = \left\lceil g\left(\zeta(n_L)\right)\right\rceil
  \end{align}
  and \(\epsilon_{\mathrm{noise}}=\min\left(\frac{1}{e\beta_{\mathrm{Mdistill}}}, \frac{1}{8}\right)\) in \cref{lemma:noisy-prep-gadget}.
  Define \(M_{\mathrm{magic}} \equiv M_l = e^{\Theta(l \log l)}\), the number of states input to the magic state distillation protocol in the statement of \cref{lemma:magic-distillation-procedure}.
  As in \cref{lemma:stab-resource-state}, we set
  \begin{align*}
    \epsilon_{*,\mathrm{magic}} = \min\left(\epsilon_{*,\mathrm{inject}}, \epsilon_{*,\mathrm{noise}}, \epsilon_{*,\mathrm{stab}}\right)~.
  \end{align*}
  Note that \(\epsilon_{*,\mathrm{inject}}\) and \(\epsilon_{*,\mathrm{noise}}\) are not necessarily the same as in \cref{lemma:noisy-prep-gadget}.

  We first use \(\gadget_\mathrm{stab}\) (\cref{lemma:stab-resource-state}) to prepare \(\Theta(M_{\mathrm{magic}})\) copies of \(\ket{\overline{\CZ(1_A,1_B)}}\), \(\ket{\overline{+}}\), and \(\ket{\overline{0}}\).
  The \(\ket{\overline{\CZ(1_A,1_B)}}\) states allow us to use transversal gates (\cref{lemma:computational-code-transversal-gates}) to execute a teleported \(\CZ\) gate on the first coordinate in the execution of the magic state distillation procedure \cref{lemma:magic-distillation-procedure}.
  These are the necessary inputs to execute the magic state distillation procedure of \cref{lemma:magic-distillation-procedure} using transversal gates (\cref{lemma:computational-code-transversal-gates}) to produce \(K_{\mathrm{magic}}\) \(\ket{\overline{\CCZ(1_A,1_B,1_C)}}\) states.
  At the end, we again use \(\gadget_\mathrm{stab}\) (\cref{lemma:stab-resource-state}) to prepare \(K_{\mathrm{magic}}\) copies of \(\ket{\overline{\SWAP(2_A,1_B)}}\) and \(\ket{\overline{\SWAP(3_A,1_B)}}\) which are used to execute teleported \(\SWAP(2_A,1_B)\SWAP(3_A,1_C)\) on \(\ket{\overline{\CCZ(1_A,1_B,1_C)}}\), converting it to \(\ket{\overline{\CCZ(1,2,3)}}\).

  The family of bad fault paths associated with these invocations of the \cref{lemma:stab-resource-state} gadget is added to \(\mathcal{G}_1\), but it does not change the asymptotic scaling with \(n_L\).
  The bound for \(\weightenum{\mathcal{G}_2}{x}\) is identical to that of \cref{lemma:stab-resource-state}.
  \(\mathcal{G}_3\) is significantly modified: For \(x \in [0, \ndep{\epsilon_{*,\mathrm{magic}}})\), we have the new bound
  \begin{align}
    \weightenum{\mathcal{G}_3}{x} \le \left(\frac{\beta_{\mathrm{Mdistill}}}{e\cdot \beta_{\mathrm{Mdistill}}}\right)^{(l+1)!} &\le e^{-(l+1)!}\le e^{-\zeta(n_L)}.
  \end{align}

  The depth is \(O(\polylog n_{L})\), and the width is \(O(n_L M_{\mathrm{magic}})\).
  For some constant \(c \in (0,1)\), we have produced
  \begin{align}
    K_{\mathrm{magic}} &= \Omega(M_{\mathrm{magic}} c^l) = \Omega\left(M_{\mathrm{magic}} c^{g\left(\zeta(n_L)\right)}\right)\\
      &= \Omega\left(M_{\mathrm{magic}} \exp\left[-O\left(\frac{\log n_L}{w\left(\log(n_L^{1/e})\right)}\right)\right]\right) \\
      &= \Omega\left(M_{\mathrm{magic}} n_L^{-\tilde{\gamma}(n_L)}\right)
  \end{align}
   \(\ket{\overline{\CCZ(1_A,1_B,1_C)}}\) states where \(\tilde{\gamma}(n_L) = O\left(\frac{1}{w\left(\log(n_L^{1/e})\right)}\right) = O_{n_L\to \infty}\left(\frac{1}{\loglog n_L}\right)= o_{n_L\to \infty}(1)\).
\end{proof}

\subsection{Constant-overhead stabilizer state distillation procedure}\label{sec:state-distill-procedure-single}
It remains to prove \cref{lemma:stab-distillation-procedure}.
We first begin by defining a procedure for a single round of distillation of a stabilizer state.

We heavily use the stabilizer formalism~\cite{gottesman1997stabilizer,aaronson2004improved}.
For an \(n\)-qubit stabilizer state \(\ket{\psi}\), we use the notation \(\stab(\psi)\) to refer to the subgroup of Pauli operators \(s\) such that \(s \ket{\psi} = \ket{\psi}\).

For an \(n \times m\) matrix \(A\), we use the notation \(A[\cdot, j]\) to refer to the \(j\)-th column, and \(A[i,\cdot]\) to refer to the \(i\)-th row.
For \(i \in [n]\), let \(e_i \in \F^n\) be the indicator vector for the \(i\)-th coordinate.
In what follows, for an \(b\) qubit Pauli operator \(P=P^{(1)}_1P^{(2)}_2\dots P^{(b)}_b\) (the operator \(P^{(i)}\) acting on the \(i\)-th qubit) and a bit string \(x \in \F^n\), we will use \(P^x \in \mathcal{P}^{n \times b}\)  to denote the Pauli operator supported on a set of qubits indexed by \([n] \times [b]\) given by
\begin{align}
  P^x \equiv \prod_{i \in [n], j \in [b]} (P^{(j)})^{x[i]}_{i,j} \in  \mathcal{P}^{n \times b}~.
\end{align}
To reduce the amount of notation, we introduce the following abuse of notation:
For a matrix \(H\in \F^{r\times n}\), and an \(b\)-qubit Pauli operator \(P\), we use \(P^H\) to denote the indexed set
\begin{align}
  P^H\equiv \left\{P^{H[i,\cdot]}\right\}_{i\in [r]} \subseteq \mathcal{P}^{n \times b} ~.
\end{align}
I.e. every row \(i \in [r]\) of \(H\) defines an operator \(h\).
For every column \(j \in [n]\) of \(H\), the corresponding row has \(h|_{\{j\} \times [b]} = P\) iff \(H[i,j] = 1\).

Let \(\qcode_{\mathrm{distill}}\) be a quantum CSS code with parameters \([[n,k,d]]\) with the following properties:
\begin{itemize}
\item There is a (check) matrix \(H \in \F^{r \times n}\) such that the set \(X^H\cup Z^H\) form a minimal generating set for the stabilizer of \(\qcode_{\mathrm{distill}}\).
\item There exists a (generator) matrix \(G \in \F^{k \times n}\) such that \(GG^T = I\) and \(HG^T = 0\).
\item There exists an efficient algorithm \(\mathcal{D}_H \colon \F^r \to \F^n\) to decode all errors of weight at most \(t\) for classical code \(\ker H\).
  I.e. for all \(x \in \F^n\) such that \(x \le t\), \(\mathcal{D}_H(Hx) = x\).
\end{itemize}
Such a quantum CSS code is sometimes said to be ``self-dual''\footnote{This notion is unrelated to the notion of the dual of a classical code.} and possesses convenient properties such as a fully transversal implementation of the Clifford group.
Quantum Hamming codes as well as color codes satisfy this property.

Let \(S_{\psi}\) be a minimal generating for the stabilizer \(\stab(\psi)\) of the \(b\)-qubit state \(\psi\).
I.e. for all \(s \in S_{\psi}\), \(\psi\) is an eigenstate with eigenvalue \(+1 \in \{\pm 1\}\).
In the following, we will associate measurement outcomes of Pauli operators with elements of \(\F\) in the canonical way: (\(+1 \mapsto 0 \in \F\), \(-1 \mapsto 1 \in \F\)).

We first compute (by Gaussian elimination) a set of operators that each anticommute with exactly one element of \(S_{\psi}\).
Concretely, we have a map \(\phi_{\psi}\colon S_{\psi} \to \mathcal{P}^b\) such that for \(s \in S_{\psi}\), \(\phi_{\psi}(s)s = -s\phi_{\psi}(s)\) and for \(S_{\psi}\ni s'\ne s\), \(\phi_{\psi}(s) s' = s' \phi_{\psi}(s)\).
This is sometimes call the ``destabilizer'' of \(S_{\psi}\).

Fix an encoding map (using only \(\CNOT\)) \(\mathcal{E}\) for \(\qcode_{\mathrm{distill}}\) such that for \(j \in [k]\), \(\mathcal{E} Z_j \mathcal{E}^{-1} = Z^{G[j,\cdot]}\) and \(\mathcal{E} X_j \mathcal{E}^{-1} = X^{G[j,\cdot]}\).
This fixes a phase of Pauli \(Y\), \(\mathcal{E} Y_j \mathcal{E}^{-1} = (i)^{|G[j,\cdot]| - 1} Y^{G[k,\cdot]} \equiv \theta_j Y_j\) where \(\theta_j \in \{\pm 1\}\).
We will need to ``fix up'' this phase in order to produce the correct state on the output.
For \(P \in \mathcal{P}^n\), let us notate the phase of an arbitrary Pauli under the encoding map \(\mathcal{E}^{\otimes b} P_j \left(\mathcal{E}^{-1}\right)^{\otimes b} = \theta_{P,j} P\).
Define the fix up operator \(\Xi\in \mathcal{P}^{k\times b}\) be a Pauli operator that corrects this phase: \(s \in S_\psi\), \(\Xi s^{e_j} = \theta_{s,j}^{b} s^{e_j} \Xi\).
\(\Xi\) can be computed by taking the product of destabilizer operators \(\phi_\psi(s)\) on position \(j\) whenever the phase is incorrect \(\left(\theta_{s,j}\right)^b = -1\).
\begin{align}
    \Xi = \prod_{\substack{j \in [k] \\ s \in S_\psi \\ \theta_{s,j} = -1}} \left(\phi_\psi(s)\right)^{e_j}
\end{align}
\RestyleAlgo{boxruled}
\begin{algorithm}[H]
  \DontPrintSemicolon
  \LinesNumbered
  \caption{Stabilizer state distillation algorithm}\label{alg:stab-state-distillation}
  \KwIn{\(n \times b\) qubit state \(\rho\)}
  \KwOut{\(k \times b\) qubit state \(\rho_{\mathrm{out}}\)}
  \Begin{
    \(m^{(X)} \in \F^{r \times b}\)
    \(m^{(Z)} \in \F^{r \times b}\)
    \(\sigma \in \F^r\)\;
    \(U \leftarrow I\)\;
    \tcc*{Measure checks of \(\qcode_{\mathrm{distill}}\)}
    \ForEach{\(j \in [b]\)}{ \label{alg:distill:measure}
      \(m[\cdot,j]^{(X)} \leftarrow\) Measurement of \(X^H\) on \(\rho|_{[n]\times\{j\}}\) \;
      \(m[\cdot,j]^{(Z)} \leftarrow\) Measurement of \(Z^H\) on \(\rho|_{[n]\times\{j\}}\) \;
    }
    \(\rho' \leftarrow\) Post-measurement state of \(\rho\)\;
    \tcc*{Correct eigenvalue of \(s^H\)}
    \ForEach{\(s \in S_{\psi}\)}{\label{alg:distill:stab-loop}
      \((a,z,x) \leftarrow z,x \in \F^b\),~\(a \in \F\) such that \(s = \pm i^a  Z^z X^x\) \tcc*{Decompose} \label{alg:distill:decompose}
      \tcc*{Compute syndrome of \(s^{H}\)}
      \ForEach{\(i \in [r]\)}{%
        \(h \leftarrow H[i,\cdot]\)\;
        \tcc*{Inferred measurement of \(s^h\)}
        \(\sigma_i \leftarrow a^{|h|/2} + \langle x, m[i,\cdot]^{(X)}\rangle + \langle z, m[i,\cdot]^{(Z)}\rangle\) \label{alg:distill:inferred-meas}
      }
      \(c \leftarrow \mathcal{D}_H(\sigma)\) \tcc*{Compute correction} \label{alg:distill:correction}
      \(P \leftarrow \phi_{\psi}(s)\)\;
      \(U \leftarrow UP^c\) \tcc*{Update correction}
    }
    \(\rho_{\mathrm{decoded}} \leftarrow \left(\mathcal{E}^{-1}\right)^{\otimes b}\circ U (\rho)\) \tcc*{Apply correction and unencode. Discarding qubits in \(([n]\setminus [k]) \times [b]\)} \label{alg:distill:unencode} \;
    \(\rho_{\mathrm{out}} \leftarrow \Xi(\rho_{\mathrm{decoded}})\) \tcc*{Fix up phases}
  }
\end{algorithm}

For \(d,n \in \mathbb{N}\), define \(\mathcal{S}^n_d = \{x \subseteq [n] \mid |x| = d\} \subseteq P([n])\) to be all subsets of \([n]\) of size \(d\).
Let \(\mathcal{F}_{\mathrm{distill}} =  \mathcal{S}^n_{t+1} \bullet \mathcal{S}^b_1\) be the family of subsets of \([n] \times [b]\) that contain one element from each of \(t+1\) rows.
\begin{prop}[Single level stabilizer state distillation]\label{prop:single-level-stab-state-distill}
  In \cref{alg:stab-state-distillation}, if \(\rho\) is \(\mathcal{F}_{\mathrm{distill}}\)-deviated from \(\psi^{\otimes n}\), then \(\rho' = \psi^{\otimes k}\).
\end{prop}
\begin{proof}
  By \cref{lemma:deterministic-errors}, it suffices for us to consider \(\rho\) that differs from \(\psi^{\otimes n}\) by a Pauli operator \(E\) supported on an \(\mathcal{F}_{\mathrm{distill}}\)-avoiding set i.e. is \(\mathcal{F}_{\mathrm{distill}}\)-diagonal Pauli-deviated from \(\psi^{\otimes n}\).
  Otherwise, it is in the kernel of the measurement channel.
  
  First, note that \(\psi^{\otimes n}\) is the simultaneous +1 eigenstate of the set of operators \(\{s^{e_i}\}_{i \in [n], s \in S_{\psi}}\).
  I.e. the measurement of the set of operators \(s^H\) on \(\psi^{\otimes n}\) is \(0\).
  Thus, measurement of \(s^H\) on \(\rho\) (an eigenstate) will produce a vector \(\sigma \in \F^r\) such that \footnote{The phase should be interpreted ``coordinate-wise.''} \(s^HE = (-1)^\sigma E s^H\).
  We will indirectly measure these operators.
  
  The measurements at \cref{alg:distill:measure} all commute.
  When a measured operator \(O\) is not in the stabilizer \(S\) of the state, the subgroup \(S'\) of \(S\) commuting with \(O\) is retained and the new stabilizer is \(\langle S', (-1)^a O \rangle\) for a random phase \(a\) (the measurement outcome).
  The decomposition at \cref{alg:distill:decompose}, allows us to write \(s^h = (\pm i^a)^{|h|} (Z^z)^h (X^x)^h\) at \cref{alg:distill:inferred-meas}.
  By assumption of the code properties, \(H^TH = 0\), so \(|h|\) is even.
  Furthermore, \((Z^z)^h (X^x)^h\) is proportional to the products of a subset of \(\cup_{i \in [b]} (X_i)^h\) and \(\cup_{i \in [b]} (Z_i)^h\).
  These operators all pairwise commute and were previously measured.
  Thus, the quantity computed at \cref{alg:distill:inferred-meas} is the current eigenvalue of \(s^h\) on the post-measurement state (i.e. is the variable \(\sigma\) from the previous paragraph).
  The correction operator computed at \cref{alg:distill:correction} modifies only the eigenvalues of \(s^H\) (and of \(s^G\)).
  By definition of the decoding map and \(\phi_{\psi}\),
  \begin{align}
    s^H P^cE = (-1)^{Hc+\sigma} P^c E s^H= P^cE s^H \label{prop:single-level-stab-state-distill:PcESyndrome}
  \end{align}
  and for any \(s' \in S_{\psi}\) with \(s \ne s'\), \((s')^H P^c = P^c (s')^H\).

  It remains to analyze \(R:=UE\).
  We will show that \(R\) commutes with \(\{s^{e_i}\}_{i \in [n], s \in S_{\psi}}\).
  Since \(E\) was supported on a \(\mathcal{F}_{\mathrm{distill}}\)-avoiding set, there exists a subset \(I \subseteq [n]\) with \(|I| \le t\) such that \(\supp E \subseteq I \times [b]\).

  Consider an operator \(s \in S_{\psi}\) and let all variables take their values at the end of the loop trip (starting at \cref{alg:distill:stab-loop}) corresponding to \(s\).
  Define the vectors \(x,y \in \F^n\) such that for \(i \in [n]\)
  \begin{align}
    s^{e_i} E &= (-1)^{x[i]}E s^{e_i}\\
    s^{e_i} (P^cE) &= (-1)^{y[i]}(P^cE)s^{e_i}
  \end{align}
  Note that \(Hx = \sigma\).
  Since \(E\) is supported on at most \(t\) columns (\(|I| \le t\)), \(|x| \le t\) and so \(\mathcal{D}_{H}(\sigma ) = \mathcal{D}_{H}(Hx) = x = c\).
  Thus, \(y = c+x = 0\), so \(R\) commutes with \(\{s^{e_i}\}_{i \in [n], s \in S_{\psi}}\), the generators of  \(\stab(\psi^{\otimes n})\).

  For a Pauli operator \(P\) and a set \(B \subseteq \mathcal{P}^n\), denote \(\{PbP \mid b \in B\}\) by \(PBP\).
  Since \(\stab(\rho) = E\stab(\psi^{\otimes n})E\),  the post measurement state has stabilizer given by\footnote{The phase on the last two terms is not important.}
  \begin{align}
    \stab(\rho') = \langle \{Es^{G}E\}_{s \in S_{\psi}}, \cup_{j \in [b]} (-1)^{m^{(X)}[\cdot,j]}(X_j)^H, \cup_{j \in [b]} (-1)^{m^{(Z)}[\cdot,j]}(Z_j)^H \rangle
  \end{align}
  Using the commutativity of \(R\) with \(s^{G} \subseteq \stab(\psi^{\otimes n})\), the post correction state has stabilizer
  \begin{align}
    \stab(U(\rho')) = \langle \{s^{G}\}_{s \in S_{\psi}}, \pm \cup_{j \in [b]} (X_j)^H, \pm \cup_{j \in [b]} (Z_j)^H \rangle
  \end{align}
  \((\mathcal{E}^{-1})^{\otimes b}\) unencodes the blocks.
  That is, for each row \(i \in [k]\) of \(G\), each \(j \in [b]\), and \(P \in \{X,Z\}\), \((\mathcal{E}^{-1})^{\otimes b}(P_j^{G[i,\cdot]}) = P_{(i,j)}\).
  Any operator of the form \(P_j^H\) is mapped to an operator supported only on the qubits \(\left([n]\setminus[k]\right)\times [b]\) which are discarded.
  Thus \(\stab(\rho_{\mathrm{decoded}}) = \langle \{\theta_{s,i}^{b} s^{e_i}\}_{i \in [k], s \in S_{\psi}}\rangle\).
  Finally, the application of \(\Xi\) fixes the phases so that \(\stab(\rho_{\mathrm{decoded}}) = \langle \{s^{e_i}\}_{i \in [k], s \in S_{\psi}} = \stab(\psi^{\otimes k})\). \qedhere
  
\end{proof}

\subsubsection{Proof of \cref{lemma:stab-distillation-procedure}}\label{sec:state-distill-procedure-proof}
\lemmaStabDistillationProcedure*
\begin{proof}
The quantum Hamming code satisfies the conditions required of \(C_\mathrm{distill}\) in \cref{sec:state-distill-procedure-single} \cite{steane1996simple,yamasaki2024time}.
\(H\) is the check matrix of a classical Hamming code and t=1.
For \(r \ge 3\), the parameters of the quantum Hamming code are \([[n=2^r - 1, k=2^r - 2r - 1, 3]]\).

We will repeatedly apply the state distillation procedure \cref{sec:state-distill-procedure-single} with quantum Hamming codes of increasing size.
In particular, we will use the sequence \((r_1, r_2, \dots, r_l)\) where \(r_\ell = \lfloor 2 \log_2(4\ell)\rfloor\).
Let \(n_\ell = 2^{r_\ell} - 1\), \(k_i = 2^{r_\ell} - 2r_\ell - 1\), \(M_l := \prod_{\ell=1}^l n_\ell\), and \(K_l := \prod_{\ell=1}^l k_\ell\)

We will start with blocks of \(b\) qubits indexed by \([n_l] \times \dots \times [n_1]\) initialized as \(M_l:=\prod_{\ell=1}^l n_\ell\) noisy copies of \(\psi\).
Iterating from \(\ell=1\) to \(\ell=l\), for each \(I \in [n_l] \times \dots \times [n_{\ell+1}]\) and \(J \in [k_{\ell-1}]\times\dots\times [k_1]\) we apply the state distillation procedure \cref{sec:state-distill-procedure-single} on the set of blocks \(\{I\} \times [n_{\ell}] \times \{J\}\) to get a new set of blocks with labels \(\{I\} \times [k_{\ell}] \times \{J\}\).
Let \(\mathcal{S}_{d}^{n}\) be the family of all subsets of \([n]\) of size \(d\).
At step \(\ell\), \cref{prop:single-level-stab-state-distill} gives that the output is \(\psi^{\otimes k_{\ell}}\) if the input is \(\mathcal{S}^{n_\ell}_2\bullet \mathcal{S}^{b}_1\)-deviated from \(\psi^{\otimes n_{\ell}}\).
We use the fact that, by construction, no two outputs from a single application of the state distillation procedure are used together in a following state distillation procedure: If the output of a state distillation in step \(\ell\) is not \(\psi^{\otimes k_{\ell}}\), then the input is not \(\mathcal{S}^{n_{\ell}}_2\)-deviated from \(\psi^{\otimes n_{\ell}}\).
Since each input is from separate state distillation procedures, \(t+1\) distinct state distillation procedures at step \(\ell-1\) must have had inputs that are not \(\mathcal{S}^{n_{\ell-1}}_2\)-deviated from \(\psi^{\otimes n_{\ell-1}}\).
Inducting from \(\ell=1\), the output of the final iteration is \(\psi^{K_l}\) if the input is \(\mathcal{A}_l:=\mathcal{S}^{n_l}_2\bullet\dots\bullet\mathcal{S}^{n_1}_2 \bullet \mathcal{S}^b_1\)-deviated from \(\psi^{\otimes M_l}\).

\(\weightenum{\mathcal{S}^n_d}{x} = \binom{n}{d} x^d \le n^d x^d\), so using \cref{lemma:composition-upper-bound}, for \(x \ge 0\), we can upper bound 
\begin{align}
    \weightenum{\mathcal{A}_l}{x} &= \weightenum{\mathcal{S}_2^{n_l}}{\weightenum{\mathcal{S}_2^{n_{l-1}}}{ \dots \weightenum{\mathcal{S}_2^{n_1}}{\weightenum{\mathcal{S}_1^b}{x}}}}\\
    &\le (b x)^{2^l}\left(\prod_{i=1}^l n_i^{2^{(l-i)}}\right)^2\\
                         &\le (b x)^{2^l}\left(\prod_{i=1}^l (4 i)^{2^{(l-i)}}\right)^2\\
                         &\le (64 b x)^{2^l}
\end{align}
Where we have used \(\sum_{i=1}^{\infty} 2^{-i} \log_2(4i) \le 3\).
Then, \(\beta_{\mathrm{distill}} = 64 b\) and \(\epsilon_{*,distill} = 1/\beta_{\mathrm{distill}}\).

We now turn our attention to the other parameters.
First, \((\ell!)^2 4^{2\ell} = \prod_{i=1}^{\ell} 2^{\log_2 (4\ell)^2} \le M_l \le \prod_{i=1}^{\ell} \frac{1}{4} 2^{\log_2 (4\ell)^2} = (\ell!)^2 4^{\ell}\), so \(M_l=e^{\Theta(l \log l)}\).
For the rate,
\begin{align}
  \frac{K_l}{M_l} &= \prod_{\ell=1}^l \frac{k_{\ell}}{n_{\ell}} \\
                  &\ge \prod_{\ell=1}^\infty \left(1-\frac{2r_{\ell}}{n_{\ell}}\right) \\
                  &\ge \frac{1}{7}\prod_{\ell=2}^\infty \left(1-\frac{8\log_2(4\ell)^2}{(4\ell)^2}\right)\\
                  &\ge \frac{1}{300}
\end{align}
where we have upper bounded the first term in the product separately and numerically evaluated \(\prod_{\ell=2}^\infty \left(1-\frac{8\log_2(4\ell)^2}{(4\ell)^2}\right) \ge \frac{1}{300}\).

To prove the circuit properties, first note that using a length-\(n\) CSS code, syndrome measurement followed by unencoding can be done in depth \(O(n^2)\) using only \(\CNOT\), and \(X\)/\(Z\) measurement without any ancilla qubits:
For the encoding with logical operators \(X^G\) and \(Z^G\) used in the previous section, there exists a unitary encoding circuit \(U\) supported on \(n\) qubits of depth \(O(n^2)\) \cite{cleve1997efficient,gottesman1997stabilizer} that uses only \(\CNOT\).
For an \(n\)-qubit state and \(r_x,r_z\) such that \(r_x+r_z=n-k\), \(U \ket{\psi}\ket{0}^{r_z} \ket{+}^{r_x} = \ket{\overline{\psi}}\).
\(U^{-1}\) maps a set of stabilizer generators of the code to \(\{Z_i\}_{i=n-k}^{n-k+r_z}\cup \{X_i\}_{i=n-k+r_z}^n\), so executing the circuit \(U^{-1}\) followed by measurement of the last \(n-k\) qubits (in either the \(X\) or \(Z\) basis) will produce a measurement outcome \(a \in \F^{n-k}\) and state on the first \(k\) qubits that is equivalent to first measuring the stabilizer generators of the code and then applying \(U^{-1}\).

Using this circuit, the overall circuit depth is \(O(l \cdot n_l^2) = O(l^{3})\) and uses only \(\CNOT\), \(X\) and \(Z\) basis measurements, and classically controlled Pauli gates. The main classical computation includes decoding of the quantum Hamming code, which is of similar time complexity.
With the exception of the classically controlled Pauli gates, the circuit is completely determined by \(l\) and \(b\).
Since no ancillas are prepared, the width is \(O(M_l)\).
\end{proof}

\section{Magic state distillation with almost-constant spacetime overhead}\label{sec:magic-state-distillation-code}
We will distill the magic state \(\ket{\CCZ}\).
To avoid large overhead, we will require a very efficient distillation scheme.
The majority of the section will be devoted to establishing the following theorem about punctured quantum Reed-Solomon codes over extension fields of \(\F\).

\begin{restatable}{theorem}{thmPqrsGoodCode}\label{thm:pqrs-good-code}
    For each \(q = 2^l\) with \(l\ge 3\), there exists a CSS qudit code (punctured quantum Reed-Solomon code) with parameters
    \begin{align}
        \left[\left[\frac{3q}{4}, \frac{q}{4}, \left\lfloor \frac{q}{3} \right\rfloor - \frac{q}{4} + 1\right]\right]_q,
    \end{align}
    such that \(\CCZ^{\otimes 3q/4}_{q}\) acts as logical \(\overline{\CCZ_q}^{\otimes q/4}\) on a basis of logical qudits.
\end{restatable}

The construction is similar to that of Krishna and Tillich \cite{krishna2019towards} with the main difference being the use of extension fields.
To simplify notation, we will specialize to characteristic 2 though this fact is used in a non-essential way.
In~\Cref{sec:distill-ccz-cczq} we will employ the qudit PQRS code to construct a distillation protocol for the qubit CCZ state with almost-constant spacetime overhead, proving~\Cref{lemma:magic-distillation-procedure}.

\subsection{Polynomials over finite fields}
We begin by making some standard definitions and recalling some standard results about polynomials over finite fields.
Fix some prime power \(q \in \mathbb{N}\).

Notate the set of polynomials in the variable \(x\) with coefficients in \(\Fq\) by \(\Fq[x]\).
For \(k \in \mathbb{N}\), define \(\Fq[X]_{<k} = \{P \in \Fq[x] \mid \deg P < k\}\) with degree of the zero polynomial defined to be \(-\infty\). We use \(\Fq^*\) to denote the set of non-zero elements of \(\Fq\).

For a classical code \(C\), the star product is given by coordinate-wise multiplication.
I.e. for \(a=(a_1,a_2,\dots, a_n) \in C\), \(b=(b_1,b_2,\dots, b_n) \in C\), 
\begin{align*}
    a\ast b := (a_1b_1, a_2 b_2, \dots a_nb_n).
\end{align*}
We use \(C^{\ast 2}\) to denote the set \(\{a\ast b \mid a,b \in C\}\).

For some \(k\in\mathbb{N}\), fix a set of coordinates \(A=\{\alpha_i\}_{i=1}^k \subseteq \Fq\) that will later be used for a systematic encoding for a Reed-Solomon code.
Recall that a Reed-Solomon code is evaluations of the set \(\Fq[X]_{<k}\) on \(\Fq\).
In order to obtain a systematic encoding, we will parameterize \(\Fq[x]_{<k}\) in terms of the evaluations on points of \(A\).
\begin{definition}\label{defn:interpolate-poly}
    The Lagrange interpolation polynomials for the set \(A\) are
    \begin{align*}
        \ell_{i}^{(A)}(x) = \prod_{j \in {[k]\setminus \{i\}}} \frac{x-\alpha_j}{\alpha_i - \alpha_j}.
    \end{align*}
\end{definition}
One can see that the interpolation polynomial \(\ell_i^{(A)}(x)\) is zero on \(A \setminus \{\alpha_i\}\) and one on \(\alpha_1\).
\begin{definition}[Message polynomial]\label{defn:message-poly}
    For a message \((m_1, m_2, \dots m_k) \in \Fq^k\), the message polynomial \(p_m^{(A)}(x)\) corresponding to \(m\) is
    \begin{align*}
        p_m^{(A)}(x) = \sum_{i\in [k]} m_i \ell_{i}^{(A)}(x).
    \end{align*}
\end{definition}
A message polynomial uniquely specifies a codeword of the code and also gives a handle on coordinate-wise multiplication of codewords.
It follows from \cref{defn:interpolate-poly} that \(p_m^{(A)}(\alpha_i) = m_i\).
Our Reed-Solomon codes are evaluations of these message polynomials, so we will find it useful to convert between codewords and the corresponding message polynomials.
\begin{definition}[Evaluation map]
    For an evaluation set \(M = \{\alpha_i\}_{i=1}^n \subseteq \Fq\), we define the evaluation map \(\phi_M \colon \Fq[x] \to \Fq^n\)
    \begin{align*}
        P \mapsto (P(\alpha_1), P(\alpha_2), \dots, P(\alpha_n)).
    \end{align*}
    We will occasionally write \(\phi^{-1}_M\) to denote the map that returns the unique lowest degree polynomial with the given evaluation set.
    I.e. for any \(c \in \Fq^n\), \(P = \phi^{-1}_M(c)\) satisfies \(\deg P \le n-1\) and \(\phi_M(P) = c\).
\end{definition}

Our main tool to get a handle on products of polynomials will be \cref{prop:low-deg-sum}, which states that the average over all field elements of low-degree polynomials is zero.
\begin{prop}\label{prop:element_power}
    For any non-zero element of \(a \in \Fq\),
    \begin{align*}
        a^{q-1} = 1.
    \end{align*}
\end{prop}
\begin{proof}
    The product over all non-zero field elements is non-zero and invariant under shifts by \(a\), so it must be the case that \(a^{q-1} = 1\).
    \begin{align*}
        \prod_{x \in \Fq^*} x = \prod_{x \in \Fq^*} ax = a^{q-1} \prod_{x \in \Fq^*} x \in \Fq^*. & \qedhere
    \end{align*}
    
\end{proof}

\begin{prop}\label{prop:low-deg-sum}
    For a polynomial \(p \in \Fq[x]\) with leading coefficient \(\beta\)
    \begin{align*}
        \sum_{x \in \Fq} p(x) = \begin{cases}
            0 & 0 \leq \deg p < q-1 \\
            -\beta & \deg p = q-1
        \end{cases}.
    \end{align*}
\end{prop}
\begin{proof}
    Let \(\alpha \in \Fq\) be a primitive element for \(\Fq\) (which always exists). %
    Consider a monomial of degree \(k\) summed over the field \(\Fq\).
    This sum is invariant under shifts by \(\alpha\),
    \begin{align*}
        \sum_{x \in \Fq} x^k = \sum_{x \in \Fq} (\alpha x)^k = \alpha^k \sum_{x \in \Fq} x^k.
    \end{align*}
    When \(k=0\), the summand is constant and the number of terms in the sum is a multiple of the characteristic.
    For \(k \in (0, q-1)\), \(\alpha^k\ne 1\), so it must be the case that the sum is zero.

    When \(k=q-1\), \cref{prop:element_power} gives that there are \(q-1\) non-zero terms in the sum, each equal to \(1\).
    The sum is then equal to \(q-1\), the additive inverse of \(1\).
\end{proof}

\subsection{Extension fields and qudits}
Later, in order to perform distillation, we will need to represent our qudits over the alphabet \(\FF_{p^l}\) in terms of qudits over \(\FF_{p}\).
For the remainder of the section fix \(q = p^l\) for some prime \(p\).
To do this, we will need to introduce a basis.
This will not enter in any essential way until \cref{sec:distill-ccz-cczq}.
\subsubsection{Bases}
For an extension field \(\FF_{p^l}\) of degree-\(l\) over \(\FF_p\), a basis for \(\FF_{p^{l}}\) is a set \(\{\alpha_i\}_{i \in [l]}\) such that every element \(x \in \FF_{p^{l}}\) has a unique decomposition \(x_1\alpha_1 + x_2 \alpha_2 + \dots + x_l\alpha_l\) with \(\{x_i\}_{i \in [l]}\subseteq \FF_p\).
We will also make use of the trace map \(\tr \colon \FF_{p^{l}}\to \FF_{p}\) given by %
\begin{align}
  \tr(x) = \sum_{i=1}^l x^{p^{i-1}}.
\end{align}
The trace map is \(\FF_p\)-linear.

In this section, we will use two different bases which simplifies different operations.
\begin{fact}[Polynomial basis \cite{mullen2013handbook}]
  Let \(\alpha\) be a root of a degree \(l\) polynomial irreducible over \(\FF_p\).
  Then \(\{1, \alpha, \alpha^2,\dots, \alpha^l\}\) forms a basis, known as a polynomial basis for \(\FF_{p^{l}}\) over \(\FF_{p}\).
\end{fact}
This fact arises from the standard construction of \(\FF_{p^l}\) as the quotient of \(\FF_p[X]\) by an irreducible polynomial.
It will lend itself to easy multiplication of elements of \(\FF_{p^l}\) using operations in \(\FF_p\).

The second basis we will use is known as a self-dual basis.  
\begin{fact}[Existence of self-dual bases \cite{seroussi1980factorization}]
  Let \(q = p^l\) for some prime \(p\).
  For \(q\) even or \(q\) and \(l\) odd, there exists an self-dual basis \(\{\sigma_i\}_{i \in [l]} \subset \FF_q\) for \(\FF_q\) over \(\FF_p\)
  This basis satisfies
  \begin{align}
      \tr \left(\alpha_i \alpha_j\right) = \begin{cases} 1 & i=j \\ 0 & i \ne j \end{cases}.
  \end{align}
\end{fact}

For an element \(x \in \FF_{p^l}\) and a basis \(N=\{\alpha_i\}_{i\in [l]}\) of \(\FF_{p^{l}}\) over \(\FF_p\), we use \(x_i^{(N)} \in \FF_p\), \(i \in [l]\) to denote the expansion of \(x\) in the basis \(N\) i.e. \(x = \sum_{i\in [l]} \alpha_i x_i^{(N)}\).
When clear from context, we will drop the superscript \(N\) to avoid cluttering the notation.

\subsubsection{\(\Fq\)-Qudits and CSS codes over \(\Fq\)}
We will define CSS codes over qudits of dimension \(q\). %
Let \(\omega\) be a \(p\)-th roof of unity.
For an element \(a \in \FF_p\), we will use the abuse of notation \(\omega^a\) to indicate exponentiation after taking the canonical inclusion \(\FF_p \to \mathbb{Z}\).

For \(a,b,x \in \Fq\), define the generalized Pauli operators (where addition is taken in \(\Fq\))
\begin{align*}
    X^{(q)}(a) \ket{x} = \ket{x+a}, && Z^{(q)}(b) \ket{x} = \omega^{\tr(bx)} \ket{x}.
\end{align*}
So that
\begin{align*}
    \omega^{\tr(ab)} X^{(q)}(a)Z^{(q)}(b) = Z^{(q)}(b)X^{(q)}(a).
\end{align*}

For \(x,y,z \in \Fq\), define
\begin{align*}
  \CNOT^{(q)}\ket{x}\ket{y} = \ket{x}\ket{x+y},
\end{align*}
\begin{align*}
    \CZ^{(q)}\ket{x}\ket{y} = \omega^{\tr(xy)}\ket{x}\ket{y}, & & \CCZ^{(q)}\ket{x}\ket{y}\ket{z} = \omega^{\tr(xyz)}\ket{x}\ket{y}\ket{z}.
\end{align*}

For linear codes \(C_1,C_2\subseteq \Fq\) such that \(C_2^\perp \subseteq C_1\), the CSS code \(\mathrm{CSS}(C_1,C_2)\) is the span of the states
\begin{align}
    c \in C_1& &\ket{c+C_2^\perp} = \frac{1}{\sqrt{|C_2^\perp|}}\sum_{\alpha \in C_2^\perp} \ket{c+\alpha}.
\end{align}

\subsubsection{\(\Fq\)-Qudits using \(\FF_p\)-qudits}
We would like to construct our qudits out of smaller qudits over the base field.
We will use \(\ket{\psi}_q\) to denote a qudit state with local dimension \(q\) (when different from 2).
For a basis \(N\), we will use the following representation of the qudit state
\begin{align}
    \ket{x}_q^{(N)} = \ket{x_1^{(N)}} \ket{x_2^{(N)}} \dots \ket{x_l^{(N)}}.
\end{align}
When \(N\) is self-dual, we have that \(\tr(xy) = \sum_{i\in[l]} x^{(N)}_i y^{(N)}_i\).
This simplifies certain Clifford gates.
In particular,
\begin{align}
  X^{(q)}(a) &= X^{(p)}(a_1) \otimes X^{(p)}(a_2) \otimes \dots \otimes X^{(p)}(a_m), \\
  Z^{(q)}(a) &= Z^{(p)}(a_1) \otimes Z^{(p)}(a_2) \otimes \dots \otimes Z^{(p)}(a_m), \\
  \CZ^{(q)}(a) &= \CZ^{(p)} \otimes \CZ^{(p)} \otimes \dots \otimes \CZ^{(p)}, \\
  \CNOT^{(q)}(a) &= \CNOT^{(p)} \otimes \CNOT^{(p)} \otimes \dots \otimes \CNOT^{(p)}.
\end{align}
We will leave the construction of \(\CCZ^{(q)}\) using \(\CCZ^{(p)}\) to a later section.
For the case $q=2$, we omit the subscript/superscript.

\subsection{Punctured quantum Reed-Solomon code}
Fix \(q = 2^l\) and parameters \(k,m,s \in \mathbb{N}\) such that \(q/3 \ge k \ge m \ge s > 0\) and \(2k \le q-m\).
Fix an set of coordinates \(A \subseteq \Fq\) of size \(k\).
We further will introduce a new set of coordinates \(B \subseteq A\) such that \(|B| = s\).
Fix an evaluation set \(M = \{\alpha_i\}_{i=1}^{q-s} = \Fq \setminus B\).
\begin{definition}[Punctured quantum Reed-Solomon Code]\label{defn:pqrs}
    We begin by defining two sets of polynomials
    \begin{align}
        \mathcal{P}_1 &= \Fq[x]_{<k}, \\
        \mathcal{P}_2^\perp &= \{P \in \Fq[x]_{<m} \mid P(B) = 0\}.
    \end{align}
    Define \(C_1 = \phi_M(\mathcal{P}_1)\) to be the evaluations of \(\mathcal{P}_1\) on \(M\) and \(C_2^\perp= \phi_M(\mathcal{P}_2^\perp)\) to be the evaluations of \(\mathcal{P}_2^\perp\) on \(M\).
    The punctured quantum Reed-Solomon code is defined to be \(\mathcal{C} = \mathrm{CSS}(C_1, C_2)\).

    We will pick a basis parameterized by \(m=(m_1, m_2, \dots, m_{|A|}) \in \Fq^{|A|}\) defined in the following way:
    For any \(m\), the corresponding codeword \(c_m\) is \(\phi_M(p_m^{(A)})\), the evaluation of the corresponding message polynomial with systematic encoding positions \(A\).
    The code state \(\ket{\overline{m}}\) is then defined to be
    \begin{align}\label{eq:pqrs-codestates}
        \ket{\overline{m}} = \frac{1}{\sqrt{|C_2^\perp|}}\sum_{\alpha \in C_2^\perp} \ket{c_m + \alpha} = \ket{c_m + C^\perp_2}.
    \end{align}
\end{definition}

Before proving properties of this code, it is helpful to first compute \(C_2\):
\begin{lemma}\label{lemma:compute-C2}
    \(C_2\) is the evaluations of \(\mathcal{P}_2 = \Fq[x]_{<q-m}\) on \(M\).
\end{lemma}
\begin{proof}
    We first write the orthogonal complement of \(C_2^\perp\) as the evaluations of the set of polynomials:
    \begin{align}
        \mathcal{P}_2' = \{P \in \Fq[x]_{<q} \mid \forall Q \in \mathcal{P}_2^\perp~\sum_{\alpha \in M} P(\alpha) Q(\alpha) = 0\}.
    \end{align}
    Clearly, \(C_2 = \phi_M(\mathcal{P}_2')\).
    We will show that this set is \(\Fq[x]_{<q-m}\).
    First note that for any non-zero \(Q \in \mathcal{P}_2^\perp\), it is divisible by the polynomial \(R = \prod_{b \in B} (x-b)\), so there is a unique polynomial \(Q' \in \Fq[x]\) such that \(Q = Q' R\) and \(\deg Q' = \deg Q - s\).
    \(R\) is zero on \(B\), so we may extend the sum to write \(\mathcal{P}_2'\) as
    \begin{align}
        \mathcal{P}_2' &= \{P \in \Fq[x]_{<q} \mid \forall Q' \in \Fq[x]_{<m-s} ~\sum_{\alpha \in \Fq} P(\alpha) Q'(\alpha) R(\alpha)= 0\}.
    \end{align}
    In the condition, if \(\deg P \ge q-m\), then there exists some \(Q' \in \Fq[x]_{<m-s}\) such that \(\deg P + \deg Q' + s = q-1\), making the sum nonzero by \cref{prop:low-deg-sum}.
    Otherwise, the sum is zero, so 
    \begin{align}
        \mathcal{P}_2' = \Fq[x]_{<q-m} \equiv \mathcal{P}_2.
    \end{align}
\end{proof}

\begin{prop}\label{prop:pqrs-params}
    \(\mathcal{C}\) is a quantum CSS code with parameters 
    \begin{align*}
        [[q-s, k-m+s, \min(q-k, m)-s+1]]_q.
    \end{align*}
\end{prop}
\begin{proof}
    We have the degree bounds \(k \le m\) so \(\mathcal{P}_2^\perp \subseteq \mathcal{P}_1\) implying \(C_2^\perp \subseteq C_1\). Thus \(\mathcal{C}\) is a valid quantum CSS code.
    By dimension counting, \(\mathcal{C}\) encodes \(k - m - s\) qubits.

    For the distance, recall that the distance of \(\mathrm{CSS}(C_1, C_2)\) is lower bounded by the minimum of the distances of \(C_1\) and \(C_2\).
    \(C_1\) is the puncturing of \(\text{RS}(q,k)\) on \(s\) coordinates and so has distance at least \((q - k + 1) - s\).
    Likewise, by \cref{lemma:compute-C2}, \(C_2\) is the puncturing of \(\mathrm{RS}(q,q-m)\) and has distance at least \((m+1)-s\).
\end{proof}

\subsubsection{\(\CCZ\) for punctured quantum Reed-Solomon codes}
\(\CCZ\) requires us to analyze the coordinate-wise product of codewords.
In particular, we will encounter products of the form \(\langle C_1\ast C_1, C_2^{\perp}\rangle\) that may cause the action of \(\CCZ^{\otimes n}\) to be non-constant on cosets of \(C_2^{\perp}\).
A sufficient condition for these products to vanish is given by the following lemma.
\begin{lemma}\label{lemma:C12-in-C2}
    \(C_1^{\ast 2} \subseteq C_2\).
\end{lemma}
\begin{proof}
    Fix two codewords \(c_1, c_2 \in C_1\), and let \(p_1, p_2 \in \mathcal{P}_1\) be the corresponding polynomials such that \(p_1 = \phi_M^{-1}(c_1)\) and \(p_2 = \phi_M^{-1}(c_2)\).
    Then, \(c_1\ast c_2=\phi_M(p_1 p_2)\) is the evaluation of \(p_1 p_2\) on \(M\), and \(\deg (p_1 p_2) \le \deg p_1 + \deg p_2 < 2k\).
    By assumption on the parameters, \(2k \le q-m\).
    It follows that \(p_1 p_2 \in \mathcal{P}_2\), so \(c_1\ast c_2\) is an element of \(C_2\) by \cref{lemma:compute-C2}.
\end{proof}

Before proceeding to multiplication of codewords of \(C_1\), it is convenient to first start with multiplication of the interpolation polynomials.
\begin{lemma}\label{prop:interpolation-poly-ortho}
    For three interpolation polynomials
    \begin{align}
        \sum_{x \in M} \ell_{a}^{(A)}(x) \ell_{b}^{(A)}(x) \ell_{c}^{(A)}(x) =
        \begin{cases}
          -1 & a = b = c \in B\\
          0 & \text{otherwise}
        \end{cases}.
    \end{align}
\end{lemma}
\begin{proof}
    By construction, the product has degree strictly less than \(q-1\), so it follows from \cref{prop:low-deg-sum} that
    \begin{align}
        \sum_{x \in \Fq} \ell_{a}^{(A)}(x) \ell_{b}^{(A)}(x) \ell_{c}^{(A)}(x) = 0~.
    \end{align}
    This allows us to break up the sum into the evaluated coordinates and the removed coordinates for which we have good control over the values of the interpolation polynomials:
    \begin{align}
        \sum_{x \in M} \ell_{a}^{(A)}(x) \ell_{b}^{(A)}(x) \ell_{c}^{(A)}(x) = -\sum_{x \in B} \ell_{a}^{(A)}(x) \ell_{b}^{(A)}(x) \ell_{c}^{(A)}(x).
    \end{align}
    By construction of the interpolation polynomials, if \(a=b=c\in B\), the summand vanishes on all elements of \(B\) except for one element where it takes the value \(1\).
    In any other case, the product is identically zero on \(B\).
\end{proof}

This straightforwardly allows us to multiply codewords.
\begin{cor}\label{lemma:prod-codewords}
    Let \(p_{m_1}^{(A)}(x)\), \(p_{m_2}^{(A)}(x)\), and \(p_{m_3}^{(A)}(x)\) be three message polynomials with the corresponding codewords of \(C_1\) denoted by \(c_1\), \(c_2\), and \(c_3\).
    Then
    \begin{align}
        \sum_{i \in M} (c_1)_i (c_2)_i (c_3)_i = -\sum_{i \in B} (m_1)_i (m_2)_i (m_3)_i.
    \end{align}
\end{cor}
\begin{proof}
    We write
    \begin{align}
        \sum_{i \in M} (c_1)_i (c_2)_i (c_3)_i = \sum_{x \in M} p_{m_1}^{(A)}(x) p_{m_2}^{(A)}(x) p_{m_3}^{(A)}(x),
    \end{align}
    and then use \cref{defn:message-poly} to expand the message polynomials in terms of the interpolation polynomials.
    The result follows after application of \cref{prop:interpolation-poly-ortho}.
\end{proof}

\begin{prop}\label{prop:logical-ccz}
    For three codestates \(\ket{\overline{m}_1}\), \(\ket{\overline{m}_2}\), and \(\ket{\overline{m}_3}\),
    \begin{align*}
        \CCZ^{\otimes n} \ket{\overline{m}_1}\ket{\overline{m}_2}\ket{\overline{m}_3} &= (-1)^{\tr \langle m_1^{(B)} \ast m_2^{(B)}, m_3^{(B)}\rangle} \ket{\overline{m}_1}\ket{\overline{m}_2}\ket{\overline{m}_3}\\
        &\equiv \overline{\CCZ}^{(B)} \ket{\overline{m}_1}\ket{\overline{m}_2}\ket{\overline{m}_3}.
    \end{align*}
    Where, \(i \in {1,2,3}\), \(m_i^{(B)}\) refers to the restriction of \(m_i\) to the coordinates in \(B\), and \(\overline{\CCZ}^{(B)}\) refers to a logical \(\CCZ\) acting coordinate-wise on the logical qubits labeled by coordinates of \(B\).
\end{prop}
\begin{proof}
    First we expand the definition of codestates (\cref{eq:pqrs-codestates}) into a sum over elements of \(C_2^\perp\). Writing \(c_i\) for the codeword of \(C_1\) corresponding to the evaluations of \(p_{m_i}(x)\), we have
    \begin{align}
        \CCZ^{\otimes n} \ket{\overline{m}_1}\ket{\overline{m}_2}\ket{\overline{m}_3} &\propto \sum_{\alpha,\beta,\gamma \in C_2^\perp} (-1)^{\tr\langle (c_1+\alpha)\ast (c_2 + \beta), (c_3+\gamma)\rangle}\ket{c_1 + \alpha} \ket{c_2 + \beta} \ket{c_3 + \gamma}.
    \end{align}
    After distributing star product, using linearity of the inner product, and \cref{lemma:C12-in-C2}, we find that the phase does not depend on \(\alpha\), \(\beta\), or \(\gamma\); it is constant on cosets of \(C_2^\perp\). We can then invoke \cref{lemma:prod-codewords}.
    \begin{align}
        \tr\langle (c_1+\alpha)\ast (c_2 + \beta), (c_3+\gamma)\rangle &= \tr\langle c_1\ast c_2, c_3\rangle\\
                                                                      &= -\tr \langle m_1^{(B)} \ast m_2^{(B)}, m_3^{(B)}\rangle.
    \end{align}
    The result follows after using the field characteristic.
\end{proof}

\subsubsection{Proof of \cref{thm:pqrs-good-code}}\label{sec:pqrs-code-thm}
We are now ready to assembly the previous results.
\thmPqrsGoodCode*
\begin{proof}
    The code is as defined in  \cref{defn:pqrs} with 
    \begin{align*}
        k &= m = \lfloor q/3 \rfloor, \\
        s &= q/4.
    \end{align*}
    These choices satisfy \(q/3 \ge k \ge m \ge s > 0\) and \(2k \le q-m\) when \(q = 2^l\) with \(l \ge 2\).
    The code parameters are proven in \cref{prop:pqrs-params} and the logical CCZ is proven in \cref{prop:logical-ccz}.
\end{proof}

\subsection{Distillation of qubit \(\ket{\CCZ}\)}\label{sec:distill-ccz-cczq}
Having established \cref{thm:pqrs-good-code}, we will use this to construct extremely efficient magic state distillation for qubits.
The magic state we would like to distill is \(\ket{\CCZ} = \CCZ\ket{+}\ket{+}\ket{+}\).
As in the previous section, let \(q = 2^l\).
Fix a self-dual basis \(N = \{\alpha_i\}_{i \in [l]}\) for \(\Fq\) over \(\F\).

\subsubsection{Qubits as \(\Fq\)-qudits}

Our codes have only transversal \(\CCZ^{(q)}\), so we will need a means to implement the operation using \(\ket{\CCZ}\) and to convert the final output states to \(\ket{\CCZ}\).

Define the states
\begin{align}
  \ket{+}_q = \frac{1}{\sqrt{q}}\sum_{x \in \Fq} \ket{x}_q, & & \ket{\CCZ^{(q)}}_q = \CCZ^{(q)}\ket{+}_q\ket{+}_q\ket{+}_q.
\end{align}

The conversion of \(\ket{\CCZ^{(q)}}_q\) to \(\ket{\CCZ}\) is straightforward in the appropriate basis.
We first begin by constructing a circuit to change the basis of the field extension.
This will give us greater flexibility when implementing operations later.
\begin{prop}[Change of basis]\label{lemma:F2n-basis-change}
  For two bases \(N = \{\alpha_i\}_{i \in [l]}\) and \(M = \{\beta_{i}\}_{i \in [l]}\) for \(\Fq\) over \(\F\), there is a unitary that maps
  \begin{align}
    \ket{x}^{(N)}  \mapsto \ket{x}^{(M)}
  \end{align}
  for all \(x \in \Fq\) using depth \(O(l)\), \(O(l^2)\) \(\CNOT\) gates, and \(l\) input ancilla qubits initialized to \(\ket{0}\).
\end{prop}
\begin{proof}
  \(N\) and \(M\) are bases for \(\Fq\) over \(\F\), so there exists an invertible matrix \(A \in \F^{l \times l}\) such that for \(x \in \Fq\),
  \begin{align*}
    x^{(M)}_j = \sum_{i \in [l]} A_{ij} x_i^{(N)}.
  \end{align*}
  The initial state is \(\ket{x}^{(N)}\ket{0}^{\otimes l}\).
  For every non-zero entry \((i,j) \in [l]\times [l]\) of \(A\) perform a \(\CNOT\) controlled on the \(j\)-th qubit of the first register and targeted on the \(i\)-th qubit of the second register.
  This maps \(\ket{x}\ket{0}^{\otimes l} \mapsto \ket{x}\ket{y}\).
  Then, for every non-zero entry \((i,j) \in [l]\times [l]\) of \(A^{-1}\) perform a \(\CNOT\) controlled on the \(j\)-th qubit of the second register and targeted on the \(i\)-th qubit of the first register.
  This maps \(\ket{x}\ket{y} \mapsto \ket{x + A^{-1} y} \ket{y} = \ket{0}^{\otimes l}\ket{y}\).
\end{proof}

\begin{lemma}\label{lemma:ccz-naive-reduction}
  Let \(M = \{\alpha_i\}_{i\in[l]}\) be a basis for \(\Fq\) over \(\F\) such that \(\alpha_1 = 1\).
  Then, there exists a circuit to convert \(\ket{\CCZ^{(q)}}_q^{(M)}\) to \(\ket{\CCZ}\) using \(3(l-1)\) measurements and \(3\) classically controlled \(\CZ\) and \(Z\) gates.
\end{lemma}
\begin{proof} Since $\{\alpha_i\}_{i \in [l]}$ form a basis, there exists a $j$ such that $\tr(\alpha_j)=1$ (if $l$ is odd then $\alpha_1=1$ works).
  Using \(\F\)-linearity of the trace, for \(x,y,z \in \Fq\), we can separate the product into 
  \begin{align}
    \tr(xyz) = x_1y_1z_j + \tr(\dots),
  \end{align}
  where the omitted terms on the right are at most quadratic in \(x_1\),\(y_1\), and \(z_j\).
  Then, starting with the state
  \begin{align}
    \ket{\CCZ^{(q)}}^{(M)}_q = \sum_{x,y,z \in \Fq} (-1)^{\tr(xyz)}\ket{x}^{(M)}\ket{y}^{(M)}\ket{z}^{(M)},
  \end{align}
  the computational basis measurement of the last \(l-1\) qubits in each of the first two registers and of the qubits $[l]\backslash\{j\}$ in the third register gives the measurement outcomes \(\{x_i\}_{i=2}^l\), \(\{y_i\}_{i=2}^l\), and \(\{z_i\}_{i \neq j}\).
  The post measurement state is
  \begin{align}
      \frac{1}{2^{3/2}} \sum_{x_1,y_1,z_j \in \F} (-1)^{x_1 y_1 z_j + p(x_1,y_1,z_j)} \ket{x_1}\ket{y_1}\ket{z_j},
  \end{align}
  where \(p(x_1,y_1,z_j)\) is a polynomial over \(\F\) that is at most quadratic in  \(x_1\),\(y_1\), and \(z_j\).
  The coefficients can be classically computed from the measurement outcomes, so the phases can be fixed by application of \(\CZ\) and \(Z\) such that the final state is \(\ket{\CCZ}\).
\end{proof}

We can also implement \(\CCZ^{(q)}\) using qubit $\CCZ$ gates.
\begin{lemma}\label{lemma:naive-ccz}
    \(\CCZ^{(q)}\) has an implementation using at most \(l^3\) \(\CCZ\) gates.
\end{lemma}
\begin{proof}
    For \(x,y,z \in \Fq\) consider 
    \begin{align}
        \CCZ^{(q)} \ket{x}^{(N)}\ket{y}^{(N)}\ket{z}^{(N)} = (-1)^{\tr(xyz)}\ket{x}^{(N)}\ket{y}^{(N)}\ket{z}^{(N)}.
    \end{align}
    \(\tr(xyz)\) is a degree-3 polynomial in the coefficients of the basis expansion.
    Let \(N = \{\alpha_i\}_{i \in [l]}\).
    Then we may expand \(\tr(xyz)\) as 
    \begin{align}
        \tr(xyz) = \sum_{i,j,k \in [l]} x^{(N)}_i y^{(N)}_j z^{(N)}_k \tr(\alpha_i \alpha_j \alpha_k).
    \end{align}
    For every triple \((i,j,k)\) for which \(\tr(\alpha_i \alpha_j \alpha_k)\) is non-zero, we perform \(\CCZ\) between the \(i\)-th qubit of the \(x\) register, the \(j\)-th qubit of the \(y\) register, and the \(k\)-th qubits of the \(z\) register.
\end{proof}

\subsubsection{Single-step state distillation}
We begin by giving the construction for a single round of magic state distillation.
Let \(N\) be a self-dual basis, and \(M\) be a polynomial basis.
Fix a field size \(q=2^l\) and a quantum code \(\mathcal{C}\) with parameters \([[n,k,d]]_q\) from \cref{thm:pqrs-good-code}.
Let \(\mathcal{E}\) be the encoding map for \(\mathcal{C}\).

\begin{figure}[H]
\centering
\begin{algorithm}[H]
  \DontPrintSemicolon
  \LinesNumbered
  \caption{Magic  state distillation algorithm\label{alg:magic-state-distillation}}
  \KwIn{\(n\cdot l^3\) \(\ket{\CCZ}\), \(3n\cdot l\) \(\ket{+}\), and \(3n\cdot l\) \(\ket{0}\)}
  \KwData{\(3n\) qubit register \(A\)}
  \KwOut{\(k\) \(\ket{\CCZ}\)}
  \Begin{
    Initialize register \(A\) with input state \(A \gets \left(\ket{+}^{(N)}_q\right)^{\otimes 3n} = \ket{+}^{\otimes 3nl}\)\;
    \(\sigma_1 \gets\) measurements of checks (sequentially) of the \(3\) blocks of register \(A\)\;
    \(U_{corr,1} \gets \) operator \(U_{corr,1}\) such that \(U_{corr,1}\) corrects \(\sigma_1\)\;
    Apply \(U_{corr,1}\) to register \(A\)\;
    Consume \(\ket{\CCZ}^{\otimes n\cdot l^3}\) to apply \(\left(\CCZ^{(q)}\right)^{\otimes n}\) to register \(A\)\;
    \(\sigma_2 \gets\) measurements of checks of \(C\) on the \(3\) blocks of register \(A\)\;
    \(U_{corr,2} \gets \) operator \(U_{corr,2}\) such that \(U_{corr,2}\) corrects \(\sigma_2\)\;
    Apply \(U_{corr,2}\) to register \(A\)\;
    Apply \(\mathcal{E}^{-1}\) to register \(A\)\;
    Use \(3n\cdot l\) \(\ket{0}\) as space to apply change of basis \(\ket{\CCZ}_q^{(N)} \to \ket{\CCZ}_q^{(M)}\) to \(A\)  \;
    Apply \(\ket{\CCZ}_q^{(M)}\) to \(\ket{\CCZ}\) conversion to \(A\)\;
    Output \(A\)\;
  }
\end{algorithm}
\end{figure}

For \(d,n \in \mathbb{N}\), define \(\mathcal{S}^n_d = \{x \subseteq [n] \mid |x| = d\} \subseteq P([n])\) to be all subsets of \([n]\) of size \(d\).
\begin{prop}\label{prop:single-level-magic-state-distill}
  If \(\mathcal{C}\) is a \(t\)-error correcting code, then the output of~\Cref{alg:magic-state-distillation} is \(\ket{\CCZ}^{\otimes k}\) if the input is \(\mathcal{S}^{n}_{t+1}\bullet \mathcal{S}^{l^3}_1 \bullet \mathcal{S}^3_1\) deviated from \(\ket{CCZ}^{\otimes nl^3}\).
\end{prop}
\begin{proof}
  In the self-dual basis \(N\), qudit \(\CNOT\) can be applied using transversal qubit \(\CNOT\). 
  The procedure first prepares the logical code state \(\mathcal{E}\left(\left(\ket{+}_q^{(N)}\right)^{\otimes k}\right)^{\otimes 3}\).
  For each of the \(n\) positions, we consume \(l^3\) \(\ket{CCZ}\) states to perform \(\CCZ^{(q)}\) using \cref{lemma:naive-ccz}.
  If any one of these states is faulty, then an error is incurred on the position.
  However, the code is \(t\)-error correcting, so as long as at most \(t\) positions are faulty, the state is successfully corrected to the encoded logical \(\mathcal{E}^{\otimes 3}\left(\left(\ket{\CCZ^{(q)}}_q^{(N)}\right)^{\otimes k}\right)\).
  We apply \((\mathcal{E}^{-1})^{\otimes 3}\) to get \(\left(\ket{\CCZ^{(q)}}^{(N)}_q\right)^{\otimes k}\)
  We then can then use the \(\ket{0}\) ancillas to convert this state to \(\left(\ket{\CCZ^{(q)}}^{(M)}_q\right)^{\otimes k}\) (\cref{lemma:F2n-basis-change}).
  This basis allows us to use \cref{lemma:ccz-naive-reduction} to obtain \(\ket{\CCZ}^{\otimes k}\).
\end{proof}

\begin{remark}
  This is a magic state distillation method for qubits for which the magic state distillation exponent
  \begin{align}
    \gamma\equiv \frac{\log(n/k)}{\log d} = \frac{\log(3\log_2^3q)}{\log(q/12)} = O\left(\frac{\log{l}}{l}\right)
  \end{align}
  can be made arbitrarily small.
\end{remark}

\begin{lemma}\label{lemma:concat-threshold}
  Fix two functions \(f,g \colon \mathbb{N}\mapsto \mathbb{R}\) such that and consider the family of functions defined as
  \begin{align*}
    W_i(x) &= (f(i) W_{i-1}(x))^{g(i)},\\
    W_0(x) &= x
  \end{align*}
  if \(f\) and \(g\) satisfy the asymptotic bound
  \begin{align}
    \frac{\log(f(i))}{\prod_{j=1}^{i-1} g(j)} = O_{i \to \infty}(i^{-2})
  \end{align}
  then there exists a constant \(\beta > 0\) such that
  \begin{align}
    W_i(x) \le (\beta x)^{\prod_{j=1}^{i} g(j)}
  \end{align}
  for \(x \in [0,1)\).
\end{lemma}
\begin{proof}
  For \(L \in \mathbb{N}\), define \(h(i)=\prod_{j=i}^L g(j)\), so that \(W_L(x)\) has the closed form
  \begin{align}
    W_L(x) &= \left(\prod_{i=1}^{L} \left(f(i)\right)^{h(i)}\right) x^{h(1)}. \\
  \end{align}
  Define \(\bar{h}(i) = h(1)/h(i) = \prod_{j=1}^{i-1} g(j)\).
  Then, we bound
  \begin{align}
    W_L(x) &= x^{h(1)}\exp\left[h(1)\sum_{i=1}^L \frac{\log(f(i))}{\bar{h}(i)}\right]\\
           &\le x^{h(1)}\exp\left[h(1)\sum_{i=1}^\infty \frac{\log(f(i))}{\bar{h}(i)}\right] \\
           &\le x^{h(1)}\exp\left[h(1) (\mathrm{const.})\right] \\
           &= \left(x \cdot (\mathrm{const.})\right)^{h(L)},
  \end{align}
  where the last line follows for a constant independent of \(L\) due to the assumption \(\frac{\log(f(i))}{\prod_{j=1}^{i-1} g(j)} = O_{i \to \infty}(i^{-2})\).
\end{proof}

\subsubsection{Proof of \cref{lemma:magic-distillation-procedure}}\label{lemma:magic-distillation-procedure-proof}
\lemmaMagicDistillationProcedure*
\begin{proof}
  The proof proceeds nearly identically to the case of distillation of stabilizer states. 
  For each extension field \(\Fq\), we pick a self-dual basis over \(\F\).
  We will pick a sequence of codes from \cref{thm:pqrs-good-code} with \(q_i = 64\cdot2^{\lceil \log_2 i\rceil}\), so that \(64 i\le q_i \le 128i\).
  The \(i\)-th code from the sequence has \(n_i=3q_i/4 \le 96 i\), \(k_i\ge 16i\), and is a \(t_i = \lfloor\frac{d_i-1}{2}\rfloor \ge i\)-error correcting code.
  We parallel the argument from \cref{lemma:stab-distillation-procedure}:
  Let \(\tilde{n}_i = n_i (\log_2(q_i))^3 = O(i (\log i)^3)\).

  Let \(M_l := \prod_{\ell=1}^l \tilde{n}_\ell \), and \(K_l := \prod_{\ell=1}^l k_\ell\).

  We start with a set of \(\ket{\CCZ}\) states indexed by \([\tilde{n}_l]\times \dots \times [\tilde{n}_1] \simeq [M_l]\).
  Iterating from \(\ell=1\) to \(\ell=l\), for each \(I \in [\tilde{n}_l] \times \dots \times [\tilde{n}_{\ell+1}]\) and \(J \in [k_{\ell-1}]\times\dots\times [k_1]\) we apply the state distillation procedure \cref{alg:magic-state-distillation} on the set of states \(\{I\} \times [\tilde{n}_{\ell}] \times \{J\}\) to get a new set of states with labels \(\{I\} \times [k_{\ell}] \times \{J\}\) (Consuming input ancilla states and applying \(\CZ\)).
  Let \(\mathcal{S}_{d}^{n}\) be the family of all subsets of \([n]\) of size \(d\).
  At step \(\ell\), \cref{prop:single-level-magic-state-distill} gives that the output is \(\ket{\CCZ}^{\otimes k_{\ell}}\) if the input is \(\mathcal{S}^{\tilde{n}_\ell}_2\bullet \mathcal{S}^{3}_1\)-deviated from \(\ket{\CCZ}^{\otimes \tilde{n}_{\ell}}\).
  We use the fact that, by construction, no two outputs from a single application of the state distillation procedure are used together in a following state distillation procedure: If the output of a state distillation in step \(\ell\) is not \(\ket{\CCZ}^{\otimes k_{\ell}}\), then the input is not \(\mathcal{S}^{\tilde{n}_{\ell}}_{t_{\ell}+1}\)-deviated from \(\ket{\CCZ}^{\otimes \tilde{n}_{\ell}}\).
  Since each input is from separate state distillation procedures, \(t+1\) distinct state distillation procedures at step \(\ell-1\) must have had inputs that are not \(\mathcal{S}^{\tilde{n}_{\ell-1}}_{t_{\ell-1}+1}\)-deviated from \(\ket{\CCZ}^{\otimes \tilde{n}_{\ell-1}}\).

  Inducting from \(\ell=1\), the output of the final iteration is \(\ket{\CCZ}^{\otimes K_l}\) if the input is \(\mathcal{A}_l:=\mathcal{S}^{\tilde{n}_l}_{t_l+1}\bullet\dots\bullet\mathcal{S}^{\tilde{n}_1}_{t_1+1} \bullet \mathcal{S}^3_1\)-deviated from \(\ket{\CCZ}^{\otimes M_l}\).

  Let \(S_i(x) = \left(96\cdot 7^3 i \log(i)^3 x\right)^{i+1}\) which satisfies \(\weightenum{\mathcal{S}^{\tilde{n}_i}_{t_i+1}}{x} \le \left(\tilde{n}_i x\right)^{t_{i}+1} \le S_i(x)\) when \(x \in [0,1/\tilde{n}_i]\).
  When \(\weightenum{\mathcal{A}_{l-1}}{x} \le 1/\tilde{n}_l\), we can bound \(\weightenum{\mathcal{A}_l}{x}\) as 
  \begin{align}
    \weightenum{\mathcal{A}_l}{x} &\le S_l\circ S_{l-1}\circ \dots \circ S_2 \circ S_1 (3x).
  \end{align}
  This recursion satisfies the preconditions of \cref{lemma:concat-threshold},\footnote{\cref{lemma:concat-threshold} gives a superexponential error suppression at every step of the recursion, so we do not need to be worried about the polynomially small precondition \(\weightenum{\mathcal{A}_{l-1}}{x} \le 1/\tilde{n}_l\). It suffices to take \(\beta_{\mathrm{Mdistill}}\) to be \(1/2\) the constant promised by \cref{lemma:concat-threshold}.} so there exists a constant \(\beta_{\mathrm{Mdistill}}\) (independent of \(l\)) such that
  \begin{align}
    \weightenum{\mathcal{A}_l}{x} \le (\beta_{\mathrm{Mdistill}} x)^{\prod_{i=1}^l (i+1)} = (\beta_{\mathrm{Mdistill}} x)^{(l+1)!}.
  \end{align}
  
  The three operations with depth \(\omega(1)\) are the application of \(\mathcal{E}\), \(\mathcal{E}^{-1}\), and \(\CCZ^{(q)}\).
  Using a self-dual basis, in step \(i\), the encoding and unencoding unitiaries have depth \(O(q_i^3) = O(i^3)\) while the \(\CCZ^{(q)}\) has depth \(O\left((\log(q_i))^3\right)=O\left((\log i)^3\right)\), so the overall depth is at most \(O(l^4 (\log l)^3)\).
  A number of ancilla qubits \(O(M_l)\) is needed to sequentially measure the syndrome, so the total space required is \(O(M_l)\) qubits.

  Finally, \(M_l\le \prod_{i=1}^l \alpha i (\log i)^3 = O\left(\alpha^l l! \cdot l (\log l)^3\right)\) for some \(\alpha > 0\) and \(K_l\ge \prod_{i=1}^l 16i = 16^l l!\), so for some \(c > 0\),
  \begin{align}
    \frac{K_l}{M_l} \ge \frac{16^ll!}{\alpha^ll!\cdot l (\log l)^3} = \Omega\left(c^l\right).
  \end{align}

      \paragraph{Classical depth.} We now describe a time-efficient decoder for the PQRS code, which implies to the claimed classical depth. Recall that the X checks correspond to the code $C_2^\perp$. According to~\Cref{lemma:compute-C2}, $C_2$ consists of the evaluations of $\mathbb{F}_q[x]_{< q-m}$ on a set $M$, where in the proof of~\Cref{thm:pqrs-good-code} we set $m = \lfloor q/3 \rfloor$ and $M \subset \mathbb{F}_q$ is an evaluation set of size $|M|=3q/4$. 
      Hence, we can apply standard RS decoders. Here we use the Berlekamp-Welch algorithm (see Theorem 12.1.6 in~\cite{guruswami2022}) which decodes $C_2$ up to error weight $\lfloor (3q/4 - (q-m)+1)/2 \rfloor= \lfloor q/12 \rfloor$ in time $O(q^3)$ (if no codeword is found the decoder outputs FAIL). Similarly, the Z checks correspond to the code $C_1^\perp$, where $C_1$ is $\mathbb{F}_q[x]_{< k}$ evaluated on $M$, and $k= \lfloor q/3 \rfloor$ is chosen in~\Cref{thm:pqrs-good-code}'s proof. So X errors can be corrected with the Berlekamp-Welch algorithm up to error weight $\lfloor (3q/4 - k +1)/2 \rfloor= \lfloor 5q/24 \rfloor$.\footnote{Alternatively, we can also use the standard twirling argument in the MSD literature~\cite{bravyi2005universal}, applying the Clifford $\CCZ^{(q)} (X^{(q)}(a_1)X^{(q)}(a_2)X^{(q)}(a_3)) (\CCZ^{(q)})^\dagger$ for randomly chosen $a \in \mathbb{F}_q^{3}$, to restrict our decoding problem to only Z-type errors.}
      It follows that the total classical depth in the $l$-level distillation procedure is $\sum_{i=1}^{l} O(q_i^3)= \sum_{i=1}^l i^3 = O(l^4).$
\end{proof}

\section{Quantum LTC construction and single-shot decoder}\label{sec:qltc}

We now prove the final remaining piece, an almost-good quantum locally-testable code with a parallel single-shot decoder~\Cref{thm:single-shot-qltc}. We will start by overviewing the cubical complex quantum code construction from~\cite{dinur2024expansion} and important notions such as tensor code robustness~\cite{kalachev2025maximally}. We give the sequential and parallel single-shot decoders for Z syndrome in~\Cref{sec:Z-decoder} and for X syndrome in~\Cref{sec:X-decoder}. The results there are stated generally in terms of the parameters of the code construction. We instantiate them to prove~\Cref{thm:single-shot-qltc} in~\Cref{subsec:proof-single-shot-qltc}

\subsection{Quantum codes from cubical complex}
We give an overview of a recent construction of almost-good quantum locally testable code~\cite{dinur2024expansion, kalachev2025maximally}. The construction is a higher-dimensional generalization of existing good quantum LDPC codes on square complexes~\cite{panteleev2022asymptotically, leverrier2022quantum, dinur2022good}. It is based on two ingredients. The first is a \emph{cubical complex} $X$, which is a purely geometric structure generalizing square complexes and that additionally has certain expansion properties. The second is a system of local coefficients, called \emph{sheaf}, that is based on a collection of constant-sized classical codes with certain robustness properties.

We start by recalling the definition of cubical complexes, following~\cite{dinur2024expansion}.

\begin{definition}[Cubical complex]
\label{def:cubical-complex}
Let $G$ be a set of size $|G|=N$ and $A_1,\hdots,A_t$ be finite subsets of permutations of $G$ of size $|A_i|=n_i $ (for simplicity we assume they have the same cardinality), such that each $A_i$ is closed under inverse and such that the permutations from different sets $A_i$ commute. A $t$-dimensional cubical complex $X(G, \{A_i\})$, or $X$ for short, is defined as a collection of cubical faces of variable dimensions, which are grouped according to their dimensions: $X = X(0) \cup \hdots \cup X(t)$.
\begin{itemize}
    \item The set of 0-dimensional faces (vertices) is denoted $X(0)$. There are $2^t |G|$ vertices, which are elements from $2^t$ copies of $G$. A vertex is referred to by a tuple $[g;\emptyset;b]$, where $g \in G$ labels the associated element and $b \in \{0,1\}^t$ labels which copy of $G$ the element is in. The dummy $\emptyset$ symbol reserves the second tuple element for higher-dimensional face notations, see below.
\end{itemize}
Higher-dimensional faces are (hyper)cubes formed by connecting vertices between different copies of $G$ using edges from $\{A_1,\hdots,A_t\}$:
\begin{itemize}
    \item Consider a subset $S \subseteq [t]$ of size $|S|=k$, a $k$-dimensional face $f$ (or $k$-face for short) of type $S$ is uniquely specified by a tuple $f=[g; (a_i)_{i \in S}, (b_j)_{j \in \overline{S}}]$, where $a_i \in A_i$ and $b_j \in \{0,1\}$. Such a $k$-face contains $2^k$ vertices ($0$-faces) specified by:
    \begin{align*}
        \{[(g \cdot \textstyle \prod_{i:b'_i=1} a_i ); \emptyset, b'\|b ]\text{ for } b' \in \{0,1\}^{S} \},
    \end{align*}
    where $b'\|b$ denotes string concatenation.
    The `$b$' part can be understood as describing the `orientation' of the face. For example, all vertices have the same type: $S=\emptyset$ and have $2^t$ possible orientations, edges (1-dimensional faces) have $t$ possible types and $2^{t-1} t$ possible orientations, etc. See~\Cref{fig:faces} for an illustration when $t=3$. We denote by $X(k)$ the set of faces of any type $S$ with $|S|=k$.
     \end{itemize}
Additionally, the cubical complex is endowed with the following partial order and incidence structure:
\begin{itemize}
    \item Given two faces $f, f'$ and $j \in [t]$ we write $f' \precdot_j f$ if it holds that: (1) $f'_j \in \{0,1\}$ while $f_j \in A_j$, (2) $f'_i=f_i$ for all $i \notin \{0,j\}$, and (3) $f'_0 = \begin{cases}
        f_0 & f'_j=0 \\
        f_0 \cdot f_j& f'_j=1 
    \end{cases}$. \\
    We write $f' \precdot f$ if $f' \precdot_j f$ for some $j$. We also write $f' \prec f$ if there is a sequence $f' \precdot \hdots \precdot f$ and $f' \preceq f$ if either $f' \prec f$ or $f'=f$.
    \item The up incidence map takes a $k$-face $f$ to all $(k+1)$-faces containing $f$: $\mathsf{up}(f) = \{f'| f' \succdot f \} $. Similarly the down map is $\mathsf{down}(f) = \{f'| f' \precdot f \} $. See~\Cref{fig:cubical-complex} for an illustration when $t=3$.
\end{itemize}
\end{definition}
\begin{figure}
    \centering
    \includegraphics[width=0.5\textwidth]{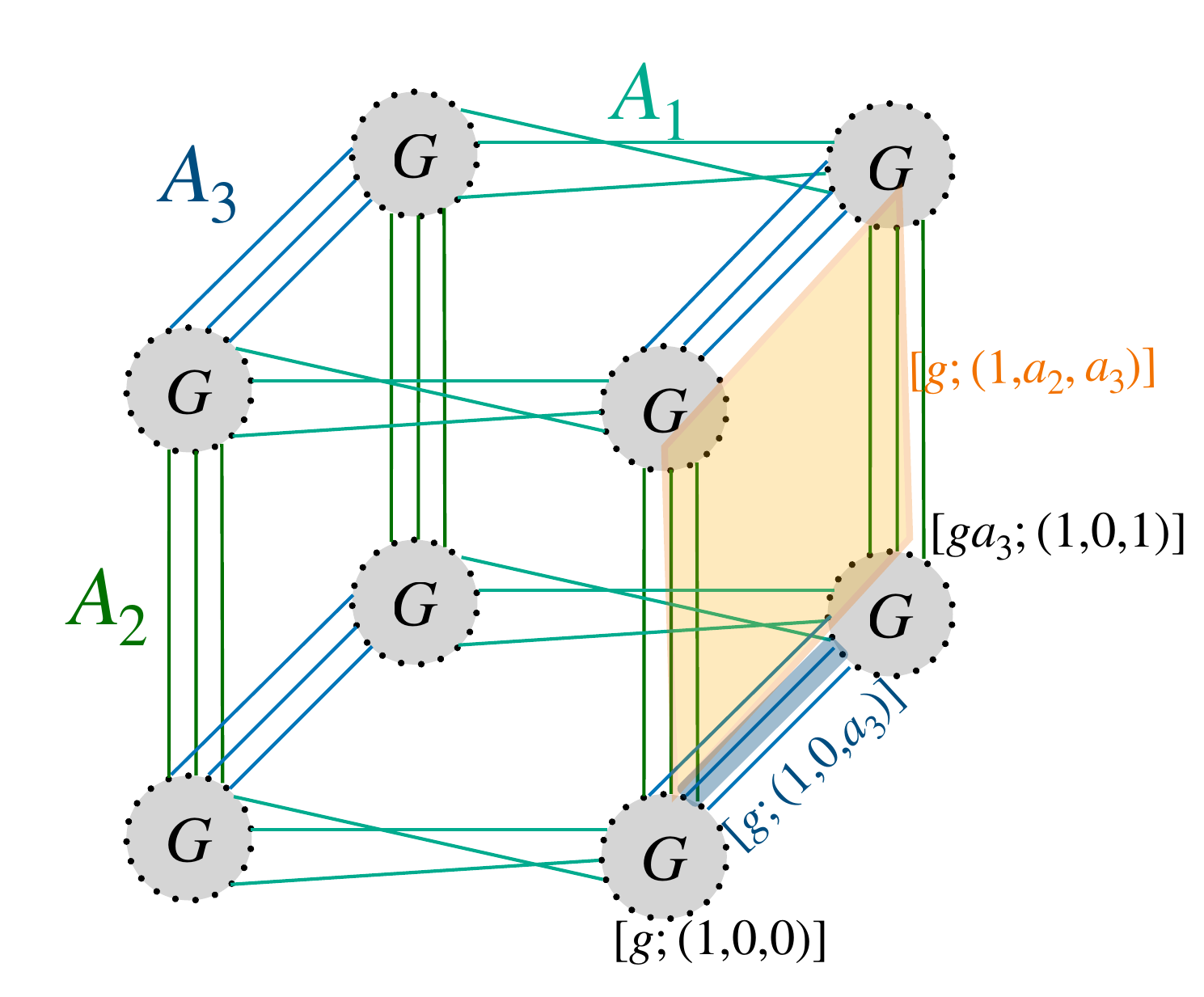}
    \caption{(Similar to Figure 1 in~\cite{dinur2024expansion}) A cubical complex with $t=3$ and three sets of permutations $A_1,A_2,A_3$. The underlying graph can be seen to be $2^t=8$-partite. The faces have label of the form as $[g;f_1,f_2,f_3]$, where $g\in G$. For a $k$-face, k elements out of $f_1,f_2,f_3$ take values from $A_i$ and the remaining elements take binary values.}
    \label{fig:faces}
\end{figure}
As examples, the case $t=1$ simply corresponds to graphs, while the case $t=2$ includes the square complex recently used in constructions of good qLDPC codes~\cite{panteleev2022asymptotically, dinur2022good, leverrier2022quantum, dinur2022locally}.

\paragraph{Links} We will often consider the `local view', or `link', of a face in the complex: The \emph{upward link} of a face $v \in X$ is the subcomplex consisting of all faces containing $v$, denoted $X_{\geq v} = \{ f\in X| f \succeq v\}$. We also denote $X_{\geq v}(k) = \{ f\in X(k)| f \succeq v\}$ and, for a type $S \supseteq \mathsf{type}(f)$, $X_{\geq v}(S) = \{ f\in X(S)| f \succeq v\}$. Similarly, the \emph{downward link} is $X_{\leq v} = \{ f\in X| f \preceq v\}$. 
We denote by $\mathsf{up}^{\ell}(f)$ the set of faces reachable from $f$ with $\ell$ up-walks, and $\mathsf{down}^{\ell}(f)$ is defined similarly.

We have the following claim regarding the geometric incidence properties of the cubical complex:
\begin{claim}\label{lem:incidence} Let $0 \leq i \leq t$ and $f \in X(i)$. It holds that
\begin{align*}
    |\mathsf{down}(f)| = 2i, \qquad |\mathsf{up}(f)| = (t-i)n,
\end{align*}
and, for $0 \leq \ell \leq i \leq k \leq t$,
\begin{align*}
    |X_{\leq f}(\ell)| = {{i}\choose {\ell}} 2^{i-\ell}, \qquad |X_{\geq f}(k)| = {{t-i}\choose {k-i}} n^{k-i}.
\end{align*}
As a consequence, the number of $k$-faces is
\begin{align*}
    |X(k)| =  \frac{1}{2^k} \sum_{v \in X(0)} |X_{\geq v}(k)| = {{t}\choose {k}} 2^{t-k} n^k |G|.
\end{align*}
\end{claim}
\begin{proof}
    Consider an $i$-face $f=[g;a,b]$ has type $S$ such that $|S|=i$. To specify an element in $\mathsf{down}(f)$ we need to specify an index $j \in S$ and a binary value $b_j$, giving $2i$ choices. This can be generalized to $X_{\leq f}(\ell)$, we need to specify $i-\ell$ indices in $S$ and $i-\ell$ binary values associated to them.
    
    To specify an element in $\mathsf{up}(f)$ we specify an index $j \in \overline{S}$ and a value $a_j \in [n]$, giving $(t-i)n$. Generalizing this to $X_{\geq f}(k)$, we need to specify $k-i$ indices in $\overline{S}$ and $k-i$ elements in from the corresponding permutations in $\{A_1,\hdots,A_t\}$.
    Finally, in the formula for the number of $k$-faces, the factor $2^{-k}$ is because each $k$-face contains $2^k$ vertices.
\end{proof}

The incidence structure of the cubical complex itself forms a chain complex. However, that chain complex does not have a nontrivial (co)homology for a good quantum code. To resolve this,~\cite{dinur2022good, dinur2024expansion} augment the faces in the cubical complex with local vector spaces (i.e., `sheaves') along with local maps with extra expansion structure, forming a chain complex called sheaf complex~\cite{first2022good}.

\begin{definition}[Sheaf complex on cubical complex]
\label{def:sheaf-complex}
Given a cubical complex $X$, a sheaf complex $C(X,\mathcal{F})$ can be constructed by attaching a vector space $V_f$ (i.e., a sheaf) to each face $f$ in $X$. Following~\cite{dinur2024expansion}, the vector space $V_f$ associated to a face $f$ depends only on $\mathsf{type}(f)$. In particular, let $m_i \geq 1$ be an integer and $\hat{A}_i = [m_i]$ for each $i \in [t]$. The vector space $V_f$ for $\mathsf{type}(f)=S$ is
\begin{align*}
    V_f = V_S = \mathbb{F}_q^{\prod_{j \in \overline{S}} \hat{A}_j},
\end{align*}
where in this paper we will restrict ourselves to a finite field $\mathbb{F}_q$ of characteristic 2. So the local vector space of a $t$-dimensional face is $V_{[t]}=\mathbb{F}_q$.
We denote by $\mathcal{F}=\{V_f\}$ the collection of local vector spaces.
This results in a $(t+1)$-term chain complex whose $k$-th chain and cochain spaces are
\begin{align*}
    C_k(X, \mathcal{F}) = \bigoplus_{f \in X(k)} V_f, \qquad C^k(X, \mathcal{F}) = \bigoplus_{f \in X(k)} V_f^*,
\end{align*}
where $V_f^* \cong V_f$ denotes the dual vector space.

The boundary $\partial_k$ and coboundary maps $\delta^k$ between the (co)chain spaces are defined by combining the incidence structure of the original cubical complex with linear operators on the local vector spaces.
For $i \in [t]$, let $h_i \in \mathbb{F}_q^{m_i \times n_i}$. We assume that $m_i \leq n_i$ and $h_i$ is full row rank.
\begin{itemize}
    \item The $k$-th coboundary map $\delta^k : C^k(X,\mathcal{F}) \mapsto C^{k+1}(X,\mathcal{F})$ is defined, for a cochain $z \in C^k(X,\mathcal{F})$, as
\begin{align*}
    \delta^k (z)(f) = \sum_{i \in \mathsf{type}(f)} \sum_{f' \precdot_i f} (\id_{-i} \otimes (h_i[f_i]^\top)) z(f'), \qquad \forall f \in X(k+1).
\end{align*}
In words, to obtain the component at a face $f \in X(k+1)$ of $\delta^k(z)$, we go over all faces $f' \precdot_i f$ below $f$, map the component $z(f')$ to an element in $V^*_f$ by applying the local operator $h_i[f_i]^\top \in \mathbb{F}_q^{1 \times m_i}$ on $z(f')$ (where $h_i[f_i]$ denotes the column of $h_i$ labeled by element $f_i \in A_i$), and then sum all the contributions from each $f'$.
\item Conversely, the $k$-th boundary map $\partial_k : C_{k}(X,\mathcal{F}) \mapsto C_{k-1}(X,\mathcal{F})$ is defined as
\begin{align*}
    \partial_k (z)(f) = \sum_{i \notin \mathsf{type}(f)} \sum_{f' \succdot_i f} z(f') \otimes h_i[f'_i], \qquad \forall f \in X(k-1),
\end{align*}
where $h_i [f'_i]$ denotes the column labeled by $f'_i \in A_i$.
\end{itemize}
Overall, we write the chain complex as
\begin{align*}
    C^*(X,\mathcal{F}) &= C^0(X,\mathcal{F}) \overset{\delta^0}{\rightarrow} C^1(X,\mathcal{F}) \overset{\delta^1}{\rightarrow} \hdots \overset{\delta^{t-1}}{\rightarrow} C^t(X,\mathcal{F}), \\
    C_*(X,\mathcal{F}) &= C_0(X,\mathcal{F}) \overset{\partial_1}{\leftarrow} C_1(X,\mathcal{F}) \overset{\partial_2}{\leftarrow} \hdots \overset{\partial_{t}}{\leftarrow} C_t(X,\mathcal{F}).
\end{align*}
We will often write $C(X)$ for short. We will also sometimes suppress the subscript $k$, writing $\delta$ or $\partial$, when it is clear what level the (co)chains are in.
\end{definition}

For quantum LDPC codes and quantum LTC, we specify to the cases of $t=2$ and $t=4$, respectively. However, all proofs presented in this section holds for arbitrary $t$.

\paragraph{(Block-)Hamming weight.} The block-Hamming weight of a cochain $z \in C^k(X,\mathcal{F})$ is the number of faces on which its restriction is non-zero, denoted as $|z| = |\{f \in X(k) \ | \  z(f)\neq 0 \}|$. The Hamming weight of $z$ is the sum over all faces $f$ of the Hamming weights of $z(f)$, and is denoted $|z|_H$. Note that $|z|_H \leq (\max_i m_i)^t \cdot |z|$. Similar definitions apply for chains.

\paragraph{Spectral expansion.} The expansion properties of the underlying graphs will determine the parameters of the codes in the cubical complex construction. Recall that given a regular connected graph $G=(V,E)$ on $N=|V|$ vertices, we say that it is $\lambda$-expanding if $\max\{|\lambda_2|,|\lambda_N|\} \leq \lambda$, where $\lambda_2$ and $\lambda_N$ are the second and $N$-th eigenvalues of the normalized adjacency matrix of $G$. We will need a slight generalization to non-connected graphs. We say $G$ is $\lambda$-expanding up to size $r|V|$, for $r \in (0,1)$, if the graph is a disjoint union of graphs of size $\geq r|G|$ each of which is $\lambda$-expanding. In the cubical complex construction, let $\operatorname{Cay}(G,A_j)$ denote the graph whose vertex set is $G$ and such that there is an edge between $(g,g')$ if and only if $g'=ag$ for some $a \in A_j$. Note that $\operatorname{Cay}(G,A_j)$ is undirected and $n$-regular since $A_j$ is assumed to be closed under inverse. Later, we will work with instantiations of the cubical complex construction in which $\operatorname{Cay}(G,A_j)$ is $\lambda$-expanding up to size $r|G|$ for each $j \in [t]$, where $r<1$.

\begin{figure}
    \centering
    \includegraphics[width=0.9\textwidth]{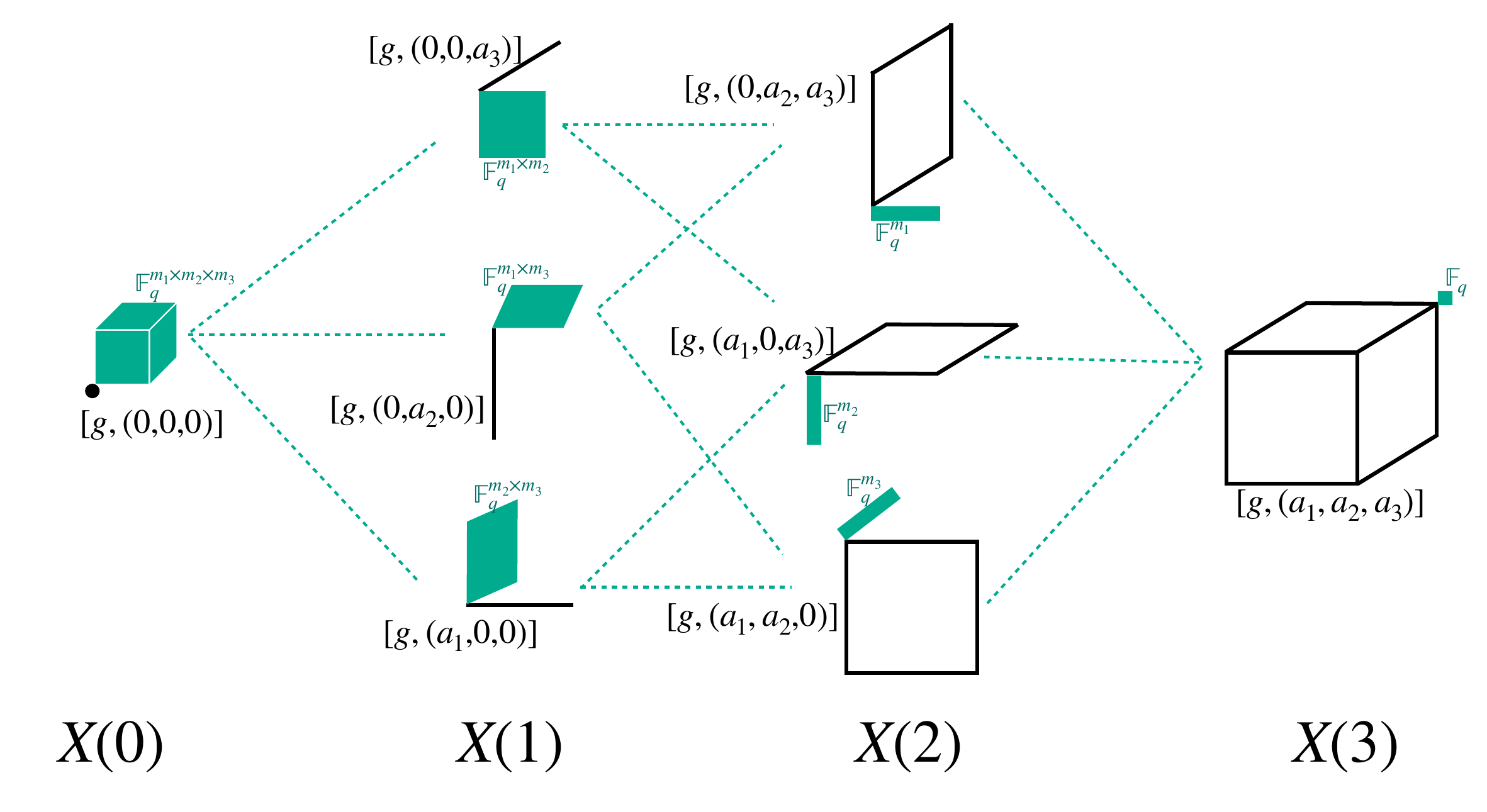}
    \caption{(Similar to Figure 2 in~\cite{dinur2024expansion}) Faces in a cubical complex ($t=3$), with sheaves (local coefficient systems) attached in teal.}
    \label{fig:cubical-complex}
\end{figure}

\subsection{Local views in cubical complex}

In this subsection, we will show that the upward local view of each face corresponds to a \emph{local chain complex}. This fact will be used crucially in the proofs in later sections. In particular, for each face type $\overline{S} \subseteq [t]$, we will show that its local view is isomorphic to a chain complex $C(L_S , \mcF)$ that we first define below.

We first consider the case $S=\{i\}$. The chain complex $C(L_{\{i\}},\mathcal{F})$, or $C(L_{\{i\}})$ for short, is defined as follows. It is a 2-term chain complex that is constructed by attaching sheaves $\mathcal{F}$ to the faces in a geometric object $L_{\{i\}}$. There is a single $0$-dimensional face which we denote $\emptyset$, i.e., $L_{\{i\}}(0)=\{\emptyset\}$. There are $n_i$ $1$-dimensional faces which we label using elements of $A_i$, i.e., we identify the set of 1-faces $L_{\{i\}}(1) \cong A_i$. The incidence structure of $L_{\{i\}}$ is simply $\emptyset \precdot a$ for any $a \in A_i$. The local vector space attached to the 0-face is $V_{\emptyset}=\mathbb{F}_q^{\hat{A}_i}$ (recall that $\hat{A}_i = [m_i]$). For each 1-face that is labeled by an element $a \in A_i$, we attach a local vector space $V_{a} = \mathbb{F}_q$. Hence, the cochain spaces are $C^{0}(L_{\{i\}})=\mathbb{F}_q^{\hat{A}_i}$ and $C^{1}(L_{\{i\}})=  \oplus_{a \in A_i} V_{a} = \mathbb{F}_q^{A_i}$.

The coboundary map $\delta_{\{i\}}$ is defined, for each $x \in  C^{0}(L_{\{i\}})= \mathbb{F}_q^{\hat{A}_i}$ and $a \in A_i$, as
\begin{align*}
    \delta_{\{i\}} (x) (a) =  h_i[a]^\top x \in \mathbb{F}_q.
\end{align*}
And the boundary map $\partial_{\{i\}}$ is defined, for $y \in C^{1}(L_{\{i\}}) =  \mathbb{F}_q^{A_i}$, as
\begin{align*}
    \partial_{\{i\}}(y) (\emptyset) = h_i y.
\end{align*}

For a general type $S \subseteq [t]$, the chain complex $C(L_S)$ is defined to be the homological product~\cite{bravyi2014homological} of $C(L_{\{i\}})$ for $i \in [S]$. In particular,
\begin{align}
    C(L_S) : C^0(L_S) \rightarrow C^1(L_S) \rightarrow \hdots  \rightarrow C^{|S|}(L_S),
\end{align}
where
\begin{align*}
    C^0(L_S) &= \otimes_{i \in S} C^0(L_{\{i\}})\\
    C^1(L_S) &= \bigoplus_{i \in S}  C^1(L_{\{i\}}) \otimes_{j \in S-\{i\}} C^0(L_{\{j\}})\\
    C^2(L_S) &= \bigoplus_{T \subseteq S: |T|=2}  \otimes_{i \in T} C^1(L_{\{i\}}) \otimes_{j \in S-T} C^0(L_{\{j\}})\\
    ...\\
    C^{|S|}(L_S) &= \otimes_{i \in S} C^1(L_{\{i\}}).
\end{align*}

For $0 \leq k  \leq |S|$, the $k$-faces of $L_S$ are $L_S(k) = \cup_{T \subseteq S, |T|=k} L_S(T)$, where faces of type $T$ are identified with elements of $\prod_{i \in T} A_i$. The local vector space associated to a face $f$ of type $T$ is $V_f = \mathbb{F}_q^{\prod_{j \in S-T} \hat{A}_j}$. Hence, $C^k(L_S)= \bigoplus_{T \subseteq S: |T|=k} C(L_S(T))$, where $C(L_S(T))= \bigoplus_{f \in \prod_{i \in T} A_i} \mathbb{F}_q^{\prod_{j \in S-T} \hat{A}_j}$.

The incidence structure is given by $f \precdot f'$ iff $f'$ is of type $T'$ such that $|T'|=|T|+1$ and $T \subseteq T'$. The coboundary map of $C(L_S)$ can be verified, for a cochain $x \in C^i(L_S)$ and face $f$ of type $T \subseteq S$ and $|T|=i+1$, to be
\begin{align}
    \delta_S(x)(f)= \sum_{j \in T} \sum_{f' \precdot_j f} (\id_{-j} \otimes h_{j}[f_j]^\top) x(f').
    \label{eq:local-chain-coboundary-map}
\end{align}
Similarly, the boundary map, for a chain $x \in C_{i+1}(L_S)$ and a face $f$ of type $T \subseteq S$ with $|T|=i$, is
\begin{align}
    \partial_S(x)(f) = \sum_{j \in S-T} \sum_{f' \succdot_j f} x(f') \otimes (h_j [f'_j]).
    \label{eq:local-chain-boundary-map}
\end{align}
Here, for brevity and consistency with the notations in~\cite{dinur2024expansion}, we omit the (co)chain level index in the notations $\delta_S, \partial_S$ (i.e., $S$ only labels the local chain complex). In later uses of these notations, it will be clear from context which (co)chain level the maps act on.

The reason we define these local chain complexes is because they are indeed what the faces in the (global) cubical complex `see' locally, as shown in the following lemma.

\begin{lemma}[Local chain structure, Lemma 4.2 in~\cite{dinur2024expansion}]\label{lemma:local-chain}
For any face $f$ of type $S$ in the sheaf complex $C(X, \mathcal{F})$, there is an isomorphism between $C(X_{\geq f})$ and $C(L_{\overline{S}})$. Specifically, for every $T \subseteq \overline{S}$, it holds that $C(X_{\geq f}(S \cup T)) \cong C(L_{\overline{S}}(T))$.
\end{lemma}

\begin{proof}
    Any face $f' \in X_{\geq f}(S \cup T)$ (i.e., the face is in $X_{\geq f}$ and of type $S \cup T$) is specified by some tuple $a_T \in A_T \triangleq \prod_{i \in T} A_i$, i.e., $f'_i=\begin{cases}
        f_i &  i \notin T\\
        a_i & i \in T
    \end{cases}$.
    By definition of the sheaf complex, the local vector space is $V_{f'} =  \mathbb{F}_q^{\prod_{j \notin (S \cup T) \hat{A}_j } } = V_{S \cup T}$. Hence,
    \begin{align*}
        C(X_{\geq f}(S \cup T)) = \bigoplus_{f' \in X_{\geq f}(S \cup T)} V_{f'} = \bigoplus_{a \in \prod_{i \in T} A_i} \mathbb{F}_q^{\prod_{j \notin (S \cup T) \hat{A}_j } },
    \end{align*} 
    which coincides with $C(L_{\overline{S}}(T))$ in the definition of the local chain complex.

    We can straighforwardly verify the (co)boundary maps coincide. Consider the restriction of the coboundary map of the global complex onto $C(X_{\geq f})$. For a cochain $z \in C^k(X)$ and face $f' \in X(k+1)$, it is defined in~\Cref{def:sheaf-complex} that    
    \begin{align*}
        \delta^k (z)(f') = \sum_{i \in \mathsf{type}(f')} \sum_{f'' \precdot_i f'} (\id_{-i} \otimes (h_i[f'_i]^\top)) z(f'').
    \end{align*}
    The restriction onto $C(X_{\geq f})$ reads, for $z \in C^k(X_{\geq f})$ and $f' \succ f$ (such that $\mathsf{type}(f') \supset S$), as
    \begin{align*}
        \delta^k (z)(f') = \sum_{i \in \mathsf{type}(f')-S} \sum_{f \preceq f'' \precdot_i f'} (\id_{-i} \otimes (h_i[f'_i]^\top)) z(f''),
    \end{align*}
    which can be seen to be isomorphic to~\Cref{eq:local-chain-coboundary-map}.

    The isomorphism of the boundary maps can be checked similarly. Recall that the $k$-th boundary map $\partial_k : C_{k}(X,\mathcal{F}) \mapsto C_{k-1}(X,\mathcal{F})$ is defined in~\Cref{def:sheaf-complex} for a chain $z \in C_k(X)$ and $f' \in X(k-1)$ as
    \begin{align*}
        \partial_k (z)(f') = \sum_{i \notin \mathsf{type}(f')} \sum_{f'' \succdot_i f'} z(f'') \otimes h_i[f''_i],
    \end{align*}
    The restriction onto $C(X_{\geq f})$ reads, for $z \in C_k(X_{\geq f})$ and $f' \succeq f$ (such that $\mathsf{type}(f') \supseteq S$), as
    \begin{align*}
        \partial_k (z)(f') = \sum_{i \notin \mathsf{type}(f')} \sum_{f'' \succdot_i f' \succeq f} z(f'') \otimes h_i[f''_i],
    \end{align*}
    which can be seen to be isomorphic to~\Cref{eq:local-chain-boundary-map}.
\end{proof}

\begin{lemma}
\label{lem:exactness-tensorcode} The following properties hold for any $S \in [t]$:
\begin{itemize}
    \item (Exactness) The chain complex $C(L_S)$ is exact at each level $i=0,\hdots, |S|-1$. In other words, for any $x \in C_i(L_S)$ such that $\partial_S(x)=0$, there is $y \in C_{i+1}(L_S)$ such that $\partial_{S}(y)=x$.
    \item (Tensor code) For any $x \in C_{|S|}(L_{S})$ such that $\partial
_S(x)=0$, $x$ must be in the tensor code $\bigotimes_{j \in S} \ker(h_j)$.
\end{itemize}
\end{lemma}
\begin{proof}
    For the exactness property, we can apply Künneth theorem because $C(L_S)$ is the homological product of $C(L_{\{i\}})$. Note the trivial homology $\operatorname{dim} H_0(L_{\{i\}}) = 0$. Hence, for $0\leq \ell \leq |S|-1$, we have $\operatorname{dim} H_\ell (L_S) = \sum_{T \subseteq S: |T|=\ell} (\prod_{j \in T}\operatorname{dim} H_1(L_{\{j\}}) (\prod_{j \notin T}\operatorname{dim} H_0(L_{\{j\}}))=0$.

    For the tensor code property, we note that $C_{|S|}(L_S) \cong \mathbb{F}_q^{\prod_{j \in S} A_j}$. For any face $f$ of type $T=S-\{i\} $, according to~\Cref{eq:local-chain-boundary-map} we have
\begin{align*}
    \partial_S(x)(f) = \sum_{f' \succdot_i f} x(f') \otimes (h_i [f'_i])=0.
\end{align*}
    This implies that after fixing $|S|-1$ coordinates specified by $f$, the restriction of $x$ on $f$ is in the kernel of $h_i$. In other words, when viewing $x$ as a tensor in $\mathbb{F}_q^{\prod_{j \in S} A_j}$, it holds that $x[\hdots,f_{i-1},\cdot,f_{i+1},\hdots] \in \ker(h_i)$ for all $f$ and 
    $i$. Therefore, $x \in \bigotimes_{j \in S} \ker(h_j)$.
\end{proof}

\subsection{Robustness property of local maps}

To use the cubical complex for good quantum code constructions, we need a set of local codes that satisfy the \emph{two-way robustness property} defined with respect to the local chain complexes from the previous subsection. Recall that the (co)boundary maps of the cubical complex code construction is obtained from a set of local codes: for $i \in [t]$, let $h_i \in \mathbb{F}_q^{m_i \times n_i}$ be a full-row-rank parity-check matrix of the $i$-th code (we will slightly abuse notation and refer to the code as $h_i$). Below and in the rest of this section, we assume $n_i=n$ for simplicity.

\begin{definition}[Minimality]
\label{def:minimal}
For $1\leq k \leq |S|$, a cochain $x \in C^{k}(L_S)$ is called minimal if $|x +\delta^{k-1}(y)| \geq |x|$ for all $y \in C^{k-1}(L_S)$. Here, the $|\cdot |$ notation refers to the block-Hamming weight.
\end{definition}

\begin{definition}[Robustness]\label{def:robustness} Let $S \subseteq [t]$ and $\ell = |S|$. For $0 \leq k \leq \ell -1$, let $\kappa_{\ell,k}$ be a positive real number. We say that the local complex $C(L_S)$ is $\kappa_{\ell,k}$-robust at level $k$ if for any $x \in C^k(L_S)$ such that $x$ is minimal it holds that
\begin{align*}
    |\delta_S(x)| \geq \kappa_{\ell,k} n |x|.
\end{align*}
We say $C(L_S)$ is $\kappa$-robust if it is $\kappa$-robust at each level.
\end{definition}
Note that for the case $S=\{i\}$ (so $\ell=1$) and $k=0$, the minimality requirement is vacuous and hence the robustness is simply the distance of the code $h_i$.

\begin{definition}[Two-way robustness]
\label{def:two-way-robustness}
Let $\kappa >0$. We say that the set of matrices the collection of $\{h_1,\hdots h_t\}$ is two-way $\kappa$-robust if the following conditions hold
    \begin{itemize}
        \item When based on the matrices $\{h_1,\hdots,h_t\}$ the complex $C(L_S)$ is $\kappa$-robust for each $S \subseteq [t]$ and at each level $k \in 0,\hdots,|S|-1$.
        \item When based on the matrices $\{h_1^\perp,\hdots,h_t^\perp\}$ the complex $C(L_S)$ is $\kappa$-robust for each $S \subseteq [t]$ and at each level $k \in 0,\hdots, |S|-1$.
    \end{itemize}
\end{definition}

It is known that the notion of level-$k$ robustness is equivalent to a related notion of coboundary expansion at level $k$, see Remark 4.7 in~\cite{dinur2024expansion}.
Furthermore, Appendix B of~\cite{kalachev2022two} shows that level-$(|S|-1)$ robustness of $C(L_S)$ coincides with the notion of two-sided product-expansion of the collection of codes $\{h_i\}_{i \in S}$. The existence of such collection of matrices with two-sided product-expansion, over large alphabets, has been recently proved using the probabilistic method~\cite{kalachev2025maximally}.

\begin{theorem}[Product-expansion of random codes~\cite{kalachev2025maximally}] Let $\ell>0$. For each collection of intervals $I_1,\hdots,I_\ell \subseteq (0,1)$, there exists $\rho >0$ such that for all sufficiently large $n \in \mathbb{N}$, taking uniformly random parity check matrices $h_i \in \mathbb{F}_q^{(1-\rho_i)n \times n}$, where $q = 2^{(n+3)^\ell}$, such that $\rho_i \in I_i$ with high probability give rise to $\{h_1,\hdots,h_\ell\}$ and $\{h_1^\perp,\hdots,h_\ell^\perp\}$ that are both $\rho$-product-expanding.
\end{theorem}

Although the alphabet is exponentially large in the codeword length, this result will suffice for applications of constructing quantum codes because we will only need constant-sized two-way robust local codes. Moreover, product-expansion a priori seems to be weaker than robustness (\Cref{def:two-way-robustness} asks for robustness at all levels, while product-expansion only concerns level $|S|-1$). However,~\cite[Section 4.3]{dinur2024expansion} shows that two-way robustness is indeed implied by product-expansion, establishing the following result.

\begin{theorem}[Existence of two-way robust codes~\cite{dinur2024expansion}]
\label{thm:DLV-robustness}
For each collection of intervals $I_1,\hdots,I_t \subseteq (0,1)$, there exists $\kappa >0$ such that for all sufficiently large $n \in \mathbb{N}$ there exist parity check matrices $h_i \in \mathbb{F}_q^{(1-\rho_i)n \times n}$, where $q = 2^{(n+3)^t}$, such that $\rho_i \in I_i$ and $\{h_1,\hdots,h_t\}$ is two-way $\kappa$-robust.
\end{theorem}

\subsection{Proof of qLTC decoder statement}
In this section, we first summarize the results from~\cite{dinur2024expansion} that give lowerbounds on the distance and soundness of a almost-c3 qLTC. We then summarize the single-shot decoding results proved in~\Cref{sec:Z-decoder} and~\Cref{sec:X-decoder}, generally stated in terms of the spectral expansions of the underlying graphs $(G,A_i)$ and the robustness of the local codes. For convenience in proving constant-overhead fault-tolerant computation, we instantiate the construction with an abelian-lifted expander family from~\cite{jeronimo2021explicit}, such that the growth rate of the code family is bounded as required in~\Cref{thm:single-shot-qltc}.

\subsubsection{Distance and soundness results from~\cite{dinur2024expansion}}
The main result of~\cite{dinur2024expansion} includes general lowerbounds on (co-)systolic distance and (co)-cycle expansion of length-$(t+1)$ chain complexes of the form $C(G;\{A_i\})$ described in the previous subsections. Their bounds are stated in terms of the spectral expansion of the graphs $(G,A_i)$ and the two-way robustness of the local codes $\{h_i\}$. These lowerbounds in turn imply to quantum code distance and quantum LTC soundness.

\begin{theorem}[Theorem 3.6 and Corollary 3.7 in~\cite{dinur2024expansion}]
\label{thm:DLVmain}
Let $t \geq 4$ and fix $2 \leq i \leq t-2$ and consider a finite field $\mathbb{F}_q$ of characteristic 2. Let $G$ be a set and $\{A_1,\hdots,A_t\}$ be collections of pairwise commuting permutations on $G$ such that the graphs $\operatorname{Cay}(G,A_j)$ are each $\lambda$-expanding up to size $r|G|$. Let $N=|G|$ and $n=|A_1|=\hdots=|A_t|$. Let $\{h_1,\hdots,h_t\}$ be a collection of full-row-rank matrices $h_j \in \mathbb{F}_q^{m_j \times n}$ satifying the two-way $\kappa$-robustness property from~\Cref{def:two-way-robustness}. Let $C^*(G;\{A_j\};\{h_j\})$ be the length-$(t+1)$ chain complex over $\mathbb{F}_q$ from~\Cref{def:sheaf-complex}. Then the dimension of the space $C^i$ (and also $C_i$) is $|C^i(X)|=Nn^i 2^{t-i} \sum_{T \in [t]: |T|=t-i} \prod_{j \in T} m_j$.

Furthermore, there exist choices
$m_1,\hdots,m_t$ for the local codes such that the following holds.  
Suppose that $\lambda < \kappa^t 2^{-\Theta(t^2)}$. We can define a $\mathbb{F}_q$-qudit CSS code by taking $\delta^i$ and $\partial_i$ as the X- and Z-type checks, respectively such that the following holds. Under a self-dual basis of $\mathbb{F}_q$, we obtain a qubit CSS code $\mcQ$ with checks of weight $O(t n \log_2 q)$ and parameters $[[N_\mcQ,K_\mcQ, D_\mcQ]]$ where
\begin{enumerate}
    \item $N_\mcQ= (\log_2 q) N \Theta(n)^t$,
    \item $K_\mcQ= \Omega((\log_2 q) N_\mcQ)$,
    \item $D_\mcQ = rN(nt)^{-O(t)} \exp(-O(t^2))$.
\end{enumerate}
Moreover, $\mcQ$ is a quantum LTC with soundness $\rho \geq r (\log_2 q)^{-1} (nt)^{-O(t)} \exp(-O(t^2))$.
\end{theorem}

If $n, t, q, r$ are all constant then  this result yields a qLTC family with constant weight parity checks, constant rate, relative distance and soundness. While the parameters $n, t, q$ can be chosen to be constants, it is currently unknown how to construct a high-dimensional cubical complex with the required expansion properties for constant $r$.

\subsubsection{Proof of~\Cref{thm:single-shot-qltc}}\label{subsec:proof-single-shot-qltc}
For convenience, we restate the definition and theorem.
\definitionSingleShotDecoding*

\theoremAlmostGoodQltcWithDecoder*

In~\Cref{sec:Z-decoder} and~\Cref{sec:X-decoder}, we will provide a linear-time single-shot decoder for quantum codes from the cubical complex construction for \emph{any} dimension $t$. The results there will in particular apply for the case DHLV code ($t=2$) and the described qLTC construction ($t=4$). Here, we will summarize the results proved there, and combine them with an explicit instantiation of the cubical complex construction to prove~\Cref{thm:single-shot-qltc}. To obtain the almost-c3 qLTC family with extra properties as required in~\Cref{thm:single-shot-qltc}, we instantiate the underlying cubical complex using a family of explicit abelian-lifted expanders from~\cite{jeronimo2021explicit}. This instantiation was described in~\cite{dinur2024expansion}. However, for constant-overhead fault tolerance applications, we need a better control on the growth rate of the constructed qLTC family, so we use a result from~\cite{jeronimo2021explicit} different from the choice used in~\cite{dinur2024expansion}.

Let $\ell$ be a positive integer and $H$ be a subgroup of the symmetric group $\operatorname{Sym}(\ell)$, recall that the $(H,\ell)$-lift of a (base) graph $G_0$ on $n$ vertices is a graph $G$ on $n \ell$ vertices where each vertex $v$ of $G_0$ is replaced by $\ell$ copies $(v,1),\hdots,(v,\ell)$ and for every edge $e=(u,v)$ of $G_0$ there are edges between $(u,i)$ and $(v,h_e(i))$ for $i \in [\ell]$, and where $h_e \in H$ is a pre-specified element assigned to the edge $e$ (such specification is called an $(H,\ell)$-signing of $G_0$). Importantly, the lifted graph $G$ has three useful properties. First, it has the same degree as the base graph. Second, if $H$ is abelian, then it inherits symmetries of $H$. Third, random lifts often have similar expansion properties to $G_0$ with high probability. For example, Agarwal et al.~\cite{agarwal2019expansion} showed that random $(\mathbb{Z}_\ell,\ell)$-lifts with $\ell \leq \exp(\Theta(n))$ are expanding. For this reason, lifting is a common way of obtaning large expanders from small ones and extensive works have been devoted to derandomizing lifts. Jeronimo et al.~\cite{jeronimo2021explicit} obtained the following derandomized abelian lift. 

\begin{lemma}[Implicit from Theorem 1.1 in~\cite{jeronimo2021explicit}]
\label{lem:abelian-lift}
There exists a universal constant $\delta \in (0,1)$ such that the following holds. For any large enough $n$ and any $\ell \leq \exp(n^\delta)$, given the generating elements of a transitive abelian group $H \leq \operatorname{Sym}(\ell)$, and any constant $\varepsilon \in (0,1)$, we can construct in deterministic polynomial time, a $d$-regular graph $G$, where $d \geq d_0(\varepsilon)$, on $\Theta(n\ell)$ vertices such that
$G$ is a $(H,\ell)$-lift of a graph $G_0$ on $\Theta(n)$ vertices and the normalized second eigenvalue of $G$ is $\lambda(G) \leq \varepsilon$.
\end{lemma}

Above, the fact that the base graph can be chosen to have size $\Theta(n)$ for any large enough $n$ is important to control the qLTC family growth rate. This is not stated in~\cite{jeronimo2021explicit}'s Theorem 1.1, but can be found in its proof in their section 5.2, which is in turn based on~\cite{mohanty2020explicit}.

As a corollary, we obtain a cubical complex family with a similar growth rate.

\begin{cor}
\label{cor:cubical-complex}
Let $t\geq 2$ and $\lambda \in (0,1)$. There exists an $i$-indexed explicit cubical complex family of the form $(G,\{A_1,\hdots,A_t\})$ as in~\Cref{def:cubical-complex} such that its size grows as $|G_i|=2^{O(i^{\delta'})}$ for some constant $\delta'<1$ and $|A_1|=\hdots=|A_t|=n$. Furthermore, for each $t \in [t]$, we have $n= n(\lambda)$ and the graph $\operatorname{Cay}(G,A_j)$ is $\lambda$-expanding up to size $r|G|$ where $r = \Theta((\log |G|)^{-(t-1)/\delta})$, where $\delta<1$ is the universal constant from~\Cref{lem:abelian-lift}.
\end{cor}

\begin{proof}
Consider the expander family from~\Cref{lem:abelian-lift}. Take the base graph to be an $n$-regular $G_0=(V_0,E_0)$ on $n_0=|V_0|$ vertices. Let $s: E_0 \mapsto H$ be the edge signing associated to the $(H,\ell)$-lift in~\Cref{lem:abelian-lift} 
with $\ell=\exp(n_0^\delta)$. The underlying vertex set of the cubical complex is defined to be $G= H \times (V_0)^{\times t}$. The permutation set $A_j = \{\pi_j^1,\hdots,\pi_j^n\}$ is defined as
\begin{align*}
    (h,v_1,\hdots,v_j,\hdots,v_t) \overset{\pi_j^i}{\longrightarrow} (s((v_j,v'_j))h,v_1,\hdots,v'_j,\hdots,v_t),
\end{align*}
where $v'_j$ is the $i$-th neighbor of $v_j$ in $G_0$. In other words, the graph $\operatorname{Cay}(G,A_j)$ is simply $|V_0|^{t-1}$ disjoint copies of the $(H,\ell)$-lift of $G_0$ with signing function $s$. Hence, by choosing the degree $n$ of $G_0$ to be a sufficiently large constant $n(\lambda)$ (whose exact dependence is from~\Cref{lem:abelian-lift}) we are guaranteed that $\operatorname{Cay}(G,A_j)$ is $\lambda$-expanding up to size $r|G|$, with $r = |V_0|^{-(t-1)}=\Omega((\log |G|)^{-(t-1)/\delta})$. For the growth rate of this cubical complex family, we have $|G|=n_0^t \exp(n_0^\delta)= 2^{O(n_0^{\delta'})}$. The growth rate follows from the fact that this construction works for any sufficiently large $n_0$.

\end{proof}

We are now ready to prove~\Cref{thm:single-shot-qltc}.

\begin{proof}[Proof of~\Cref{thm:single-shot-qltc}]
    The cubical complex with appropriate growth rate is from~\Cref{cor:cubical-complex}. The parameter $r= 1/\polylog(|G|)=1/\polylog(N)$. 
    Importantly the degree $n=O(1)$ can be chosen to be a sufficiently large constant such that the expansion $\lambda$ is small enough compared to the two-way robustness $\kappa$. This is possible because the universal constant $\kappa$ from~\Cref{thm:DLV-robustness} does not depend on $n$ (only on $t$), whereas $\lambda$ goes to zero as $n$ grows. Because $n$ is a constant, we can exhaustively search for a tuple $\{h_1,\hdots,h_t\}$ that is two-way $\kappa$-robust whose existence is guaranteed by~\Cref{thm:DLV-robustness}
    
    From~\Cref{thm:DLVmain} we have
    the rate, distance, and LTC soundness guarantees.
    For example, we take $t=4$ and place the qubits at level 2 of the above cubical complex construction. The alphabet size from~\Cref{thm:DLV-robustness} is $q=2^{(n+3)^{t}}=O(1)$.
    The LDPC sparsity of the code is roughly $\Delta= n^{O(t)}=O(1)$. The rate is also constant. The relative distance is $\Theta(r)=1/\polylog(N)$ and so is the soundness.

    The linear-time single-shot sequential decoder is given in~\Cref{prop:noisy-Zdecoder} (Z syndrome) and~\Cref{prop:noisy-Xdecoder} (X syndrome). There it is shown that the decoder is $(\alpha, \beta(N), \gamma, 0)$-single-shot, where $\alpha, \gamma$ are constants and $\beta = \Theta(r) = 1/\polylog(N)$.

    The single-shot parallel decoder is given in~\Cref{prop:noisy-parallel-Z} (Z syndrome) and~\Cref{prop:noisy-parallel-X} (X syndrome). There it is shown that in a constant runtime, the decoder is $(\alpha', \beta'(N), \gamma', \frac{1}{8})$-single-shot, where $\alpha', \gamma'$ are constants and $\beta' = \Theta(r) = 1/\polylog(N)$, and (as a consequence) in $O(\log N)$ runtime the parallel decoder is $(\alpha', \beta'(N), \gamma', 0)$-single-shot

    Again, we note that the required relations between $\lambda, t, \kappa$ in all of the above decoder results can be satisfied by choosing $\lambda$ to be sufficiently small (choosing $n$ to be a large constant).
\end{proof}

\subsection{Z-syndrome decoder}\label{sec:Z-decoder}
In this subsection, we present the Z-syndrome decoder in~\Cref{alg:Zdecoder}, which is a generalization of the `co-decoder' in~\cite{dinur2022good} to higher-dimensional cubical complexes and was briefly sketched in~\cite[Section 6]{dinur2024expansion}. We will first define some auxiliary tools and lemmas that were introduced in the X distance proof of~\cite{dinur2024expansion}. We then describe the decoder and prove its correctness, single-shot property, and time complexity. Finally, we give a parallel version of the decoder.

\subsubsection{Auxiliary definitions and lemmas}\label{sec:Z-distance}

We start with the notion of co-locally minimality (co-LM).

\begin{definition}[Local minimality]
\label{def:co-LM}
For $k>0$, a cochain $x \in C^{k}(X,\mathcal{F})$ is co-locally minimal if $|x +\delta^{k-1}(y)| \geq |x|$ for all $y \in C^{k-1}(X, \mathcal{F})$ supported on $X_{\geq v}(k-1)$ for some vertex $v \in X(0)$. Here, the $|\cdot |$ notation refers to the block-Hamming weight.
\end{definition}

We will use expansion properties of a natural random walk on the cubical complex.

\paragraph{The random walk} For $0 \leq \ell \leq k$ we define the following walk $W^{(k,\ell)}$ on $X(k)$. Informally, starting at a face $f \in X(k)$, we perform $\ell$ down-walks, reaching a face $f' \in X(k-\ell)$, we then take a random step to a face $f'' \in X(k-\ell)$ neighboring to $f'$, and finally perform $\ell$ up-walks from $f''$ back to $X(k)$. We give the detailed definition below.

\begin{definition}[Random walk] For $0 \leq \ell \leq k$ the walk $W^{(k,\ell)}$ is defined, starting $f \in X(k)$, as
\begin{enumerate}
    \item Choose a uniformly random $(k-1)$-face belonging to $f$ (there are $2k$ choices), then choose a random $(k-2)$-face belonging to that face (there are $2(k-1)$ choices), and so forth $\ell$ times in total. Let $f'$ be the $(k-\ell)$-face we end up with. (If $\ell=k$ then simply set $f'=f$).
    \item Let $f'=[g;a,b] \in X(k-\ell)$ be the face reached after step 1 and $S=\mathsf{type}(f')$. Choose a random direction $j \in  \overline{S}$ and random $a_j \in A_j$. Let $g'=ga_j$, $b'_j=1-b_j$ and $b'_i=b_i$ for $i \in \overline{S}\backslash\{j\}$. Let $f''=[g';a,b']$.
    \item Perform $\ell$ up-walks from $f''$ similarly to step 1, but in reverse. Note that an up-walk from $X(i)$ to $X(i+1)$ has $(t-i)\Delta$ choices. 
\end{enumerate}
\end{definition}

We have the following lemma regarding the expansion of the walk $W^{(k,\ell)}$.

\begin{lemma}[Lemma 5.12 in~\cite{dinur2024expansion}]\label{lem:walk-expansion} Suppose that the graphs $(G,A_j)$ are $\lambda$-spectral-expanding up to all vertex sets of size below $r|G|$. Let $0 \leq \ell \leq k$ and $\mathcal{A} \subseteq X(k)$, and let $\mathbf{1}_\mathcal{A}$ be the all-ones vector supported on $\mathcal{A}$, then
\begin{align*}
    \langle \mathbf{1}_\mathcal{A}, W^{(k,\ell)} \mathbf{1}_\mathcal{A} \rangle \leq \lambda |\mathcal{A}| + (4n)^{t} \frac{|\mathcal{A}|^2}{r|X(k)|}. 
\end{align*}
\end{lemma}
The intuition for~\Cref{lem:walk-expansion} is as follows. When $\frac{|\mathcal{A}|}{|X(k)|} \ll \frac{r\lambda}{(4n)^t}$, the RHS is dominated by $\lambda |\mathcal{A}|$. In other words, small sets must expand under the walk $W^{(k,\ell)}$.

Next, we define several notions of `next-nearest neighbors' of a face.
\begin{definition}\label{def:neighbors} Let $0 \leq \ell \leq k <t$ and $v \in X(\ell)$.
Define the following set of $k$-faces that intersect but do not contain $v$ is
\begin{align*}
    \mathsf{Nb}_v (k) = \{ f \in X(k) \ | \ \exists f' \in X(k+1), v \prec f', f \precdot f', f \cap v \neq \varnothing, v \slashed{\prec} f \}.
\end{align*}
Define the following set of $k$-faces that do not intersect $v$ is
\begin{align*}
    \mathsf{Op}_v (k) = \{ f \in X(k) \ | \ \exists f' \in X(k+1), v \prec f', f \precdot f', f \cap v = \varnothing \}.
\end{align*}
$\mathsf{Nb}$ stands for `neighbor' and $\mathsf{Op}$ stands for `opposite'. An illustration of the above definitions of neighbors can be found in~\Cref{fig:neighbors}.
\end{definition}

\begin{figure}[H]
    \centering
    \includegraphics[width=0.9\textwidth]{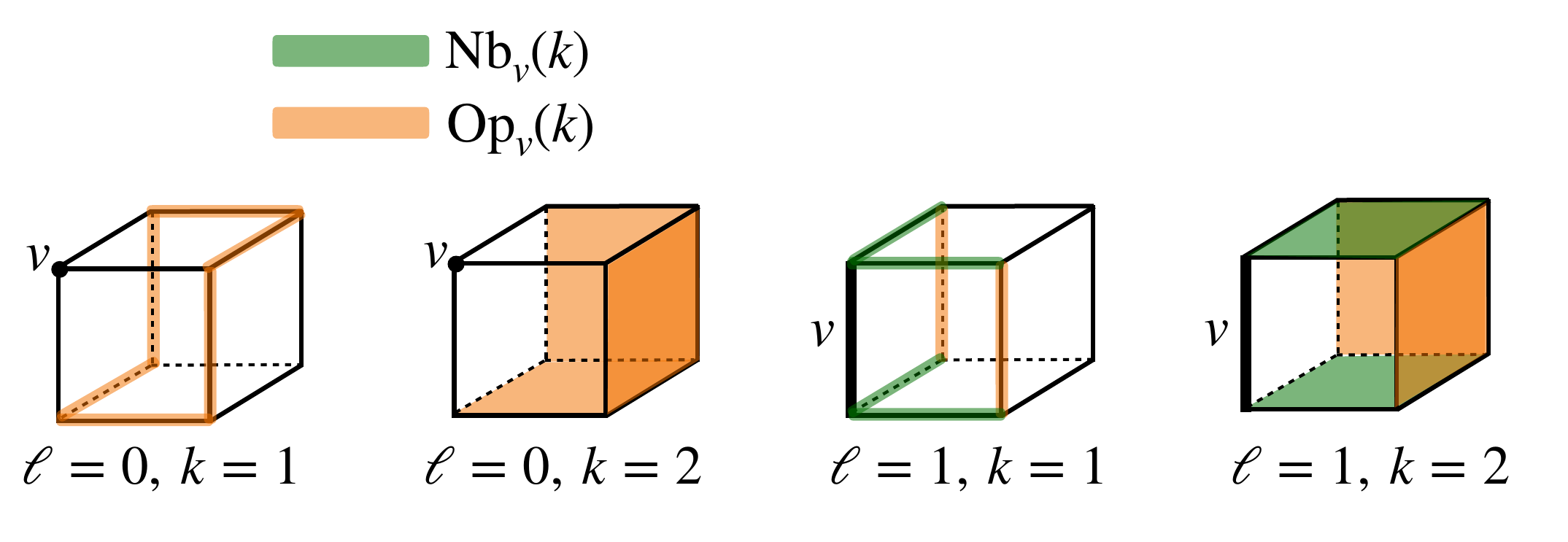}
    \caption{Illustration of `adjacent' neighbors and `opposite' neighbors defined in~\Cref{def:neighbors}. Note that vertices ($1$-faces) can only have opposite-type neighbors.}
    \label{fig:neighbors}
\end{figure}

It is clear from the above definitions that
\begin{align}
\label{eq:neighbor-contribution}
    \mathsf{down}( \mathsf{up}^{k+1-\ell} (v_\ell)) =  X_{\geq v_\ell}(k) \cup \mathsf{Nb}_{v_\ell}(k) \cup \mathsf{Op}_{v_\ell}(k).
\end{align}

Next, we note the following claims that will be used in the decoder's proof. We emphasize that this claim only depends on the geometric structure of the underlying cubical complex but not the attached sheaves. The notation $|\cdot|$ refers to the block-Hamming weight.

\begin{claim}\label{claim:neighbors} For any cochain element $x \in C^{k}(X)$, any $k$-face $v_k \in X(k)$, and any $\ell$-face $v_\ell \in X(\ell)$ with $ 1 \leq \ell \leq k$, it holds that
\begin{align}
\label{eq:neighbor-small}
|x(\mathsf{Nb}_{v_\ell}(k))| \leq 2^t \sum_{v_{\ell-1}\prec v_k \atop v_{\ell-1} \in X(\ell-1)} |x(X_{\geq v_{\ell-1}}(k))|.
\end{align}
In addition, let $x \in C^k(X)$ and $\mathcal{A}$ be the active $k$-faces of $x$, it holds that
\begin{align}
         \sum_{v_\ell \preceq v_k} |x (\mathsf{Op}_{v_\ell}(k))| \leq t8^t n^{k+1-\ell} \langle \mathbf{1}_\mathcal{A}, W^{(k,\ell)} \mathbf{1}_{v_k} \rangle.
         \label{eq:op-neighbors-small}
\end{align}
\end{claim}

\begin{proof}
    To see~\Cref{eq:neighbor-small}, we first make the observation that a face $f \in \mathsf{Nb}_{v_\ell}(k)$ must intersect with $v_\ell$ on some $(\ell-1)$-face according to~\Cref{def:neighbors}. This is because $\mathsf{type}(f)$ contains at least $\ell-1$ elements in $\mathsf{type}(v_\ell)$, but $v \slashed{\prec} f$, so $\mathsf{type}(f)$ contains exactly $\ell-1$ elements in $\mathsf{type}(v_\ell)$. Therefore, letting $v_k$ be any $k$-face containing $v_\ell$, we have
\begin{align*}
    |x(\mathsf{Nb}_{v_\ell}(k))| &\leq \sum_{v'_\ell: v'_\ell \prec v_k } |x(\mathsf{Nb}_{v'_\ell}(k))| \tag{trivial since $v_\ell \prec v_k$}\\
    &= \sum_{v'_\ell \prec v_k } \sum_{v_{\ell-1} \prec v'_\ell} \sum_{f_k \succ v'_\ell} \id_{f_k \cap v'_\ell = v_{\ell-1}} \id_{f_k \in x(\mathsf{Nb}_{v'_\ell}(k))} \tag{by observation above} \\
    &\leq  \sum_{v_{\ell-1} \prec v_k} \sum_{v'_\ell: v_{\ell-1} \prec v'_\ell \prec v_k } |x(X_{\geq v_{\ell-1}}(k))| \tag{bounding sum over $f_k$}\\
    &\leq  2^t \sum_{v_{\ell-1} \prec v_k} |x(X_{\geq v_{\ell-1}}(k))|,
\end{align*}
where the subscript implicitly indicates the face's dimension, and the last inequality is because any face has at most $2^t$ subfaces (which is a crude bound for the number of $v_\ell$ such that $v_{\ell-1} \prec v_\ell \prec v_k$, but it suffices for our purposes.)

To see~\Cref{eq:op-neighbors-small}, notice from definitions that $\sum_{v_\ell \preceq v_k} |\mathsf{Op}_{v_\ell}(k)|$ is upperbounded by the number of reachable vertices in $X(k)$ by performing the walk $W^{(k,\ell)}$ starting from $v_k$, hence the same argument in~\cite[Claim 5.11]{dinur2024expansion} can be applied. The factor $t8^t n^{k+1-\ell}$ is a simple upperbound for the factor ${k \choose \ell}{{t-\ell} \choose {k-\ell}}(t-\ell)2^{k-\ell} n^{k+1-\ell}$ in their paper.
\end{proof}

\subsubsection{Sequential decoder}
In this section, we give a linear-time decoder for $Z$-type syndromes, which was briefly sketched in~\cite{dinur2024expansion}. Our proof generalizes the ``co-decoder'' in~\cite{dinur2022good} to arbitrary $t$-dimensional cubical complex for any $t\geq 2$. The decoder below is also a generalization of small-set flip algorithm to high-dimensional cubical complexes~\cite{leverrier2015quantum}.Furthermore, we show that the decoder is single-shot can be parallelized.

\RestyleAlgo{boxruled}
\LinesNumbered
\begin{algorithm}[H]
    \setstretch{1.35}
    \caption{Sequential Z-syndrome decoder (level $k$) with parameter $\varepsilon \in (0, 1]$. \label{alg:Zdecoder}}
 \KwIn{Syndrome $\sigma = \delta^k (\hat{e}) \in C^{k+1}(X)$ of error $\hat{e} \in C^{k}(X)$.}
 \KwOut{A proposed correction $e \in C^{k}(X)$.}
        (Initialization) Set $\sigma^{(0)}=\sigma$;
        
        (Update loop) In $i$-th iteration, if there is a vertex $v^{(i)} \in X(0)$ and a cochain $w^{(i)} \in C^{k}(X)$ supported on $X_{\geq v^{(i)}}(k)$ such that $|\sigma^{(i-1)}| - |\sigma^{(i-1)}+\delta (w^{(i)})|  > (1-\varepsilon) |\delta (w^{(i)})| $, then set $\sigma^{(i)} = \sigma^{(i-1)} + \delta (w^{(i)})$;
        
        (Termination) Stop if no such $w^{(i)}$ is found. Output $e =  \sum_i w^{(i)}$.
\end{algorithm}

It is instructive to consider the case $\varepsilon=1$ of the above decoder. In that case, it is a simple greedy algorithm that finds a vertex whose upward link supports a local `flip' that strictly reduces the syndrome weight. This simple update rule already suffices to obtain a linear-time sequential decoder, but here we introduce the extra parameter $\varepsilon$ and analyze this general case in order to prove the parallel decoder in the next subsection. We start with proving the performance guarantee of~\Cref{alg:Zdecoder} when the syndrome is noiseless.

\begin{prop}[Sequential Z-syndrome decoder with noiseless syndrome]
\label{prop:noiseless-Zdecoder} Let $\varepsilon$ be any constant in $ \frac{2}{\kappa} \lambda^{1/t} 8^{t} < \varepsilon \leq 1$.  Consider a $t$-dimensional cubical complex construction where the underlying graphs $\operatorname{Cay}(G,A_j)$ are $\lambda$-expanding up to size $r|G|$ and the local codes $\{h_1,\hdots,h_t\}$ are two-way $\kappa$-robust.
Let $1\leq k \leq t-1 $ and consider the quantum code where qubits correspond to level $k$ of the chain complex as described in~\Cref{thm:DLVmain}. Let $\hat{e} \in C^{k}(X)$ be the cochain (over $\mathbb{F}_q$) corresponding to a Pauli X error such that $|\hat{e}|< \frac{8^{-t^2}\kappa^{k+1}-\lambda}{2^{t+1} (4n)^t} r |X(k)| \cdot \min (\frac{(8^{-t^2}\kappa^{k+1} - \lambda)^2}{2^{t+3}}, (\frac{\varepsilon \kappa}{2})^{k+1}8^{-t^2}-\lambda )$. 
Given the syndrome $\sigma = \delta^k (\hat{e}) \in C^{k+1}(X)$ for the error $\hat{e}$, then~\Cref{alg:Zdecoder} runs in time $O(q^{n^{O(t)}} |X(k)|)$ and outputs a correction cochain $e$ that is homologically equivalent to $\hat{e}$.
\end{prop}

\begin{proof} 
Let us first describe the high-level idea of the proof. We track the decoding progress with cochains $\hat{x}^{(i)} \in C^k(X)$ (unknown to decoder) that represent the remaining uncorrected error after step $i$. Our goal is to show that after the final iteration $T$, we have $\hat{x}^{(T)}=0$. Up to stabilizers, these cochains can be chosen such that they are co-locally minimal, allowing us to use the $\kappa$-robustness property. In particular, the termination condition and the robustness property together imply $\hat{x}^{(T)}$ only expands a little (\Cref{eq:349}). On the other hand, the spectral expansion of the underlying cubical complex implies that this can only happen when either $\hat{x}^{(T)}=0$ or $\hat{x}^{(T)}$ is large (\Cref{eq:350}). Finally, we show that the latter cannot happen if the initial error is small.

    \paragraph{Auxiliary variables.} We first define auxilliary variables that are used in the analysis but unknown to the decoder. These variables will be distinguished by a hat symbol ``$\string^$''. Let $\hat{x}^{(0)}$ be the minimal cochain that is homologically equivalent to the underlying error $\hat{e}$, i.e., $\hat{x}^{(0)}$ is the minimal block-Hamming weight cochain among those of the form $\hat{e} +\delta^{k-1}(y)$, where $y \in C^{k-1}(X)$. Clearly such $\hat{x}^{(0)}$ is also co-locally minimal (according to~\Cref{def:co-LM}). And $\hat{x}^{(0)}$ and $\hat{e}$ share the same syndrome $\sigma$. Similarly let $\hat{x}^{(i)}$ be the minimal version of $\hat{e} + \sum_{i'=1}^i w^{(i')}$. Let $\sigma^{(i)} = \sigma + \sum_{i'=1}^{i} \delta^k (w^{(i)})$, with $\sigma^{(0)}= \sigma$.

\paragraph{Useful weight inequalities.} Next, we note several useful block-Hamming weight inequalities (\cref{eq:345},~\cref{eq:346},~\cref{eq:347} below) that hold true throughout the algorithm. 

First, using the robustness property of the local maps and the fact that $\hat{x}^{(i)}$ is co-LM, we claim that the following inequality holds.
\begin{claim}
    For every iteration $i$, and
    $\forall \ell \leq k,\ \forall v_\ell \in X(\ell)$, $S= \mathsf{type}(v_\ell)$,
\begin{align}
    \kappa n |\hat{x}^{(i)}(X_{\geq v_\ell}(k))| \leq |\delta_{\overline{S}}(\hat{x}^{(i)}(X_{\geq v_\ell}(k))|,
    \label{eq:345}
\end{align}
where $\delta_{\overline{S}}$ denotes the coboundary map of the local complex from~\Cref{lemma:local-chain}.
\end{claim}

\begin{proof}
    Indeed, according to~\Cref{lemma:local-chain},  $C(X_{\geq v_\ell})$ is isomorphic to the length-$(t-\ell)$ local chain complex $C(L_{\overline{S}})$. So we can view $z \triangleq \hat{x}^{(i)}(X_{\geq v_\ell}(k))$ as a $(k-\ell)$-cochain of $C(L_{\overline{S}})$, i.e., $z \in C^{k-\ell}(L_{\overline{S}})$. Furthermore, the fact that $\hat{x}^{(i)}$ is co-LM implies that $z$ is minimal with respect to $C(L_{\overline{S}})$, in the sense of~\Cref{def:minimal}. Indeed, the case $\ell=k$ is trivial since $z$ is a 0-cochain and thus is automatically minimal. Consider the case $\ell \leq k-1$, let $y \in C^{k-\ell-1}(L_{\overline{S}})$ be an arbitrary $(k-\ell-1)$-cochain in the local chain complex. From the perspective of the global cubical complex $y$ corresponds to a $(k-1)$-cochain supported on $X_{\geq v_\ell}(k-1)$ and similarly $\delta_{\overline{S}}(y)$ corresponds to a $k$-cochain supported on $X_{\geq v_\ell}(k)$. We have that 
\begin{align*}
   |\hat{x}^{(i)} + \delta^{k-1}(y)| &= |z + \delta_{\overline{S}}(y)| + |\hat{x}^{(i)} - z| \tag{decomposing $\hat{x}^{(i)} + \delta^{k-1}(y)$}
    \\ 
    \Rightarrow |z + \delta_{\overline{S}}(y)| & =|\hat{x}^{(i)} + \delta^{k-1}(y)| - |\hat{x}^{(i)} - z| \\
    &\geq  |\hat{x}^{(i)}| - |\hat{x}^{(i)} - z| \tag{using $\hat{x}^{(i)}$ is co-LM}\\
    &\geq |z|. 
\end{align*}

Having shown that $z$ is a minimal cochain in $C(L_{\overline{S}})$, we can then apply the robustness property from~\Cref{def:robustness} to obtain~\Cref{eq:345}.
\end{proof}

Next, recall from~\Cref{eq:neighbor-contribution} that $\mathsf{down}( \mathsf{up}^{k+1-\ell} (v_\ell)) =  X_{\geq v_\ell}(k) \cup \mathsf{Nb}_{v_\ell}(k) \cup \mathsf{Op}_{v_\ell}(k)$, which means that, for \emph{all} cochain $x$ (regardless of whether it is co-LM or not), the restriction of $\delta^k(x)$ onto $X_{\geq v_\ell}(k+1)$ only receives contributions from the images under $\delta^k$ of $x(X_{\geq v_\ell}(k))$,  $x(\mathsf{Nb}_{v_\ell}(k))$, and $x(\mathsf{Op}_{v_\ell}(k))$. Therefore, the following triangle inequalities hold
    \begin{align}
        |\sigma^{(i)}(X_{\geq v_\ell}(k+1))| &\geq |\delta_{\overline{S}} \hat{x}^{(i)}(X_{\geq v_\ell}(k))| - (|\hat{x}^{(i)}(\mathsf{Op}_{v_\ell}(k))| + |\hat{x}^{(i)}(\mathsf{Nb}_{v_\ell}(k))| ), \label{eq:346}\\
        |\sigma^{(i)}(X_{\geq v_\ell}(k+1)) + \delta_{\overline{S}} \hat{x}^{(i)}(X_{\geq v_\ell} (k))| &\leq |\hat{x}^{(i)}(\mathsf{Op}_{v_\ell}(k))| + |\hat{x}^{(i)}(\mathsf{Nb}_{v_\ell}(k))|.\label{eq:347}
    \end{align}

\paragraph{Termination condition implies either no errors remain or final error is large.} Suppose that no small-set flip is found in iteration $T$ so that~\Cref{alg:Zdecoder} terminates at iteration $T$. Our goal is to show that $\hat{x}^{(T)}=0$, which would mean that the proposed correction $\sum_{i=1}^{T} w^{(i)}$ is homologically equivalent to the true error $\hat{e}$ as desired. To this end, let us first show that the termination condition of the algorithm implies either $\hat{x}^{(T)}=0$ or it has to have large block-Hamming weight. 
    
Observe that, for each face $v_\ell \in X(\ell)$ where $0 \leq \ell \leq k$, modifying $\hat{x}^{(i)}$ on $X_{\geq v_\ell}(k)$ only affects the syndrome on $X_{\geq v_\ell}(k+1)$.
In particular, flipping $\hat{x}^{(i)}(X_{\geq v_\ell}(k))$ with a $k$-cochain $f$ supported on $X_{\geq v_\ell}(k)$ changes the syndrome values on $X_{\geq v_\ell}(k+1)$ by $\delta_{\overline{S}} f$.
Since no flip is found at iteration $T$ and $\hat{x}^{(T)}(X_{\geq v_\ell} (k))$ is itself a candidate for $w^{(T)}$, it holds that 
    \begin{align}
        |\sigma^{(T)}(X_{\geq v_\ell}(k+1))| - |\sigma^{(T)}(X_{\geq v_\ell}(k+1)) + \delta_{\overline{S}} \hat{x}^{(T)}(X_{\geq v_\ell} (k))| \leq  (1-\varepsilon)|\delta_{\overline{S}} \hat{x}^{(T)}(X_{\geq v_\ell} (k))|.
        \label{eq:345-5}
    \end{align}
    Combining~\Cref{eq:345-5} with~\Cref{eq:346} and~\Cref{eq:347} we have
    \begin{align}
        |\delta_{\overline{S}} \hat{x}^{(T)}(X_{\geq v_\ell}(k))| \leq \frac{2}{\varepsilon} (| \hat{x}^{(T)}(\mathsf{Op}_{v_\ell}(k))| + |\hat{x}^{(T)}(\mathsf{Nb}_{v_\ell}(k))|).
        \label{eq:348}
    \end{align}
    Moreover, combining~\Cref{eq:348} with~\Cref{eq:345}, it holds that, $\forall \ell \leq k,\ \forall v_\ell \in X(\ell)$
        \begin{align}
       \kappa n |\hat{x}^{(T)}(X_{\geq v_\ell}(k))| \leq \frac{2}{\varepsilon} (| \hat{x}^{(T)}(\mathsf{Op}_{v_\ell}(k))| + |\hat{x}^{(T)}(\mathsf{Nb}_{v_\ell}(k))|).
        \label{eq:349}
    \end{align}

Intuitively,~\Cref{eq:349} upperbounds the weight of a local patch in the support of $\hat{x}^{(T)}$ by the size of its neighborhood. This local expansion is governed by the robustness $\kappa$ of the local maps. Furthermore, such local expansion can be turned into a global expansion when combined with the spectral expansion of the underlying cubical complex, so that bounds on the block-Hamming weight of the global $\hat{x}^{(T)}$ can be obtained. Indeed, let us first apply~\Cref{eq:349} to relate $|\hat{x}^{(T)}|$ to the random walk defined in the previous subsection.

Consider a face $v_k \in X(k)$. We recursively apply~\Cref{eq:349}, starting from $\ell=k$ to $\ell=0$ as follows (below, the subscript $\ell$ in $v_\ell$ implicitly means $v_\ell \in X(\ell)$)
\begin{align*}
    |\hat{x}^{(T)}(X_{\geq v_k}(k))| &\overset{\eqref{eq:349}}{\leq} \frac{| \hat{x}^{(T)}(\mathsf{Op}_{v_k}(k))|}{\varepsilon \kappa n/2} + \frac{|\hat{x}^{(T)}(\mathsf{Nb}_{v_k}(k))|}{\varepsilon \kappa n/2}\\
    &\labelrel\leq{eq:122} \frac{| \hat{x}^{(T)}(\mathsf{Op}_{v_k}(k))|}{\varepsilon \kappa n/2} + 2^t  \sum_{v_{k-1} \prec v_k} \frac{|\hat{x}^{(T)}(X_{\geq v_{k-1}}(k))|}{\varepsilon \kappa n/2}\\
    &\leq \frac{|\hat{x}^{(T)} (\mathsf{Op}_{v_k}(k))|}{\varepsilon \kappa n/2} + 2^t  \sum_{v_{k-1} \prec v_k} \left( \frac{|\hat{x}^{(T)}(\mathsf{Op}_{v_{k-1}}(k))|}{ (\varepsilon \kappa  n/2)^2 } + \frac{|\hat{x}^{(T)}(\mathsf{Nb}_{v_{k-1}}(k))|}{(\varepsilon \kappa n/2)^2} \right)\\
    &\labelrel\leq{eq:123} \frac{| \hat{x}^{(T)}(\mathsf{Op}_{v_k}(k))|}{\varepsilon \kappa n/2} + 2^t  \sum_{v_{k-1} \prec v_k} \frac{|\hat{x}^{(T)}(\mathsf{Op}_{v_{k-1}}(k))|}{(\varepsilon \kappa n/2)^2 } +2^{2t} \sum_{v_{k-2} \prec v_k}\frac{|\hat{x}^{(T)}(X_{\geq v_{k-2}}(k))|}{( \varepsilon \kappa n/2)^2}\\
    &\leq \hdots \\
    & \leq \sum_{\ell=0}^k \sum_{v_\ell \preceq v_k} \frac{2^{t(k-\ell)}}{(\varepsilon \kappa n/2)^{k-\ell+1}} |\hat{x}^{(T)} (\mathsf{Op}_{v_\ell}(k))|, \numberthis \label{eq:x_vk}
\end{align*}
where~\eqref{eq:122},~\eqref{eq:123} used~\Cref{eq:neighbor-small} in~\Cref{claim:neighbors}

Next, let $\mcA$ be the support of $\hat{x}^{(T)}$, then according to~\Cref{eq:op-neighbors-small} in~\Cref{claim:neighbors} we have
\begin{align}
    \sum_{v_\ell \preceq v_k} |\hat{x}^{(T)} (\mathsf{Op}_{v_\ell}(k))| \leq \langle \mathbf{1}_\mathcal{A}, t8^t n^{k+1-\ell} W^{(k,\ell)} \mathbf{1}_{v_k} \rangle .
    \label{eq:x_vk-2}
\end{align}
Combining~\Cref{eq:x_vk} and~\Cref{eq:x_vk-2} we get
\begin{align*}
    |\hat{x}^{(T)}| &= \sum_{v_k \in X(k)} |\hat{x}^{(T)}(X_{\geq v_k}(k))| \\
    &\overset{\eqref{eq:x_vk}}{\leq} \sum_{ v_k \in X(k)} \sum_{\ell=0}^k \sum_{v_\ell \preceq v_k} \frac{2^{t(k-\ell)}}{(\varepsilon \kappa n/2)^{k-\ell+1}} |\hat{x}^{(T)} (\mathsf{Op}_{v_\ell}(k))| \\
    &\overset{\eqref{eq:x_vk-2}}{\leq} \sum_{\ell=0}^k \frac{2^{t(k-\ell)}}{(\varepsilon \kappa n/2)^{k-\ell+1}} \langle \mathbf{1}_\mathcal{A}, t8^t n^{k+1-\ell} W^{(k,\ell)} \mathbf{1}_{\mcA} \rangle \\
    &\leq  \sum_{\ell=0}^k  \frac{2^{t(k-\ell)}}{(\varepsilon \kappa/2)^{k-\ell+1}} t8^t \left( \lambda|\mathcal{A}| + (4 n)^t \frac{|\mathcal{A}|^2}{r|X(k)|} \right) \hspace{0.5cm} \tag{using \Cref{lem:walk-expansion}} \\
    &\leq  \frac{ 8^{t^2}}{(\varepsilon \kappa/2)^{k+1}}\left( \lambda|\hat{x}^{(T)}| + (4 n)^t \frac{|\hat{x}^{(T)}|^2}{r|X(k)|} \right), \numberthis \label{eq:350}
\end{align*}
where we assume $t \geq 2$ to get the simple (but loose and somewhat arbitrary) upperbound $8^{t^2}$ for the sum. It follows that, either $\hat{x}^{(T)}=0$ or $|\hat{x}^{(T)}| \geq \frac{r |X(k)|}{(4n)^t} \left( \frac{(\varepsilon \kappa/2)^{k+1}}{8^{t^2}} -\lambda \right)$.

The previous few paragraphs also automatically imply the following lemma which will be referred to in the parallel decoder proof.

\begin{lemma}[Existence of local flip]
    \label{lem:existence-local-flip} 
    For every iteration $i$, if the current stabilizer-reduced error satisfies $|\hat{x}^{(i)}| < \frac{r |X(k)|}{(4n)^t} \left( \frac{(\varepsilon \kappa/2)^{k+1}}{8^{t^2}} -\lambda \right) $, then there exist
    $\ell \leq k,\ v_\ell \in X(\ell)$ such that, letting $S= \mathsf{type}(v_\ell)$,
    \begin{align}
        |\sigma^{(i)}(X_{\geq v_\ell}(k+1))| - |\sigma^{(i)}(X_{\geq v_\ell}(k+1)) + \delta_{\overline{S}} \hat{x}^{(i)}(X_{\geq v_\ell} (k))| >  (1-\varepsilon)|\delta_{\overline{S}} \hat{x}^{(i)}(X_{\geq v_\ell} (k))|.
    \end{align}
where $\delta_{\overline{S}}$ denotes the coboundary map of the local chain complex from~\Cref{lemma:local-chain}.
\end{lemma}

\paragraph{Update rule guarantees error cannot grow too large from initially small error.} To conclude that $\hat{x}^{(T)} =0$ is the case, we now show that if the initial underlying error $e$ is small, then $|\hat{x}^{(i)}|$ will remain small, and therefore $\hat{x}^{(T)}=0$. Indeed, this follows from two properties: (1) the syndrome's block-Hamming weight strictly decreases in each iteration, $|\sigma^{(i)}| < |\sigma^{(i-1)}|$, and (2) $|\sigma^{(i)}|$ provides an upperbound on $|\hat{x}^{(i)}|$ as long as they are sufficiently small \footnote{This small-set condition differentiates small-set co-LM expansion from local testability, which instead requires a similar relation between error weight and syndrome weight for \emph{all} errors. Of course, when $k \geq 2 $ (e.g., $t=4, k=2$), we can simply use the inverse-polylog local testability soundness from~\cite{dinur2024expansion} to immediately conclude the decoder proof. However, here we aim to give decoder for the lower-dimensional cubical complexes such as $t=2$, $k=1$, which include the good qLDPC code from~\cite{dinur2022good}.}. The first property is by design of~\Cref{alg:Zdecoder}, while the second property is similar to a notion of small-set co-boundary expansion in~\cite{dinur2022good} and holds for any sufficiently small co-locally minimal cochain.

\begin{lemma}[Small-set co-LM expansion]\label{lem:small-set-coLM-expansion} For $1\leq k \leq t-1$ and $\lambda < 8^{-t^2}\kappa^{k+1}$. 
     Let $x \in C^{k}(X)$ be a co-locally minimal cochain with $|x| <  \frac{(8^{-t^2}\kappa^{k+1} - \lambda)^2}{2^{t+2}(4n)^t} r |X(k)|$ and $\varsigma=\delta^k(x)$. Then it holds that $|x| \leq \frac{2}{n(8^{-t^2}\kappa^{k+1}-\lambda )} |\varsigma|$.
\end{lemma}
\begin{proof}

Recall~\Cref{eq:345} which actually holds for any co-LM cochain and~\Cref{eq:346} which holds for any cochain. Combining them we get, $\forall \ell \leq k,\ \forall v_\ell \in X(\ell)$, $S= \mathsf{type}(v_\ell)$,
\begin{align}
    \kappa n |x(X_{\geq v_\ell}(k))| \leq |\delta_{\overline{S}} x(X_{\geq v_\ell}(k))| \leq  | x(\mathsf{Op}_{v_\ell}(k))| + |x(\mathsf{Nb}_{v_\ell}(k))| + |\varsigma(X_{\geq v_\ell}(k+1))|.
    \label{eq:351}
\end{align}
We next apply the same argument in~\Cref{eq:x_vk},~\Cref{eq:x_vk-2}, but replacing~\Cref{eq:349} with~\Cref{eq:351}. Doing so we obtain
    \begin{align*}
       |x| &= \sum_{v_k \in X(k)} |x(X_{\geq v_k}(k))| \\
       &\overset{\eqref{eq:351}}{\leq} \sum_{v_k \in X(k)} \frac{1}{\kappa n}  \big(| x(\mathsf{Op}_{v_\ell}(k))| + |x(\mathsf{Nb}_{v_\ell}(k))| + |\varsigma(X_{\geq v_\ell}(k+1))| \big)\\
       &\leq \sum_{ v_k \in X(k)} \sum_{\ell=0}^k \sum_{v_\ell \preceq v_k} \frac{2^{t(k-\ell)}}{(\kappa n)^{k-\ell+1}} \left( |x (\mathsf{Op}_{v_\ell}(k))| +  |\varsigma(X_{\geq v_\ell}(k+1))| \right)  \tag{similar to~\eqref{eq:x_vk}}\\
       &\leq  \frac{8^{t^2}}{\kappa^{k+1}}\left( \lambda|x| + (4 n)^t \frac{|x|^2}{r|X(k)|} \right) + \frac{8^{t^2}}{n\kappa^{k+1}}  |\varsigma|,  \numberthis\label{eq:553}
    \end{align*}
    where in the last expression, the first term follows from a similar calculation to~\Cref{eq:350}, and in the second term we used the fact that each $(k+1)$-face has a bounded number of subfaces as follows
    \begin{align*}
        &\sum_{ v_k \in X(k)} \sum_{\ell=0}^k \sum_{v_\ell \preceq v_k} \frac{2^{t(k-\ell)}}{(\kappa n)^{k-\ell+1}} |\varsigma(X_{\geq v_\ell}(k+1))|  \numberthis \label{eq:222}\\
        \leq & |\varsigma|\cdot \max_{f \in X(k+1) }\left( \sum_{\ell=0}^{k} \frac{2^{t(k-\ell)}}{(\kappa n)^{k-\ell+1}} \sum_{v_\ell \prec f} \sum_{v_k: v_\ell \preceq v_k} 1   \right)\\
        \leq & |\varsigma|\cdot \left( \sum_{\ell=0}^{k} \frac{2^{t(k-\ell)}}{(\kappa n)^{k-\ell+1}}  {{k+1} \choose {\ell}}  2^{k+1-\ell} \cdot {{t -\ell} \choose {k-\ell}} n^{k-\ell} \right) \tag{\Cref{lem:incidence}}\\
        \leq & |\varsigma|\cdot \left( \frac{2\cdot 2^{2t}}{\kappa n} \sum_{\ell=0}^{k} \frac{2^{(t+1)(k-\ell)}}{\kappa ^{k-\ell}} \right)\\
        \leq & \frac{8^{t^2}}{n\kappa^{k+1}}|\varsigma|  \tag{assuming $t \geq 2$.}
    \end{align*}
    Let $B_1 \triangleq  \frac{8^{t^2}}{\kappa^{k+1}}$ and $B_2 \triangleq \frac{(4n)^t}{r}$. We rearrange~\Cref{eq:553}:
    \begin{align*}
        (1-B_1\lambda) |x| - B_1 B_2 \frac{|x|^2}{|X(k)|} \leq \frac{B_1}{n} |\varsigma|.
    \end{align*}
    Using that $|\varsigma| \leq 2^t n |x|$ and the assumption $|x| \leq  \frac{(1-B_1\lambda)^2}{2^{t+2}B_1^2 B_2} |X(k)|$, we solve the above inequality and find that either
    \begin{align*}
    |x| &\leq \frac{1}{2}\left(\frac{1-B_1 \lambda}{B_1 B_2}|X(k)| - \sqrt{\left(\frac{1-B_1 \lambda}{B_1 B_2}|X(k)|\right)^2-\frac{4}{nB_2 }|\varsigma||X(k)|} \right)  \\
    & \leq \frac{2 B_1}{n(1-B_1 \lambda)}|\varsigma|,
    \end{align*}
    or 
    $$|x| \geq \frac{1}{2}\left(\frac{1-B_1 \lambda}{B_1 B_2}|X(k)| + \sqrt{\left(\frac{1-B_1 \lambda}{B_1 B_2}|X(k)|\right)^2-\frac{4}{nB_2 }|\varsigma||X(k)|} \right) \geq \frac{(1-B_1\lambda)}{2 B_1 B_2}|X(k)|.$$
    The latter solution is excluded by the lemma assumption.
    \end{proof}

    We now use~\Cref{lem:small-set-coLM-expansion} to conclude the decoder proof. Assuming $|\hat{e}| \leq \frac{1}{4} \frac{(8^{-t^2}\kappa^{k+1} - \lambda)^3}{2^{2t+2}(4n)^t} r |X(k)|$, then for every iteration $i$ we have that
    \begin{align*}
        |\sigma^{(i)}| < |\sigma^{(0)}| \leq  2^{t} n |\hat{e}| \leq \frac{1}{4} \frac{(8^{-t^2}\kappa^{k+1} - \lambda)^3}{2^{t+2}(4n)^t} nr |X(k)|.
    \end{align*}
    Since $|\hat{x}^{(1)}|$ is larger than $|\hat{x}^{(0)}|\leq |\hat{e}| \leq \frac{1}{2} \frac{(8^{-t^2}\kappa^{k+1} - \lambda)^2}{2^{t+2}(4n)^t} r |X(k)|$ by at most $(2n)^t$ (since~\Cref{alg:Zdecoder} only performs a small-set flip of size at most $(2n)^t$ per iteration),~\Cref{lem:small-set-coLM-expansion} implies that $|\hat{x}^{(1)}| \leq \frac{2}{n(8^{-t^2}\kappa^{k+1}-\lambda )} |\sigma^{(1)}| \leq \frac{1}{2} \frac{(8^{-t^2}\kappa^{k+1} - \lambda)^2}{2^{t+2}(4n)^t} r |X(k)|$. Repeating the same argument on $|\hat{x}^{(2)}|$ and so on, we conclude for all $i$ that
    \begin{align}
        |\hat{x}^{(i)}| \leq \frac{2}{n(8^{-t^2}\kappa^{k+1}-\lambda )} |\sigma^{(i)}|. \label{eq:near-end-sequential-Z}
    \end{align}

    On the other hand, $|\hat{x}^{(i)}| \leq \frac{2}{n(8^{-t^2}\kappa^{k+1}-\lambda )} |\sigma^{(i)}| \leq  \frac{2^{t+1}}{8^{-t^2}\kappa^{k+1}-\lambda }|\hat{e}|$ is strictly smaller than $\frac{(\varepsilon \kappa/2)^{k+1}8^{-t^2}-\lambda}{(4n)^t} r |X(k)|$ (which is a solution of~\Cref{eq:350} that we want to rule out) as long as $|\hat{e}| < \frac{((\varepsilon \kappa/2)^{k+1}8^{-t^2}-\lambda)(8^{-t^2}\kappa^{k+1}-\lambda)}{2^{t+1} (4n)^t} r |X(k)|$.
    
    Overall it suffices to have $|\hat{e}|< \frac{8^{-t^2}\kappa^{k+1}-\lambda}{2^{t+1} (4n)^t} r |X(k)| \cdot \min \left(\frac{(8^{-t^2}\kappa^{k+1} - \lambda)^2}{2^{t+3}}, (\frac{\varepsilon \kappa}{2})^{k+1}8^{-t^2}-\lambda  \right)$, proving the convergence of the Z-syndrome decoder.

    \paragraph{Linear runtime.} It is clear that the algorithm terminates in $O(|\sigma|)=O(|N_\mcQ|)$ iterations. However, as currently stated, each iteration takes linear time to search for a flip, so the total runtime would naively be quadratic. However, we can introduce a preprocessing step to resolve this, as done in the original classical bit-flip algorithm~\cite{sipser1996expander}. In particular, before starting the algorithm, we go over all vertices in $X(0)$ and create a set $S$ of `flippable' vertices. This step takes time $O(q^{(2n)^t \prod_{j=1}^t m_j} |X(k)|)$.  Then, in each iteration of the algorithm, we take a vertex $v$ in the set $S$, perform the flip on $X_{\geq v}(k)$, and update the set $S$ by checking the flippability of the neighbors of $v$. Here, the neighbors are vertices $v'$ such that $X_{\geq v}(k) \cap X_{\geq v'}(k) \neq \emptyset$ and there are at most $(4n)^t$ of them, hence each iteration takes time $O(q^{(2n)^t \prod_{j=1}^t m_j} (4n)^t)$. When $t, q, n, m_1,\hdots,m_t$ are constants, the total runtime of the modified algorithm is linear. \qedhere

\end{proof}

\subsubsection{Noisy syndrome case}
We next consider the noisy syndrome case. For simplicity, in this subsection we choose the parameter $\varepsilon=1$ since this does not lose much of the conceptual point of how the sequential decoder works under noisy syndrome measurements.

\begin{prop}[Sequential Z-syndrome decoder with noisy syndrome]
\label{prop:noisy-Zdecoder}
Consider the same setting in~\Cref{prop:noiseless-Zdecoder}. Let $\hat{e} \in C^{k}(X)$ be the cochain (over $\mathbb{F}_q$) corresponding to a Pauli X error. 
Let $\hat{\sigma} = \delta^k (\hat{e}) \in C^{k+1}(X)$ be the ideal syndrome. And let $\sigma = \hat{\sigma} + \hat{m}$ be the observed syndrome, where the measurement error $\hat{m} \in C^{k+1}(X)$ satisfies $|\hat{m}| \leq \frac{n}{4} \frac{(8^{-t^2}(\kappa/2)^{k+1}-\lambda)^2}{(4n)^t} r|X(k)|$ and $2^t n |\hat{e}|  + 2|\hat{m}| <  \frac{n}{4} \frac{(8^{-t^2} \kappa^{k+1} -\lambda)}{(4n)^t} r|X(k)| \cdot  \min  (\frac{(8^{-t^2} \kappa^{k+1} -\lambda)^2}{2^{t+2}} ,8^{-t^2}(\frac{\kappa}{2})^{k+1}-\lambda  )$. Then~\Cref{alg:Zdecoder} with parameter $\varepsilon=1$ runs in time $O(q^{n^{O(t)}} |X(k)|)$ and outputs a correction cochain $e$ such that $e + \hat{e}$ has stabilizer-reduced block-weight at most $\frac{2}{n(8^{-t^2}(\kappa/2)^{k+1}- \lambda)} |\hat{m}|$.
\end{prop}

\begin{proof}
 Similarly to the noiseless case proof, we define variables used in the analysis but unknown to the decoder and distinguish them with a hat symbol. Also note that we set $\varepsilon=1$ in this proof. Suppose we are given the noisy syndrome $\sigma^{(0)} = \hat{\sigma}^{(0)} + \hat{m}$, where $\hat{m} \in C^{k+1}(X)$ is the unknown measurement error and $\hat{\sigma}=\hat{\sigma}^{(0)}$ is the unknown true syndrome. Let $\hat{x}^{(0)}$ be a minimal cochain that is homologically equivalent to the underlying error $\hat{e}$, so that $\delta^{k}(\hat{x}^{(0)})=\delta^k(\hat{e})=\hat{\sigma}$ and note that $\hat{x}^{(0)}$ is also co-LM. Similarly define $\hat{x}^{(i)}$ to be the minimal version of $\hat{e}+ \sum_{i'=1}^i w^{(i')}$, where $w^{(i)}$ is the local flip performed at iteration $i$ of the algorithm when starting with the noisy syndrome $\sigma^{(0)}$. Also let $\hat{\sigma}^{(i)}= \delta^k (\hat{x}^{(i)})$ be the true syndrome at iteration $i$ and let $\sigma^{(i)}= \hat{\sigma}^{(i)} + \hat{m}$ be the syndrome that the decoder `thinks' it is correcting at iteration $i$.

Suppose the algorithm terminates at iteration $T$. Similarly to the argument before~\Cref{eq:348}, the termination condition implies that $|\sigma^{(T)}(X_{\geq v_\ell}(k+1)) + \delta_{\overline{S}} \hat{x}^{(T)}(X_{\geq v_\ell} (k))| \geq |\sigma^{(T)}(X_{\geq v_\ell}(k+1))|$.
Making the substitution $\sigma^{(i)}(X_{\geq v_\ell}(k+1)) = \hat{\sigma}^{(i)}(X_{\geq v_\ell}(k+1)) + \hat{m}(X_{\geq v_\ell}(k+1))$, we obtain 
\begin{align}
    |\hat{\sigma}^{(T)}(X_{\geq v_\ell}(k+1)) + \delta_{\overline{S}} \hat{x}^{(T)}(X_{\geq v_\ell} (k))| + 2 |\hat{m}(X_{\geq v_\ell}(k+1))| \geq |\hat{\sigma}^{(T)}(X_{\geq v_\ell}(k+1))|.
\end{align}
Then we follow the same derivation as before to obtain the following modification of~\Cref{eq:349} (note that we need to replace $\sigma^{(i)}$ by $\hat{\sigma}^{(i)}$ in~\cref{eq:345},~\cref{eq:346},~\cref{eq:347} because $\sigma^{(i)}$ in those equations was meant to be the true syndrome of $\hat{x}^{(i)}$):
 \begin{align}
       \kappa n |\hat{x}^{(T)}(X_{\geq v_\ell}(k))| \leq 2 (|\hat{x}^{(T)}(\mathsf{Op}_{v_\ell}(k))| + |\hat{x}^{(T)}(\mathsf{Nb}_{v_\ell}(k))| + |\hat{m}(X_{v_\ell}(k+1))| ).
        \label{eq:554}
    \end{align}
    
Next, we follow the derivation of~\Cref{eq:350} in~\Cref{prop:noiseless-Zdecoder}'s proof but this time replace~\Cref{eq:349} with~\Cref{eq:554} we obtain
\begin{align*}
    |\hat{x}^{(T)}| &= \sum_{v_k \in X(k)} |\hat{x}^{(T)}(X_{\geq v_k}(k))| \\
    &\overset{\eqref{eq:554}}{\leq} \sum_{v_k \in X(k)} \frac{2}{\kappa n}  \big(| \hat{x}^{(T)}(\mathsf{Op}_{v_\ell}(k))| + |\hat{x}^{(T)}(\mathsf{Nb}_{v_\ell}(k))| + |\hat{m}(X_{\geq v_\ell}(k+1))| \big)\\
    &\leq \sum_{ v_k \in X(k)} \sum_{\ell=0}^k \sum_{v_\ell \preceq v_k} \frac{2^{t(k-\ell)}}{(\kappa n/2)^{k-\ell+1}} \left( |\hat{x}^{(T} (\mathsf{Op}_{v_\ell}(k))| +  |\hat{m}(X_{\geq v_\ell}(k+1))| \right) \tag{similar to~\eqref{eq:x_vk}}\\
    &\leq  \frac{8^{t^2}}{(\kappa/2)^{k+1}}\left( \lambda|\hat{x}^{(T)}| + (4 n)^t \frac{|\hat{x}^{(T)}|^2}{r|X(k)|} \right) + \frac{8^{t^2}}{n(\kappa/2)^{k+1}}  |\hat{m}|, \numberthis\label{eq:555}
\end{align*}
where in the final expression the first term is the same as from~\Cref{eq:350} and the second term follows the same calculation in~\Cref{eq:222}. Using the promise $|\hat{m}| \leq \frac{n}{4} \frac{(8^{-t^2}(\kappa/2)^{k+1}-\lambda)^2}{(4n)^t} r|X(k)|$ and solving for~\Cref{eq:555} we find that either $x^{(T)} \leq \frac{2}{n(8^{-t^2}(\kappa/2)^{k+1}- \lambda)} |\hat{m}|$ or $x^{(T)} \geq \frac{(8^{-t^2}(\kappa/2)^{k+1}-\lambda)}{2 (4n)^t}r|X(k)|$ (similar to the solutions of~\Cref{eq:553}, but here we replace $\kappa$ by $\kappa/2$, $\hat{x}^{(i)}$ by $\hat{x}^{(T)}$, $\sigma^{(i)}$ by $\hat{m}$.)

Finally, to argue that $\hat{x}^{(T)} \leq \frac{2}{n(8^{-t^2}(\kappa/2)^{k+1}- \lambda)} |\hat{m}|$ must be the case, we can reuse~\Cref{eq:553} which is only a statement about co-LM cochains and their (noiseless) syndromes, repeated below (here, we replace from~\Cref{eq:553} $x$ by $\hat{x}^{(i)}$ and $\varsigma$ by $\hat{\sigma}^{(i)})$:

\begin{align*}
    |\hat{x}^{(i)}| \leq  \frac{8^{t^2}}{\kappa^{k+1}}\left( \lambda|\hat{x}^{(i)}| + (4 n)^t \frac{|\hat{x}^{(i)}|^2}{r|X(k)|} \right) + \frac{8^{t^2}}{n\kappa^{k+1}}  |\hat{\sigma}^{(i)}|.
\end{align*}
This time we note that 
\begin{align*}
    |\hat{\sigma}^{(i)}| \leq |\sigma^{(i)}|+|\hat{m}| < |\sigma^{(0)}|+|\hat{m}| \leq |\hat{\sigma}^{(0)}|+2|\hat{m}| \leq 2^t n |\hat{e}|  + 2|\hat{m}|.
\end{align*}
Thus, using $2^t n |\hat{e}|  + 2|\hat{m}| \leq \frac{n}{4} \frac{(8^{-t^2} \kappa^{k+1} -\lambda)^3}{2^{t+2} (4n)^t} r|X(k)|$ we can follow the same solutions to~\Cref{eq:553} in~\Cref{lem:small-set-coLM-expansion} and~\Cref{prop:noiseless-Zdecoder}'s proof and obtain that $|\hat{x}^{(i)}| \leq \frac{2}{n(8^{-t^2}\kappa^{k+1}-\lambda )} |\hat{\sigma}^{(i)}| \leq \frac{1}{2}  \frac{(8^{-t^2} \kappa^{k+1}-\lambda)^2}{2^{t+2} (4n)^t}r|X(k)|$ for all $i$. 
On the other hand we also want $|\hat{x}^{(T)}| < \frac{(8^{-t^2}(\kappa/2)^{k+1}-\lambda)}{2 (4n)^t}r|X(k)|$. 
Therefore, choosing $2^t n |\hat{e}|  + 2|\hat{m}| <  \frac{n}{4} \frac{(8^{-t^2} \kappa^{k+1} -\lambda)}{(4n)^t} r|X(k)| \cdot  \min  \left(\frac{(8^{-t^2} \kappa^{k+1} -\lambda)^2}{2^{t+2}} ,8^{-t^2}(\frac{\kappa}{2})^{k+1}-\lambda  \right)$ ensures that $\hat{x}^{(T)} \leq \frac{2}{n(8^{-t^2}(\kappa/2)^{k+1}- \lambda)} |\hat{m}|$, concluding the proof.
\end{proof}

\subsubsection{Parallel decoder}

In this subsection, we prove the parallel Z-syndrome decoder under noiseless syndrome in~\Cref{prop:noiseless-parallel-Z}. Then we will show in~\Cref{prop:noisy-parallel-Z} that the single-shot property essentially directly follows from the strong form of parallelization that we prove for the parallel decoder in~\Cref{lem:Z-parallel}.

The idea of the parallel decoder is to realize that many local flips in the sequential decoder may be performed simultaneously. However, we have to be careful when the upward links of two vertices $v,v' \in X(0)$ overlap. Hence, we schedule a parallel decoding round into sub-rounds such that in each of which the vertices under consideration do not have this issue.

\begin{claim} The set of vertices $X(0)$ can be partitioned into $\chi =(4n)^{t}+1$ subsets $X(0)= \bigsqcup_{s=1}^\chi V_{s}$ such any two vertices from different subsets have non-intersecting upward links.
\end{claim}
\begin{proof}
    According to~\Cref{lem:incidence}, a vertex $v \in X(0)$ has  $\binom{t}{k}n^{k} \leq (2n)^t$ faces in its upward link. Each of these faces contains at most $2^k \leq 2^t$ vertices. Hence, there are at most$(4n)^t$ other vertices whose upward links intersect with $v$'s. This intersection structure gives rise to a graph whose vertices are $X(0)$ and whose degree is at most $(4n)^t$. A greedy vertex coloring using $(4n)^t+1$ colors achieves the claim.
\end{proof}

\RestyleAlgo{boxruled}
\LinesNumbered
\begin{algorithm}[H]
    \setstretch{1.35}
    \caption{Z-syndrome parallel decoder (level $k$).} \label{alg:Z-parallel}
    \KwIn{Syndrome $\sigma = \delta^k (\hat{e}) \in C^{k+1}(X)$ of error $\hat{e} \in C^{k}(X)$. Number of decoding rounds $\tau$.}
    \KwOut{A proposed correction $w \in C^{k}(X)$.}

        (Initialization) Set $\sigma^{(0)}=\sigma$ and $w^{(0)}=0$;

        \For{$i \in [\tau \chi]$}{
        Set $s = i \mod \chi$ \tcc*{Color for current subround}

        Parallel for each vertex $v$ in $V_s$: Find local flip $w^\shortparallel_v \in C^{k}(X_{\geq v})$ such that $|\sigma^{(i)}|- |\sigma^{(i)}+\delta^k (w^\shortparallel_v )| \geq \frac{|\delta^k (w^\shortparallel_v)|}{2}$, and $\delta^k (w^\shortparallel_v)$  has maximal block-weight among such $w^\shortparallel_v$;
        
        Set $\sigma^{(i+1)} \leftarrow  \sigma^{(i)} + \bigsqcup_{v \in V_s} \delta^k (w^\shortparallel_v)$;
        
        Set $w^{(i+1)} \leftarrow  w^{(i)} + \bigsqcup_{v \in V_s} w^\shortparallel_v$;

        }

        Output $w=w^{(\tau\chi)}$.
\end{algorithm}

The performance guarantee of the parallel decoder rests on the following main lemma. It is important to emphasize that the condition $\frac{2}{\kappa} \lambda^{1/t} 8^{t} < \varepsilon < 4^{-(t+2)}$ below is satisfiable by choosing the underlying graphs in the cubical complex to be of sufficiently large degree such that $\lambda$ is sufficiently small (this can be done independently of $\kappa$, a universal constant depending only on $t$ from~\Cref{thm:DLV-robustness}).

\begin{lemma}\label{lem:Z-parallel}
     Let $\frac{2}{\kappa} \lambda^{1/t} 8^{t} < \varepsilon < 4^{-(t+2)}$. Suppose that the current syndrome $\sigma$ (before a round of the parallel decoder) is such that the sequential decoder with parameter $\varepsilon$ successfully corrects the error. Then after one round of parallel decoding (which consists of $\chi$ subrounds), the syndrome weight is reduced by $\frac{1}{16}\left(1 - 4^{t+2}\varepsilon \right)|\sigma|$.
\end{lemma}

\begin{proof}
    We begin by showing the existence of a `nice' sequence of flips $w^{(1)},...,w^{(T)}$ that form a valid execution of the parameter-$\varepsilon$ sequential decoder upon input $\sigma$. Here, `nice' means that these flips are disjoint, and many of them have large overlaps with the flips that are performed in the parallel decoding round. These properties allow us to prove the lemma.

    To describe this sequence, we recall some notations and use lemmas from the proof of~\Cref{prop:noiseless-Zdecoder}. Let $\hat{x}^{(0)}$ denote a minimal-block-weight cochain that is homologically equivalent to the underlying error $\hat{e}$, so that $\delta^k(\hat{x}^{(0)})= \delta^k(\hat{e}) =\sigma$. Let $\hat{e}^{(i)}=\hat{e} + \sum_{i'=1}^i w^{(i')}$ with $\hat{e}^{(0)}=\hat{e}$ and $\sigma^{(i)} = \sigma + \sum_{i'=1}^{i} \delta^k (w^{(i)})$, with $\sigma^{(0)}= \sigma$. Similarly let $\hat{x}^{(i)}$ be a minimal version of $\hat{e}^{(i)}$.

    Since $|\hat{x}^{(0)}| < \frac{r |X(k)|}{(4n)^t} \left( \frac{(\varepsilon \kappa/2)^{k+1}}{8^{t^2}} -\lambda \right) $ by assumption, we can apply~\Cref{lem:existence-local-flip} which states that there exist
    $\ell \leq k,\ v_\ell \in X(\ell)$ such that, letting $S= \mathsf{type}(v_\ell)$,
    \begin{align}
        |\sigma^{(0)}(X_{\geq v_\ell}(k+1))| - |\sigma^{(0)}(X_{\geq v_\ell}(k+1)) + \delta_{\overline{S}} \hat{x}^{(0)}(X_{\geq v_\ell} (k))| >  (1-\varepsilon)|\delta_{\overline{S}} \hat{x}^{(0)}(X_{\geq v_\ell} (k))|.
    \end{align}
    In other words, $w^{(1)} \triangleq  \hat{x}^{(0)}(X_{\geq v_\ell} (k))$ is a valid flip that could be performed by the sequential decoder. Next, we claim that a minimal weight version of $\hat{e}^{(1)} = \hat{e}^{(0)}+ w^{(1)}$ can be chosen to be $\hat{x}^{(1)} \triangleq \hat{x}^{(0)}+w^{(1)}$. This follows from the following claim.

    \begin{claim}\label{claim:block-reduced}
        Let us say $z \in C^k(X)$ is $|\cdot|$-block-reduced if its block-Hamming weight is minimal among all $k$-cochains homologically equivalent to it. For another $k$-cochain $z'$, we say $z' \subseteq z$ if $\supp(z') \subseteq \supp(z)$ and $z' = z (\supp(z'))$. If $z$ is $|\cdot|$-block-reduced and $z' \subseteq z$, then $z'$ is also $|\cdot|$-block-reduced.
    \end{claim}
    \begin{proof}[Proof of~\Cref{claim:block-reduced}]
        By definition, it holds for all $y \in \Im(\delta^{k-1})$ that
        \begin{align}
            0 \leq |z + y| - |z| = |A| - |B|,  
        \end{align}
        where $A$ denotes the faces in $\supp(y)\backslash \supp(z)$ and $B$ denotes the faces on which $y(B)=z(B)$. Let $A'$, $B'$ be the similar sets defined between $y$ and $z' \subseteq z$. Since $|A'| \geq |A|$ and $|B'| \leq  |B|$, it holds that
        \begin{align}
            0 \leq |A|-|B| \leq |A'| - |B'| = |z'+y| -|z'|,
        \end{align}
        implying the claim.
    \end{proof}
    
    Since $\hat{x}^{(0)}+w^{(1)} \subseteq \hat{x}^{(0)}$,~\Cref{claim:block-reduced} implies that $\hat{x}^{(0)}+w^{(1)}$ is a valid choice for $\hat{x}^{(1)}$. Repeating the same argument (applying~\Cref{lem:existence-local-flip} and~\Cref{claim:block-reduced}), we can choose $w^{(i)} \subseteq \hat{x}^{(i-1)}$ as valid flip at iteration $i$ and
    $\hat{x}^{(i)} = \hat{x}^{(i-1)} + w^{(i)}$ as a minimal block-weight chain homologically equivalent to $\hat{e}^{(i)}$. By construction, $\hat{x}^{(0)}= \sum_{i=1}^{T} w^{(i)}$, where $w^{(1)},...,w^{(T)}$ are disjoint, and each is supported on the upward link of some face $v_\ell^{(i)} \in X(\ell)$ for $\ell \leq k$ (and hence the upward link of some vertex $v_0^{(i)} \in X(0)$). Let $\nu^{(i)} \triangleq \delta^k(w^{(i)})$ be the associated updates to the syndrome. The update rule of the parameter-$\varepsilon$ sequential decoder gives 
    \begin{align}
        |\sigma| = \sum_{i=1}^{T} (|\sigma_{i-1}| - |\sigma_{i}|) \geq \sum_{i=1}^{T} (1-\varepsilon)|\nu^{(i)}|.
        \label{eq:4913}
    \end{align}
    On the other hand, a union bound gives $|\sigma| = |\delta^k(\hat{x}^{(0)})| \leq \sum_{i=1}^{T} |\nu^{(i)}|$, suggesting that the $\nu^{(i)}$ are mostly non-intersecting, a fact that we formalize in~\Cref{claim:good-indices} below.

    Consider for each $i$ a maximal $\nu'^{(i)} \subseteq \nu^{(i)}$ such that $\nu'^{(i)}$ does not intersect with any other $\nu^{(j)}$.  Note that $\nu'^{(i)}$ are disjoint and $\nu'^{(i)} \subseteq \sigma$ for all $i$. Let $G$ be the set of `good' indices $1\leq i \leq T$ such that
    \begin{align}
        |\nu'^{(i)}| \geq ( 1-\varepsilon')|\nu^{(i)}|
        \label{eq:11131}
    \end{align}
    for $\varepsilon' = \varepsilon \cdot 4^{t+1}$. The following claim asserts that the `good' changes cover a large portion of $\sigma$. 
    \begin{claim}\label{claim:good-indices}
        It holds for $c= \frac{\varepsilon'}{4^{t+2}}$ that
        \begin{align}
            |\bigsqcup_{i\in G} \nu'^{(i)} | \geq c\cdot |\sigma|.
            \label{eq:11132}
        \end{align}
    \end{claim}
    \begin{proof}[Proof of~\Cref{claim:good-indices}] 
        Suppose for contradiction that $\sum_{i \in G} |\nu'^{(i)}| < c |\sigma|$. Then,
        \begin{align}
            |\sigma| = |\sum_i \nu^{(i)}| \leq \sum_{i \in G}|\nu^{(i)}| +  |\sum_{i \in B} \nu^{(i)} | \leq  \frac{c}{1-\varepsilon'} |\sigma| + |\sum_{i \in B} \nu^{(i)} |.
            \label{eq:4916}
        \end{align}
        Our strategy will be to upper bound $|\sum_{i \in B} \nu^{(i)} |$ and arrive at a contradiction with~\Cref{eq:4913}. We have
        \begin{align}
            |\sum_{i \in B} \nu^{(i)} | &\leq \sum_{i \in B} |\nu^{(i)} | - \sum_{f \in X(k+1)} \mathbf{1}_{|\{j \in B: f \in \supp(\nu^{(j)}) \}| >1}\\
            &\leq  \sum_{i \in B} |\nu^{(i)} | - \frac{1}{4^t} \sum_{i \in B} \sum_{f \in \supp(\nu^{(i)})} \mathbf{1}_{|\{j \in B: f \in \supp(\nu^{(j)}) \}| >1}.
            \label{eq:4917}
        \end{align}
        Above, the last inequality follows from the fact that the flips $w^{(i)}$ are disjoint and each of them is supported in the upward link of some face $v_\ell^{(i)} \in X(\ell)$ for $\ell \leq k$. Since a face $f$ in $\nu^{(i)}$ contains less than $4^t$ sub-faces, it appears in at most $4^t$ other $\nu^{(j)}$. Next, we lower bound the sum over indicators
        \begin{align*}
            &\sum_{i \in B} \sum_{f \in \supp(\nu^{(i)})} \mathbf{1}_{|\{j \in B: f \in \supp(\nu^{(j)}) \}| >1}\\
            \equiv &\sum_{i \in B} |\{ f \in \nu^{(i)}: \exists j \in B\backslash \{i\} , f \in \supp(\nu^{(j)})\}| \\
            \geq &\sum_{i \in B} \left( |\{ f \in \nu^{(i)}: \exists j \neq i , f \in \supp(\nu^{(j)})\}| - |\{ f \in \nu^{(i)}: \exists j \in G, f \in \supp(\nu^{(j)})\}| \right) \\
            \geq &\sum_{i \in B} \left( \varepsilon' |\nu^{(i)}| - |\{ f \in \nu^{(i)}: \exists j \in G, f \in \supp(\nu^{(j)})\}| \right) \\
            \geq &\sum_{i \in B}  \varepsilon' |\nu^{(i)}| - \sum_{i \in B}  \sum_{f \in \supp(\nu^{(i)}) } \sum_{j \in G} \mathbf{1}_{f \in \supp(\nu^{(j)})}  \\
            = &\sum_{i \in B}  \varepsilon' |\nu^{(i)}| - \sum_{i \in B}  \sum_{f \in \supp(\nu^{(i)}) } \sum_{j \in G} \sum_{f' \in \supp(\nu^{(j)}) } \mathbf{1}_{f = f'} \\\
            = &\sum_{i \in B}  \varepsilon' |\nu^{(i)}| -  \sum_{j \in G} \sum_{f' \in \supp(\nu^{(j)}) } \left( \sum_{i \in B}  \sum_{f \in \supp(\nu^{(i)}) } \mathbf{1}_{f = f'} \right) \\
            \labelrel\geq{eq:112123} &\sum_{i \in B}  \varepsilon' |\nu^{(i)}| -  \sum_{j \in G} \sum_{f' \in \supp(\nu^{(j)}) } 4^t \\
            \geq &\sum_{i \in B}  \varepsilon' |\nu^{(i)}| -  \frac{c}{1-\varepsilon'}|\sigma| 4^t, \numberthis
            \label{eq:4918}
        \end{align*}
        where \eqref{eq:112123} is for the same reason as in~\Cref{eq:4917}.

        Substituting~\Cref{eq:4917} and~\Cref{eq:4918} into~\Cref{eq:4916} we obtain
        \begin{align}
            &|\sigma| \leq \frac{2 c}{1-\varepsilon'} |\sigma| + \sum_{i \in B} (1-\varepsilon'/4^t) |\nu^{(i)} | \\
            \Rightarrow & |\sigma| \leq \frac{(1-\varepsilon'/4^{t})}{1- 2c/(1-\varepsilon')} \sum_{i=1}^T |\nu^{(i)}| \leq (1-\varepsilon'4^{-t}/2) \sum_{i=1}^T |\nu^{(i)}| \leq (1-2\varepsilon) \sum_{i=1}^T |\nu^{(i)}|,
        \end{align}
        which contradicts~\Cref{eq:4913}. \qedhere

    \end{proof}

    The next claim states that the parallel decoding round produces a syndrome change that has a large overlap with each of the steps of the sequential decoder and its proof is inspired by~\cite{leverrier2023decoding}.

    \begin{claim}\label{claim:overlap-with-good-indices}
        The following holds for every $0\leq i \leq T$. Let $v$ be the vertex associated to the update $\nu^{(i)}$ in the sequential decoder upon input syndrome $\sigma$. Now consider one round of the parallel decoder. Let $\sigma^{\shortparallel}_v=\delta^k(w^\shortparallel_v)$ be the syndrome change proposed by vertex $v$ during the parallel decoding round (which could be zero). Let $U = \bigcup_{v \in X(0)} \supp(\sigma^\shortparallel_v)$.
        Then,
        \begin{align}
            |\nu'^{(i)}| - |U \cap \supp(\nu'^{(i)})| < \frac{3}{4} |\nu^{(i)}|.
            \label{eq:11133}
        \end{align}
    \end{claim}
    \begin{proof}[Proof of~\Cref{claim:overlap-with-good-indices}]
        Suppose for the sake of contradiction that $|\nu'^{(i)}| - |U \cap \supp(\nu'^{(i)})| \geq \frac{3}{4} |\nu^{(i)}|$.

        Let $z\in C^{k+1}(X_{\geq v})$ be the syndrome seen locally by vertex $v$ immediately before its turn to make an update in the parallel decoding step. We claim that the following holds
        \begin{align}\label{eq:overlap-good-indices}
            |z + \sigma^\shortparallel_v| - |z + \sigma^\shortparallel_v + \nu^{(i)}| \geq |\nu^{(i)}|/2.
        \end{align}
        Let us prove~\Cref{eq:overlap-good-indices}. First, since $\supp(\nu'^{(i)})\backslash U $ is left untouched throughout the parallel decoding round, we have 
        \begin{align}
            \label{eq:11134}
            \supp(\nu'^{(i)})\backslash U \subseteq \supp(z),
        \end{align}
        and furthermore,
        \begin{align}
            \nu'^{(i)}(\supp(\nu'^{(i)})\backslash U) =\sigma(\supp(\nu'^{(i)})\backslash U)= z(\supp(\nu'^{(i)})\backslash U).
            \label{eq:11135}
        \end{align}
        Similarly, since $\supp(\sigma^\shortparallel_v) \subseteq U$ by definition, we have
        \begin{align}
            \label{eq:11136}
            \supp(\nu'^{(i)})\backslash U \subseteq \supp(z + \sigma^\shortparallel_v).
        \end{align}
        Combining~\cref{eq:11134},~\cref{eq:11135}, and~\cref{eq:11136} and recalling that $\nu'^{(i)} \subseteq \nu^{(i)}$ we see that adding $\nu^{(i)}(\supp(\nu'^{(i)})\backslash U)$ to $z+ \sigma^\shortparallel_v$ removes at least $|\nu'^{(i)}| - |U \cap \supp(\nu'^{(i)})| \geq \frac{3}{4} |\nu^{(i)}|$ faces from it. This implies that adding $\nu^{(i)}$ to $z+ \sigma^\shortparallel_v$ removes at least $\frac{3}{4} |\nu^{(i)}|$ faces from $z+ \sigma^\shortparallel_v$ and therefore adds at most $\frac{1}{4} |\nu^{(i)}|$ faces to it, which in turn implies~\Cref{eq:overlap-good-indices}.

        On the other hand, the parallel decoding rule requires
        \begin{align}
            |z| - |z+\sigma^\shortparallel_v| \geq |\sigma^\shortparallel_v|/2.
            \label{eq:13784}
        \end{align}
        Combining~\Cref{eq:13784} and~\Cref{eq:overlap-good-indices} we obtain
        \begin{align}
            |z| - |z + \sigma^\shortparallel_v + \nu^{(i)}| \geq |\sigma^\shortparallel_v|/2+|\nu^{(i)}|/2 \geq |\sigma^\shortparallel_v + \nu^{(i)}|/2.
        \end{align}
        Note that $\sigma^\shortparallel_v + \nu^{(i)} \in \delta^k(x)$ for some cochain $x \in C^k(X_{\geq v})$ because each of $\sigma^\shortparallel_v $ and $ \nu^{(i)}$ has this form by construction. Thus, $\sigma^\shortparallel_v + \nu^{(i)}$ is also a valid syndrome update for the parallel decoder. Furthermore, $|\sigma^\shortparallel_v + \nu^{(i)}|> |\sigma^\shortparallel_v|$ since at least $\frac{3}{4}|\nu^{(i)}|$ faces in $\supp(\nu^{(i)})$ are not included in $\supp(\sigma^\shortparallel_v)$. This is a contradiction because we have instructed the parallel decoder to pick a local syndrome change $\sigma^\shortparallel_v$ with maximal weight. Therefore,~\Cref{claim:overlap-with-good-indices} holds.
    \end{proof}

    We now combine the two claims to prove the~\Cref{lem:Z-parallel}.~\Cref{claim:overlap-with-good-indices} implies that, for each good index $i \in G$,
    \begin{align}
        |U \cap \supp(\nu'^{(i)})|  > |\nu'^{(i)}| - \frac{3}{4} |\nu^{(i)}| \geq  \frac{1}{4}\left(1 - 4^{t+2}\varepsilon \right) |\nu^{(i)}|.
    \end{align}
    Since $\nu'^{(i)}$ are disjoint, the above expression implies that
    \begin{align}
        |U| \geq \frac{1}{4}\left(1 - 4^{t+2}\varepsilon \right) \sum_{i \in G} |\nu^{(i)}|.
    \end{align}
    Using $\sum_{i \in G}|\nu^{(i)}|\geq \sum_{i \in G} |\nu'^{(i)}|$ and applying~\Cref{claim:good-indices} we obtain that 
    \begin{align}
        |U| \geq \frac{1}{8}\left(1 - 4^{t+2}\varepsilon \right) |\sigma|.
    \end{align}
    Finally, notice that the total number of faces on which the syndrome is updated (where a face may be flipped by more than one sub-rounds of the parallel decoding round) is at least $|U|$. Hence, the update rule of~\Cref{alg:Z-parallel} implies that the decrease in the syndrome weight is at least $|U|/2$. This concludes the proof of~\Cref{lem:Z-parallel}.
\end{proof}

Having established~\Cref{lem:Z-parallel} we can now proceed to prove the parallel decoder performance guarantee.
\begin{prop}[Noiseless parallel Z-syndrome decoder]\label{prop:noiseless-parallel-Z} Consider the same setting as in~\Cref{prop:noiseless-Zdecoder}.
Let $\frac{2}{\kappa} \lambda^{1/t} 8^{t} < \varepsilon < 4^{-(t+2)}$. Let $\hat{e} \in C^{k}(X)$ be the cochain (over $\mathbb{F}_q$) corresponding to a Pauli X error such that $|\hat{e}|_R< \frac{8^{-t^2}\kappa^{k+1}-\lambda}{2^{t+1} (4n)^t} r |X(k)| \cdot \min (\frac{(8^{-t^2}\kappa^{k+1} - \lambda)^2}{2^{t+3}}, (\frac{\varepsilon \kappa}{2})^{k+1}8^{-t^2}-\lambda )$. Given the syndrome $\sigma = \delta^k (\hat{e}) \in C^{k+1}(X)$ for the error $\hat{e}$, then~\Cref{alg:Z-parallel} takes time $O(q^{n^{O(t)}})$ per round and achieves the following guarantees:
\begin{itemize}
    \item After $\tau = \log (\frac{2^{t+s+1}}{8^{-t^2}\kappa^{k+1}-\lambda})/\log (\frac{1}{15/16 + 4^t \varepsilon})$ rounds,~\Cref{alg:Z-parallel} outputs a correction cochain $w$ such that $w + \hat{e}$ has stabilizer-reduced block-weight at most $\frac{1}{2^s}|\hat{e}|_R$ (i.e., the error weight is reduced by a factor of $2^s$ after every $\tau$ rounds).
    \item As a consequence, if the number of rounds is $O(\tau \log |X(k)|)$, then~\Cref{alg:Z-parallel} outputs a correction cochain $w$ that is homologically equivalent to $\hat{e}$.
\end{itemize}
\end{prop}

\begin{proof}
    We first observe the following:
    after any parallel decoding round, if we were to replace the parallel decoder by the parameter-$\varepsilon$ sequential decoder from~\Cref{prop:noiseless-Zdecoder}, then the sequential decoder would successfully correct the remaining error. This can be seen by using a similar argument to the paragraphs near the end of~\Cref{prop:noiseless-Zdecoder}'s proof as follows. By construction, the syndrome block-weight is always at most $|\sigma^{(i)}| \leq 2^t n |\hat{e}|$ throughout the parallel decoding algorithm. Imagine serialize the parallel local flips within a parallel decoding round, then these flips themselves form a valid sequence of local flips that could have been executed by the sequential decoder (with parameter $\varepsilon=1/2$). Hence, we can use the small-set co-LM expansion property from~\Cref{lem:small-set-coLM-expansion} and a similar argument to~\Cref{eq:near-end-sequential-Z} to conclude that, after every parallel decoding round\footnote{In fact, this holds after every subround.}, the remaining error $\hat{e}+\sum_{j=1}^{i\chi} w^{(j)}$ (for all $i$) has stabilizer-reduced block-weight satisfying
    \begin{align}
        |\hat{x}^{(i\chi)}|   \leq \frac{2}{n(8^{-t^2}\kappa^{k+1}-\lambda )} |\sigma^{(i\chi)}| \leq \frac{r |X(k)|}{(4n)^t} \cdot \min \left(\frac{(8^{-t^2}\kappa^{k+1} - \lambda)^2}{2^{t+3}}, (\frac{\varepsilon \kappa}{2})^{k+1}8^{-t^2}-\lambda \right),
    \end{align}
    where $\hat{x}^{(i\chi)}$ is the minimal block-weight cochain that is homologically equivalent to $\hat{e}+\sum_{j=1}^{i\chi} w^{(j)}$.  This suffices to establish the convergence of the parameter-$\varepsilon$ sequential decoder as in~\Cref{prop:noiseless-Zdecoder}.

    We now apply~\Cref{lem:Z-parallel}. The initial syndrome block-weight is bounded by $|\sigma| \leq 2^t n |\hat{x}^{(0)}|$, so after $\tau$ rounds of parallel decoding the remaining syndrome is
    \begin{align}
        |\sigma^{(\tau \chi)}| \leq (1 - \frac{1}{16}(1 - 4^{t+2}\varepsilon))^{\tau} 2^t n |\hat{e}|.
    \end{align}
    It follows that
    \begin{align}
        |\hat{x}^{(i\chi)}|   \leq \frac{2^{t+1}}{(8^{-t^2}\kappa^{k+1}-\lambda )} \left(\frac{15}{16} + 4^{t}\varepsilon \right)^{\tau} |\hat{x}^{(0)}|.
    \end{align}
    Therefore, the error block-weight decreases by a factor of $2^s$ after every $\tau = \log (\frac{2^{t+s+1}}{8^{-t^2}\kappa^{k+1}-\lambda})/\log (\frac{1}{15/16 + 4^t \varepsilon})$ parallel decoding rounds. And $O(\tau \log(|X(k)|))$ rounds suffice to correct all errors.

    \paragraph{Runtime.} It is clear to see that each local flip can be bruteforce searched in time $q^{n^{O(t)}}$ (see the runtime analysis at the end of~\Cref{prop:noiseless-Zdecoder}'s proof), and hence each parallel round takes time $q^{n^{O(t)}} \cdot \chi = q^{n^{O(t)}}$.
\end{proof}

Finally we consider the noisy syndrome case. The single-shot property quite directly follows from the above strong form of parallelization. This simple proof also appears in~\cite[Section 6]{anshu2024circuit}, which is in turn inspired by~\cite{pippenger1985networks}. \footnote{A more careful analysis of~\Cref{lem:Z-parallel} in the noisy syndrome case can give better constants for the single-shot parallel decoder; e.g., we can avoid the exponential dependence in $q^{n^{O(t)}}$ in the residual error.}

\begin{prop}[Noisy parallel Z-syndrome decoder]\label{prop:noisy-parallel-Z} Consider the same setting as in~\Cref{prop:noiseless-Zdecoder}.
    Let $\frac{2}{\kappa} \lambda^{1/t} 8^{t} < \varepsilon < 4^{-(t+2)}$ and set $\tau = \log (\frac{2^{t+s+1}}{8^{-t^2}\kappa^{k+1}-\lambda})/\log (\frac{1}{15/16 + 4^t \varepsilon})$ and $E_0 = \frac{8^{-t^2}\kappa^{k+1}-\lambda}{2^{t+1} (4n)^t} r |X(k)| \cdot \min (\frac{(8^{-t^2}\kappa^{k+1} - \lambda)^2}{2^{t+3}}, (\frac{\varepsilon \kappa}{2})^{k+1}8^{-t^2}-\lambda )$. Let $\hat{e} \in C^{k}(X)$ be the cochain (over $\mathbb{F}_q$) corresponding to a Pauli X error such that $|\hat{e}|_R< E_0 $ and $\hat{\sigma} = \delta^k (\hat{e}) \in C^{k+1}(X)$ be the ideal syndrome. And let $\sigma = \hat{\sigma} + \hat{m}$ be the observed syndrome, where the measurement error $\hat{m} \in C^{k+1}(X)$ satisfies $\frac{1}{2^s}|\hat{e}|_R + 2^{O(\tau q^{n^{O(t)}})}  |\hat{m}| < E_0$. Then~\Cref{alg:Z-parallel} upon input the noisy syndrome $\sigma$ takes time $O(q^{n^{O(t)}})$ per round and achieves the following guarantees:
    \begin{itemize}
        \item After $\tau$ decoding rounds,~\Cref{alg:Z-parallel} outputs a correction cochain $w$ such that $|w +\hat{e}|_R \leq \frac{1}{2^s}|\hat{e}| + 2^{O(\tau q^{n^{O(t)}})} |\hat{m}| $.
        \item As a consequence, if the number of rounds is $O(\tau \log |X(k)|)$, then~\Cref{alg:Z-parallel} outputs a correction cochain $w$ such that $|w+\hat{e}|_R \leq 2^{O(\tau q^{n^{O(t)}})} |\hat{m}|$.
    \end{itemize}
\end{prop}

\begin{proof}   
   The main observation is that each decoding round of~\Cref{alg:Z-parallel} can be implemented by a circuit of depth $q^{n^{O(t)}}$ and hence we can use a lightcone argument in this circuit viewpoint. More specifically, the $\tau$-round guarantee from~\Cref{prop:noiseless-parallel-Z} can be interpreted as there exists a classical circuit $\mathsf{EC}_\tau$ with 2-bounded gates of depth $\mathsf{D}_\tau= O(\tau q^{n^{O(t)}})$ such that upon input the ideal syndrome $\hat{\sigma}$ it outputs a correction vector $w_1 \in C^k(X)$ with $|w_1+\hat{e}|_R \leq \frac{1}{2^s}|\hat{e}|_R$. Measurement errors correspond to input faults in $\mathsf{EC}_\tau$. These faults cause at most $|\hat{m}| 2^{O(\tau q^{n^{O(t)}})}$ faults in the output of $\mathsf{EC}_\tau$.
   
   Now we show the second item. Choosing $s=1$, we can repeat this $\tau$-round analysis under the assumption $|\hat{e}+w_1|_R \leq \frac{1}{2}|\hat{e}|_R + 2^{O(\tau q^{n^{O(t)}})} |\hat{m}| < E_0$ as follows. Before the next $\tau$ rounds, the decoder would ideally `see' the syndrome of $\hat{e}+w_1$, which is $\hat{\sigma} + \delta^{k}(w_1)$. However, what it actually sees is $\sigma + \delta^{k}(w_1)$, which differs from $\hat{\sigma} + \delta^{k}(w_1)$ by $\hat{m}$. Hence, we are in the same situation of decoding noisy syndrome with measurement noise $\hat{m}$ as before. The next $\tau$ rounds output a correction $w_2$ such that
   \begin{align}
    |\hat{e}+w_1+w_2|_R \leq \frac{1}{2}\left(\frac{1}{2}|\hat{e}|_R + 2^{O(\tau q^{n^{O(t)}})} |\hat{m}|\right) + 2^{O(\tau q^{n^{O(t)}})}|\hat{m}|.
   \end{align}
    And hence after $\Theta(\tau \log |X(k)|)$ rounds the stabilizer-reduced block-weight of the residual error is bounded by
   \begin{align}
        \frac{1}{2^{\Theta(\log |X(k)|)}}|\hat{e}|_R + (1+ \frac{1}{2}+\hdots+\frac{1}{2^{\Theta(\log |X(k)|)}}) 2^{O(\tau q^{n^{O(t)}})}|\hat{m}| \leq 3 \cdot 2^{O(\tau q^{n^{O(t)}})} |\hat{m}|.
   \end{align}
\end{proof}

\subsection{X-syndrome decoder}\label{sec:X-decoder}
In this subsection, we present the X-syndrome decoder in~\Cref{alg:Xdecoder}. Due to the asymmetry between the X and Z directions in the cubical complex construction, we cannot simply use the Z-syndrome decoder to decode X syndromes.
The main idea in the X-decoder is a reduction to the Z-decoding problem in a related code. This idea was introduced in~\cite{dinur2022good}, and our decoder is a natural generalization of theirs to higher-dimensional cubical complexes with $t\geq 2$. We further show that our decoder is single-shot. We will first define some auxiliary tools and lemmas some of which were introduced in~\cite{dinur2024expansion}. We then describe the X-syndrome decoder and prove its correctness, single-shot property, and time complexity. Finally, we show how the decoder can be parallelized.

\subsubsection{Auxiliary definitions and lemmas}

In~\cite{dinur2022good}, the X decoder is obtained via a reduction to the Z decoder on a `dual' chain complex. Similarly, the distance proof in~\cite{dinur2024expansion} is also done via the same reduction.

\begin{definition}\label{def:dual-chain-complex}
    The dual chain complex $\Tilde{C}(X, \Tilde{\mathcal{F}})$ is defined similarly to $C(X,\mathcal{F})$, except that the matrices $h_1,\hdots,h_t$ are replaced by $h_1^\perp,\hdots, h_t^\perp$, where $h_i^\perp \in \mathbb{F}_q^{k_i \times n_i}$ is a full-row-rank parity check matrix for the dual code, i.e., a generating matrix for the code $\ker(h_i)$. Accordingly, the local vector space associated to a face of type $S$ is $\mathbb{F}_q^{\prod_{j \in \overline{S}} [k_j] }$ (in contrast to $\mathbb{F}_q^{\prod_{j \in \overline{S}} [m_j] }$ in the original $C(X,\mathcal{F})$.)
\end{definition}

We begin by introducing yet other new chain complexes by attaching various kinds of local sheaves to the cubical complex $X$. These sheaves correspond to local information of the faces in the original complex $C(X,\mcF)$, and will facilitate proving the mentioned reduction.

\begin{definition}[Upward-view chain complex] Let $0 \leq k \leq t$.
\begin{itemize}
    \item (Upward-view sheaf) For each face $f \in X(k')$ with $k'\leq k$, we define a sheaf for $f$ by $\mathcal{F}_k(f)=C_k(X_{\geq f})$. In other words, the new sheaf is an `aggregation' of all sheaves associated to the $k$-faces in the upward link of $f$ in the original complex $C(X,\mathcal{F})$. The elements of $\mathcal{F}_k(f)$ can themselves be viewed as $k$-chains of the original complex $C(X,\mcF)$, whose support is limited to $X_{\geq f}(k)$. If $k' > k$, we define $\mathcal{F}_k(f)=\{0\}$.
    Furthermore, recall that~\Cref{lemma:local-chain} asserts that space $\mathcal{F}_k(f)=C_{k}(X_{\geq f})$ for a face $f$ of type $S$ is also isomorphic to $C_k(L_{\overline{S}})$.
    \item (Upward-view complex) We denote the resulting cochain complex as (\textbf{boldface}) $\bs C_{ k}(X, \mathcal{F}_k)$. Note that such a cochain complex is defined for each $k$, and hence the extra subscript $k$\footnote{In other words, the subscript $k$ does \emph{not} label the chain spaces of the complex. This is not an issue because we will only use the cochain spaces of them, which are labeled by a superscript.}.
    \item (Coboundary map) The $i$-th coboundary map (\textbf{boldface}) $\bs \Delta^i_{ k}: \bs C^i_{ k}(X, \mathcal{F}_k) \mapsto \bs C^{i+1}_{ k}(X, \mathcal{F}_k) $ for $\bs C_{ k}(X, \mathcal{F}_k)$ is defined in a natural way. For a cochain $y \in \bs C^i_{ k}(X, \mathcal{F}_k)$, the component of $\bs \Delta^i_{ k} (y)$ on each $f' \in X(i+1)$, is the sum over $f \precdot f'$ of the restrictions of $y(f) \in \mathcal{F}_k(f)=C_k(X_{\geq f})$ onto $X_{\geq f'}(k)$, i.e.,
    \begin{align}
        \bs \Delta^i_{ k}(y)(f') = \sum_{f \precdot f'} y(f) \big|_{X_{\geq f'}(k)}.
        \label{eq:local-view-sheaf}
    \end{align}
\end{itemize}
Similar to before, the block-Hamming weight for $y \in \bs C^i_{ k}(X, \mathcal{F}_k)$, denoted $|\cdot|$, is defined to be the number of faces on which $y$ is nonzero.
\end{definition}

The above complex can be interpreted as follows. The sheaf $\mathcal{F}_k(f)$ contains the face $f$'s `opinion' about its upward link in the original chain complex $C(X,\mathcal{F})$. Let us show that $\bs \Delta^i_{ k}$ are indeed valid coboundary maps (compose to zero). In fact, they further satisfy the exactness property. The following claim is [Lemma 7.4 in~\cite{dinur2024expansion}], and we include its proof for completeness and to explain the added statement about time complexity (which was not stated in their paper).

\begin{claim}
\label{claim:exactness-agg}
For each $1 \leq i \leq k$, the map $\bs \Delta^i_{ k}$ is exact on $i$-chains. In other words, for any $y \in \bs C^i_{ k}(X, \mathcal{F}_k)$ such that $\bs \Delta^i_{ k}(y)=0$, there exists a $z \in \bs C^{i-1}_{ k}(X, \mathcal{F}_k)$ such that $\bs \Delta^{i-1}_{ k} (z)=y$. Furthermore, such $z$ exists with $|z|\leq (4n)^t |y|$ and can be constructed in sequential linear time $O(q^{n^{O(t)}}|X(k)|)$ or in parallel time $O(q^{n^{O(t)}})$.
\end{claim}
\begin{proof} The main observation is that $\bs \Delta^i_{ k}$ is a direct sum of local maps corresponding to the \emph{downward} link of the $k$-faces. In particular, note that $\bs C^i_{ k}(X, \mathcal{F}_k)$ can be decomposed as follows
\begin{align*}
    \bs C^i_{ k}(X, \mathcal{F}_k) = \bigoplus_{f \in X(i)} C_k(X_{\geq f}) =  \bigoplus_{f \in X(i)} \left(\bigoplus_{u \in X(k): f \preceq u} V_u  \right) = \bigoplus_{u \in X(k)} C^i(X_{\leq u}, V_u),
\end{align*}
where we recall $X_{\leq u}$ denotes the faces in the downward link of $u$ and here $C^i(X_{\leq u}, V_u) \triangleq \bigoplus_{f \in X(i): f \preceq u} V_u$ stores the opinion of each $i$-face in $X_{\leq u}$ about the value on the sheaf $V_u$. Note that $C^i(X_{\leq u}, V_u) \cong V_u^{\oplus |X_{\leq u}(i)|}= V_u^{\oplus {k \choose i} 2^{k-i}}$ according to~\Cref{lem:incidence}. See an illustration of the decomposition below.
\begin{center}
\includegraphics[width=0.5\linewidth]{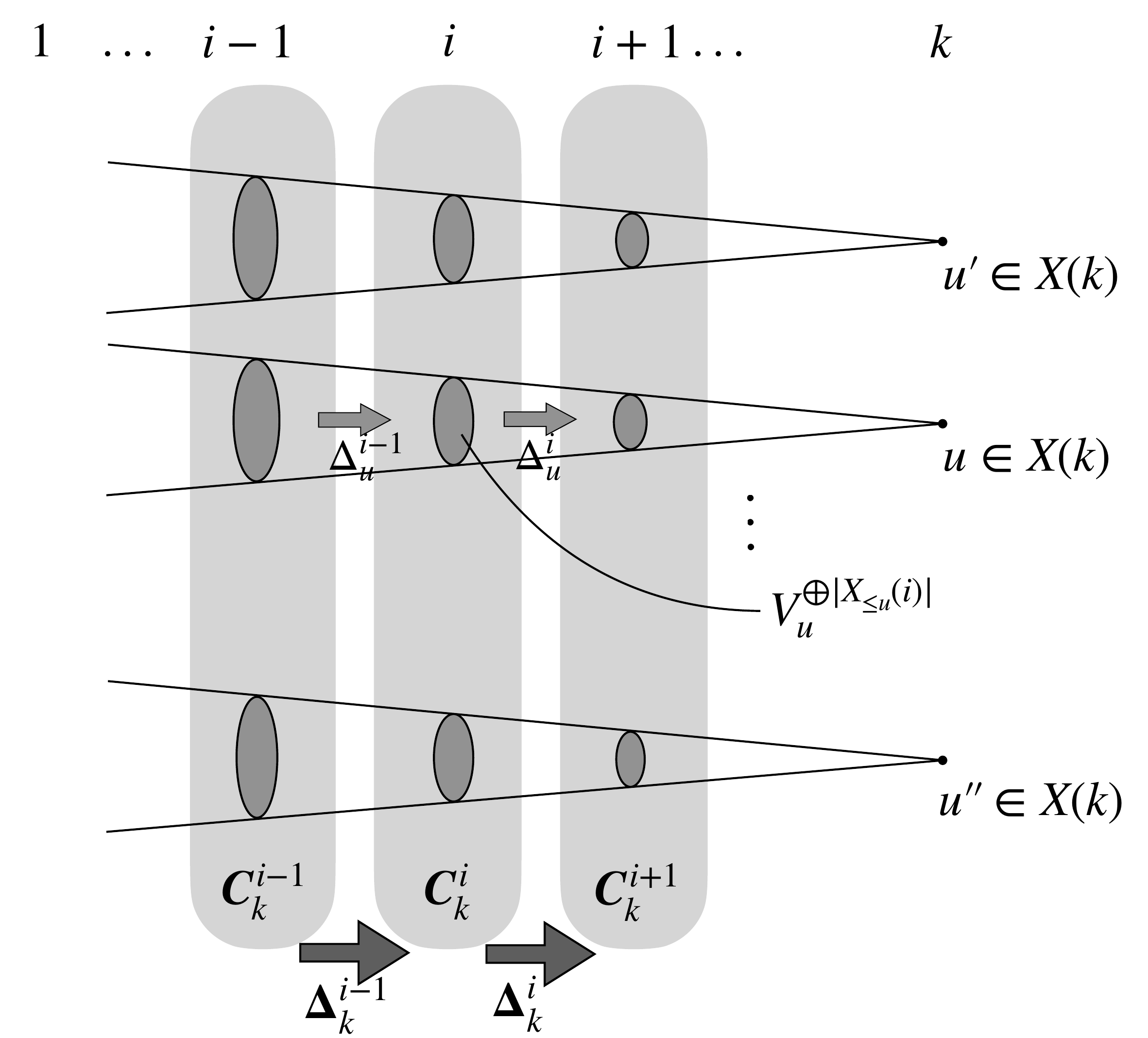}
\end{center}
Therefore, the map $\bs \Delta^i_{ k}$ decomposes as a direct sum $\bs \Delta^i_{ k} = \bigoplus_{u \in X(k)} \Delta^i_u$ of local maps $\Delta^i_u: C^i(X_{\leq u}, V_u) \mapsto C^{i+1}(X_{\leq u}, V_u)$. Hence, it suffices to verify that $\Delta_u^{i+1}\circ \Delta_u^i =0$ and that $\Delta_u^i$ is exact.

The cochain complex $C^i(X_{\leq u}, V_u)$ is obtained by attaching the sheaf $V_u$ to each face of a ``Hypercube'' complex. Let $S = \mathsf{type}(u)$. For each $i \in S$, define a 2-term cochain complex $C^*(H_{\{i\}})$ over $V_u$, where $H_{\{i\}}(0) = \{0,1\}$, $H_{\{i\}}(1)=\{e_i\}$ (a single element) and the coboundary map $\delta_{\{i\}}: C^0(H_{\{i\}}) \mapsto C^1(H_{\{i\}})$ is defined by $(\delta_{\{i\}} x)(e_i) = x(0) + x(1) \in V_u$. Then the hypercube complex $C^*(H_S)$ is obtained by taking the homological product of $C^*(H_{\{i\}})$ over $i \in S$, i.e.,  $C^\ell(H_S) = \bigoplus_{T \subseteq S: |T|=\ell}  \otimes_{j \in T} C^1(H_{\{j\}}) \otimes_{j \in S-T} C^0(H_{\{j\}})$. It has a single $|S|$-dimensional face which is the hypercube $u$, and the lower-dimensional faces are subfaces of $u$. The map $\Delta_u$ is isomorphic to the coboundary map of this hypercube construction, so $\Delta_u \circ \Delta_u=0$ at all levels. Next, we verify the exactness property. The cohomology groups of $C^*(H_{\{i\}})$ have dimensions $\operatorname{dim} H^0(H_{\{i\}}) = 1$ and $\operatorname{dim} H^1(H_{\{i\}}) = 0$. Hence, the Künneth formula gives $\operatorname{dim} H^j(H_S) = 0$ for all $0 < j \leq |S|$. In other words, $C^*(H_S)$ is exact at all levels $i > 0$.

Finally, provided with $y$ such that $\bs \Delta^i_{ k}(y)=0$, we can construct a preimage $z$ by going over each face $u \in X(k)$ and using the exactness of $\Delta_u$ to locally find $z[u]$ from $y[u]$ (here $z[u]$ denotes the restriction of $z$ onto $C^i(X{\leq u}, V_u)$). Since the hypercube complex has size $O(2^t \prod_{j=1}^t m_j)= n^{O(t)}$, we can simply bruteforce for each face $u$ in time $O(q^{n^{O(t)}})$\footnote{This calculation can actually be done faster than bruteforce since it is equivalent to solving a linear system over $\mathbb{F}_q$.} and hence $z$ can be constructed in sequential linear time $O(q^{n^{O(t)}}|X(k)|)$ or parallel constant time $O(q^{n^{O(t)}})$. The block-Hamming weight of the constructed $z$ can be straightforwardly bounded
\begin{align*}
    |z| \leq \sum_{u \in X(k)}  |z[u]| \leq \sum_{u \in X(k)} 2^t |y[u]| \leq \sum_{f \in X(i)} \sum_{u \in X(k): f \preceq u} 2^t |y(f)[u]| \leq \sum_{f \in X(i)} 2^t (2n)^t |y(f)| = (4n)^t|y|,
\end{align*}
where the last inequality uses $|X_{\geq f}(k)| = {{t-i}\choose {k-i}} n^{k-i} \leq 2^t n^t$ from~\Cref{lem:incidence}.
\end{proof}

Since the sheaf $\mathcal{F}_k(f)$ contains the face $f$'s `opinion' about its upward view in the original chain complex, the coboundary maps $\bs \Delta_{ k}$ serve as a measure of the inconsistency among these opinions. This interpretation can be made formal by the following claim from Claim 7.3 in~\cite{dinur2024expansion}. Again, we include the proof to explain the added statement about time complexity.

\begin{claim}
\label{claim:consistent} Let $0\leq k \leq t$.
Let $z \in \bs C^0_{ k}(X,\mathcal{F}_k)$ be such that $\bs \Delta^0_{ k}(z)=0$. Then there is a chain $z' \in C_k(X)$ in the original complex such that $z' \big |_{X_{\geq v}(k)}=z_v$ for all $v \in X(0)$. In other words, $z$, viewed as a collection of upward local views from each vertex $v \in X(0)$, corresponds to a valid global chain in the original $C_k(X, \mcF)$. Furthermore, there exists such a $z'$ with $|z'| \leq (2n)^t|z|$ and it can be constructed in sequential linear time $O((\log_2 q) n^{O(t)} |X(k)|)$ or in parallel constant time $O((\log_2 q) n^{O(t)})$.
\end{claim}

\begin{proof} The case $k=0$ is trivial because the $X_{\geq v}(0)=\{v\}$ and indeed $\bs \Delta^0_{ 0}(z)=0$ for any $z \in \bs C^0_{0}(X)$ by definition. So we only need to consider $k\geq 1$.

    We construct $z'$ as follows. For each $k$-face $f \in X(k)$ we choose an arbitrary vertex $v \in X(0)$ such that $v \prec f $ and define $z'(f)=z_v(f)$. Clearly, this can be done in sequential time $O((\log_2 q) n^{O(t)} |X(k)|)$ and in parallel time $O((\log_2 q) n^{O(t)})$. Let us show that the chain $z'$ defined in this way is independent of the choice of $v$. Indeed, consider two vertices $v,v' \prec f$. If $\{v,v'\}$ forms an edge $e \in X(1)$, then it holds that
    \begin{align*}
        \bs \Delta_{  k}^{0}(z)(e)=z_v(f)+z_{v'}(f)=0.
    \end{align*}
    So $z_v(f)=z_{v'}(f)$. On the other hand, if $v, v'$ are not neighbors, then there is still an edge-path $v=v_0, v_1,\hdots,v' $ between them. The assumption $\bs \Delta_{  k}^{0}(z)=0$ gives $z_{v}(f)=z_{v_1}(f)=\hdots=z_{v'}(f)$.

    The stated block-Hamming weight bound on $z'$ follows from~\Cref{lem:incidence}, which asserts that each upward link $X_{\geq v}(k)$ contains at most ${{t}\choose {k}} n^{k} \leq (2n)^t$.
\end{proof}

The final definition needed for the X-decoder proof is a transformation from $\bs C_{ k}(X,\mathcal{F}_k)$ to $\bs C_{ k-1}(X,\mathcal{F}_{k-1})$. Recall from the discussion around~\Cref{lemma:local-chain} about the local chain structure of $C(X,\mathcal{F})$. In particular,~\Cref{lemma:local-chain} asserts that space $\mathcal{F}_k(f)=C_{k}(X_{\geq f})$ for a face of type $S$ is isomorphic to $C_k(L_{\overline{S}})$. Thus, there is a naturally defined map (\textbf{boldface}) $\bs \partial_{L}: \bs C_{ k}(X,\mathcal{F}_k) \mapsto \bs C_{ k-1}(X,\mathcal{F}_{k-1})$ that applies \emph{face-wise} the maps $\partial_{\overline{\mathsf{type}(f)}}: \mathcal{F}_k(f) \mapsto \mathcal{F}_{k-1}(f)$ from~\Cref{lemma:local-chain} (the subscript `L' in $\bs \partial_L$  refers to `local', and we suppress the dependence on the cochain level).

\begin{claim}
\label{claim:partialL}
The following properties hold for $\bs \partial_L$:
\begin{itemize}
    \item $\bs \partial_L \circ \bs \partial_L=0$,
    \item $\bs \partial_L$ is exact. In other words, if $z \in \bs C^{i}_{ k}(X,\mathcal{F}_{k})$ such that $\bs \partial_L(z)= 0$ then there exists $y \in \bs C^{i}_{k+1}(X,\mathcal{F}_{k+1})$ such that $z= \bs \partial_L(y)$ and $|y| \leq |z|$, and such $y$ can be found in sequential time $O(q^{n^{O(t)}}|X(k)|)$ or parallel time $O(q^{n^{O(t)}})$,
    \item The coboundary maps $\bs \Delta^i_{ k}$ and $\bs \partial_L$ commute. In other words, the following diagram is commutative.
\begin{figure}[H]

\begin{center}
		\begin{tikzcd}[arrows=rightarrow]
    \bs C^{i+1}_{ k}(X,\mathcal{F}_{k}) \ar[r,"\bs \partial_L"]  & \bs C^{i+1}_{ k-1}(X,\mathcal{F}_{k-1})  \\
    
    \bs C^i_{ k}(X,\mathcal{F}_{k}) \ar[u,"\bs \Delta^i_{ k}"]\ar[r,"\bs \partial_L"]  &\bs C^{i}_{ k-1}(X,\mathcal{F}_{k-1})   \ar[u,"\bs \Delta^i_{ k-1}"]
\end{tikzcd}
.
\end{center}
    \caption{The double complex formed from the upward-view complexes $\bs C_{ k}(X,\mathcal{F}_{k})$.}
    \label{fig:commutative}
\end{figure}
\end{itemize}
\end{claim}
\begin{proof}
    The first item is because $\partial_{\overline{\mathsf{type}(f)}} \circ \partial_{\overline{\mathsf{type}(f)}} = 0$ for each face $f$. 
    The second item is because of the exactness of the local chain shown in~\Cref{lem:exactness-tensorcode}. In particular, for each $k \geq 1$  and each face $f \in X(k)$ we can bruteforce for a preimage within $X_{\geq f}$ which contains at most $(4n)^t$ faces, hence the claimed time complexity.
    The third item is shown in~\cite[Lemma 7.5]{dinur2024expansion}.
\end{proof}

\subsubsection{Sequential decoder}
As in~\cite{dinur2022good}, the X-syndrome decoder is obtained via a reduction to the Z-syndrome decoder on the dual chain complex $\Tilde{C}(X,\Tilde{\mathcal{F}})$. We first describe the noiseless syndrome case.

\RestyleAlgo{boxruled}
\LinesNumbered
\begin{algorithm}[H]
    \setstretch{1.35}
    \caption{X-syndrome decoder (level $k$). \label{alg:Xdecoder}
    \\\textbf{Mode}: sequential or $\tau$-round parallel}
    \KwIn{Syndrome $\sigma = \partial_k (\hat{e}) \in C_{k-1}(X)$ of error $\hat{e} \in C_{k}(X)$}
    \KwOut{A proposed correction $g \in C_{k}(X)$.}

    (Step 1: Local guesses) For each vertex $v \in X(0)$, find a minimal-weight local chain $g_v \in C_k (X_{\geq v})$ consistent with local syndrome $\sigma(X_{\geq v}(k-1))$;

    (Step 2: Compute syndrome explanation sequence) View $\{g_v\}_v$ as a cochain $g^{(0)} \in \bs C^0_{k}(X)$ in the upward-view complex. Compute the sequence $g^{(i)} \in \bs C^{i}_{ k+i}(X),  \sigma^{(i+1)}= \bs \Delta^i_{k+i}(g^{(i)})$ defined in~\Cref{fig:tracing} up to $i=t-k$;

    (Step 3: Z-syndrome decoding on dual complex)

    \Begin{
    Map $\sigma^{(t-k+1)}$ to a $(t-k+1)$-cochain $\tilde{\sigma}^{(t-k+1)}$ in the dual complex $\tilde{C}(X,\tilde{\mathcal{F}})$.
    
    Call Z-syndrome decoder (\Cref{alg:Zdecoder} with parameter $\varepsilon=1$ \footnote{The choice $\varepsilon=1$ is only for simplicity.} if sequential mode and~\Cref{alg:Z-parallel} with $\tau$ rounds if parallel mode) on $\tilde{\sigma}^{(t-k+1)}$ to obtain a $(t-k)$-cochain $\tilde{C}(X,\tilde{\mathcal{F}})$. 
    
    Map this cochain back to the upward-view complex to obtain a `correction' $w^{(t-k)} \in \bs C^{t-k}_{t}(X)$.
    }

    (Step 4: Reverse syndrome explanation sequence) Set $g'^{(t-k)}=g^{(t-k)}+ w^{(t-k)}$, then reverse the syndrome sequence computation in~\Cref{fig:tracing} and obtain $g'^{(0)} \in \bs C^0_{k}(X)$;
    
    Output a correction $g \in C_{k}(X)$ according to $g'^{(0)}$.

    \textbf{Note 1:} in noiseless syndrome case, the algorithm may be allowed to end early in lines 1 or 2 (see Steps 1 and 2 in~\Cref{prop:noiseless-Xdecoder}'s proof).

    \textbf{Note 2:} in noisy syndrome case, lines 1, 2, 5, 9, 10 need a modification, see~\Cref{prop:noisy-Xdecoder}.
\end{algorithm}

\begin{prop}[Sequential X-syndrome decoder with noiseless syndrome]
\label{prop:noiseless-Xdecoder} 
    Consider a $t$-dimensional cubical complex construction where the underlying graphs $\operatorname{Cay}(G,A_j)$ are $\lambda$-expanding up to size $r|G|$ and the local codes $\{h_1,\hdots,h_t\}$ are two-way $\kappa$-robust.
        Let $1\leq k \leq t-1 $ and consider the quantum code where qubits correspond to level $k$ of the chain complex, similar to described in~\Cref{thm:DLVmain}. Let $\hat{e} \in C_{k}(X)$ be the chain (over $\mathbb{F}_q$) corresponding to a Pauli Z error such that $|\hat{e}|_R \leq (nt)^{-O(t^2)} (8^{-t^2}\kappa^{t-k+1}-\lambda)r |X(t-k)| \cdot \min (\frac{(8^{-t^2}\kappa^{t-k+1} - \lambda)^2}{2^{t+3}}, (\frac{\kappa}{2})^{t-k+1}8^{-t^2}-\lambda )$. 
        Given the syndrome $\sigma = \partial_k (\hat{e}) \in C_{k-1}(X)$ for the error $\hat{e}$, then~\Cref{alg:Xdecoder} in sequential mode runs in time $O(q^{n^{O(t)}} |X(k)|)=O(q^{n^{O(t)}} |X(t-k)|)$ and outputs a correction cochain $e$ that is homologically equivalent to $\hat{e}$.
\end{prop}

\begin{proof}
    As in the Z-syndrome decoder proof, we keep track of variables used in the proof but unknown to the decoder by a hat symbol. For example, $\hat{e} \in C_k(X)$ is the unknown error. Wlog, we assume $\hat{e}$ is (block-Hamming) stabilizer-reduced.

    \paragraph{Step 1: Local guesses.} In the first step, we go over each vertex $v$ in $X(0)$ and locally find the minimal block-Hamming weight chain $g_v \in C_k (X_{\geq v})$ that is consistent with the observed syndrome $\sigma(X_{\geq v}(k-1))$ on $X_{\geq v}(k-1)$. If the syndrome is noiseless, such a $g_v$ exists for each vertex $v$. Furthermore, since $X_{\geq v}$ has size at most $n^{O(t)}$, this step is a local bruteforce in sequential time $O(q^{n^{O(t)}} |X(0)|)$ or parallel time $O(q^{n^{O(t)}})$.
    
    Let $S=\mathsf{type}(v)$. By construction, we have $ \partial_{\overline{S}}(\hat{e}(X_{\geq v}(k))) =  \partial_{\overline{S}} (g_v)=\sigma(X_{\geq v}(k-1))$. So according to the exactness of $ \partial_{\overline{S}}$ (\Cref{lem:exactness-tensorcode}),
    there exists an (unknown to decoder) chain $\hat{z}_v \in C_{k+1}(X_{\geq v})$ such that 
    \begin{align}
        g_v + \partial_{\overline{S}} (\hat{z}_v) &= \hat{e}(X_{\geq v}(k)).\label{eq:4437}
    \end{align}
    Let $g^{(0)}$ denote the cochain in $\bs C_{ k}^{0}(X)$
    corresponding to $\{g_v\}_{v \in X(0)}$. Similarly, let $\hat{x}^{(0)} \in \bs C_{ k}^{0}(X)$ denote the cochain representing the local views of $\hat{e}$ from the vertices. So~\Cref{eq:4437} can be rewritten as
    \begin{align}
        \bs \partial_L (\hat{z}^{(0)}) = g^{(0)} + \hat{x}^{(0)},
        \label{eq:4440}
    \end{align}
    and hence,
    \begin{align}
        \bs \partial_L(g^{(0)})= \bs \partial_L(\hat{x}^{(0)})= \sigma^{(0)}.
        \label{eq:4441}
    \end{align}
    We note the following block-Hamming weight bounds
    \begin{align}
        |g^{(0)}| \leq |\sigma^{(0)}| \leq  |\hat{x}^{(0)}|\leq 2^t |\hat{e}|,
    \end{align}
where the first inequality is because $g_v=0$ if $\sigma_v=0$ and the last inequality is because each $k$-face participates in the upward link of at most $2^{k} \leq 2^t$ vertices (\Cref{lem:incidence}).
    
    Now, if it turns out that $\{g_v\}_{v \in X(0)}$ are consistent, i.e., they correspond to the local views of a valid global chain in $C_k(X, \mcF)$, then we are done. Acording to~\Cref{claim:consistent}, this is equivalent to the condition $\bs \Delta^{0}_{ k} (g^{(0)}) = 0$, and we can  obtain a global guess $g \in C_k(X)$ in time $O((\log_2q)n^{O(t)}|X(k)|)$ or in parallel time $O((\log_2q)n^{O(t)})$, such that $\partial_k(g) = \partial_k (\hat{e}) = \sigma$ and $|g| \leq (2n)^t|g^{(0)}|$. Furthermore, the guess $g$ and the underlying error $\hat{e}$ will be homologically equivalent. This is because
    \begin{align*}
        |g + \hat{e}| \leq (2n)^t |g^{(0)}| + |\hat{e}| \leq (2n)^t|\sigma^{(0)}| + |\hat{e}| \leq ((4n)^t   + 1) |\hat{e}|.
    \end{align*}
    Hence, $|g+ \hat{e}| \leq \mu_\mathrm{syst}(k)$ provided that $|\hat{e}| \leq \frac{1}{(4n)^t+1} \mu_\mathrm{syst}(k)$. We summarize the above discussion into the following claim.

    \begin{claim}\label{claim:local-guess-work}
    Suppose that the local guesses are already consistent, i.e., $\bs \Delta^0_{ k}(g^{(0)}) = 0$. Then we can find a chain $g \in C_{k}(X)$ that has syndrome $\sigma$ in sequential time (i.e., $O(q^{n^{O(t)}}|X(0)|)$) or in parallel time $O(q^{n^{O(t)}})$. Furthermore, assuming that $|\hat{e}| \leq \frac{1}{(4n)^t+1} \mu_\mathrm{syst}(k)$, then $g$ is guaranteed to be homologically equivalent to the underlying error $\hat{e}$.
    \end{claim}

    \paragraph{Step 2: Syndrome explanation sequence.} However, it is likely that the local guesses $g^{(0)}$ made in Step 1 are inconsistent. In that case we need to find a `correction' for the guesses $\{g_v\}$. In this step, we will attempt to do so via procedure that traces the inconsistencies at higher levels of the cubical complex shown in~\Cref{fig:tracing}.

    We first define the `inconsistency' at the bottom level in $\{g_v\}$.
    \begin{align}
        \sigma^{(1)} = \bs \Delta^{0}_{ k} (g^{(0)}) \in \bs C_{ k}^{1}(X, \mcF_k).
    \end{align}
    Let $\hat{z}^{(0)}$ denote the cochain element in $\bs C_{ k+1}^{0}(X, \mcF_{k+1})$ corresponding to $\{\hat{z}_v\}_{v \in X(0)}$. It is natural to ask whether the elements of $\hat{z}^{(0)}$ correspond to the local views of a valid chain in $\bs C_{k+1}(X)$? Notice that if $g^{(0)}$ were consistent, then a consistent $\hat{z}^{(0)}$ satisfying~\Cref{eq:4440} would have existed\footnote{Let $g$ be the chain from~\Cref{claim:local-guess-work}, then $g+\hat{e}=\partial_{k-1}(\hat{z})$ for some chain $\hat{z}\in C_k(X, \mcF)$. Then we can simply take $\hat{z}^{(0)}$ to be the local views of $\hat{z}$.}. So there seems to be a connection between the inconsistencies in $g^{(0)}$ and in $\hat{z}^{(0)}$. The inconsistency in $\hat{z}^{(0)}$ is quantified by
    \begin{align}
        \hat{x}^{(1)} = \bs \Delta^{0}_{ k+1} (\hat{z}^{(0)}) \in C_{ k+1}^{1}(X).
    \end{align}
    Indeed, we can verify that $\sigma^{(1)}$ is the `syndrome' of $\hat{x}^{(1)}$:
    \begin{align*}
        \boldsymbol{\partial}_L (\hat{x}^{(1)}) &=  \boldsymbol{\partial}_L (\bs \Delta^{0}_{ k+1} (\hat{z}^{(0)}))= \bs \Delta^{0}_{ k} (\boldsymbol{\partial}_L ( \hat{z}^{(0)})) \tag{using \Cref{fig:commutative}} \\
        &= \bs \Delta^{0}_{ k} (g^{(0)} + \hat{x}^{(0)} ) \tag{using \Cref{eq:4440}}\\
        &= \bs \Delta^{0}_{ k} (g^{(0)}) \tag{$\hat{x}^{(0)}$ consistent,~\Cref{claim:consistent}}\\
        &= \sigma^{(1)}. \numberthis \label{eq:4445}
    \end{align*}
    Below we will show that knowing $\hat{x}^{(1)}$ (up to elements in $\Im(\bs \Delta_{k+1}^0) $)
     would allow us to fix the inconsistency in $g^{(0)}$ quite straightforwardly.
    Therefore, we will attempt to guess $\hat{x}^{(1)}$ from its syndrome $\sigma^{(1)}$.

    Since $\bs \partial_L(\sigma^{(1)})= \bs \partial_L (\bs \partial_L(x^{(1)}))=0$, we can use the exactness of $\bs \partial_L$ from~\Cref{claim:partialL} to construct $g^{(1)} \in \bs C_{ k+1}^{1}(X)$ such that $\bs \partial_L(g^{(1)}) = \sigma^{(1)}$ and $|g^{(1)}| \leq |\sigma^{(1)}|$.
    This step can be done in sequential time $O(q^{n^{O(t)}}|X(0)|)$ or in parallel time $O(q^{n^{O(t)}})$.
    Howerver, such a locally made guess might not correspond to a cochain homologically equivalent to $\hat{x}^{(1)}$. From definition, we have that $\bs \Delta^1_{ k+1}(\hat{x}^{(1)})=0$. However, this is not necessarily the case for our guess $g^{(1)}$. The quantity
    \begin{align}
        \sigma^{(2)} = \bs \Delta^1_{ k+1}(g^{(1)})
    \end{align}
    could be nonzero.
    
    Since
    \begin{align}
        \bs \partial_L (g^{(1)} + \hat{x}^{(1)}) = \sigma^{(1)} + \sigma^{(1)} =0.
    \end{align}
    the exactness of $\bs \partial_L$ implies that there exists a $\hat{z}^{(1)} \in \bs C^{1}_{  k+2}(X)$ such that
    \begin{align}
       \bs \partial_L(\hat{z}^{(1)}) = g^{(1)} + \hat{x}^{(1)} \qquad \text{and}
        \qquad |\hat{z}^{(1)}| \leq |g^{(1)}+\hat{x}^{(1)}|.
    \end{align}

\begin{figure}
\begin{center}
\begin{tikzcd}[arrows=rightarrow]
&0&1&2&3&\\
k-1 & \sigma^{(0)} &&&&\\
k&     (g^{(0)}, \hat{x}^{(0)}) \ar[r,"\bs \Delta^0_{ k}", "1\mapsto" below] \ar[u,"\bs \partial_L", "//" right] & \sigma^{(1)} &  & &\\
k+1&     \hat{z}^{(0)} \ar[u,"\bs \partial_L", "+" right] \ar[r,"\bs \Delta^0_{ k+1}", "\mapsto 2" below] & (g^{(1)},\hat{x}^{(1)}) \ar[u,"\bs \partial_L", "//" right] \ar[r,"\bs \Delta^1_{ k+1}", "1\mapsto" below] & \sigma^{(2)}& &\\
k+2&     & \hat{z}^{(1)} \ar[rd, dashed, rightarrow, no head] \ar[u,"\bs \partial_L", "+" right] \ar[r,"\bs \Delta^1_{ k+2}", "\mapsto 2" below]& (g^{(2)},\hat{x}^{(2)}) \ar[rd, dashed, rightarrow, no head] \ar[u,"\bs \partial_L", "//" right]&&\\
 & & & ~& \makebox[\widthof{$x^{(3)}$}]{~}&~
\end{tikzcd}
\end{center}
\caption{Illustration of the syndrome explanation sequence in~\Cref{def:tracing}. Here, $\sigma^{(0)}$ is the local views of the syndrome $\sigma$ and $\hat{x}^{(0)}$ is the local views of the underlying error $\hat{e}$. Both $g^{(i)}$ and $\hat{x}^{(i)}$ are cochains in $\bs C^{i}_{ k+i}(X, \mcF_{k+i})$. The arrows have additional labels: ``$//$'' means both $g^{(i)}$ and $\hat{x}^{(i)}$ are mapped to $\sigma^{(i)}$ under $\bs \partial_L$, ``+'' means $\hat{z}^{(i)}$ is mapped to $g^{(i)} + \hat{x}^{(i)}$, ``$1\mapsto$'' means $\sigma^{(i+1)}$ is the image of $g^{(i)}$ (the first element of the tuple) under $\bs \Delta^i_{ k+i}$, ``$\mapsto 2$" means $\hat{x}^{(i+1)}$ (the second element of the tuple) is the image of $\hat{z}^{(i)}$ under $\bs \Delta^i_{ k+i+1}$. In this step of the X-syndrome decoder, we attempt to fix the disagreements in the local guesses in $g^{(0)}$ by generating the sequence $\sigma^{(1)}, g^{(1)},\sigma^{(2)}, g^{(2)}$, etc. until reaching a syndrome $\sigma^{(r)}=0$. If this succeeds then~\Cref{claim:syndrome-consistent} asserts that a consistent set of new local guesses $g'^{(0)}$ can be obtained immediately. Otherwise, we move on to the next step of calling Z-syndrome decoder subroutine.}
\label{fig:tracing}
\end{figure}
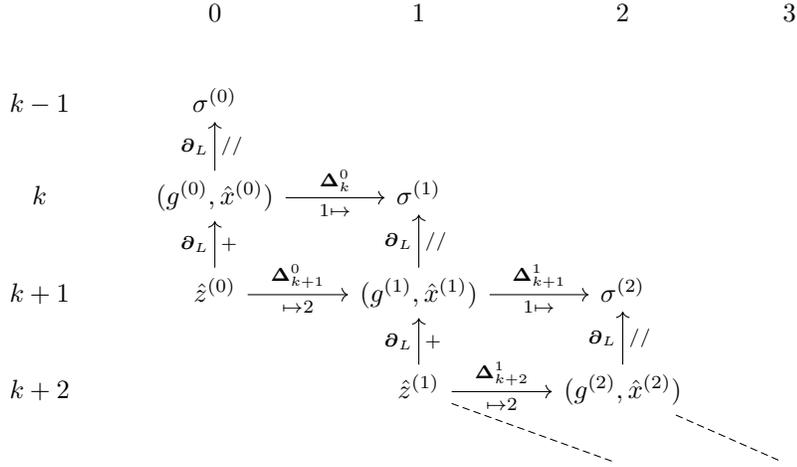

We can similarly define, if needed, $g^{(2)},\sigma^{(2)},\hat{x}^{(2)}, \hat{z}^{(2)}$, etc. as illustrated in~\Cref{fig:tracing}. More formally, they are inductively defined as follows.

\begin{definition}[Syndrome explanations]\label{def:tracing} Let $0 \leq r \leq t-k$. A sequence $(g^{(0)},\hat{x}^{(0)}),\hdots, (g^{(r)},\hat{x}^{(r)})$, where for each $i$, $g^{(i)}, \hat{x}^{(i)} \in \bs C^{i}_{ k+i}(X)$, is said to explain the syndrome $\sigma^{(0)} \equiv \bs \partial_L (\hat{x}^{(0)})$ if the following holds. There are $\sigma^{(1)},\hdots,\sigma^{(r+1)}$ where $\sigma^{(i)} \in \bs C^{i}_{ k+i-1}(X)$ and $\hat{z}^{(0)},\hdots, \hat{z}^{(r)}$ where $\hat{z}^{(i)} \in \bs C^{i}_{ k+i+1}(X)$, such that $\sigma^{(r+1)}=0$, and for all $i=0,\hdots,r$,
\begin{align}
    \sigma^{(i)} = \bs \partial_L (g^{(i)}), \qquad \sigma^{(i+1)}= \bs \Delta^{i}_{ k+i}(g^{(i)})   \tag{known to decoder},
\end{align}
and for all $i=0,\hdots,r-1$,
\begin{align}
    g^{(i)} + \hat{x}^{(i)}= \bs \partial_L(\hat{z}^{(i)}), \qquad \hat{x}^{(i+1)} = \bs \Delta^{i}_{ k+i+1} (\hat{z}^{(i)}). \tag{unknown to decoder}
\end{align}
Additionally, we require, for all $i=0,\hdots,r$, that
\begin{align}
    &|g^{(i)}| \leq |\sigma^{(i)}|,   & |\sigma^{(i+1)}| \leq  nt |g^{(i)}|
    \label{eq:weight-bound-1} \\
    &|\hat{z}^{(i)}| \leq |g^{(i)}|+ |\hat{x}^{(i)}|,   &|\hat{x}^{(i+1)}| \leq nt |\hat{z}^{(i)}|.
    \label{eq:weight-bound-2}
\end{align}
Furthermore, given $g^{(i)}$ we can construct $\sigma^{(i+1)}$ in sequential time $O(q^{n^{O(t)}}|X(0)|)$ or in parallel time $O(q^{n^{O(t)}})$. The same complexity holds for constructing $g^{(i)}$ provided with $\sigma^{(i)}$.
\end{definition}

Note that in~\Cref{def:tracing},
there exist $\hat{z}^{(i)}$ and $g^{(i+1)}$ satisfying the mentioned constraints because of the exactness of $\bs \partial_L$ (\Cref{claim:partialL}) and $\bs \partial_L (\sigma^{(i+1)})=0$, $\bs \partial_L (g^{(i)}+\hat{x}^{(i)})=0$. Indeed, the latter equalities can proved inductively from definitions. The base case $i=0$ were shown in~\Cref{eq:4441} and~\Cref{eq:4445}. For the induction step, we have
\begin{align*}
    \bs \partial_L(\sigma^{(i+1)})&= \bs \partial_L(\Delta^{i}_{ k+i} (g^{(i)}))\\
    &= \bs \Delta^{i}_{ k+i-1} (\bs \partial_L(g^{(i)})) \tag{using~\Cref{fig:commutative}}\\
    &= \bs \Delta^{i}_{ k+i-1} (\sigma^{(i)})\\
    &= \bs \Delta^{i}_{ k+i-1} (\bs \Delta^{i-1}_{ k+i-1}(g^{(i-1)}))\\
    &=0,
\end{align*}
and similarly,
\begin{align*}
    \bs \partial_L (g^{(i)} + \hat{x}^{(i)}) &= \sigma^{(i)} + \bs \partial_L(\bs \Delta^{i-1}_{ k+i} (\hat{z}^{(i-1)}))\\
    &= \sigma^{(i)} + \bs \Delta^{i-1}_{ k+i-1} (\bs \partial_L(\hat{z}^{(i-1)}))\\
    &= \sigma^{(i)} + \bs \Delta^{i-1}_{ k+i-1} (g^{(i-1)}+\hat{x}^{(i-1)})\\
    &= \sigma^{(i)} + \sigma^{(i)} + \bs \Delta^{i-1}_{ k+i-1}(\bs \Delta^{i-2}_{ k+i-1}(\hat{z}^{(i-2)}))\\
    &=0.
\end{align*}
In addition, the block-Hamming weight bounds $|\hat{z}^{(i)}| \leq |\hat{x}^{(i)}|$ and $|\hat{x}^{(i+1)}| \leq nt |\hat{z}^{(i)}|$ hold due to~\Cref{lem:incidence}. Clearly, the stated time complexity in~\Cref{def:tracing} can be achieved because all operations are done within local views of faces.

If we inductively perform the above procedure and end up with a valid syndrome explanation sequence, i.e., $\sigma^{(r+1)}=0$ for some $r \leq t-k$, then~\Cref{claim:syndrome-consistent} below allows us to find a `correction' to the original inconsistent local guesses $g^{(0)}$, reducing to the case of~\Cref{claim:local-guess-work} and in turn obtaining a valid global chain in $C_k(X, \mcF)$.

\begin{claim}
\label{claim:syndrome-consistent}
Suppose that $0 \leq k < t$ and there is some $0\leq r \leq t-k$ such that a sequence $(g^{(0)},\hat{x}^{(0)}),\hdots, (g^{(r)},\hat{x}^{(r)})$ explains the syndrome $\sigma^{(0)}$ according to~\Cref{def:tracing}. Then, provided with $g^{(0)},\hdots,g^{(r)}$, we can construct, in sequential time $O(q^{n^{O(t)}}|X(k)|)$ or in parallel time $O(q^{n^{O(t)}})$, a cochain $g'^{(0)} \in \bs C^0_{ k}(X, \mcF_k)$ such that $\bs \partial_{L}(g'^{(0)})= \sigma^{(0)}$ and $\bs \Delta^0_{ k}(g'^{(0)}) = 0$. Furthermore, it is guaranteed that $|g'^{(0)}| \leq 8^{t^2}n^{t(t+1)} |g^{(0)}|$.
\end{claim}

\begin{proof}[Proof of~\Cref{claim:syndrome-consistent}]
The case $r=0$ is~\Cref{claim:local-guess-work} where we simply take $g'^{(0)}=g^{(0)}$. So we consider $r \geq 1$ and proceed by reducing to the $r=0$ case. In particular, we will construct a cochain $g'^{(r-1)}$ such that $g^{(0)},g^{(1)},\hdots,g'^{(r-1)}$ explain $\sigma^{(0)}$. We will then repeat this procedure to reduce to the case $r=0$.

First, we use the assumption $\sigma^{(r+1)} \equiv \bs \Delta^{r}_{ k+r} (g^{(r)})=0$ and the exactness of $\bs \Delta_{ k+r}$ to construct a cochain $u^{(r-1)} \in \bs C^{r-1}_{ k+r}(X)$ such that $\bs \Delta^{r-1}_{ k+r}(u^{(r-1)}) = g^{(r)}$. In particular,~\Cref{claim:exactness-agg} asserts that such a cochain can be constructed from $g^{(r)}$ in sequential time $O(q^{n^{O(t)}}|X(k+r)|)=O(q^{n^{O(t)}}|X(k)|)$ or in parallel time $O(q^{n^{O(t)}})$, and furthermore, it holds that
$|u^{(r-1)}| \leq (4n)^t|g^{(r)}|$. Next, we claim that $g'^{(r-1)} \triangleq g^{(r-1)}+\bs \partial_L(u^{(r-1)})$ is the desired new element such that $g^{(0)}, g^{(1)},\hdots,g'^{(r-1)}$ explain $\sigma^{(0)}$. Indeed, we have
\begin{align*}
    \bs \Delta^{r-1}_{ k+r-1}(g'^{(r-1)}) &= \bs \Delta^{r-1}_{ k+r-1}(g^{(r-1)}) + \bs \Delta^{r-1}_{ k+r-1}(\bs \partial_L(u^{(r-1)})) \\
    &= \sigma^{(r)} + \bs \partial_L(\Delta^{r-1}_{ k+r}(u^{(r-1)})) \tag{using \Cref{fig:commutative}}\\
    &= \sigma^{(r)} + \bs \partial_L(g^{(r)} )= 0,
\end{align*}
and
\begin{align*}
    \bs \partial_{L}(g'^{(r-1)})&= \bs \partial_{L}(g^{(r-1)}) + \bs \partial_{L}(\bs \partial_L(u^{(r-1)}))\\
    &= \sigma^{(r-1)}.
\end{align*}
Note the following bound on the block-Hamming weight of $g'^{(r-1)}$
\begin{align*}
    |g'^{(r-1)}| \leq |g^{(r-1)}|+|\bs \partial_L(u^{(r-1)})| \leq |g^{(r-1)}| + (4n)^t|g^{(r)}| \leq (1+ nt(4n)^{t} )|g^{(r-1)}|,
\end{align*}
where the last inequality follows from the block-Hamming weight bound $|g^{(i+1)}| \leq nt |g^{(i)}|$ in~\Cref{def:tracing}.

Using the same procedure we obtain an explanation sequence $g^{(0)},\hdots, g'^{(r-2)}$ with $|g'^{(r-2)}| \leq |g^{(r-2)}| + (4n)^t |g'^{(r-1)}| \leq |g^{(r-2)}| + (4n)^t (1+ nt(4n)^{t} )|g^{(r-1)}| \leq (1+ nt(4n)^{t} + (nt(4n)^t)^2 )|g^{(r-2)}|$. Repeating this procedure until $r=0$ we obtain the desired $g'^{(0)}$ with the following weight bound
\begin{align*}
   |g'^{(0)}| \leq (1+nt(4n)^t+\hdots+(nt(4n)^t)^r)|g^{(0)}| \leq (nt(4n)^t)^{r+1} |g^{(0)}| \leq 8^{t^2}n^{t(t+1)}|g^{(0)}|.
\end{align*}
The time complexity is as claimed since the above reduction step is done no more than $t$ steps.
\end{proof}

However, if Step 2 doesn't succeed, i.e., we end up with $\sigma^{(t-k+1)}\neq 0$, then we instead proceed by finding a `correction' to $g^{(t-k)}$ to reduce to the case $\sigma^{(t-k+1)}=0$. This is done via a reducing to a Z-syndrome decoding problem on the dual chain complex from~\Cref{def:dual-chain-complex} and then invoking the Z-syndrome decoder (\Cref{alg:Zdecoder} with parameter $\varepsilon=1$), as shown in~\Cref{claim:apply-Zdecoder} below.

\begin{claim}
\label{claim:apply-Zdecoder}
Let $1 \leq k <t$. Suppose that the procedure in~\Cref{def:tracing} yields a sequence $(g^{(0)},\hat{x}^{(0)}),\hdots, (g^{(t-k)},\hat{x}^{(t-k)})$ \footnote{We emphasize that here $\hat{x}^{(i)}$ are unknown to the decoder.} with $\sigma^{(t-k+1)}\neq 0$. Furthermore, suppose that
\begin{align}
    |\hat{e}| \leq \frac{1}{2^t(t-k+2)(nt)^{t-k}}\frac{8^{-t^2}\kappa^{t-k+1}-\lambda}{2^{t+1} (4n)^t} r |X(t-k)| \cdot \min (\frac{(8^{-t^2}\kappa^{t-k+1} - \lambda)^2}{2^{t+3}}, (\frac{\kappa}{2})^{t-k+1}8^{-t^2}-\lambda ).
\end{align}
Then we can find a chain $g \in C_{k}(X)$ that is homologically equivalent to the underlying error $\hat{e}$ in time equal to the runtime of the Z-syndrome decoder call, plus $O(q^{n^{O(t)}}|X(k)|)$ sequential time (or $O(q^{n^{O(t)}})$ parallel time). 
\end{claim}

\begin{proof}[Proof of~\Cref{claim:apply-Zdecoder}]
Our goal is to find a `correction' co-chain $w^{(t-k)} \in \bs C^{t-k}_{ t}(X)$ to $g^{(t-k)}$ such that the sequence $g^{(0)},g^{(1)},\hdots,g'^{(t-k)} \triangleq g^{(t-k)}+w^{(t-k)}$ explains $\sigma^{(0)}$, which essentially brings us back to the setting of~\Cref{claim:syndrome-consistent}. As per~\Cref{def:tracing}, we want the following to hold
\begin{align}
    \bs \Delta^{t-k}_{ t}(g'^{(t-k)})=0, \qquad \text{and} \qquad 
    \bs \partial_L (g'^{(t-k)})=\sigma^{(t-k)}.
    \label{eq:995}
\end{align}

    We now describe how to construct such a correction cochain.
    According to~\Cref{def:tracing}, it holds that
    \begin{align}
        \bs \partial_L (g^{(t-k)} + \hat{x}^{(t-k)}) = 0.
    \end{align}
    
     \paragraph{Step 3: Map to dual complex and apply Z-decoder} Here recall that $g^{(t-k)}, \hat{x}^{(t-k)} \in \bs C^{t-k}_{ t}(X, \mathcal{F}_{t})$ and their restriction on each face $f \in X(t-k)$ corresponds to an upward local view in the original chain complex: $(g^{(t-k)}+ \hat{x}^{(t-k)})(f) \in C_t(X_{\geq f}) \cong \mathbb{F}_q^{\prod_{j \notin \mathsf{type}(f)} n_j}$. It follows, due to~\Cref{lem:exactness-tensorcode} that, the restriction $(g^{(t-k)} + \hat{x}^{(t-k)})(f)$ is a tensor codeword in $\otimes_{j \notin \mathsf{type}(f)} \ker(h_j)$. Therefore, we can define a $(t-k)$-cochain $\Tilde{c}^{(t-k)}\in \Tilde{C}^{t-k}(X,\Tilde{\mathcal{F}})$ in the dual chain complex (\Cref{def:dual-chain-complex}) by replacing the tensor codeword $(g^{(t-k)} + \hat{x}^{(t-k)})(f)$ by its un-encoded logical message $\Tilde{c}^{(t-k)} \in \mathbb{F}_q^{\prod_{j \notin \mathsf{type}(f)} k_j}$ (according to the pre-specified generating matrix $h_j^{\perp} \in \mathbb{F}^{k_j \times n_j}$). Of course, $\Tilde{c}^{(t-k)}$ is a priori unknown to the decoder. However, if we knew $\Tilde{c}^{(t-k)}$, we would know $g^{(t-k)}+ \hat{x}^{(t-k)}$ (by going over each face and re-encoding $\Tilde{c}^{(t-k)}(f)$ back into the tensor code in sequential time $O(q^{n^{O(t)}}|X(k)|)$ or in parallel time $O(q^{n^{O(t)}})$), which can be easily verified to be such a desired correction $w^{(t-k)}$.
    
    Therefore, we attempt to find $\tilde{c}^{(t-k)}$. Remarkably, it turns out that we have a handle on its syndrome,  $\tilde{\delta}^{t-k} (\tilde{c}^{(t-k)})$, with respect to the coboundary map of $\Tilde{C}(X, \Tilde{\mathcal{F}})$. Indeed, for each face $f \in X(t-k+1)$, it holds that
    \begin{align*}
        \left(\bigotimes_{j \notin \mathsf{type}(f)} (h_j^\perp)^\top  \right) \big(  \tilde{\delta} (\tilde{c}^{(t-k)}) (f) \big)
        &=\sum_{i \in \mathsf{type}(f)} \sum_{f' \precdot_i f} \left( \left( \bigotimes_{j \notin \mathsf{type}(f)} (h_j^\perp)^\top \right) \otimes (h_{i}^\perp [f_i])^\top \right) \big(\tilde{c}^{(t-k)}  (f') \big) \tag{using \Cref{def:sheaf-complex}; here $[f_i]$ means $f_i$-th column of $h_i^\perp$}\\
   &= \sum_{i \in \mathsf{type}(f)} \sum_{f' \precdot_i f} \big(g^{(t-k)} + \hat{x}^{(t-k)} \big) (f') \big|_{X_{\geq f}(t)} \tag{$h_j^\perp$ re-encodes}\\
   &= \big(\bs \Delta_{ t}^{t-k} (g^{(t-k)}+\hat{x}^{(t-k)}) \big)(f) \tag{using~\Cref{eq:local-view-sheaf}}\\
   &= \sigma^{(t-k+1)},
    \end{align*}
where the last equality is because $\bs \Delta_{ t}^{t-k} (\hat{x}^{(t-k)})= \bs \Delta_{ t}^{t-k}(\bs \Delta^{t-k-1}_{ t} (\hat{z}^{(t-k-1)}))=0$. In other words, $\sigma^{(t-k+1)}$, which we know, is in fact an encoding of the syndrome $\tilde{\delta}^{t-k} (\tilde{c}^{(t-k)})$. So we can first un-encode the local tensor code in $\sigma^{(t-k+1)}$ to find $\Tilde{\sigma}^{(t-k+1)} =\tilde{\delta}^{t-k} (\tilde{c}^{(t-k)})$ (in sequential time $O(q^{n^{O(t)}}|X(k)|)$ or in parallel time $O(q^{n^{O(t)}})$).  
Note the following block-Hamming weight bound
    \begin{align*}
     |\tilde{c}^{(t-k)}| &\leq |g^{(t-k)}| + |\hat{x}^{(t-k)}| \\
     &\leq |\sigma^{(t-k)}| + nt(|g^{(t-k-1)}| + |\hat{x}^{(t-k-1)}|)  \tag{using~\Cref{eq:weight-bound-2}}\\
     &\leq \hdots\\
     &\leq \sum_{i=0}^{t-k} (nt)^i |\sigma^{(t-k-i)}| + (nt)^{t-k} (|g^{(0)}|+ |\hat{x}^{(0)}|)\\
     &\leq \sum_{i=0}^{t-k} (nt)^i (nt)^{t-k-i}|\sigma^{(0)}| + (nt)^{t-k} (|g^{(0)}|+ |\hat{x}^{(0)}|)  \tag{using~\Cref{eq:weight-bound-1}} \\
     &\leq  (t-k+2)(nt)^{t-k}|\hat{x}^{(0)}|.
      \numberthis\label{eq:dual-chain-weight-bound}
    \end{align*}
Hence, using~\Cref{claim:apply-Zdecoder}'s promise and $|\hat{x}^{(0)}|\leq 2^t |\hat{e}|$ we can apply the Z-syndrome decoder~\Cref{alg:Zdecoder} (with parameter $\varepsilon=1$ for simplicity) on $\Tilde{C}(X,\Tilde{\mathcal{F}})$, whose performance guanrantee from~\Cref{prop:noiseless-Zdecoder} applies because of the two-way $\kappa$-robustness of the local maps. According to~\Cref{prop:noiseless-Zdecoder}, upon the syndrome input $\Tilde{\sigma}^{(t-k+1)}$, the Z-decoder outputs 
\begin{align}
    \tilde{c}'^{(t-k)} = \tilde{c}^{(t-k)} + \tilde{\delta}^{t-k-1}(\tilde{r}^{(t-k-1)}),
\end{align}
where $\tilde{r}^{(t-k-1)} \in \Tilde{C}^{t-k-1}(X)$. Our correction to $g^{(t-k)}$ will be the re-encoded version of this, i.e., let $w^{(t-k)}(f) \triangleq \big( \otimes_{j \notin \mathsf{type}(f)} (h_j^\perp)^\top \big)\tilde{c}'^{(t-k)}(f)$, for each face $f \in X(t-k)$.

Let us verify that this correction satisfies~\Cref{eq:995}. The re-encoded version of $\tilde{c}^{(t-k)}$ is $g^{(t-k)}+ \hat{x}^{(t-k)}$. Let $\hat{v}^{(t-k)} \in \bs C^{t-k}_{ t}(X)$ be the re-encoded version of $\Tilde{\delta}(\tilde{r}^{(t-k-1)})$, we have that
\begin{align*}
    \bs \Delta_{ t}^{t-k} (g^{(t-k)}+ w^{(t-k)}) &= \bs \Delta_{ t}^{t-k} (\hat{x}^{(t-k)}) + \bs \Delta_{ t}^{t-k} (\hat{v}^{(t-k)})\\
    &= \bs \Delta_{ t}^{t-k} (\hat{v}^{(t-k)}).
\end{align*}
We can verify that $\bs \Delta_{ t}^{t-k} (\hat{v}^{(t-k)})=0$. For each face $f \in X(t-k)$ we have
\begin{align*}
    \hat{v}^{(t-k)}(f)&=\big(\otimes_{j \notin \mathsf{type}(f)} (h_j^\perp)^\top  \big) \big(  \Tilde{\delta}(\Tilde{r}^{(t-k-1)})  (f) \big)\\
   &= \sum_{i \in \mathsf{type}(f)} \sum_{f' \precdot_i f} \left(\bigotimes_{j \notin \mathsf{type}(f)} (h_j^\perp)^\top \otimes (h_{i}^\perp [f_i])^\top \right) \big(  \Tilde{r}^{(t-k-1)} (f') \big)\\
   &= \sum_{i \in \mathsf{type}(f)} \sum_{f' \precdot_i f} \hat{r}^{(t-k-1)}(f') \big|_{X_{\geq f}(t)} \tag{$\hat{r}^{(t-k-1)}$ is re-encoded version of $\Tilde{r}^{(t-k-1)}$}\\
   &= \bs \Delta_{ t}^{t-k-1} (\hat{r}^{(t-k-1)}) (f). \numberthis \label{eq:r-t-k-1}
\end{align*}
Hence, $\bs \Delta_{ t}^{t-k} (g^{(t-k)}+ w^{(t-k)}) = \bs \Delta_{ t}^{t-k}(\bs \Delta_{ t}^{t-k-1} (\hat{r}^{(t-k-1)}))=0$, satisfying the first requirement from~\Cref{eq:995}. Next, we verify the second requirement from~\Cref{eq:995}
\begin{align*}
    \bs \partial_L (g^{(t-k)}+ w^{(t-k)})
    &= \bs \partial_L (\hat{x}^{(t-k)}) + \bs \partial_L (\hat{v}^{(t-k)})\\
    &= \sigma^{(t-k)} + \bs \partial_L (\hat{v}^{(t-k)})\\
    &= \sigma^{(t-k)},
\end{align*}
where the last equality is because $\hat{v}^{(t-k)}(f)$ is an image of the map $\otimes_{j \notin \mathsf{type}(f)} (h_j^\perp)^\top$ and thus $\bs \partial_L (\hat{v}^{(t-k)})=0$.

To summarize, we have found a correction $w^{(t-k)} \in \bs C_{t}^{t-k}(X)$  such that the new guess $g'^{(t-k)} \triangleq g^{(t-k)}+ w^{(t-k)}$ has the following form
\begin{align}
    g'^{(t-k)} = \hat{x}^{(t-k)} +\hat{v}^{(t-k)},
    \label{eq:g'-step3}
\end{align}
where $\hat{v}^{(t-k)} \in \Im(\bs \Delta_{ t}^{t-k-1})$ and $\hat{v}^{(t-k)} \in \ker(\bs \partial_L)$.
And the updated sequence $g^{(0)},g^{(1)},\hdots, g'^{(t-k)}$ gives us a new syndrome explanation sequence~\Cref{def:tracing} (only almost so, since the weight bounds from~\Cref{def:tracing} don't hold anymore, due to the extra `boundary term' $\hat{v}^{(t-k)}$ in $g'^{(t-k)}$.)

\paragraph{Step 4: Obtaining final guess.}
We can perform similar calculations to the proof of~\Cref{claim:syndrome-consistent} in Step 2 to convert the sequence $g^{(0)},g^{(1)},\hdots, g'^{(t-k)}$ obtained in Step 3 into a final guess $g'^{(0)}$ that is consistent. To repeat what happened in~\Cref{claim:syndrome-consistent}'s proof, we use the exactness of $\bs \Delta^i_k$ for $1\leq i \leq k$ from~\Cref{claim:exactness-agg} to find $u^{(t-k-1)} \in \bs C^{t-k-1}_t(X)$ such that $\bs \Delta^{t-k-1}_t (u^{(t-k-1)}) = g'^{(t-k)}$ and set $g'^{(t-k-1)} = g^{(t-k-1)} + \bs \partial_L(u^{(t-k-1)})$. The new sequence $g^{(0)},g^{(1)},\hdots, g'^{(t-k-1)}$ is still a valid explanation sequence. This is repeated to obtain $g'^{(t-k-2)},\hdots,g'^{(0)}$ such that $\bs \Delta^{i}_{k+i}(g'^{(i)}) = 0$ and $\bs \partial_L (g'^{(i)}) = \sigma^{(i)}$. Each iteration takes either linear sequential time or parallel constant time as per~\Cref{claim:exactness-agg}. Now, we can invoke~\Cref{claim:consistent} on $g'^{(0)}$ to obtain a chain $g \in C_k(X, \mcF)$ as done in Step 1. By construction, $g$ has syndrome $\sigma$, so it remains to prove that $g$ is homologically equivalent to the underlying error $\hat{e}$.
We cannot use the simple block-Hamming weight bound argument as done in Step 1 and Step 2 since those weight bounds don't hold anymore, due to the extra `boundary term' ($v^{(t-k)}$ in~\Cref{eq:g'-step3}) coming from applying the Z-decoder in Step 3. So we use a different argument that involves opening up the structure of $g'^{(0)}$ and the proofs of~\Cref{claim:exactness-agg} and~\Cref{claim:consistent}.

First, we claim that, for each $0\leq i \leq t-k$, the new guess $g'^{(i)}$ can be written as $g'^{(i)} = \hat{x}^{(i)} + \hat{v}^{(i)}$, where $\hat{v}^{(i)} \triangleq g'^{(i)} + \hat{x}^{(i)} \in \Im(\bs \Delta^{i-1}_{k+i}) \cap \Im(\bs \partial_L)$. Indeed, by construction it holds that $\bs \Delta^{i}_{k+i} (\hat{v}^{(i)}) = \bs \Delta^{i}_{k+i}(g'^{(i)}) -   \bs \Delta^{i}_{k+i}(\hat{x}^{(i)})=0$ and $\bs \partial_L (\hat{v}^{(i)})= \bs \partial_L (g'^{(i)}) - \bs \partial_L (\hat{x}^{(i)})= \sigma^{(i)} - \sigma^{(i)}=0$. Hence, the claimed statement follows from the exactness of $\bs \Delta_{k+i}$ and $\bs \partial_L$. Let us write
\begin{align}
    g'^{(1)}= \hat{x}^{(1)} + \bs \Delta^0_{k+1}(\hat{r}^{(0)}),
\end{align}
 where $\hat{r}^{(0)} \in \bs C^0_{k+1}(X)$ (which is unknown to us). \footnote{More generally, we can write $g'^{(i)}= \hat{x}^{(i)} + \bs \Delta^{i-1}_{k+i}(\hat{r}^{(i-1)})$, such that when $i=t-k$, $\hat{r}^{(t-k-1)}$ coincides with~\Cref{eq:r-t-k-1}.} Next, recall that $g'^{(0)} = g^{(0)} + \bs \partial_L(u^{(0)})$ where $u^{(0)}$ is such that $\bs \Delta^{0}_{k+1}(u^{(0)})= g'^{(1)}=\hat{x}^{(1)} + \bs \Delta^0_{k+1}(\hat{r}^{(0)})$. We now show that the $k$-chain $g\in C_{k}(X,\mcF)$ obtained from $g'^{(0)}$ has the form $g = \hat{e} + \partial_{k+1}(\hat{r}')$ for some $\hat{r}' \in C_{k+1}(X,\mcF)$ related to $\hat{r}^{(0)}$. To do this, we need to know exactly how $g$ and $u^{(0)}$ are constructed.

Let us look at $u^{(0)}$. Recall from the proof of~\Cref{claim:exactness-agg} that the map $\bs \Delta^0_{k+1}$ can be decomposed into a direct sum of local maps $\Delta^0_f$ over faces $f \in X(k+1)$. In particular, it was shown that $\bs C^0_{k+1}(X, \mathcal{F}_{k+1}) =  \bigoplus_{f \in X(k+1)} C^0(X_{\leq f}, V_f)$, where $C^0(X_{\leq f}, V_f)$ is a cochain space in a `hypercube' complex defined within the downward link of $f$. Hence, for each 1-face (edge) $e=(v,v') \in X_{\leq f}$ we have 
\begin{align}
    g'^{(1)}(e)[f] = u^{(0)}(v)[f] + u^{(0)}(v')[f],
\end{align}
where, as in~\Cref{claim:exactness-agg}'s proof, $(\cdot)[f]$ denotes the restriction onto the local opinions for $f$). On the other hand, because $\bs \Delta^{1}_{k+1}(g'^{(1)}) = 0$, we have, for each $2$-face (square) $s=(v_{00},v_{01},v_{10},v_{11}) \in X_{\leq f}$, that
\begin{align}
    g'^{(1)}((v_{00},v_{01}))[f] + g'^{(1)}((v_{01},v_{11}))[f] + g'^{(1)}((v_{00},v_{10}))[f] + g'^{(1)}((v_{10},v_{11}))[f] =0.
    \label{eq:square}
\end{align}
Hence, one solution to the equation $\bs \Delta^{0}_{k+1}(u^{(0)})= g'^{(1)}$ is as follows. Consider a face $f \in X(k+1)$, let $v_a$ denote the vertices in $f$, where the label $a$ goes over all bitstrings in $\{0,1\}^{k+1}$. We set 
\begin{align}
    u^{(0)}(v_{0^{k+1}})[f]=0,
\end{align}
where $0^{k+1}$ is the all-zeros bitstring. And for each vertex $v_a$ in f, if $(e=(v_{0^k},v),\hdots,e'=(v', v_a))$ are the edges on a path from $v_{0^k}$ to $v_a$, we set 
\begin{align}
u^{(0)}(v_a)[f]= g'^{(1)}(e)[f] + \hdots + g'^{(1)}(e')[f].
\label{eq:642}
\end{align}
Although there can be multiple paths between $v_{0^k}$ and $v_a$, it is easy to see that the definition $u^{(0)}(v_a)[f]$ is independent of the chosen path because one path can be deformed to another by adding squares which has no effect to the assigned value due to~\Cref{eq:square}.

Next we look at how $g$ is constructed. Because $g'^{(1)}= \hat{x}^{(1)} + \bs \Delta^0_{k+1}(\hat{r}^{(0)})= \bs \Delta^0_{k+1}(\hat{z}^{(0)} + \hat{r}^{(0)})$, it holds for each edge $e=(v,v')$ that
\begin{align}
    g'^{(1)}((v,v'))[f]=   \hat{z}^{(0)}(v)[f] + \hat{z}^{(0)}(v')[f] +  \hat{r}^{(0)}(v)[f] + \hat{r}^{(0)}(v')[f].
    \label{eq:643}
\end{align}
Therefore, combining~\Cref{eq:642} and~\Cref{eq:643}, we obtain, for each face $f \in X(k+1)$ and vertex $v \prec f$, that
\begin{align}
    u^{(0)}(v)[f]= (\hat{z}^{(0)} + \hat{r}^{(0)})(v)[f] + (\hat{z}^{(0)} + \hat{r}^{(0)})(v_{0^k})[f].
\end{align}
So we can write
\begin{align}
    u^{(0)}=\hat{z}^{(0)} + \hat{r}'^{(0)},
\end{align}
where
\begin{align}
    \hat{r}'^{(0)}(v)[f] \triangleq \hat{z}^{(0)}(v_{0^k})[f] + \hat{r}^{(0)}(v_{0^k})[f] + \hat{r}^{(0)}(v)[f].\label{eq:645}
\end{align}
Define the cochain $\hat{z}'^{(0)} \in \bs C^0_{k+1}(X)$ with $ \hat{z}'^{(0)}(v)[f]=\hat{z}^{(0)}(v_{0^k})[f] + \hat{r}^{(0)}(v_{0^k})[f]$. By construction this cochain is consistent and represents the local views of a chain $z' \in C_{k+1}(X,\mcF)$. We rewrite~\Cref{eq:645} as
\begin{align}
    \hat{r}'^{(0)} = \hat{z}'^{(0)} + \hat{r}^{(0)}.
    \label{eq:646}
\end{align}

Hence, we have 
\begin{align}
    g'^{(0)} = g^{(0)} + \bs \partial_L(u^{(0)}) = g^{(0)} + \bs \partial_L(\hat{z}^{(0)}) + \bs \partial_L(\hat{r}'^{(0)}) = \hat{x}^{(0)} + \bs \partial_L(\hat{r}'^{(0)}).
    \label{eq:647}
\end{align}
The final guess $g \in C_{k}(X, \mcF)$ is constructed by setting, for each face $f \in X(k)$, $g(f) = g'^{(0)}(v)[f]$, where $v$ is an arbitrary vertex in $f$ (recall this assignment is independent of the choice of $v$ because of~\Cref{claim:consistent}.)

We are now ready to show $g = \hat{e} + \partial_{k+1}(\hat{r}')$ for some $\hat{r}' \in C_{k+1}(X,\mcF)$. Indeed, for each face $f \in X(k)$ and a canonical choice of vertex $v \prec f$,
\begin{align*}
    g(f)&=g'^{(0)}(v)[f] \\
    &= \hat{x}^{(0)}(v)[f] +  \bs \partial_L(\hat{r}'^{(0)})(v)[f] \tag{using \Cref{eq:647}}\\
    &= \hat{x}^{(0)}(v)[f] +  \bs \partial_L(\hat{z}'^{(0)})(v)[f] + \bs \partial_L(\hat{r}^{(0)})(v)[f] \tag{\Cref{eq:646}}\\
    &= \hat{e}(f) + \bs \partial_L(\hat{z}'^{(0)})(v)[f]+ \bs \partial_L(\hat{r}^{(0)})(v)[f] \tag{$\hat{x}^{(0)}$ contains local views of $\hat{e}$} \\ 
    &= \hat{e}(f) +  \partial_{\overline{\mathsf{type}(v)}}(\hat{z}'^{(0)}(v)) (f)  + \bs \partial_L(\hat{r}^{(0)})(v)[f] \tag{definition of $\bs \partial_L$~\Cref{eq:local-chain-boundary-map}}\\
    &= \hat{e}(f) + \partial_{k+1}(z')(f) + \bs \partial_L(\hat{r}^{(0)})(v)[f] \tag{consistency of $\hat{z}'^{(0)}$}.
\end{align*}
Finally, we notice that $\bs \partial_L(\hat{r}^{(0)})=0$. This fact can be traced back to~\Cref{eq:r-t-k-1} which shows that for each face $f \in X(t-k-1)$, we have $\hat{r}^{(t-k-1)}(f)$ is an image of the map $\otimes_{j \notin \mathsf{type}(f)} (h_j^\perp)^\top$. Furthermore, the reverse syndrome explanation procedure is linear and can be verified to preserve this image structure. Specifically, consider each vertex $v \in X(0)$ and $i \in [t]$ and $v \preceq f \in X(k)$ such that $i \notin \mathsf{type}(f)$, then for the faces $f' \succdot_i f$, the local coefficients $\hat{r}^{(0)}(v)[f'] \in \mathbb{F}_q^{\prod_{j \notin (\mathsf{type}(f) \cup\{i\})} \hat{A}_j}$ satisfy
\begin{align}
    \{\hat{r}^{(0)}(v)[f'] \}_{f' \succdot_i f} \in \Im((h_i^\perp)^\top) \otimes \mathbb{F}_q^{\prod_{j \notin (\mathsf{type}(f) \cup\{i\})} \hat{A}_j}
\end{align}
implying $\bs \partial_L(\hat{r}^{(0)})(v)[f]=0$.  Therefore, the final guess $g$ satisfies $g= \hat{e} + \partial_{k+1}(z')$ and is homologically equivalent to the true error.

The runtime stated in~\Cref{claim:apply-Zdecoder} is because both the syndrome sequence computations and the unencoding/reencoding of the local codes can be done in sequential time $O(q^{n^{O(t)}}|X(k)|)$ (or parallel time $O(q^{n^{O(t)}})$) since all of them are done locally in neighborhoods of size $n^{O(t)}$. This concludes the proof of~\Cref{claim:apply-Zdecoder}.
\end{proof}

\paragraph{Runtime} The time complexity of the sequential Z-syndrome decoder call is $O(q^{n^{O(t)}}|X(t-k)|) =O(q^{n^{O(t)}}|X(k)|)$ due to~\Cref{prop:noiseless-Zdecoder}. Hence, the runtime of the entire~\Cref{alg:Xdecoder} is $O(q^{n^{O(t)}}|X(k)|)$ since each of steps 1, 2, 3 ,4 take this much time (sequentially). This concludes the proof of the noiseless X-syndrome decoder algorithm. 
\end{proof}

\subsubsection{Noisy syndrome case}

We now describe a modification of~\Cref{alg:Xdecoder} for the noisy syndrome case. In step 1, we now only see the noisy syndrome $\sigma = \hat{\sigma} + \hat{m}$ where $\hat{m}$ is the measurement error and $\hat{\sigma}$ is the noiseless syndrome corresponding to the underlying data error $\hat{e} \in  C_k(X, \mcF)$. Hence, a local view $\sigma(X_{\geq v}(k-1))$ for a vertex $v$ is not necessarily associated with a preimage $g_v \in C_k(X_{\geq v})$. So if we do not find a chain $g_v$ such that $\partial_{\overline{\mathsf{type}(v)}}(g_v)= \sigma(X_{\geq v}(k-1))$, we simply set $g_v = 0$ in step 1. Similar modifications are also made in other steps of the algorithm, we refer to the proof for details.

Under the above modification, we can show that the modified~\Cref{alg:Xdecoder} is single-shot by carefully tracking the propagation of the syndrome errors throughout the proof of the noiseless X-decoder~\Cref{prop:noiseless-Xdecoder} and invoking the single-shotness (\Cref{prop:noisy-Zdecoder}) of the sequential Z-decoder subroutine used within~\Cref{alg:Xdecoder}.

\begin{prop}[Sequential X-decoder with noisy syndrome]
\label{prop:noisy-Xdecoder}
    Consider the same setting in~\Cref{prop:noiseless-Xdecoder} and the X-syndrome decoder algorithm (\Cref{alg:Xdecoder}), with the modification described above. Let $\hat{e} \in C_{k}(X)$ be the chain (over $\mathbb{F}_q$) corresponding to a Pauli Z error. Let $\hat{\sigma} = \partial_k (\hat{e}) \in C_{k-1}(X)$ be the ideal syndrome. And let $\sigma = \hat{\sigma} + \hat{m}$ be the observed syndrome, where the measurement error $\hat{m} \in C_{k-1}(X)$ satisfies $\hat{m} \leq \frac{1}{(2nt)^t} \frac{n}{4} \frac{(8^{-t^2}(\kappa/2)^{t-k+1}-\lambda)^2}{(4n)^t} r|X(t-k)|$ and $2^t(t+1) n |\hat{e}|_R  + 2|\hat{m}| \leq \frac{1}{(2nt)^t} \frac{n}{4} \frac{(8^{-t^2} \kappa^{t-k+1} -\lambda)}{(4n)^t}  \cdot  \min  (\frac{(8^{-t^2} \kappa^{t-k+1} -\lambda)^2}{2^{t+2}} ,8^{-t^2}(\frac{\kappa}{2})^{t-k+1}-\lambda  )$. Then~\Cref{alg:Xdecoder} (modified version) in sequential mode runs in time $O(q^{n^{O(t)}} |X(k)|)$ and outputs a correction chain $g$ such that $g + \hat{e}$ has stabilizer-reduced weight at most $\frac{(nt)^{O(t^2)}}{(\kappa/2)^{t-k+1}- \lambda} |\hat{m}|$.
\end{prop}

\begin{proof}
    We go over the steps in the proof of~\Cref{prop:noiseless-Xdecoder} and track the propagation of the syndrome errors. We will keep track of how the variables defined in~\Cref{prop:noiseless-Xdecoder}'s proof changes in each step. In particular, in that proof, we used variables such as $g^{(i)}, \sigma^{(i)}, \hat{z}^{(i)}, \hat{x}^{(i)}$ (see~\Cref{fig:tracing}) where variables with a hat symbol are unknown to decoder. The ideal guesses and ideal syndromes were known to the decoder, but they are not in the current proof. So we will add a hat symbol $\hat{g}^{(i)}, \hat{\sigma}^{(i)}$ when we refer to results derived in~\Cref{prop:noiseless-Xdecoder}. In the current proof, we still have the nonideal guessses $g^{(i)}$ and syndromes $\sigma^{(i)}$ that the decoder `thinks' it is fixing. 

    \paragraph{Step 1: (Modified) Local guesses.} In this step we go over each vertex $v$ in $X(0)$ and locally find the minimal block-Hamming weight chain $g_v \in C_k (X_{\geq v})$ that is consistent with the observed syndrome $\sigma(X_{\geq v}(k-1))$ on $X_{\geq v}(k-1)$. If there is no such chain $g_v$ then we set $g_v=0$. If the syndrome is noiseless as in~\Cref{prop:noiseless-Xdecoder}, such a $g_v$ exists for each vertex $v$. Let us denote by $\{\hat{g}_v\}_{v \in X(0)}$ the ideal local guesses and let $\hat{g}^{(0)}$ be the corresponding cochain in $\bs C^0_{k}(X)$. In the current setting, upon receiving the noisy syndrome $\sigma= \hat{\sigma} + \hat{m}$, the modified local guess rule will output a different $g^{(0)}$. Since each $(k-1)$-face contains $2^{k-1} \leq 2^t$ vertices, the difference $\hat{h}^{(0)} \triangleq g^{(0)} -  \hat{g}^{(0)}$ can be bounded by
    \begin{align}
        |\hat{h}^{(0)}|=|g^{(0)} -  \hat{g}^{(0)}| \leq 2^t|\hat{m}|.
    \end{align}

    As before, $\hat{\sigma}^{(0)} \in \bs C^0_{k-1}(X)$ (with the hat symbol now) denotes the local views of $\hat{\sigma}$, $\sigma^{(0)}$ corresponds to $\sigma$, $\hat{x}^{(0)} \in \bs C^0_{k}(X)$ coresponds to $\hat{e}$. It holds that $\bs \partial_L(\hat{g}^{(0)})= \bs \partial_L(\hat{x}^{(0)})= \hat{\sigma}^{(0)}$. So we can still define a cochain $\hat{z}^{(0)} \in  \bs C^0_{k+1}(X)$ such that
    \begin{align}
        \bs \partial_L (\hat{z}^{(0)}) = \hat{g}^{(0)} + \hat{x}^{(0)}.
    \end{align}

    \paragraph{Step 2: (Modified) Syndrome explanation sequence.} We proceed by tracing the syndrome explanation sequence as before. However, note that $g^{(i)}$ are now constructed by the modified local guess rule: if there is no preimage for a local view, then set it to be zero. The (noisy) higher-level syndromes are defined accordingly $\sigma^{(i+1)}= \bs \Delta^i_{k+i}(g^{(i)})$. In the below diagram, we keep track of the ideal guesses $\hat{g}^{(i)}$ and syndromes $\hat{\sigma}^{(i+1)} = \bs \Delta^i_{k+i}(\hat{g}^{(i)}) $, as well as $\hat{z}^{(i)}$, together with the noisy values. The below diagram includes some notations different from~\Cref{fig:tracing}. In particular, the label `$2+3$' means the the output is equal to the sum of tuple elements 2 and 3, label `$1;1 \mapsto 2;2$' means tuple element 1 gets mapped to 1, 2 to 2, label `$\sim//$' means tuple element 1 gets mapped approximately and elements 2 and 3 gets mapped to element 2.

    \begin{center}
        \begin{tikzcd}[arrows=rightarrow]
        &0&1&2&3&\\
        k-1 & (\sigma^{(0)}, \hat{\sigma}^{(0)}) &&&&\\
        k&     (g^{(0)}, \hat{g}^{(0)} ,\hat{x}^{(0)}) \ar[r,"\bs \Delta^0_{ k}", "1;1 \mapsto 2;2" below] \ar[u,"\bs \partial_L", "\sim //" right] & (\sigma^{(1)}, \hat{\sigma}^{(1)}) &  & &\\
        k+1&     \hat{z}^{(0)} \ar[u,"\bs \partial_L", "2+3" right] \ar[r,"\bs \Delta^0_{ k+1}", "\mapsto 3" below] & (g^{(1)},\hat{g}^{(1)},\hat{x}^{(1)}) \ar[u,"\bs \partial_L", "\sim //" right] \ar[r,"\bs \Delta^1_{ k+1}", "1;1 \mapsto 2;2" below] & (\sigma^{(2)},\hat{\sigma}^{(2)})& &\\
        k+2&     & \hat{z}^{(1)} \ar[rd, dashed, rightarrow, no head] \ar[u,"\bs \partial_L", "2+3" right] \ar[r,"\bs \Delta^1_{ k+2}", "\mapsto 3" below]& (g^{(2)},\hat{g}^{(2)}, \hat{x}^{(2)}) \ar[rd, dashed, rightarrow, no head] \ar[u,"\bs \partial_L", "\sim //" right]&&\\
         & & & ~& \makebox[\widthof{$x^{(3)}$}]{~}&~
        \end{tikzcd}
        \end{center}

        We have $|\hat{h}^{(0)}|=|g^{(0)} -  \hat{g}^{(0)}| \leq 2^t|\hat{m}|$, so
        \begin{align}
            |\sigma^{(1)} - \hat{\sigma}^{(1)}| = |\bs \Delta^0_{k} (g^{(0)} + \hat{g}^{(0)})| \leq  nt |g^{(0)} -  \hat{g}^{(0)}| \leq (nt) 2^t |\hat{m}|.
        \end{align}
        We still have $\bs \partial_L (\hat{\sigma}^{(1)})=0$, while $\bs \partial_L (\sigma^{(1)})$ can be nonzero but $|\bs \partial_L (\sigma^{(1)})| \leq (nt)2^{t}|\hat{m}|$.
        Next, $\hat{g}^{(1)}$ is still defined to be such that $\bs \partial_L(\hat{g}^{(1)}) = \hat{\sigma}^{(1)}$. To construct $g^{(1)}$ from $\sigma^{(1)}$, there can be faces $f \in X(k+1)$ where no preimage exists for $\sigma^{(1)}(f)$, for which we set $g^{(1)}(f)=0$. There are at most $(nt)2^t|\hat{m}|$ such faces, so we have $|g^{(1)} - \hat{g}^{(1)}| \leq (nt)2^t|\hat{m}|$. Next, we can similarly bound $|\sigma^{(2)}- \hat{\sigma}^{(2)}| = |\bs \Delta^1_{k+1} (g^{(1)} + \hat{g}^{(1)})| \leq (nt)^2 2^t|\hat{m}|$ and $|g^{(2)} - \hat{g}^{(2)}|$, and so on. Performing this calculation recursively we obtain
        \begin{align}
            |\hat{\varsigma}^{(i)}| \triangleq |\sigma^{(i)} - \hat{\sigma}^{(i)}| \leq (nt)^i 2^t|\hat{m}|,\\
            |\hat{h}^{(i)}| \triangleq |g^{(i)} - \hat{g}^{(i)}| \leq (nt)^i 2^t|\hat{m}|.
        \end{align}
        We compute the above sequence until obtaining $\sigma^{(t-k+1)}$.

    \paragraph{Step 3: Map to dual complex and apply single-shot Z-decoder.} In this step we map the problem to a Z-decoding problem on noisy syndrome related to $\sigma^{(t-k+1)}= \hat{\sigma}^{(t-k+1)} + \hat{\varsigma}^{(t-k+1)}$.

    As in~\Cref{prop:noiseless-Xdecoder}'s proof, we notice that $\hat{g}^{(t-k)} + \hat{x}^{(t-k)} \in 
    \bs C^{t-k}_t(X, \mcF_t)$ and $\hat{\sigma}^{(t-k+1)} \in \bs C^{t-k+1}_t(X, \mcF_t)$ are associated with a cochain $\tilde{c}^{(t-k)} \in \tilde{C}^{t-k}(X,\tilde{\mcF})$ and its syndrome $\tilde{s}^{(t-k+1)} =\tilde{\delta}^{t-k}(\tilde{c}^{(t-k)})$, respectively, via applying the unencoding maps of the local codes $\{h_1^\perp,\hdots,h_t^\perp\}$ on each face. The difference is that now we only have access to $\sigma^{(t-k+1)}$. So when we try to unencode the local codes on $\sigma^{(t-k+1)}$, we only obtain a noisy version $s^{(t-k+1)} \in \tilde{C}^{t-k+1}(X,\tilde{\mcF})$ of $\tilde{s}^{(t-k+1)}$. Note that there may be faces $f \in X(k+1)$ where $\sigma^{(t-k+1)}(f)$ is not a valid codeword of the local tensor code, for which we simpy set $s^{(t-k+1)}(f)=0$. However, we have the following guarantee
    \begin{align}
        |s^{(t-k+1)} - \tilde{s}^{(t-k+1)}| \leq |\hat{\varsigma}^{(t-k+1)}| \leq (nt)^{t-k+1}2^t |\hat{m}| \leq (2nt)^t |\hat{m}|.
        \label{eq:8011}
    \end{align}
    Note that we still have the bound
    \begin{align}
        |\tilde{c}^{(t-k)}| \leq |\hat{g}^{(t-k)}+ \hat{x}^{(t-k)}| \leq  (t-k+2)(nt)^{t-k}2^t|\hat{e}|\leq (t+1)(2nt)^t|\hat{e}|
        \label{eq:8012}
    \end{align}
    from~\Cref{eq:dual-chain-weight-bound}.
    So from the given promise on the weights of $\hat{e}$ and $\hat{m}$, we are guaranteed that
    \begin{align}
        |s^{(t-k+1)} - \tilde{s}^{(t-k+1)}| \leq D_2 r|X(t-k)|\\
        2^t n|\tilde{c}^{(t-k)}| + 2 |s^{(t-k+1)} - \tilde{s}^{(t-k+1)}| \leq D_3 r |X(t-k)|, 
    \end{align}
    where $D_2= \frac{n}{4} \frac{(8^{-t^2}(\kappa/2)^{t-k+1}-\lambda)^2}{(4n)^t}$, $D_3=\frac{n}{4} \frac{(8^{-t^2} \kappa^{t-k+1} -\lambda)}{(4n)^t}  \cdot  \min  (\frac{(8^{-t^2} \kappa^{t-k+1} -\lambda)^2}{2^{t+2}} ,8^{-t^2}(\frac{\kappa}{2})^{t-k+1}-\lambda  )$ are constants from~\Cref{prop:noisy-Zdecoder}.

    Hence, we can invoke the noisy sequential single-shot Z-syndrome decoder from~\Cref{prop:noisy-Zdecoder}. In particular,~\Cref{alg:Zdecoder} when applied on the Z-syndrome decoding problem with noisy syndrome $s^{(t-k+1)}$, runs in time $O(q^{n^{O(t)}} |X(t-k)|)$ and outputs a cochain $\tilde{c}'^{(t-k)}$ such that
    \begin{align}
        \tilde{c}'^{(t-k)} = \tilde{c}^{(t-k)} + \tilde{\delta}^{(t-k-1)}(\tilde{r}^{(t-k-1)}) + \tilde{w}^{(t-k)},
    \end{align}
    where $\tilde{r}^{(t-k-1)} \in \tilde{C}^{t-k-1}(X,\tilde{\mcF})$ is some unknown cochain and
     \begin{align}
        |\tilde{w}^{(t-k)}| &\leq \frac{2}{n(8^{-t^2}(\kappa/2)^{t-k+1}- \lambda)} |s^{(t-k+1)} - \tilde{s}^{(t-k+1)}| \stackrel{\eqref{eq:8011}}{\leq}  \frac{2(2nt)^t}{n(8^{-t^2}(\kappa/2)^{t-k+1}- \lambda)}|\hat{m}|
     \end{align}
    is the residual error.

    Now, adding the re-encoded version of $\tilde{c}'^{(t-k)}$ to $g^{(t-k)}$ yields
     \begin{align}
        g'^{(t-k)} = \hat{x}^{(t-k)} + \hat{v}^{(t-k)}+ \hat{h}^{(t-k)} + \hat{w}^{(t-k)},
     \end{align}
    where $\hat{v}^{(t-k)}$ is the re-encoded version of $\tilde{\delta}^{(t-k-1)}(\tilde{r}^{(t-k-1)})$ and $\hat{w}^{(t-k)}$ is of $\tilde{w}^{(t-k)}$. And we have the following bound
    \begin{align}
        |\hat{h}^{(t-k)} + \hat{w}^{(t-k)}| \leq (nt)^i 2^t|\hat{m}| + \frac{2(2nt)^t}{n(8^{-t^2}(\kappa/2)^{t-k+1}- \lambda)}|\hat{m}| \leq \frac{4(2nt)^t}{n(8^{-t^2}(\kappa/2)^{t-k+1}- \lambda)}|\hat{m}|.
    \end{align}
    In other words, let $\hat{g}'^{(t-k)}= \hat{x}^{(t-k)} + \hat{v}^{(t-k)}$ (the `ideal' corrected guess) and $\hat{h}'^{(t-k)}= \hat{h}^{(t-k)} + \hat{w}^{(t-k)}$, then we can write
    \begin{align}
        |\hat{h}'^{(t-k)}| = |g'^{(t-k)} - \hat{g}'^{(t-k)}| \leq \frac{4(2nt)^t}{n(8^{-t^2}(\kappa/2)^{t-k+1}- \lambda)}|\hat{m}|.
        \label{eq:8888}
    \end{align}
    \paragraph{Step 4: Obtaining the final guess.} Now we reverse the syndrome sequence as similarly done in step 4 of~\Cref{prop:noiseless-Xdecoder}, starting from $g'^{(t-k)}$. In~\Cref{prop:noiseless-Xdecoder}, we were guaranteed that $\hat{g}'^{(t-k)}$ would give rise to $\hat{g}'^{(0)}$ (added hat symbols) that would provide a global guess $\hat{g} \in C_k(X)$ homologically equivalent to the underlying error $\hat{e}$. Now we instead perform this step starting with $g'^{(t-k)}$ whose deviation from $\hat{g}'^{(t-k)}$ is bounded in~\Cref{eq:8888}.

    We recall that we used the exactness of $\bs \Delta^i_k$ for $1\leq i \leq k$ from~\Cref{claim:exactness-agg} to find $\hat{u}^{(t-k-1)} \in \bs C^{t-k-1}_t(X)$ (added hat symbols) such that $\bs \Delta^{t-k-1}_t (\hat{u}^{(t-k-1)}) = \hat{g}'^{(t-k)}$ (see~\Cref{claim:syndrome-consistent}'s proof). Then we set $\hat{g}'^{(t-k-1)} = \hat{g}^{(t-k-1)} + \bs \partial_L(\hat{u}^{(t-k-1)})$. This was repeated to obtain $\hat{g}'^{(t-k-2)},\hdots,\hat{g}'^{(0)}$ such that $\bs \Delta^{i}_{k+i}(\hat{g}'^{(i)}) = 0$. In the current proof, we are no longer guaranteed that $\bs \Delta^{t-k}_{t}(g'^{(t-k)}) = 0$. By~\Cref{eq:8888}, $\bs \Delta^{t-k}_{t}(g'^{(t-k)})$ is mostly zero. Since $\bs \Delta^{t-k}_{t}$ only acts locally, i.e., it decomposes into local maps $\bs \Delta_{t} = \bigoplus_{f \in X(t)} \Delta_f $ each of which is exact (recall~\Cref{claim:exactness-agg}'s proof), we can define $u^{(t-k-1)}$ locally within each downward link of a face $f \in X(t)$. If the restriction $g'^{(t-k)}(X_{\leq f}(t-k))[f]$ is nonzero, we set the preimage to be zero. There are at most $(2n)^t|g'^{(t-k)} - \hat{g}'^{(t-k)}|$ \footnote{Here $(2n)^t$ is a crude upperbound of the number of superfaces of any face.} faces $f\in X(t)$ for whose downward link this could be the case, hence we have
    \begin{align}
        |u^{(t-k-1)} - \hat{u}^{(t-k-1)} | \leq 2^t (2n)^t|g'^{(t-k)} - \hat{g}'^{(t-k)}|.
    \end{align}
    With $g'^{(t-k-1)}= g'^{(t-k)}+ \bs \partial_L (u^{(t-k-1)})$,  it follows that
    \begin{align}
        |g'^{(t-k-1)} - \hat{g}'^{(t-k-1)}| \leq (1+(4n)^t)|g'^{(t-k)} - \hat{g}'^{(t-k)}|.
    \end{align}
    Repeating this calculation, we obtain
    \begin{align}
        |g'^{(0)} - \hat{g}'^{(0)}| &\leq (1+(4n)^t+ +\hdots + ((4n)^t)^{t-k})|g'^{(t-k)} - \hat{g}'^{(t-k)}|\\
        & \leq n^{O(t^2)}|g'^{(t-k)} - \hat{g}'^{(t-k)}|\\
        & \stackrel{\eqref{eq:8888}}{\leq } \frac{(nt)^{O(t^2)}}{8^{-t^2}(\kappa/2)^{t-k+1}- \lambda} |\hat{m}|.
    \end{align}
    Finally we obtain a global guess $g \in C_k(X,\mcF)$ from $g'^{(0)}$. Since $g'^{(0)}$ is not necessarily consistent, we can't find such $g$ consistently everywhere. For faces $f \in X(k)$ with inconsistent local views in $g^{(0)}$, we simply set $g(f)=0$. Since $\hat{g}'^{(0)}$ is consistent, the above weight bound guanrantees there are at most $(2n)^t|g'^{(0)} - \hat{g}'^{(0)}|$ $k$-faces for which $g$ differs from the ideal $\hat{g}$. Therefore, up to stabilizers, $g$ differs from the underlying error $\hat{e}$ by at most $\frac{(nt)^{O(t^2)}}{8^{-t^2}(\kappa/2)^{t-k+1}- \lambda} |\hat{m}|$.

    \paragraph{Runtime} We only made some local modifications in the steps of the modified algorithm, so the runtime is the same as in~\Cref{prop:noiseless-Xdecoder}.
\end{proof}

\subsubsection{Parallel decoder}
At this point it is straightforward to obtain the X-syndrome parallel decoder because we have seen that all the steps in~\Cref{prop:noiseless-Xdecoder}'s proof except for the call to Z-syndrome decoder can be done in $O(q^{n^{O(t)}})$ parallel time. So all we need to do now is to call the parallel Z decoder instead. A similar argument to~\Cref{prop:noisy-Xdecoder} also shows single-shotness. Hence, we will be brief.

\begin{prop}[Parallel X-decoder with noiseless syndrome]\label{prop:noiseless-parallel-X}
    Consider the same setting in~\Cref{prop:noiseless-Xdecoder}. Let $\frac{2}{\kappa} \lambda^{1/t} 8^{t} < \varepsilon < 4^{-(t+2)}$ and $\tau_0 = \log (\frac{2^{t+2}}{8^{-t^2}\kappa^{t-k+1}-\lambda})/\log (\frac{1}{15/16 + 4^t \varepsilon})$. 
    Let $\hat{e} \in C_{k}(X)$ be the chain (over $\mathbb{F}_q$) corresponding to a Pauli Z error such that $|\hat{e}|_R \leq (nt)^{-O(t^2)} (8^{-t^2}\kappa^{t-k+1}-\lambda)r |X(t-k)| \cdot \min (\frac{(8^{-t^2}\kappa^{t-k+1} - \lambda)^2}{2^{t+3}}, (\frac{\varepsilon\kappa}{2})^{t-k+1}8^{-t^2}-\lambda )$. 
        Given the syndrome $\sigma = \partial_k (\hat{e}) \in C_{k-1}(X)$ for the error $\hat{e}$, then~\Cref{alg:Xdecoder} in $\tau$-round parallel mode runs in time $O(\tau q^{n^{O(t)}})$ and achieves the following guarantees
        \begin{itemize}
            \item If $\tau=O(\tau_0 (t^2 \log(nt) + s))$, then it outputs a correction cochain $g$ such that $|g + \hat{e}|_R \leq \frac{1}{2^s}|\hat{e}|_R$.
            \item If $\tau=O(\tau_0 \log |X(t-k)|)$, then it outputs a correction cochain $g$ homologically equivalent to $\hat{e}$.
        \end{itemize}
\end{prop}
\begin{proof} We follow~\Cref{prop:noiseless-Xdecoder}'s proof, with the difference that we will perform all steps and not end early (as was allowed in its steps 1 and 2 --~\Cref{claim:local-guess-work} and~\Cref{claim:syndrome-consistent}).
     In~\Cref{prop:noiseless-Xdecoder}'s proof, the main step (step 3) is to call the Z-syndrome decoder (level $t-k$) on to correct a cochain $\tilde{c}^{(t-k)}$ on the dual chain complex. In~\Cref{eq:dual-chain-weight-bound}, the block-weight of this cochain was bounded by
    $|\tilde{c}^{(t-k)}| \leq (t-k+2)(nt)^{t-k} 2^t |\hat{e}|_R$. Hence, the weight promise on $|\hat{e}|_R$ allows us to call the parallel Z-decoder from~\Cref{alg:Z-parallel} and use the performance guarantees from~\Cref{prop:noiseless-parallel-Z}.

    We first consider the case of logarithmic many rounds. And let $\tau_0 = \log (\frac{2^{t+2}}{8^{-t^2}\kappa^{t-k+1}-\lambda})/\log (\frac{1}{15/16 + 4^t \varepsilon})$ be the constant from~\Cref{alg:Z-parallel} (note that we are decoding level $t-k$ of the dual complex). It is guaranteed that after $O(\tau_0 \log |X(t-k)|)$
    parallel decoding rounds,~\Cref{alg:Z-parallel} outputs a cochain homologically equivalent to $\tilde{c}^{(t-k)}$. The rest of the argument in~\Cref{prop:noiseless-Xdecoder} then follows.

    Next we consider the case of constant number of parallel decoding rounds. In this case there will be residual error in decoding $\tilde{c}^{(t-k)}$ and hence we need to use an argument similar to steps 3 and 4 of~\Cref{prop:noisy-Xdecoder}'s proof. After $\tau$ decoding rounds (which we assume to be an integer multiple of $\tau_0$),~\Cref{alg:Z-parallel} outputs a cochain $\tilde{c}'^{(t-k)}$ such that $|\tilde{c}'^{(t-k)} -\tilde{c}^{(t-k)}|_R \leq (\frac{1}{2})^{\tau/\tau_0} |\tilde{c}^{(t-k)}|_R$. Then we follow the same argument in step 4 of~\Cref{prop:noisy-Xdecoder}'s proof starting from~\Cref{eq:8888}. We need to replace RHS of~\Cref{eq:8888} by $(\frac{1}{2})^{\tau/\tau_0} |\tilde{c}^{(t-k)}|_R$, and this leads to a final correction (for the X-decoding problem) $g$ such that $|g+ \hat{e}|_R \leq n^{O(t^2)}(\frac{1}{2})^{\tau/\tau_0} |\tilde{c}^{(t-k)}|_R \leq (nt)^{O(t^2)}(\frac{1}{2})^{\tau/\tau_0} |\hat{e}|_R$. Hence taking $\tau = O(\tau_0 (t^2 \log(nt)+s))$ guarantees $|g+ \hat{e}|_R \leq \frac{1}{2^s}|\hat{e}|_R$. 
\end{proof}

\begin{prop}[Parallel X-decoder with noisy syndrome]
\label{prop:noisy-parallel-X}
        Consider the same setting in~\Cref{prop:noiseless-Xdecoder} and the X-syndrome decoder algorithm (\Cref{alg:Xdecoder}) with the modification described in~\Cref{prop:noisy-Xdecoder}. Let $\hat{e} \in C_{k}(X)$ be the chain (over $\mathbb{F}_q$) corresponding to a Pauli Z error. Let $\frac{2}{\kappa} \lambda^{1/t} 8^{t} < \varepsilon < 4^{-(t+2)}$ and $\tau_0 = \log (\frac{2^{t+2}}{8^{-t^2}\kappa^{t-k+1}-\lambda})/\log (\frac{1}{15/16 + 4^t \varepsilon})$ and $E_0 = (nt)^{-O(t^2)} (8^{-t^2}\kappa^{t-k+1}-\lambda)r |X(t-k)| \cdot \min (\frac{(8^{-t^2}\kappa^{t-k+1} - \lambda)^2}{2^{t+3}}, (\frac{\varepsilon\kappa}{2})^{t-k+1}8^{-t^2}-\lambda )$. Let $\hat{\sigma} = \partial_k (\hat{e}) \in C_{k-1}(X)$ be the ideal syndrome. And let $\sigma = \hat{\sigma} + \hat{m}$ be the observed syndrome, where the measurement error $\hat{m} \in C_{k-1}(X)$ satisfies $2^t(t+1) n |\hat{e}|_R  + 2^{O(\tau_0 q^{n^{O(t)}} )}|\hat{m}| < E_0$. Then~\Cref{alg:Xdecoder} (modified version) in $\tau$-round parallel mode runs in time $O(\tau q^{n^{O(t)}})$ and achieves the following guarantees
        \begin{itemize}
            \item If $\tau=O(\tau_0 (t^2 \log(nt) + s))$, then it outputs a correction cochain $g$ such that $|g + \hat{e}|_R \leq \frac{1}{2^s}|\hat{e}|_R + 2^{O(\tau_0 q^{n^{O(t)}} )}|\hat{m}|$.
            \item If $\tau=O(\tau_0 \log |X(t-k)|)$, then it outputs a correction cochain $g$ such that $|g+\hat{e}|_R \leq 2^{O(\tau_0 q^{n^{O(t)}} )} |\hat{m}|$.
        \end{itemize}
\end{prop}

\begin{proof}
We follow~\Cref{prop:noisy-Xdecoder}'s proof. In step 3, we need to perform Z-syndrome decoding on the dual chain complex to find $\tilde{c}^{(t-k)}$ under noisy syndrome, under the promise that the syndrome noise has weight at most $\leq (2nt)^t|\hat{m}|$ (\Cref{eq:8011}) and $|\tilde{c}^{(t-k)}| \leq (t+1)(2nt)^t |\hat{e}|_R$ (\Cref{eq:8012}). Hence the promise in the proposition allows us to invoke the performance guarantees from~\Cref{prop:noisy-parallel-Z} for the $\tau$-round parallel Z-decoding (level $t-k$)~\Cref{alg:Z-parallel} under noisy syndrome.

First consider the logarithmic many rounds case. We are guaranteed that after $O(\tau_0|X(t-k)|)$ parallel decoding rounds,~\Cref{alg:Z-parallel} outputs a cochain $\tilde{c}'^{(t-k)}$ such that $|\tilde{c}'^{(t-k)}-\tilde{c}^{(t-k)}|_R \leq 2^{O(\tau_0 q^{n^{O(t)}} )}|\hat{m}|$. Then we follow the same argument in step 4 of~\Cref{prop:noisy-Xdecoder}'s proof starting from~\Cref{eq:8888}. We replace RHS of~\Cref{eq:8888} by $2^{O(\tau_0 q^{n^{O(t)}} )}|\hat{m}|$, and this leads to a final correction (for the X-decoding problem) $g$ such that $|g+ \hat{e}|_R \leq 2^{O(\tau_0 q^{n^{O(t)}} )}|\hat{m}|$.

For the constant many rounds case. We are guaranteed that after $\tau$ (which we assume is a multiple of $\tau_0$) decoding rounds~\Cref{alg:Z-parallel} we obtain 
a cochain $\tilde{c}'^{(t-k)}$ such that $|\tilde{c}'^{(t-k)}-\tilde{c}^{(t-k)}|_R \leq (\frac{1}{2})^{\tau/\tau_0} |\tilde{c}^{(t-k)}|_R +  2^{O(\tau_0 q^{n^{O(t)}} )}|\hat{m}|$. Then we follow the same argument in step 4 of~\Cref{prop:noisy-Xdecoder}'s proof starting from~\Cref{eq:8888} and find that~\Cref{alg:Xdecoder} outputs a cochain $g$ such that $|g - \hat{e}|_R \leq (nt)^{O(t^2)}(\frac{1}{2})^{\tau/\tau_0} |\hat{e}|_R  +  2^{O(\tau_0 q^{n^{O(t)}} )}|\hat{m}| $. Hence taking $\tau = O(\tau_0 (t^2 \log(nt)+s))$ guarantees $|g+ \hat{e}|_R \leq \frac{1}{2^s}|\hat{e}|_R + 2^{O(\tau_0 q^{n^{O(t)}} )}|\hat{m}|$. 
\end{proof}

\printbibliography
\end{document}